\def\isarxiv{1} %

\ifdefined\isarxiv
\documentclass[11pt]{article}

\usepackage[numbers]{natbib}

\else
\documentclass{article}
\usepackage{neurips_2022}
\fi

\usepackage{amsmath}
\usepackage{amsthm}
\usepackage{amssymb}
\usepackage{algorithm}
\usepackage{algpseudocode}
\usepackage{subfig}
\usepackage{graphicx}
\usepackage{grffile}
\usepackage{wrapfig,epsfig}
\usepackage{url}
\usepackage{xcolor}
\usepackage{epstopdf}
\usepackage{tabularx}
\usepackage{enumitem}

\usepackage{bbm}
\usepackage{dsfont}

\allowdisplaybreaks

\ifdefined\isarxiv

\let\C\relax
\usepackage{tikz}
\usepackage{hyperref}  %
\hypersetup{colorlinks=true,citecolor=blue,linkcolor=blue} %
\usetikzlibrary{shapes.geometric, arrows}
\usepackage[margin=1in]{geometry}

\else

\usepackage{microtype}
\usepackage{hyperref}
\definecolor{mydarkblue}{rgb}{0,0.08,0.45}
\hypersetup{colorlinks=true, citecolor=mydarkblue,linkcolor=mydarkblue}

\fi
\graphicspath{{./figs/}}

\newtheorem{theorem}{Theorem}[section]
\newtheorem{lemma}[theorem]{Lemma}
\newtheorem{definition}[theorem]{Definition}

\newtheorem{corollary}[theorem]{Corollary}

\newtheorem{fact}[theorem]{Fact}

\newtheorem{claim}[theorem]{Claim}

\newcommand{\wh}{\widehat}
\newcommand{\wt}{\widetilde}
\newcommand{\ov}{\overline}
\newcommand{\eps}{\varepsilon}
\renewcommand{\epsilon}{\varepsilon}
\renewcommand{\phi}{\varphi}

\newcommand{\rect}{\mathrm{rect}}
\newcommand{\N}{\mathcal{N}}
\newcommand{\R}{\mathbb{R}}

\newcommand{\Z}{\mathbb{Z}}
\newcommand{\C}{\mathbb{C}}
\renewcommand{\L}{\mathsf{left}}
\newcommand{\len}{\mathsf{len}}
\newcommand{\num}{\mathsf{num}}

\renewcommand{\i}{\mathbf{i}}

\renewcommand{\hat}{\wh}

\renewcommand{\d}{\mathrm{d}}
\newcommand{\poly}{\mathrm{poly}}
\DeclareMathOperator{\sinc}{sinc}

\newcommand{\supp}{\mathrm{supp}}

\DeclareMathOperator*{\E}{{\mathbb{E}}}

\newcommand{\median}{\mathrm{median}}

\makeatletter
\newcommand*{\RN}[1]{\expandafter\@slowromancap\romannumeral #1@}
\makeatother

\definecolor{b2}{RGB}{51,153,255}
\definecolor{mygreen}{RGB}{80,180,0}

\usepackage{lineno}

\begin{document}

\ifdefined\isarxiv

\date{}

\title{Quartic Samples Suffice for Fourier Interpolation}
\author{
Zhao Song\thanks{\texttt{zsong@adobe.com}. Adobe Research. }
\and 
Baocheng Sun\thanks{\texttt{woafrnraetns@gmail.com}. Weizmann Institute of Science.
} 
\and
Omri Weinstein\thanks{\texttt{omri@cs.columbia.edu}. The Hebrew University and Columbia University.}
\and 
Ruizhe Zhang\thanks{\texttt{ruizhe@utexas.edu}. The University of Texas at Austin.}
}

\else

\title{Intern Project} 
\maketitle 
\fi

\ifdefined\isarxiv
\begin{titlepage}
  \maketitle
  \begin{abstract}

We study the  problem of interpolating a noisy Fourier-sparse signal in the time duration $[0, T]$ from  noisy samples in the same range, where  
the ground truth signal can be any $k$-Fourier-sparse signal with band-limit $[-F, F]$. 
Our main result is an efficient Fourier Interpolation algorithm that improves the previous best algorithm by [Chen, Kane, Price, and Song, FOCS 2016] in the following three aspects: 
\begin{itemize}
\item The sample complexity is improved from  $\wt{O}(k^{51})$ to $\wt{O}(k^{4})$.
\item The time complexity is improved from $ \wt{O}(k^{10\omega+40})$ to $\wt{O}(k^{4 \omega})$. 
\item The output sparsity is improved from $\wt{O}(k^{10})$ to  $\wt{O}(k^{4})$. 
\end{itemize}
Here, $\omega$ denotes the exponent of fast matrix multiplication.
The state-of-the-art sample complexity of this problem is $\sim k^4$, but was only known to be achieved by an \emph{exponential-time} algorithm.
Our algorithm uses the same number of samples but has a polynomial runtime, laying the groundwork for an efficient Fourier Interpolation algorithm. 

The centerpiece of our algorithm is a new sufficient condition for the frequency estimation task---a high signal-to-noise (SNR) band condition---which allows for efficient and accurate signal reconstruction. 
Based on this condition together with a new structural decomposition of Fourier signals  (Signal Equivalent Method), 
we design a cheap algorithm to estimate each ``significant'' frequency within a narrow range, which is then combined with a signal estimation algorithm into a new Fourier Interpolation framework to reconstruct the ground-truth signal.

  \end{abstract}
  \thispagestyle{empty}
\end{titlepage}

{\hypersetup{linkcolor=black}
\tableofcontents
}
\newpage
\setcounter{page}{1} 

\else

\begin{abstract}

We study the  problem of interpolating a noisy Fourier-sparse signal in the time duration $[0, T]$ from  noisy samples in the same range, where  
the ground truth signal can be any $k$-Fourier-sparse signal with band-limit $[-F, F]$. 
Our main result is an efficient Fourier Interpolation algorithm that improves the previous best algorithm by [Chen, Kane, Price, and Song, FOCS 2016] in the following three aspects: 
\begin{itemize}
\item The sample complexity is improved from  $\wt{O}(k^{51})$ to $\wt{O}(k^{4})$.
\item The time complexity is improved from $ \wt{O}(k^{10\omega+40})$ to $\wt{O}(k^{4 \omega})$. 
\item The output sparsity is improved from $\wt{O}(k^{10})$ to  $\wt{O}(k^{4})$. 
\end{itemize}
Here, $\omega$ denotes the exponent of fast matrix multiplication.
The state-of-the-art sample complexity of this problem is $\sim k^4$, but was only known to be achieved by an \emph{exponential-time} algorithm.
Our algorithm uses the same number of samples but has a polynomial runtime, laying the groundwork for an efficient Fourier Interpolation algorithm. 

The centerpiece of our algorithm is a new sufficient condition for the frequency estimation task---a high signal-to-noise (SNR) band condition---which allows for efficient and accurate signal reconstruction. 
Based on this condition together with a new structural decomposition of Fourier signals  (Signal Equivalent Method), 
we design a cheap algorithm to estimate each ``significant'' frequency within a narrow range, which is then combined with a signal estimation algorithm into a new Fourier Interpolation framework to reconstruct the ground-truth signal.

\end{abstract}

\fi

\section{Introduction}

Fourier transforms are the backbone of signal processing and engineering,  with profound implications to nearly every field of scientific computing and technology. This is primarily due to the discovery of the well-known Fast Fourier Transform (FFT) algorithm \cite{ct65}, which is ubiquitous in engineering applications, from image and audio processing to fast integer multiplication and optimization. The classic FFT algorithm of \cite{ct65} computes the \emph{Discrete Fourier Transform} (DFT) of a length-$n$ vector $x$, where both the time and frequency domains are assumed to be discrete. This algorithm takes $O(n)$ samples in the time domain, and constructs $\hat{x} = \mathrm{DFT}(x)$ in $ O(n \log(n))$ time. The discrete setting of DFT limits its applicability in two main aspects: The first one is that many real-world signals are continuous (analog) by nature; Secondly, many real-world applications (such as image processing) involve signals which are \emph{sparse} in the frequency domain (i.e., $\|\hat{x}\|_0= k \ll n$)  \cite{itu92, wat94, r02}. This feature underlies the  \emph{compressed sensing} paradigm \cite{crt06}, which leverages sparsity to obtain \emph{sublinear} algorithms for signal reconstruction, with time and sample complexity depending only on the sparsity $k$. 
Unfortunately, the continuous case cannot  simply be reduced to the discrete case via  standard discretization (i.e., using a sliding-window function), as it ``smears out" the frequencies and blows up the sparsity, which motivates a more direct approach for the continuous problem  \cite{ps15}.

The study of Fourier-sparse signals dates back to the work of Prony in 1795 \cite{pro95}, who studied the problem of exact recovery of the ``ground-truth" signal $x$ in the vanilla \emph{noiseless} setting. By contrast, the realistic setting of reconstruction from  \emph{noisy-samples} \cite{ps15} is a different ballgame, and   
exact recovery is generally impossible \cite{moi15}.
In the \emph{Fourier Interpolation} problem, the  ground-truth signal \begin{align*}
    x^*(t)=\sum_{j=1}^k v_j e^{2\pi\i f_j t},~~  v_j \in \C, f_j \in [-F, F]~ \forall j\in [k],
\end{align*}
is a $k$-Fourier-sparse signal with bandlimit $F$. Given noisy access to the ground truth $x(t)= x^*(t)+g(t)$ in limited time duration $t\in [0, T]$ (which means that we need to recover $x^*(t)$ by taking samples from $x(t)$), the goal is to reconstruct a $\wt{k}$-Fourier-sparse signal $y(t)$ (i.e., $y(t)=\sum_{j=1}^{\wt{k}} \wt{v}_j e^{2\pi\i \wt{f}_j t}$ for some $\wt{v}_j \in \C, \wt{f}_j \in [-F, F]$ for all $ j\in [\wt{k}]$) such that 
\begin{align*}
    \|y(t)-x^*(t)\|_T^2 \leq c ( \|g\|_T^2 + \delta \|x^*(t)\|_T^2)
\end{align*}
holds for some $c= O(1)$, where the $T$-norm of any function $f:\R\rightarrow \C$ is defined as 
\begin{align*} 
\|f(t)\|_T^2 := \frac{1}{T} \int_0^T |f(t)|^2\d t.
\end{align*}
We note that it is not necessary for $y(t)$'s frequencies and magnitudes $(\wt{f}_j, \wt{v}_j)$ being close to the ground-truth signal $x^*(t)$'s frequencies and magnitudes $({f}_{j'}, {v}_{j'})$. %

Prior to this work, the state-of-the-art algorithm for the Fourier interpolation problem was given by \cite{ckps16}, which achieves 
 $\wt{O}(k^{51}) $ sample complexity, 
 $\wt{O}(k^{10\omega+40})$ running time, $ \wt{O}(k^{10})$ output sparsity, and $c\geq 2000$ approximation ratio. 
In \cite{sswz22}, the approximation ratio was improved to $\approx 1+\sqrt{2}$, but the sample complexity remained large, and runtime remained slow.
For calibration, we note that $o(k^4)$ sample  complexity for 
Fourier interpolation
is not known to be achievable even with 
\emph{exponential} decoding time.  
In this work, we focus on improving  the \emph{efficiency} of \cite{ckps16}'s algorithm  across all aspects: (i) runtime, (ii) sample complexity, and (iii) output-sparsity. 
Our main result is: 

\begin{table}[!ht]
    \centering
    \renewcommand{\arraystretch}{1.25}
    \begin{tabular}{|l|c|c|c|} \hline 
       {\bf References} & {\bf Samples} & {\bf Time} & {\bf Output Sparsity}  \\ \hline 
        \cite{ckps16}  & $\wt{O}(k^{51}) $ & $ \wt{O}(k^{10\omega+40})$ & $ \wt{O}(k^{10})$ \\ \hline
        \cite{cp19_colt, sswz22} & $\wt{O}(k^4)$ & $\exp(k^3)$ & $k$ \\ \hline
        Ours (Theorem~\ref{thm:intro_main}) & $\wt{O}(k^{4}) $ & $\wt{O}(k^{4\omega}) $ & $\wt{O}(k^{4}) $ \\ \hline
    \end{tabular}
    \caption{Summary of the results. %
    All the algorithms obtain $O(1)$ approximation ratio. We use $\omega$ to denote the exponent of matrix multiplication, currently $\omega \approx 2.373$ \cite{w12,aw21}.}
    \label{tab:my_label}
\end{table}

\begin{theorem}[Main Theorem]\label{thm:intro_main}
Let $x(t) = x^*(t) + g(t)$, where $x^*(t)$ is $k$-Fourier-sparse signal with frequencies in $[-F, F]$.
Given samples of $x(t)$ over $[0, T]$, there is an algorithm that uses 
\begin{align*}
 k^4 \log(FT) \cdot \poly\log(k,1/\delta , 1/\rho)
\end{align*}
samples, runs in 
\begin{align*}
k^{4\omega} \log(FT)   \cdot \poly\log(k,1/\delta , 1/\rho)
\end{align*}
time, and
outputs a $k^4\cdot \poly\log(k/\delta) $-Fourier-sparse
signal $y(t)$ s.t with probability at least $1-\rho$,
\begin{align*}
\|{y(t) - x^*(t)}\|_T \lesssim \|{g}(t)\|_T + \delta\|{x^*}(t)\|_T.
\end{align*}

\end{theorem}

\subsection{Related works}
\paragraph{Sparse Fourier transform in the discrete setting}

The Fourier transform $\hat{x} \in \C^{N}$ is a vector of length $N$. The goal of a sparse DFT algorithm is, given a bunch of samples $x_{i}$ in the time domain and the sparsity parameter $k$, to output a $k$-Fourier-sparse signal $x'$ with the $\ell_{2}/\ell_{2}$-guarantee
\begin{align*}
     \| \hat{x}' - \hat{x} \|_2 ~ \lesssim ~ \min_{k\text{-sparse}~z} \| z - \hat{x} \|_2.
\end{align*}
There are two different lines of work solving the above problem. One line \cite{gms05,hikp12a,hikp12,ikp14,ik14,k16,k17} is carefully choosing samples (via hash function) and obtaining sublinear sample complexity and running time. The other line \cite{ct06,rv08,bou14,hr16,nsw19} is taking \emph{random} samples (via RIP property \cite{ct06} or others) and paying sublinear sample complexity but nearly linear running time.

\paragraph{Sparse Fourier transform in the continuous setting}

\cite{ps15} defined the sparse Fourier transform in the continuous setting. It shows that as long as the sample duration $T$ is large enough compared to the frequency gap $\eta$, then there is a sublinear time algorithm that recovers all the frequencies up to certain precision and further reconstructs the signal. \cite{jls23} improves and generalize several results in \cite{ps15}. In particular, \cite{ps15} only works for one-dimensional continuous Fourier transform, and \cite{jls23} generalizes it to $d$-dimensional Fourier transform. In order to convert the tone estimation guarantee to signal estimation guarantees, \cite{ps15} provides a positive result which shows $T=O(\log^2(k) / \eta)$ is sufficient, and \cite{moi15} shows a lower bound result where $T = \Omega(1/\eta)$. \cite{s19} asked an open question about whether this gap can be closed. \cite{jls23} made positive progress on that problem by providing a new upper bound which is $T =O(\log(k)/\eta)$. %

From the negative side, \cite{moi15} shows that in order to show tone estimation\footnote{Tone refers to a (frequency, coefficient) pair in \cite{ps15}. E.g., $(f_i,v_i)$ is a tone of the signal $x(t)=\sum_{i=1}^k v_i e^{2\pi\i f_it}$. And tone estimation means estimating each $(f_i, v_i)$ precisely.}, we have to pay a lower bound in sample duration $T$. In \cite{ps15}, it shows that once we have tone estimation, we can obtain a signal estimation guarantee.
Since \cite{ps15} and \cite{moi15}, there is an interesting question about whether we can reconstruct the signal without having a tone estimation guarantee, which is defined as the Fourier interpolation problem.  \cite{ckps16} shows a positive answer to this problem. They provide a polynomial time algorithm to solve this problem. However, both sample complexity and running time in \cite{ckps16} have a huge polynomial factor in $k$.  The major goal of our work is to significantly improve those polynomial factors.

\section{Technical Overview}

\subsection{High-level approach}
The high-level approach of Fourier Interpolation (also Fourier Signal reconstruction) has two steps: frequency estimation  and signal estimation (also called signal recovery or Fourier set query). This work mainly contributes to the first frequency estimation step. %

\paragraph{Filters and \textsc{HashToBins}}
The core technique in Fourier sparse recovery and interpolation algorithms is filtering.  %
There are two kinds of filters we are using. The first filter function applied to the signal is $H(t)$ (Figure~\ref{fig:H_filter_intro}), which is the bounded band limit approximation of the rectangular window function $\rect_T(t)$. Intuitively, since the time duration is restricted to $[0, T]$, we should view the ground truth signal as $ x^*(t) \cdot \rect_T(t)$. However, handling $\wh{\rect}_T(f)$ is not easy due to its unbounded support in the frequency domain. Therefore, we use $H(t)$ instead, which truncates the frequency domain of $\rect_T(t)$ and makes the analysis much easier. 

Another kind of filters we use is $G^{(j)}_{\sigma, b}(t)$ (Figure~\ref{fig:G_filter_intro}), which ``isolates'' the signal through the procedure \textsc{HashToBins} and extracts the one-cluster signal in the $j$-th bin. More specifically, \textsc{HashToBins} divides the frequency domain into $B=O(k)$ bins. %
We can show that with high probability over the randomized hashing function, each bin contains a single cluster of frequencies. Hence, in the following frequency estimation step, we can just focus on recovering the frequency of a one-cluster \emph{filtered signal} in each bin $j\in [B]$:
\begin{align*}
    z_j(t)=(x\cdot H)(t) *G^{(j)}_{\sigma, b}(t).
\end{align*}

\begin{figure}[!ht]
    \centering
    \subfloat[Time domain filter $H(t)$.\label{fig:H_filter_intro}]{
        \includegraphics[width=\textwidth]{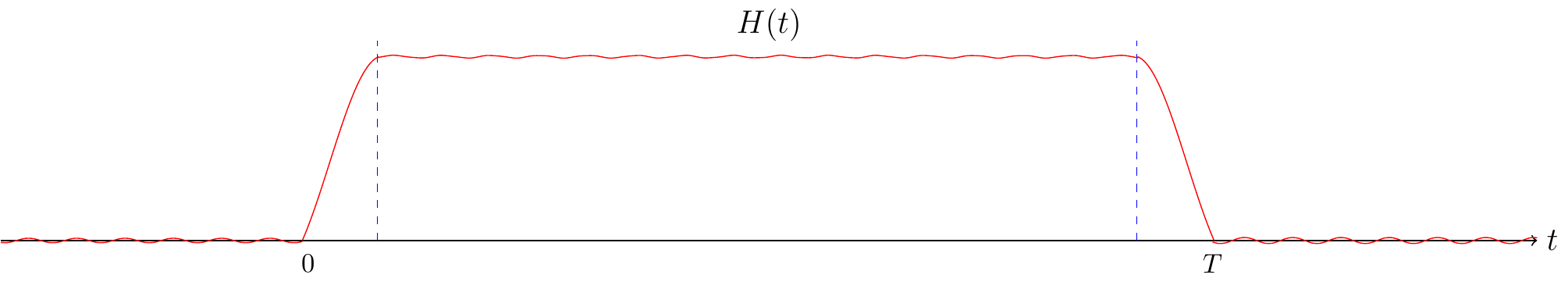}}
    \hfill  
    \subfloat[Frequency domain filter $\hat{G}_{\sigma,b}^{(j)}(f)$ for the $j$-th bin.\label{fig:G_filter_intro}]{
        \includegraphics[width=\textwidth]{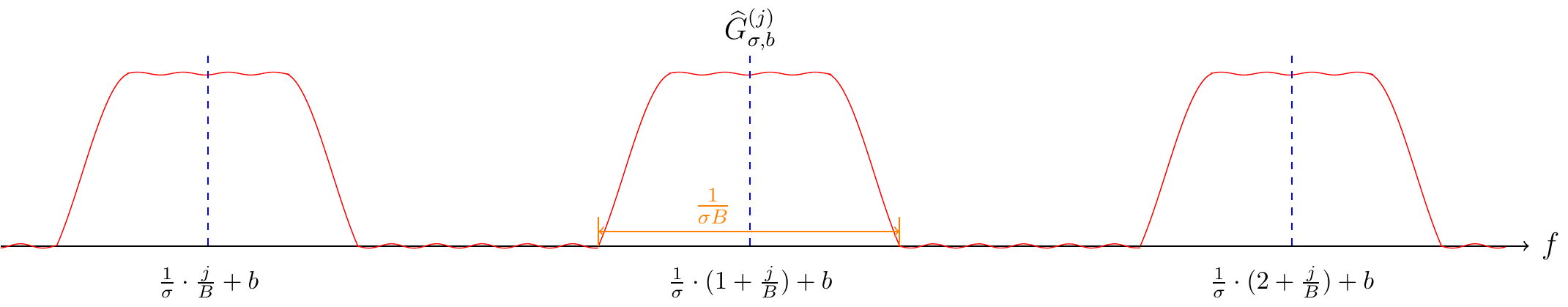}}
    \caption{Time and frequency domain filters.}
\end{figure}

\paragraph{Frequency Estimation} This step is the main focus on this work.
To estimate the frequencies, our algorithm has two levels. The first level generates significant samples of the \emph{local-test signal}:
\begin{align*}
    d_z(t)=z(t)e^{2\pi\i f^* \beta}-z(t+\beta),
\end{align*}
where $ z(t)=z_j(t)$ is the filtered signal in the $j$-th bin and $\beta$ is a perturbation parameter.
A time point $\alpha\in [0,T]$ is defined to be significant with respect to the target frequency $f^*$ if $|d_z(\alpha)| $ is small.  In this case, $z(\alpha+\beta)/z(\alpha)$ is a good approximation of $e^{2\pi\i f^* \beta}$,%
which further implies the target frequency $f^*$. %
The second level is a searching algorithm that iteratively estimates the target frequency $f^*$. In each iteration, it calls the significant sample generation algorithm and uses the significant sample to narrow the possible range of the target frequency until reaching the desired accuracy. %
Based on the two-level strategy, we design an efficient, high-accuracy frequency estimation algorithm, improving the time complexity, sample complexity, and the estimation error of the frequency estimation algorithms in previous works \cite{ckps16, cp19_icalp}. The theorem is stated as follows.

\begin{theorem}[Frequency estimation, Informal version of Theorem~\ref{thm:frequency_recovery_k_better}]\label{thm:freq_est_intro}
There exists an algorithm takes $O(k^2 \log(1/\delta) \log(FT))$ samples, runs in $O(k^2 \log(1/\delta) \log^2 (FT))$ time, returns a set $L$ of $O(k)$ frequencies such that with probability $1- \rho_0$, for any ``important frequency'' $f$, there exists an $\wt{f} \in L$ satisfying
\begin{equation*}
|f-\widetilde{f} | \lesssim \Delta,
\end{equation*}
where $\Delta=k\cdot |\supp(\wh{H})|$, where $\wh{H}$ is the Fourier transform of $H$.
\end{theorem}

\paragraph{Signal Estimation}
In signal estimation, a set of estimated frequencies of $y(t)$ has been found, and it remains to interpolate the signal under these frequencies. %
This is often done via \emph{set-query} techniques \cite{pri11}. This step is not the focus of this paper, and more discussions can be found in \cite{ckps16, sswz22}.
\footnote{We stress that this paper is self-contained and we provide all the technical details of signal estimation in Section~\ref{sec:sig_recontr}.
}

\subsection{Our techniques for frequency estimation}

In the frequency estimation part, there are two central questions that need to be answered:
\begin{enumerate}
    \item \emph{Which frequencies or hashing bins are worth  recovering?}
    \item \emph{How to recover a key frequency in a bin?}
\end{enumerate}
Our answer to these questions substantially deviates from previous works, as we discuss below.

\paragraph{Answer to the first question:}For the first question, \cite{ckps16}'s answer is the \emph{heavy-cluster condition}, which is defined as follows:
\begin{align}\label{eq:heavy_cluster_intro}
[f^*-\Delta,f^*+\Delta]~~\text{is heavy if }~~\int_{f^*-\Delta}^{f^*+\Delta} |\wh{H \cdot x^*}(f)|^2 \mathrm{d} f \ge T \cdot \N^2/k,
\end{align}
where $\N^2:=\|g\|_T^2+\delta\|x^*\|_T^2$ represents the noisy-level of $x(t)$.
However, only considering the energy of the ground-truth signal is not enough\footnote{For example, consider the ground-truth signal $x^*(t)=ve^{2\pi \i f^*t}+ve^{2\pi \i (f^*+10\Delta)t}$ and the noise $g(t)=-ve^{2\pi\i f^* t}$. Even if $f^*\pm \Delta$ is a heavy cluster, it is impossible to recover $f^*$ from the observation $x(t)=x^*(t)+g(t)$, since $\hat{x}(f)$ is zero around $f^*$.}. 
Indeed, their algorithm only works for  ``recoverable" clusters, which are defined as:
\begin{align*}
[f^*-\Delta,f^*+\Delta]~~\text{is recoverable if }~~\int_{f^*-\Delta}^{f^*+\Delta} |\wh{H \cdot x}(f)|^2 \mathrm{d} f \ge T \cdot \N^2/k.
\end{align*}
The gap between heavy clusters and recoverable clusters is a bottleneck for improving the approximation ratio of the Fourier interpolation algorithms in \cite{ckps16} to an arbitrarily small constant.  This gap also introduces many  other technical  difficulties in designing more efficient frequency estimation algorithms.

\begin{figure}[!ht]
    \centering
    \subfloat[Low-noise band recovery: high-accuracy frequency estimation is needed.\label{fig:high_snr_b}]{\includegraphics[width=\textwidth]{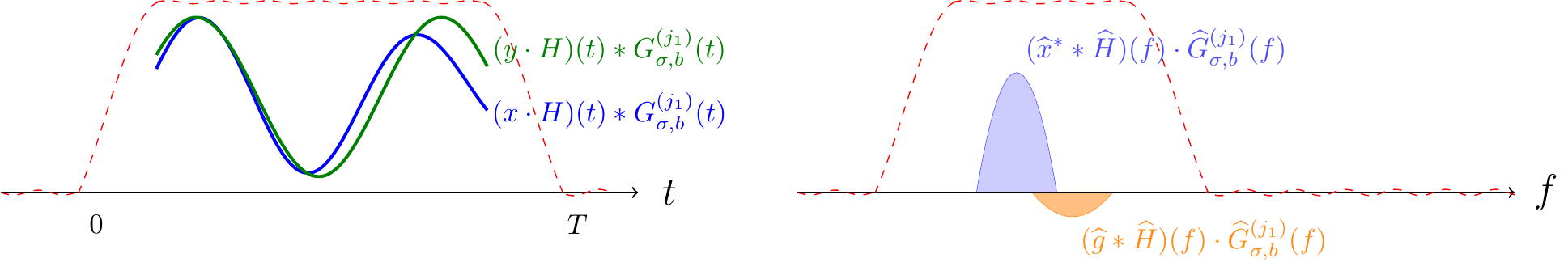}}
    \vspace{1mm}
    \subfloat[High-noise band recovery: any frequency estimation output is acceptable.\label{fig:high_snr_c}]{\includegraphics[width=\textwidth]{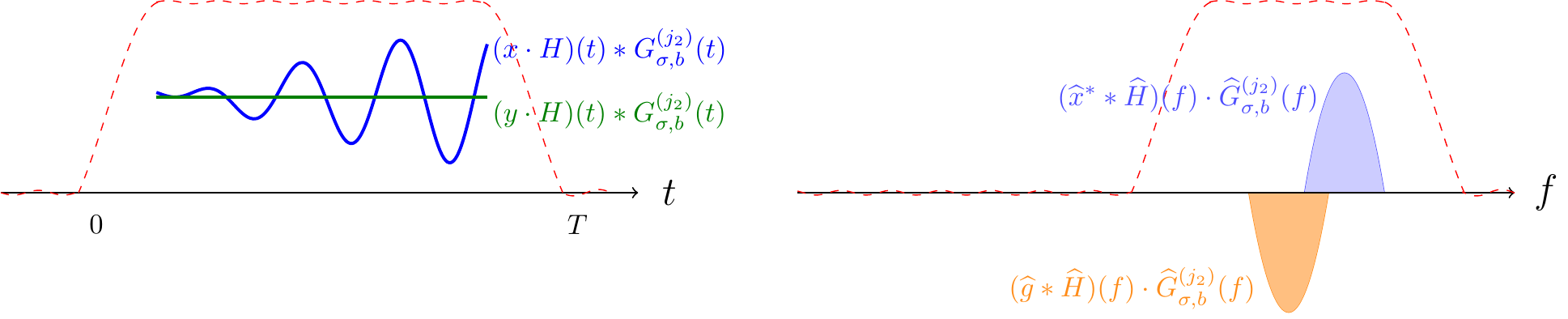}} 
    \caption{The high SNR band condition. The red curves are the filters. On the left, the blue curves are the filtered noisy observation signal in the time domain,  and the green curves are corresponding reconstructed signals. On the right, the light blue regions are the filtered frequencies of the ground-truth signal $x^*$, and the orange regions are the filtered frequencies of the noise $g$.  Figure~\ref{fig:high_snr_b} shows a high-SNR case, where we can recover a good approximation of $x^*$ in this band. Figure~\ref{fig:high_snr_c} shows an extremely low-SNR case, where $g$ has almost the same energy as $x^*$, and a trivial signal ($y(t)=\text{constant}$) suffices for the recovery of this band. %
    }
    \label{fig:high_snr}
\end{figure}

To overcome this  gap, we introduce a new criterion for the frequency bands that need to be non-trivially reconstructed, which we call the \emph{high signal-to-noise ratio (SNR) band condition}. Formally, we say a hashing bin $j\in [B]$ has a high SNR if the filtered signal $z^*_j(t)=(x^*\cdot H)*G_{\sigma,b}^{(j)}$ satisfies:
\begin{align}\label{eq:high_snr_intro}
    \|(g\cdot H)*G_{\sigma,b}^{(j)}(t)\|_T^2\leq c\cdot \|z_j^*(t)\|_T^2,
\end{align}
where $c$ is a universal small constant.  Our frequency estimation algorithm focuses solely on recovering \emph{heavy frequencies in high-SNR bins}. The intuition behind this condition is as follows: if the noise in a band (i.e., $ (g\cdot H)*G^{(j)}_{\sigma, b}(t)$) is too large, then we can simply use an all-zero signal as the reconstruction of the filtered signal. We show this new condition brings many advantages for designing more efficient frequency estimation algorithms. %
In particular, we show that the remaining frequencies in the low-SNR bins are inconsequential for the reconstruction error, and ignoring them in the signal estimation can still achieve the approximation guarantee of Fourier interpolation.\footnote{We remark  our algorithm never  attempts to decide whether a bin satisfies the high-SNR condition or not, but rather assumes all bins are ``good". The low-SNR bins may therefore produce totally wrong frequency estimates. However, for accurate signal estimation, we only need to guarantee that all the good frequencies are reconstructed by the frequency estimation algorithm, so even if the output set contains some wrong frequencies, they can be simply ignored.} %

\begin{figure}[!ht]
    \centering
    {\includegraphics[width=0.6\textwidth]{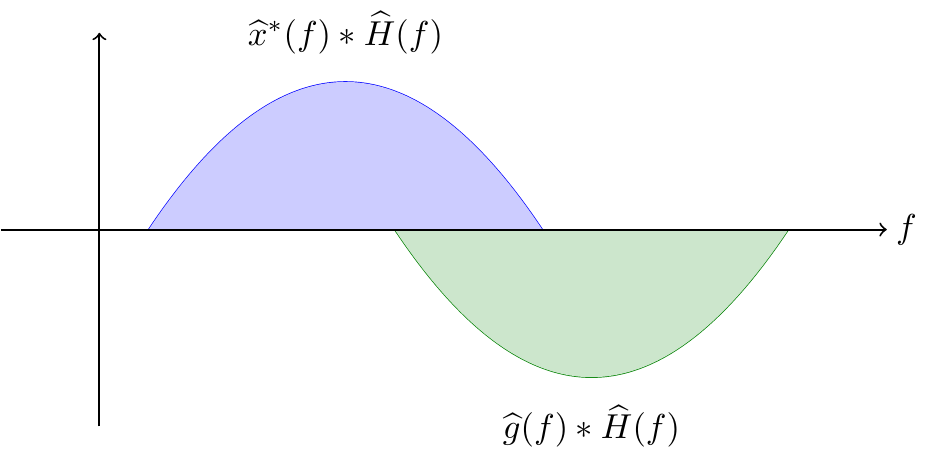}}
\caption{A case that violates our high SNR band assumption but \cite{ckps16} tries to recover. %
$\wh{x}^*(f)*\wh{H}(f)$ (in blue) is the filtered ground-truth signal, %
and $\wh{g}(f)*\wh{H}(f)$ (in green) is the filtered noise. This signal does not satisfy the high SNR band condition since the noise $\wh{g}(f)*\wh{H}(f)$ is too strong. However, the combined signal $(\wh{x}^*(f)+\wh{g}(f))*\wh{H}(f)$ still satisfies the recoverable-cluster condition since it has enough energy in the frequency domain.
}\label{fig:case_ckps_sub_ours}
    
\end{figure}

\paragraph{Answer to the second question:}
As we discussed earlier, the key to answering this question is our novel ``significant-samples"  generation procedure (which produces samples $\alpha$ such that $|z(\alpha)e^{2\pi\i f^* \beta}-z(\alpha+\beta)| $ is small, where $z(t)$ is the filtered signal and $\beta$ is a parameter). This is the content of the following lemma.

\begin{lemma}[Significant Sample Generation, Informal version of Lemma \ref{lem:significant_samples_for_each_bins}]
\label{lem:significant_samples_for_each_bins:main}
There is a Procedure \textsc{GenerateSignificantSamples} in Algorithm \ref{algo:significant_samples_z} such that for $\beta \leq O(1/\Delta)$, it takes $\wt{O}(k^2)$
samples in $x(t)$ and runs in 
$\wt{O}(k^2)$
time. For each frequency $f^*$ with $j := h_{\sigma, b}(f^*)$, if the $j$-th bin has ``high SNR'', and $f^*$ is ``heavy'', then the output $\alpha_j$ satisfies:
\begin{equation*}
|z_j(\alpha_j + \beta) - z_j(\alpha_j)e^{2 \pi \i f^* \beta}|^2 \le 0.01  |z_j(\alpha_j)|^2,
\end{equation*}
with a high constant probability, where $z_{j}(t):=(x\cdot H)*G^{(j)}_{\sigma, b}(t)$.
\end{lemma}

We first sketch the proof of Theorem~\ref{thm:freq_est_intro} using Lemma~\ref{lem:significant_samples_for_each_bins:main}.  Intuitively, if $z(t)$ is exactly one-sparse, i.e., $z(t)=e^{2\pi\i f^* t}$, then we have $z(t)e^{2\pi\i f^* \beta}-z(t+\beta)=0$, and $\frac{z(t+\beta)}{z(t)}$ gives the exact value of $e^{2\pi\i f^*\beta}$. More generally, by the guarantee of the significant sample, that ratio can well-approximate $e^{2 \pi \i f^* \beta}$, which gives a good estimate of ${f^* \beta} \mod{1}$ in a small constant range:
\begin{align*}
    f^*\approx \frac{1}{2\pi \beta}\Big(\arg\Big(\frac{z(\alpha+\beta)}{z(\alpha)}\Big) + 2\pi s\Big)
\end{align*}
for some unknown $s\in \mathbb{Z}$.
To determine $s$, we use a search technique to narrow down the potential range of $f^*$ from $[-F,F]$ to $[f^*-\Delta,f^*+\Delta]$. In each iteration, we divide the region of interest into $\num=O(1)$ regions, %
and repeatedly run the Procedure \textsc{GenerateSignificantSamples} with several different $\beta$ and pick up the heavy-hitter among all possible regions, which can exponentially increase the success probability of finding the correct interval. Now, we consider the costs of this process. The initial frequency range is $[-F, F]$, and  in the last iteration, the frequency range is $[f^*-\Theta(\Delta),f^*+\Theta(\Delta)]$. %
Thus, we can take the number of iterations to be $O(\log(F/\Delta))\leq O(\log(FT))$. In each iteration, we call Procedure \textsc{GenerateSignificantSamples} for $O(\log\log(F/\Delta) )\leq O(\log\log(FT))$ times. Note that each run of Procedure \textsc{GenerateSignificantSamples} can generate significant samples for all $B$ bins. Therefore, by Lemma~\ref{lem:significant_samples_for_each_bins:main}, the total time and sample complexity for frequency estimation is $\wt{O}(k^2)\cdot O(\log(FT)) \cdot O(\log\log(FT))=\wt{O}(k^2)$. %

\begin{algorithm}[!ht]\caption{Frequency Estimation Algorithm, Informal version of Algorithm \ref{alg:pre_compute},  \ref{alg:locateksignal_locatekinner}, and \ref{algo:freqEstX}}
\begin{algorithmic}[1]
\Procedure{FrequencyEstimationX}{$x, (\sigma, b)$} 
\For{$j\leftarrow [B]$}
\State $\wt{f}_j \leftarrow\textsc{FrequencyEstimationZ}(x, H, G^{(j)}_{\sigma, b})$ \Comment{recover the heavy frequency of $z^{(j)}$} 
\State $L\gets L\cup \{\wt{f}_j\}$
\EndFor
\State \Return $L$
\EndProcedure

\Procedure{$\textsc{FrequencyEstimationZ}$}{$x,H,G^{(j)}_{\sigma, b}$}
	\State $\num \leftarrow O(1)$\Comment{$\num$-ary search in each iteration}
	\State $D \leftarrow O(\log (\frac{FT}{ \Delta} ))$\Comment{number of iterations}
	\State $\L_1\leftarrow -F$, $\len_1 \leftarrow 2F$\Comment{initial searching interval $[\L_1, \L_1+\len_1]$}
	\For{$d\in [D]$}
	\State  $\L_{d+1} \leftarrow  \textsc{ArySearch}$($x,H,G^{(j)}_{\sigma, b}, \L_d, \len_d,  \num$)\Comment{new searching interval's left-end}
	\State $\len_{d+1}\leftarrow 5 \frac{\len_{d}}{ \num} $\Comment{new searching interval's length}
	\EndFor
	\State \Return $\L_{D+1}$
\EndProcedure

\Procedure{$\textsc{ArySearch}$}{$x,H,G^{(j)}_{\sigma, b}, \L_i, \len_i,  \num$}
	\State $I_q\leftarrow [\L_d+(q-1)\len_d/\num, \L_d+q\len_d/\num]$ for $q\in [\num]$\Comment{candidate regions}
	\State $v_{q}\leftarrow 0$ for $q\in [\num]$ \Comment{votes counter}
	\State $R \leftarrow O(\log({\log(FT)})) $ 
	\For {$r=1 \to R$}
		\State Sample $\beta \sim \text{Uniform}([\frac{1}{2}\wh{\beta}, \wh{\beta} ])$ for $\wh{\beta}=O(\frac{\num}{\len_d})$\Comment{perturbation}
        \State $z(\alpha+\beta), z(\alpha) \leftarrow \textsc{GenerateSignificantSamples}(x,H,G^{(j)}_{\sigma, b})$ \Comment{significant sample}
	    \State $\wt{S} \leftarrow \frac{1}{2\pi\beta} ( \arg(\frac{z(\alpha+\beta)}{ z(\alpha)}) +2\pi \Z)$\Comment{all possible frequencies}
	    \State $\wt{I} \leftarrow \{q\in [\num]~|~ I_q \cap \wt{S} \neq \emptyset  \}$\Comment{all possible regions}
	    \State $ v_{q}\leftarrow v_q + 1$ for $q\in \wt{I}$\Comment{add votes to these regions}
	\EndFor
\State \Return  $\L_d+(q-1)\len_d/\num$ for any $q$ such that $ v_{q} +  v_{q+1} + v_{q+2} \geq R / 2$
\EndProcedure
\end{algorithmic}
\end{algorithm}

Then, we sketch the proof of Lemma~\ref{lem:significant_samples_for_each_bins:main}, which contains three parts:
\begin{enumerate}[label=\Roman*.]
    \item A two-level sampling procedure (see Section~\ref{sec:intro_gen_sig_samples}).
    \item Energy estimation and Signal Equivalent Method (see Section~\ref{sec:energy_est_intro}). 
    \item Time-domain concentration of filtered signals (see Section~\ref{sec:intro_concentration}).
\end{enumerate}
 
\subsubsection{Two-level sampling for significant samples generation}\label{sec:intro_gen_sig_samples}

We may assume that in frequency domain, the energy of $\hat{z}(f)$ is concentrated around $f^*$: %
\begin{equation*} \int_{f^*-\Delta}^{f^*+\Delta} | \widehat{z}(f) |^2 \mathrm{d} f \geq 0.7 \int_{-\infty}^{+\infty} | \widehat{z}(f) |^2 \mathrm{d} f.
\end{equation*} 
This is a very natural and necessary assumption for the frequency estimation problem. \footnote{For the filtered signals that do not satisfy the frequency domain energy concentration assumption, it basically means that they do not contain enough information to recover $f^*$, and we can just ignore those ``useless'' clusters.}
Then we can show that:
\begin{align}\label{eq:freq_concentration_intro}
\|z(t)e^{2\pi\i f^* \beta}-z(t+\beta)\|_T^2 < \gamma \|z(t)\|_T^2
\end{align}
where $\gamma\in (0,0.001)$ is a small constant.
We show how to find an $\alpha$ such that $|z(\alpha)e^{2\pi\i f^* \beta}-z(\alpha+\beta)|^2 < \gamma |z(\alpha)|^2$. For ease of discussion, we scale the time domain from $[0, T]$ to $[-T, T]$. 

The main idea is to use a two-level sampling procedure, which is motivated by \cite{cp19_icalp}. In the first level, we take a set $S=\{t_1,\dots,t_s\}$ of $O(k\log(k))$ i.i.d. samples from the following distribution: 
\begin{align}
    D_z(t) =
\begin{cases}
{c}\cdot (1-|t/T| )^{-1}T^{-1} & \text{if } |t|\leq  T(1-1/k)\\
c \cdot  k T^{-1} & \text{if } |t|\in [T(1-{1}/k), T] 
\end{cases}~~~\forall t\in U,  \label{eq:tech_ov:D_z}
\end{align}
where $U= \{t_0\in \R~|~ H(t) > 1-\delta_1~ \forall t\in [t_0,t_0+\beta]\}$. Then, we assign weights $w_i:= 1 / (2T |S| D_z(t_i))$ 
for each sample $t_i\in S$.

In the second level of the sampling procedure, we sub-sample a $t_i$ from the set $S$ as the output according to the following distribution:
\begin{align*}
    D_S(t_i) = 
        \frac{ w_i \cdot |z(t_i)|^2}{\sum_{j\in [s]}  w_j\cdot |z(t_j)|^2} ~~~\forall i \in [s].
\end{align*}

Now, we explain why the two-level sampling procedure works. By the energy estimation method discussed in Section~\ref{sec:energy_est_intro}, we know that:
\begin{align*}
    \|z(t)\|_T^2 \approx & \|z(t)\|_{S,w}^2:=\sum_{i=1}^s w_i\cdot |z(t_i) |^2,~~\text{and}\\
    \|z(t)e^{2\pi\i f^* t}-z(t+\beta)\|_T^2 \approx& \|z(t) e^{2 \pi \i f^* \beta} - z(t+\beta)\|_{S,w}^2:=\sum_{i=1}^s w_i\cdot |z(t_i) e^{2 \pi \i f^* \beta} - z(t_i+\beta)|^2.
\end{align*}
The second level of the sampling procedure ensures that
\begin{align*}
\E_{t\sim D_S}\left[ \frac{|z(t) e^{2 \pi \i f^* \beta} - z(t+\beta)|^2}{|z(t)|^2} \right]
=  &~ \frac{\sum_{i=1}^s w_i |z(t_i) e^{2 \pi \i f^* \beta} - z(t_i+\beta)|^2}{\sum_{j=1}^s w_j|z(t_j)|^2}\\  %
=  &~ \frac{\|z(t) e^{2 \pi \i f^* \beta} - z(t+\beta)\|_{S,w}^2}{\|z(t)\|_{S,w}^2}.
\end{align*}
Hence, we get that
\begin{align*}
    \E_{t\sim D_S}\left[ \frac{|z(t) e^{2 \pi \i f^* \beta} - z(t+\beta)|^2}{|z(t)|^2} \right] \approx 
    \frac{ \|z(t) e^{2 \pi i f^* \beta} - z(t+\beta)\|_T^2}{\|z(t)\|_T^2}< \gamma,
\end{align*}
where the last step follows from Eq.~\eqref{eq:freq_concentration_intro}.
Then by Markov's inequality, we get that the sample $\alpha$ generated by the two-level sampling procedure satisfies $|z(\alpha) e^{2 \pi i f^* \beta} - z(\alpha+\beta)|^2\lesssim \gamma |z(\alpha)|^2$ with high probability. 

The costs of this two-level sampling procedure are calculated as follows. 
In the first level, we takes $|S|=\wt{O}(k)$ samples from $z(t)$, where each sample $ z(t_i)=((x\cdot H) *G^{(j)}_{\sigma, b})(t_i)$ can be computed by $|\supp(G^{(j)}_{\sigma, b}(t))| = \wt{O}(k)$ samples from $x(t)$ in $\wt{O}(k)$ time. Thus, the total time and sample complexity for the first level sampling procedure is $ \wt{O}(k)\cdot \wt{O}(k) =\wt{O}(k^2)$. In the second level, we further select one sample from the output of the first level, which can be done in $ \wt{O}(|S|)=\wt{O}(k)$ times and does not need any new sample.  

We further discuss how large $\beta$ we can choose in the sampling procedure since it controls the estimation accuracy of $f^*$.\footnote{By comparing $z(t+\beta)$ and $z(t)$, we get an estimate of $f^*\beta$ within some error $\pm b$, which implies an estimate of $f^*$ within an error $\pm b/\beta$. Hence, larger $\beta$ gives a higher accuracy of the frequency estimation. }
We note that the range of $\beta$ is determined by Eq.~\eqref{eq:freq_concentration_intro}, %
which is an underlying assumption of our sampling procedure. 
To satisfy this inequality, we need to guarantee that $|e^{2\pi\i f^* \beta} -e^{2\pi\i f \beta} |\leq  \gamma$ for any $f\in f^*\pm \Delta$, which implies that $\beta \leq O(\gamma/\Delta)$. For comparison, the upper bound of $\beta$ in \cite{ckps16} is only $O(\gamma/(\Delta\sqrt{\Delta T}))$ due to a stronger accuracy requirement there.\footnote{\cite{ckps16} give an $\ell_1$-norm error guarantee in the frequency domain, i.e., $\int_{f^*-\Delta}^{f^*+\Delta} |e^{2\pi\i f^* \beta}-e^{2\pi\i f \beta}| \d f $ is small. To obtain an $\ell_2$-norm guarantee (like Eq.~\eqref{eq:freq_concentration_intro}), they need to apply Cauchy-Schwarz inequality, which results in an extra $\sqrt{\Delta T}$ factor in their upper bound of $\beta$.} %

\subsubsection{Energy estimation and Signal Equivalent Method}\label{sec:energy_est_intro}
In this section, we show that the sampling and reweighing method we use in the significant sample generation procedure can accurately estimate the energy of $z(t)$ and $z(t)e^{2\pi\i f^* t}-z(t+\beta)$ with a sample complexity almost reaching the information-theoretic limit. 
\begin{lemma}[Informal version of Lemma~\ref{lem:significant_samples_z_below} and Lemma~\ref{lem:significant_samples_z_above}]\label{lem:energy_est_intro}
Suppose $f^*$ is a heavy frequency hashed to the $j$-th bin which satisfies the high SNR condition. Let $z^*(t)=(x^*\cdot H)*G_{\sigma,b}^{(j)}$ and $z(t)=(x\cdot H)*G_{\sigma,b}^{(j)}$. Let $U\subseteq [0,T]$ be an interval. Let $S=\{t_1,\dots,t_s\}$ be a set of $O(k\log(k))$ i.i.d. samples from the distribution $D$ defined by Eq.~\eqref{eq:tech_ov:D_z} with weights $w_i=1/(TsD(t_i))$. Then, with probability at least 0.8,
\begin{align*}
    \|z(t)\|_{S,w}^2\gtrsim \|z^*(t)\|_U^2~~\text{and}~~\|z(t)e^{2\pi\i f^* t}-z(t+\beta)\|_{S,w}^2 \lesssim \|z^*(t)\|_U^2,
\end{align*}
where $\|z(t)\|_U^2=(1/|U|)\cdot \int_U |z(t)|^2{\rm d}t$.
\end{lemma}
To prove Lemma~\ref{lem:energy_est_intro}, we develop a \emph{Signal Equivalent Method}. %
Below, we sketch the proof of the first half of Lemma~\ref{lem:energy_est_intro} on the energy estimation for $z(t)$. The second half follows similar ideas. %

Energy estimation is also used in prior works \cite{ckps16, cp19_colt, cp19_icalp,sswz22}, %
where a key component is the following energy bound for the interested function family ${\cal F}$: 
\begin{align*}
    \sup_{f\in {\mathcal{F}}}\sup_{t\in [0,T]}~\frac{|f(t)|^2}{\|f(t)\|_T^2}.
\end{align*}
However, this approach is unlikely to work directly for our filtered signal $z(t)$ since it depends on the randomized hashing function. And under some hashing parameter $(\sigma, b)$, there always exists some signal $x(t)$ such that $z(t)=(x\cdot H)(t) *G^{(j)}_{\sigma, b}(t)$ is in ill-condition (e.g., the frequencies are not well-isolated, or large offset events happen). As a result, bounding $\frac{|z(t)|^2}{\|z(t)\|_T^2}$ for all $z(t)$ of the form $(x\cdot H)*G_{\sigma,b}^{(j)}(t)$ by a small number is not easy.
We bypass the issue by proving an energy bound only for those $z(t)$ under some well-hashed conditions (e.g. frequency is isolated and do not have a large offset), and showing that such a ``refined energy bound'' is still sufficient to derive the sample complexity of our algorithm.    %

\begin{figure}[t]
    \centering
    \subfloat[Signal with non-ideal filter $\wh{G}^{(j)}_{\sigma, b}(f) $.]{\includegraphics[width=0.45\textwidth]{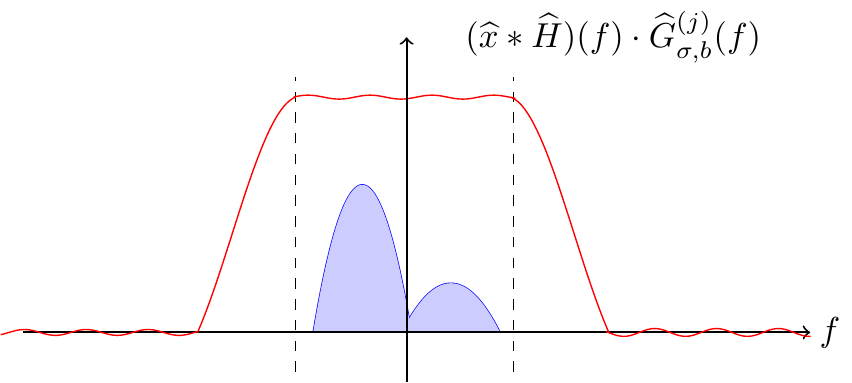}}\label{fig:sig_equ_method:1}
    \hspace{1mm}
    \subfloat[Signal with ideal filter $\wh{I}(f)$.]{\includegraphics[width=0.45\textwidth]{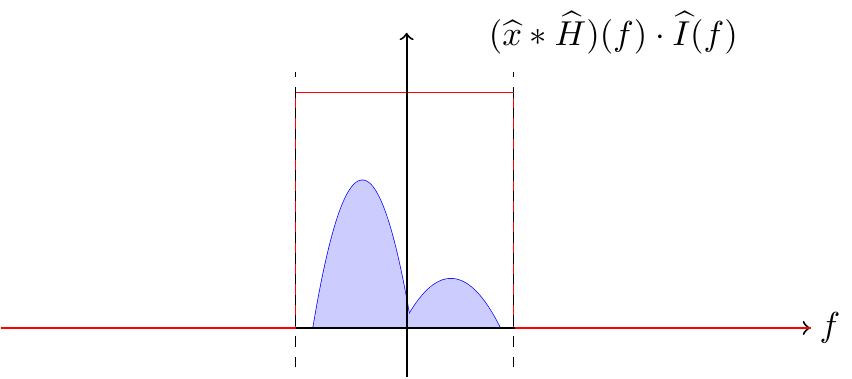}} \label{fig:sig_equ_method:2}
    \caption{The Signal Equivalent Method. This figure demonstrates that $(\wh{x}*\wh{H})(f)\cdot \wh{G}^{(j)}_{\sigma, b}(f)$ (left) can be approximated by $(\wh{x}*\wh{H})(f)\cdot \wh{I}(f)$ (right). %
    $\wh{I}(f)$ (red curve on the right) is the ideal filter that approximate $\wh{G}^{(j)}_{\sigma, b}(f)$ (red curve on the left). %
    } 
    \label{fig:sig_equ_method}
\end{figure}

The motivation of the Signal Equivalent Method comes from the special structure of $z(t)=(x\cdot H)(t)* G^{(j)}_{\sigma, b}(t)$ in the frequency domain. Notices that the observed signal $x(t)$'s Fourier transform $\wh{x}(f)$ only contains some spikes (assuming small noise). 
By convolution with $\wh{H}(f)$ (which corresponds to multiplying by $H(t)$ in the time domain), $(\wh{x}*\wh{H})(f)$ fattens the spikes in the frequency domain (and by Parseval's theorem, the area of the signal in frequency domain equals to its energy). Then, convolution with $G^{(j)}_{\sigma, b}(t)$ ``zooms-in'' to a narrow band around a single frequency. This construction of $z(t)$ motivates us to build a new signal $\ov{z}(t)=(x\cdot H)(t)*I(t)$, where $I(t)$ is a filter function such that $\wh{I}(f)=1$ when $G_{\sigma,b}^{(j)}(t)>1/2$, and $\wh{I}(f)=0$ otherwise.  To analyze the equivalent signal $\ov{z}(t)$, we improve the analysis of the filter $H(t)$ in \cite{cp19_icalp} and give a tighter bound on its value in a sub-interval of $[0,T]$.
Then, we show that the equivalent signal $\ov{z}(t)$ is \emph{almost equivalent} to $z(t)$ under  some ``good conditions'' (i.e., the frequency is isolated and no large offset). We also prove that the ideal filter has %
several useful properties that can mush simplify the analysis (e.g., the function $I(t)$ is randomized, and with high probability, $I(t)$ commutes with $H(t)$).

By the Signal Equivalent Method, we can first prove an energy bound for the equivalent signal $\ov{z}(t)$, which follows from the Fourier-sparse signals' energy bounds (see Section~\ref{sec:energy_bound}). Then, it remains to show that the equivalent signals' energy bound can approximate the original filtered signal $z(t)$'s energy bound. %
We find that the approximation error comes from two sources: the observation noise $g(t)$ and the approximation error $\ov{z}(t)-z(t)$. The first part of the error is small due to the high SNR band condition (Eq.~\eqref{eq:high_snr_intro}). And the second part of the error is mitigated by the tail-bound for $G^{(j)}_{\sigma, b}(t)$ and the heavy-cluster condition (Eq.~\eqref{eq:heavy_cluster_intro}). %
More specifically, the \textsc{HashToBins} procedure and the filter $G_{\sigma, b}^{(j)}(t)$ can bring some interference noise from other bins to $z(t)$, which is perfectly eliminated by the ideal filter $I(t)$ in the equivalent signal $\ov{z}(t)$. Hence, we need to bound this part of noise when we transfer back from the equivalent signal to the true filtered signal. The tail bound of $G_{\sigma,b}^{(j)}(t)$ ensures that adding small interference noise with frequencies far away from the center of the cluster will not drastically affect $z(t)$. However, by this argument, we can only bound the distance between $\ov{z}(t)$ and $z(t)$ by $\|x^*(t)\|_T$, which can be much larger than $\|z(t)\|_T$. Hence, we need to use the heavy-cluster assumption to ensure that $\|x^*(t)\|_T\lesssim \|z(t)\|_T$. Using these error-control techniques, we can prove that an energy bound for $\ov{z}(t)$ implies an energy bound for $z(t)$.

\begin{figure}[!t]
    \centering
    {\includegraphics[width=\textwidth]{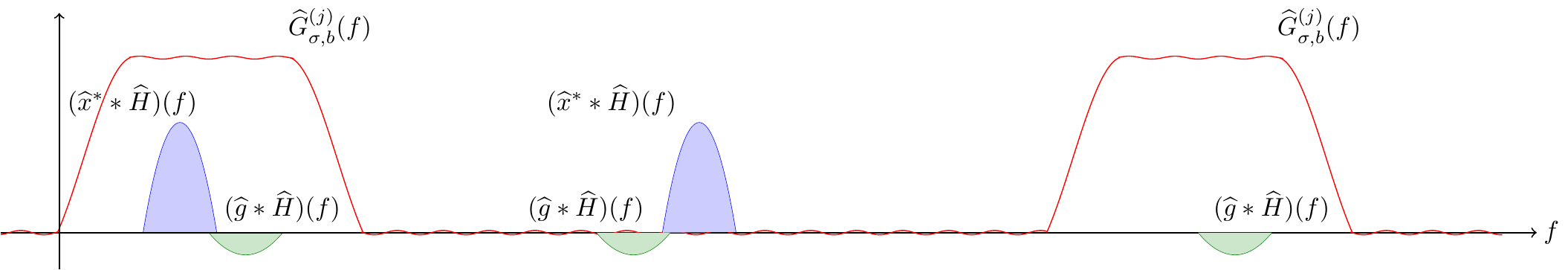}}
\caption{An illustration of a filtered noisy signal. %
$\wh{G}^{(j)}_{\sigma, b}$ (the red curve) is the \text{HashToBins} filter for the $j$-th bin. %
The noise in the filtered signal comes from two parts: one is $\wh{g}(f)*\wh{H}(f)$ (the green signal), and another is the interference by signals outside the bin (the blue and green signals in the middle). 
}\label{fig:equiv_noise}
\end{figure}
We give a comparison between ours and previous approaches for proving the energy estimation guarantee. \cite{ckps16} considers $z(t)$ as a generic signal that satisfies the time and frequency domains concentration properties\footnote{It means that most of the energy of $z(t)$ (i.e., $\|z(t)\|_2$) lies in $[0, T]$ and most of the energy of $\wh{z}(f)$ lies in a $\poly(k)/T$ length interval in frequency domain.}. We exploit ``finer'' structure of $z(t)$ and obtain a stronger energy bound and reduce the number of samples required in norm preserving.
\cite{cp19_icalp} also proves a similar  property (but only for $(x\cdot H)(t)$). However, they assume that all the frequencies of $x^*(t)$ are contained in a small interval, making the task much easier. Our filtered signal $z(t)$ does not satisfy this condition due to the interference noise caused by the \textsc{HashToBins} procedure. %

\subsubsection{Time-domain concentration of filtered signals}\label{sec:intro_concentration}

\begin{figure}[!ht]
    \centering
    {\includegraphics[width=0.6\textwidth]{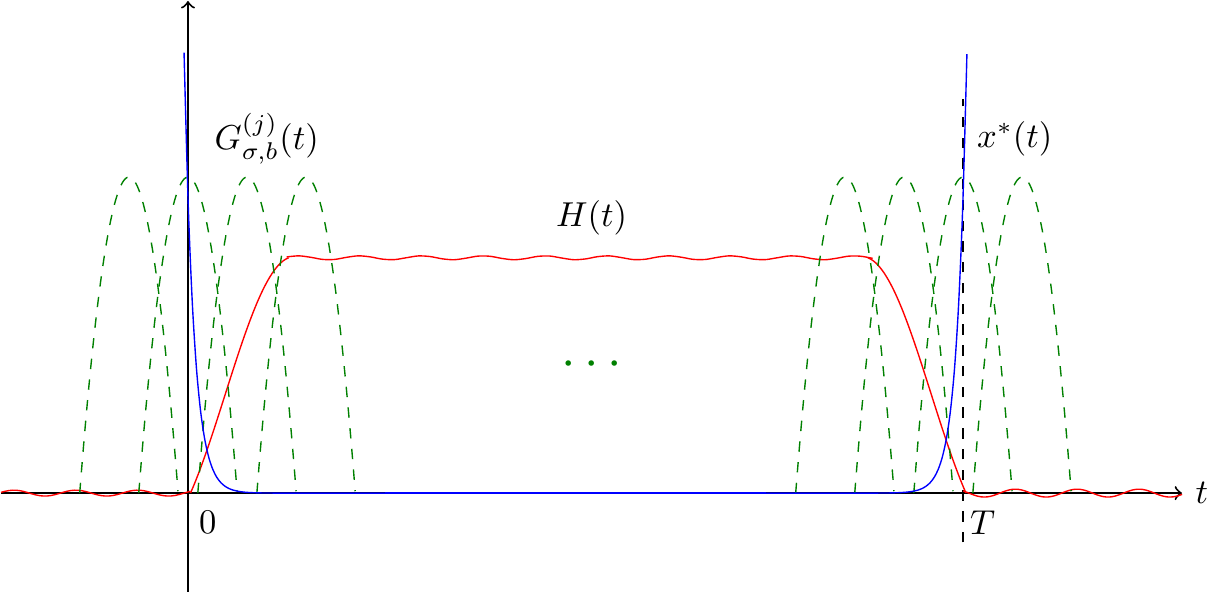}}
    \caption{A bad case that may break the time domain concentration of $z(t)=(x\cdot H)(t)*G^{(j)}_{\sigma, b}(t)$ in $[0, T]$ when the filter decay slowly. $x^*(t)$ (in blue) is a $k$-Fourier-sparse signal. $G^{(j)}_{\sigma, b}$ (in green dashed) is the filter of frequency domain. $H(t)$ (in red) is the filter in time domain. %
    On the one hand, since $x^*(t)$ is very small in $[T/\poly(k), T(1-/\poly(k))]$, and $H(t)$ is very small in $[0,T/\poly(k)] \cup [ T(1-/\poly(k)), T]$, the filtered signal has very small energy within $[0, T]$. On the other hand, since the convolution with $G^{(j)}_{\sigma, b}(t)$ can bring some energy of $x^*(t)$ passing the boundary of $[0, T]$, and the signal $x^*(t)$ could be very large outside $[0, T]$, $(x\cdot H)(t)*G^{(j)}_{\sigma, b}(t)$ may contain very large energy in $\R\backslash [0, T]$.  In this case, $\|z(t)\|_{L_2}^2\gg \|z(t)\|_T^2$.}\label{fig:z_time_cctr}
\end{figure}

The proof of Lemma~\ref{lem:significant_samples_for_each_bins:main} relies on an underlying assumption: the most of the energy of the filtered signal is contained in the observation window $[0,T]$. That is, we need the following lemma:
\begin{lemma}[Informal version of Lemma~\ref{lem:full_proof_of_3_properties_true_for_z}]\label{lem:energy_concentrate_intro}
Let $j\in [B]$ be a bin that contains a heavy frequency. Let $z(t)=(x^*\cdot H)*G_{\sigma,b}^{(j)}$ be the filtered signal. Then, we have
\begin{align*}
    \int_{-\infty}^{+\infty} |z(t) |^2 \mathrm{d} t \leq 1.35\int_{0}^{T} | z(t) |^2 \mathrm{d} t.
\end{align*}
\end{lemma}

A similar concentration property is also proved in \cite{ckps16}, using a very strict requirement on the $H(t)$ filter that it decays at an exponential rate near the boundary. More specifically, they require that $H(t)$ is exponentially small not only outside the time duration $[0, T]$, but also in the shrinking boundary $[0,T/\poly(k)]\cup [T-T/\poly(k),T]$. The additional constraint allows them to show that $x(t)$'s energy near the boundary cannot ``pass'' the $H$ filter, and the energy concentration of $z(t)$ easily follows. However, it also results in a large support of $H$ in the frequency domain, which leads to a large error in the frequency estimation, and further causes large output sparsity and time/sample complexity of their Fourier interpolation algorithm. 

We resolve this issue by changing the filter function to the one defined in \cite{cp19_icalp}, which has much smaller support and thus saves time and sample complexities.  However, it is exponentially small outside $[0,T]$, but only \emph{polynomially} small near the boundary. %
To prove Lemma~\ref{lem:energy_concentrate_intro}, we use our Signal Equivalent Method again. We construct an equivalent signal $\ov{z}(t)=(x^*\cdot H)*I(t) = (x^* * I)\cdot H(t)$, where $x^*(t)*I(t)$ is a Fourier-sparse signal. Then, by some finer analysis on the $H(t)$ filter (see Lemma~\ref{lem:property_of_filter_H}), we can show that most of the energy of $x^*(t)*I(t)$ is preserved in $[0,T]$, i.e.,
\begin{align*}
      \int_{-\infty}^{+\infty} |\ov{z}(t) |^2 \mathrm{d} t \leq 1.1 \int_{0}^{T} |\ov{z}(t) |^2 \mathrm{d} t. 
\end{align*}
Finally, by the approximation guarantee of Signal Equivalent Method, we get that the energy concentration of $\ov{z}(t)$ implies the energy concentration of $z(t)$.

\subsection{Our techniques for Fourier Interpolation}\label{sec:intro_prob_boost}

\begin{figure}[!ht]
    \centering
    \subfloat[High SNR, where $j_1=h_{\sigma_1,b_1}(f^*)$]{\includegraphics[width=0.45\textwidth]{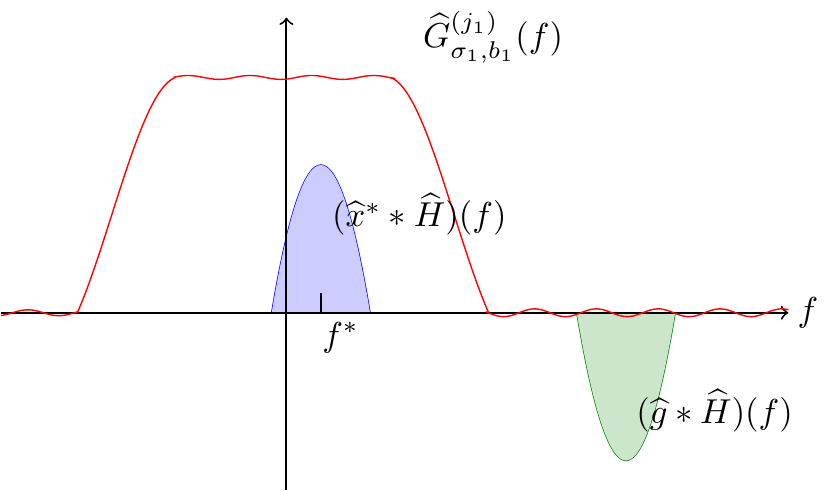}}\label{fig:hash_hsbr:1}
    \hspace{1mm}
    \subfloat[Low SNR, where $j_2=h_{\sigma_2,b_2}(f^*)$]{\includegraphics[width=0.45\textwidth]{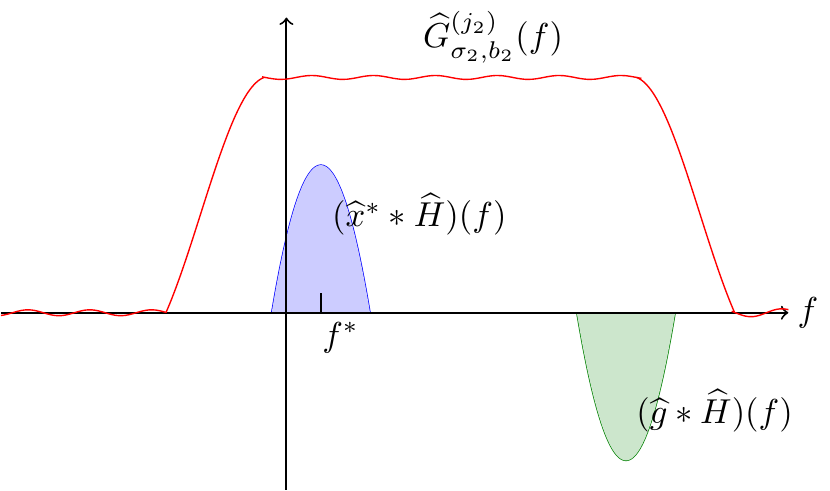}} \label{fig:hash_hsbr:2}
    \caption{The SNR of the same signal changes with different hash functions. In (a), the noise $g\cdot H$ is hashed outside the bin and suppressed by $G^{(j_1)}_{\sigma_1,b_1}$. Thus, this bin has high SNR. In (b), by a different hash function, the noise is hashed inside the bin and  $G^{(j_2)}_{\sigma_2,b_2}$ preserves its energy. Thus, the SNR of the bin becomes very low, even if the signal doesn't change. %
    } 
    \label{fig:hash_hsbr}
\end{figure}
In this section, we discuss how to obtain a Fourier interpolation algorithm with improved efficiency and output sparsity (Theorem~\ref{thm:intro_main}) based on our frequency estimation algorithm (Theorem~\ref{thm:freq_est_intro}).

We first remark that simply applying the original framework of Fourier Interpolation (e.g., \cite{ckps16,sswz22}) and combining with an existing signal estimation algorithm is still not enough to improve the previous algorithm, since the frequency estimation algorithm has a low success probability and we cannot apply the success probability boosting trick in \cite{ckps16} to increase it to $1-\rho$.
More specifically, \cite{ckps16} first boosts the success probability of their frequency estimation algorithm by their merge-stage algorithm (which runs the frequency estimation algorithm for $R=\log(1/\rho)$ times, sorts all recovered frequencies, and picks every $R/2$-th entry of the sorted list), and then runs the signal estimation algorithm. It does not work here because our high SNR band condition makes frequency estimation and signal estimation ``entangled''. More specifically, whether a frequency is contained in a high SNR bin (which \emph{needs} to be recovered) or not depends on the randomized hash function. %
However, if the outputs of multiple runs of the frequency estimation algorithm are mixed together, it is hard to justify which frequencies are necessary, since different runs use different hash functions, resulting in different high SNR bins.
In other words, if we still use \cite{ckps16}'s boosting strategy, we cannot guarantee the final output of the frequency estimation satisfies the requirement of the signal estimation algorithm.

We propose a new Fourier Interpolation framework that boosts the success probability after the signal estimation step. That is, in each run of the constant success probability frequency estimation algorithm, we reconstruct the signal immediately. Let $y_1,\dots,y_{R_p}$ denote the reconstructed signals of $R_p$ runs. Then, we boost the total success probability by outputting the signal $y_{j^*}$:
\begin{align*}
    j^* = \underset{j \in [R_p] }{\arg\min}~\underset{i\in [R_p]}{\median}~ \| y_j(t) -y_i(t) \|_T^2.
\end{align*}
By Chernoff bound, there are more than a half of $y_i$'s being good approximations of the ground-truth signal $x^*(t)$. Using the median trick, we can show that $y_{j^*}$ satisfies the recovery guarantee with an exponentially small failure probability.

It remains to estimate the distance $\|y_j(t) - y_i(t)\|_T^2$ between different reconstructed signals. Naively, it takes $O(\wt{k}^2)$-time since $y_1,y_2$ are $\wt{k}$-Fourier sparse, and it is enough to obtain the time complexity of our Fourier interpolation algorithm in Theorem~\ref{thm:intro_main}. We further propose an $\wt{O}(\wt{k}\cdot k)$-time approximation algorithm for estimating a Fourier-sparse signal's energy, which could be of independent interest. The main idea is to use $\|y_i(t) - y_j(t)\|_{S,w}^2$ to approximate $\|y_i(t) - y_j(t)\|_T^2$, where the sample set $S$ and weights $w$ are defined by the significant sample generation procedure in Section~\ref{sec:intro_gen_sig_samples}. We show that if we take $|S|=\wt{O}(\wt{k})$, we can achieve a constant approximation ratio in $\wt{O}(\wt{k}\cdot k) $ time. In addition, we prove that even if we use the approximated distances, the output signal $y_{j^*}$ still satisfies the recovery guarantee of Fourier interpolation. %

\section{Organization}

In Section \ref{sec:preli}, we define our notations in this paper. In Section \ref{sec:energy_bound}, we review several energy bounds for Fourier-sparse signal.  In Section \ref{sec:f_filter},  we define and show several properties of the frequency domain filters $G_{\sigma,b}^{(j)}$.  In Section \ref{sec:hash_stratgy}, we review the  \textsc{HashToBins} strategy and prove that bad events only happen with small probability.  In Section \ref{sec:t_filter},
we define and show some properties of the time domain filter $H(t)$.

Based on the analysis of the filters, in Section \ref{sec:approx_f_SFS}, we study the ideal filter and develop the Signal Equivalent Method. In Section \ref{sec:concentr}, we show that the filtered signal satisfies some concentration properties in both time and frequency domains. 

Based on the Signal Equivalent Method and the concentration properties, in Section \ref{sec:struct_FS}, we prove an energy bound for filtered Fourier-sparse signals. In Section \ref{sec:twist_struct_FS}, we further extend the energy bound for local-test signals. Then, in Section \ref{sec:emp_energy_est}, we apply the energy bounds and describe how to use samples to empirically estimate the energy of filtered signals and local-test signals. In Section \ref{sec:gen_sigi_sam}, we introduce our algorithm for generating significant samples. In Section \ref{sec:freq_est}, we use the significant samples to do frequency estimation for Fourier sparse signals. Finally, in Section \ref{sec:sig_recontr}, we combine our frequency estimation algorithm with a signal estimation procedure and boost the success probability of Fourier Interpolation. 

Section~\ref{sec:flowchart} presents a flowchart of the key theorems/lemmas for our Fourier interpolation algorithm.

\newpage

\newpage
\onecolumn
\appendix

\section{Preliminaries}
\label{sec:preli}

For any positive integer $n$, we define $[n]$ to be the set $\{1,2,\cdots,n\}$. We define $\i$ to be $\sqrt{-1}$. For a complex number $z=a+b\i$, we define $|z|$ to be the magnitude of $z$, i.e., $|z|=\sqrt{a^2+b^2}$.
For a function $f$, we use $\supp(f)$ to denote the support set of $f$. We use $f \lesssim g$ to denote that there exists a constant $C$ such that $f\leq C\cdot g$. We use $f\eqsim g$ to denote that $ f\lesssim g \lesssim f$. 
For any function $f$, we use $\poly(f)$ to denote $f^{O(1)}$, and $\wt{O}(f)$ to denote $f\cdot \poly\log(f)$. 
For an interval $U\subseteq \mathbb{R}$, we use $|U|$ to denote the size of the interval, and we use $ \text{Uniform}(U)$ to denote the uniform distribution over $U$. 

We use $\omega$ to denote the exponent of matrix multiplication, i.e., $n^{\omega}$ denote the time of multiplying an $n \times n$ matrix with another $n \times n$ matrix. Currently $\omega \approx 2.373$ \cite{w12,aw21}.

For two functions $f$ and $g$, we use $(f*g)(t)=\int_{-\infty}^\infty f(s)g(t-s) \d s$ to denote the convolution of two functions $f$ and $g$. And we use $f^{*l}$ to denote the $l$-fold convolution of $f$, i.e., $f^{*l}(t)=f(t)*f(t)*\cdots *f(t)$. 
For $a\in \R_+$, we use $\rect_a(t)$ to denote the box function with support set length $a$, i.e., $\rect_a(t)=\mathbf{1}_{[-a/2, a/2]}(t)$.
For $a\in \R$, we use $\delta_a(f)$ to denote $\delta(f-a)$, where $\delta(f)$ is the Dirichlet function. We use ${\rm round}(x)$ to denote rounding $x\in \mathbb{R}$ to the nearest integer. For $x\in\R, y \in \R_+$, we use $x\pmod y $ to denote the smallest positive $z\in\R_+$ such that $z\in x+y\Z$.

We say $x(t)$ is $k$-Fourier-sparse if:
\begin{align*}
x(t)=\sum_{j=1}^{k} v_j e^{2\pi\i f_j t}.
\end{align*}
We define $\wh{x}(f)$ to be the Fourier transform of $x(t)$:
\begin{align*}
  \wh{x}(f) =\int_{-\infty}^{\infty} x(t) e^{-2\pi\i f t} \d t.  
\end{align*}
We use ${\cal F}_{k,F}$ to denote the following family of signals:
\begin{align*}
    {\cal F}_{k,F}:=\Big\{ x(t)=\sum_{j=1}^k v_j \cdot e^{2 \pi \i f_j t} ~\Big|~ f_j \in [-F,F], v_j\in \mathbb{C}~\forall j\in[k] \Big\}. 
\end{align*}
Then, we define several norms for signal. 
\begin{itemize}
\item For any discrete set $S\subseteq \mathbb{R}$, the \emph{discrete norm} of $x$ with respect to a set $S$ is defined as 
\begin{align*}
    \|x(t)\|^2_S= \frac{1}{|S|} \sum_{t\in S}|x(t)|^2,
\end{align*}
and the \emph{weighted discrete norm} with weights $w\in \mathbb{R}^S$ is defined as
\begin{align*}
    \|x(t)\|^2_{S,w}=  \sum_{t\in S}w_t|x(t)|^2.
\end{align*}
\item For any continuous interval $U\subset \R$, the continuous $U$-norm of $x$ is defined as
\begin{align*}
    \|x(t)\|_U^2 = \frac{1}{|U|}\int_U |x(t)|^2\d t.
\end{align*}
\item For any $T>0$, the \emph{continuous $T$-norm} is defined as 
\begin{align*}
    \|x(t)\|^2_T=\frac{1}{T} \int_0^T |x(t)|^2\d t.
\end{align*}
\item Let $D$ be a probability distribution over $\mathbb{R}$. The \emph{continuous $D$-norm} is defined as
\begin{align*}
    \|x(t)\|^2_D=\int_{-\infty}^\infty D(t) |x(t)|^2\d t.
\end{align*}
\item The $L_2$-norm of $x(t)$ is defined as
\begin{align*}
    \|x(t)\|^2_{L_2}=\int_{-\infty}^\infty |x(t)|^2\d t.
\end{align*}
\end{itemize}

Throughout this paper, we assume that $x^*(t)\in {\cal F}_{k,F}$ is our ground-truth signal. And the observation signal is $x(t)=x^*(t)+g(t)$, where $g(t)$ is an arbitrary noise function. Furthermore, we assume that $x(t)$ can be observed at any point in $[0, T]$. 

\begin{lemma}[Chernoff Bound \cite{chernoff1952}]\label{lem:chernoff_bound}
Let $X_1, X_2, \cdots, X_n$ be independent random variables. Assume that $0\leq X_i \leq 1$ always, for each $i \in [n]$. Let $X= X_1+X_2+\cdots+X_n$ and $\mu = \mathbb{E}[X] = \overset{n}{  \underset{i=1}{\sum} } \mathbb{E}[X_i]$. Then for any $\epsilon>0$,
\begin{align*}
    \mathsf{Pr} [ X \geq (1+\epsilon) \mu ] \leq \exp(-\frac{\epsilon^2 }{2+\epsilon} \mu) \textit{ and } \mathsf{Pr} [ X \leq (1-\epsilon) \mu ] \leq \exp(-\frac{\epsilon^2 }{2} \mu).
\end{align*}
\end{lemma}

\section{Energy Bounds of Fourier Sparse Signals}
\label{sec:energy_bound}

The energy bound of a function family $\cal F$ is the largest value achieved by a function $f\in {\cal F}$ normalized by its norm (total energy) $\|f\|_T$. It connects the extreme value and the average value of the functions in ${\cal F}$, and is very useful in analyzing the concentration property.  

In our setting, we take ${\cal F}={\cal F}_{k,F}$ to be the set of $F$-band-limit, $k$-sparse Fourier signals.
\cite{k08} showed an energy bound that only depends on the sparsity $k$, without any dependence on the time point $t$, band-limit $F$, time duration $T$, frequency gap $\eta={\min}_{i\neq j}|f_i-f_j|$:

\begin{theorem}[\cite{k08}] %
\label{thm:energy_bound}
For  any $t \in [0,T]$,
\begin{align*}
    \underset{x \in {\cal F}_{k,F}}{\sup} \frac{|x(t)|^2}{\|x\|_T^2} \lesssim k^2.
\end{align*}
\end{theorem}

The $k^2$ energy bound can be further improved if we only consider the functions' value at a fixed time point $t$: %
\begin{theorem}[\cite{cp19_colt, be06}] %
\label{thm:bound_k_sparse_FT_x_middle}
Let $D:=\mathrm{Uniform}([-1, 1])$. For  any $t \in (-1,1)$,
\begin{align*}
    \underset{x \in {\cal F}_{k,F}}{\sup} \frac{|x(t)|^2}{\|x\|_D^2} \lesssim \frac{k}{1-|t|}.
\end{align*}
\end{theorem}

\section{Filter in Frequency Domain}
\label{sec:f_filter}
Filtering is one of the most important techniques in sparse Fourier transform literature. 
In this section, we introduce the frequency domain filter function $\wh{G}^{(j)}_{\sigma, b}(f)$, which is the key to implement the \textsc{HashToBins} strategy.
We first review the the construction given by \cite{ckps16} with some different parameter settings and show some known properties (see Section \ref{sec:f_filter_pre}). %
Then, we prove a new property of the filter functions: the frequency domain covering property (see Section \ref{sec:f_filter_cover}). %

\subsection{Frequency domain filter construction}
\label{sec:f_filter_pre}

In this section we review the construction and several basic properties of the frequency domain filter $G^{(j)}_{\sigma, b}(t)$, $\wh{G}^{(j)}_{\sigma, b}(f)$.

\begin{definition}[$G$-filter's construction, \cite{ckps16}]\label{def:define_G_filter}
Given $B >1$, $\delta >0$, $\alpha>0$. Let $l := \Theta( \log(k/\delta) )$.
Define $G_{B,\delta,\alpha}(t)$ and its Fourier transform $\wh{G_{B,\delta,\alpha}}(f)$ as follows:
\begin{eqnarray*}
G_{B,\delta,\alpha}(t):= &~ b_0 \cdot (\rect_{ \frac{B}{(\alpha \pi)}} (t) )^{\star  l} \cdot  \sinc(t \frac{\pi}{2B}), \\
\widehat{G_{B,\delta,\alpha}}(f):= &~ b_0 \cdot ( \sinc(\frac{B}{\alpha \pi} f) )^{ l} * \rect_{\frac{\pi}{2B}}(f),
\end{eqnarray*}
where $b_0 = \Theta(B \sqrt{l}/\alpha)$ is the normalization factor such that $\wh{G}(0)=1$. %
\end{definition}

\begin{definition}[Filter for bins]\label{def:G_j_sigma_b}
Given $B >1$, $\delta >0$, $\alpha>0$, let 
\begin{align*}
    \wh{G}(f):=\wh{G}_{B,\delta,\alpha}(2\pi(1-\alpha) f)
\end{align*}
where $G_{B,\delta,\alpha}$ is defined in Definition~\ref{def:define_G_filter}. For any $\sigma>0, b\in\R$ and $j\in [B]$, 
define
\begin{align*}
G^{(j)}_{\sigma,b}(t) := &~ \frac{1}{\sigma} G(t/\sigma) e^{2\pi\i t(j/B-\sigma b)/\sigma},
\end{align*}
and its Fourier transformation:
\begin{align*}
\widehat{G}^{(j)}_{\sigma,b}(f)  = \sum_{i\in \Z}\widehat{G}( \sigma f +\sigma b-i-\frac{j}{B}).
\end{align*}
\end{definition}

Then, we provide several properties of $G $ and $G^{(j)}_{\sigma, b}(t)$, which is proven by \cite{ckps16}.

\begin{lemma}[$G$-filter's properties, \cite{ckps16}]
Given $B >1$, $\delta >0$, $\alpha>0$, let $G:=G_{B,\delta,\alpha}(t)$ be defined in Definition~\ref{def:define_G_filter}. Then, $G$ satisfies the following properties:%
\begin{eqnarray*}
&\mathrm{Property~\RN{1}} : &\widehat{G}(f) \in [1 - \delta/k, 1] , \text{~if~} |f|\leq (1-\alpha)\frac{2\pi}{2B}.\\
&\mathrm{Property~\RN{2}}: &\widehat{G}(f) \in [0,1], \text{~if~}  (1-\alpha)\frac{2\pi}{2B} \leq |f| \leq \frac{2\pi}{2B}.\\
&\mathrm{Property~\RN{3}} : &\widehat{G}(f) \in [-\delta /k, \delta/k], \text{~if~}  |f| >  \frac{2\pi}{2B}
.\\
&\mathrm{Property~\RN{4}} : &\supp(G(t) ) \subset [\frac{l}{2} \cdot \frac{-B}{\pi\alpha}, \frac{l}{2} \cdot \frac{B}{\pi\alpha}].\\
&\mathrm{Property~\RN{5}} : & \underset{t}{\max} |G(t)| \lesssim  \poly(B,l).%
\end{eqnarray*}
\end{lemma}

\begin{lemma}\label{lem:property_of_filter_G}
Let $G^{(j)}_{\sigma, b}(t)$ be defined in Definition~\ref{def:G_j_sigma_b}.
Let offset function 
\begin{align*}
    o_{\sigma, b}(f) = |(\sigma f +\sigma b -\frac{j}{B}) -\frac{1}{2} {\pmod 1}| -\frac{1}{2}.
\end{align*}
Then, 
\begin{eqnarray*}
&\mathrm{Property~\RN{1}} : & \wh{G}_{\sigma,b}^{(b)}(f) \in [1 - \delta/k, 1] , \text{~if~} |o_{\sigma, b}(f)|\leq  (1-\alpha)\frac{2\pi}{2B}.\\
&\mathrm{Property~\RN{2}}: & \wh{G}_{\sigma,b}^{(b)}(f)\in [0,1], \text{~if~}  (1-\alpha)\frac{2\pi}{2B} \leq |o_{\sigma, b}(f)| \leq \frac{2\pi}{2B}.\\
&\mathrm{Property~\RN{3}} : & \wh{G}_{\sigma,b}^{(b)}(f)\in [-\delta /k, \delta/k], \text{~if~}  |o_{\sigma, b}(f)| >  \frac{2\pi}{2B}
.\\
&\mathrm{Property~\RN{4}} : &\supp({G}_{\sigma,b}^{(b)}(t) ) \subset [\frac{l}{2} \cdot \frac{-B}{\pi\alpha}, \frac{l}{2} \cdot \frac{B}{\pi\alpha}].\\
&\mathrm{Property~\RN{5}} : & \underset{t}{\max} |{G}_{\sigma,b}^{(b)}(t)| \lesssim  \poly(B,l).%
\end{eqnarray*}
\end{lemma}

\begin{figure}[t]
    \centering
    {\includegraphics[width=\textwidth]{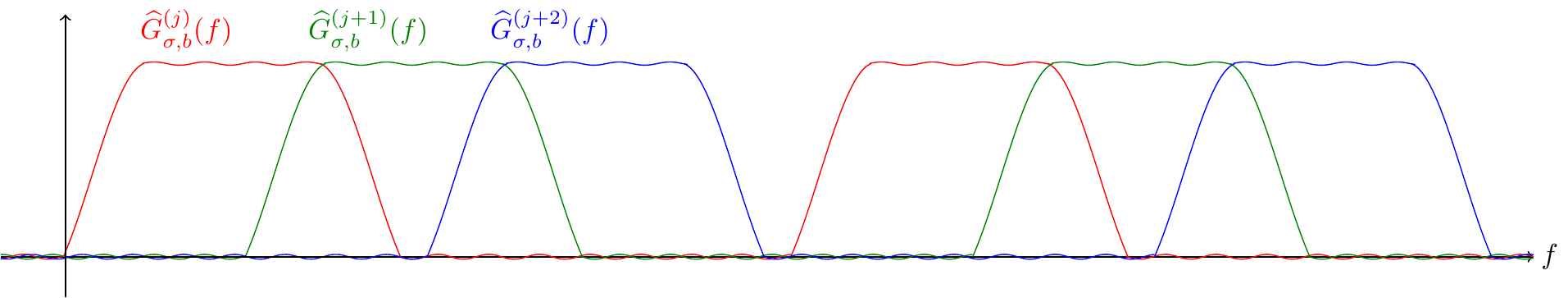}}
\caption{Filters with the frequency domain covering property. The red, green, and blue curves represent the filters $\wh{G}^{(j)}_{\sigma, b}$, $\wh{G}^{(j+1)}_{\sigma, b}$, and $\wh{G}^{(j+2)}_{\sigma, b}$, respectively.
The frequency domain covering property ensures that for each frequency $f\in \mathbb{R}$, there are \emph{at least one but no more than two} filters satisfying $\wh{G}^{(j)}_{\sigma, b}(f)\geq 1-\delta$. %
}\label{fig:covering}
\end{figure}

\subsection{Frequency domain covering}
\label{sec:f_filter_cover}

In this section, we show that the filter functions $\{\hat{G}_{\sigma,b}^{(j)}\}_{j\in [B]}$ form a proper cover for the frequency domain. Roughly speaking, for any frequency $f\in \mathbb{R}$, we show that the sum of all filters' values (squared) at $f$ is very close to one. This property is very important for our high SNR band assumption.

\begin{lemma}\label{lem:GlargerEverywhere}
For any $f\in \R$, there exists at least one $ j\in [B]$ such that
\begin{align*}
    \widehat{G}^{(j)}_{\sigma,b}(f) \geq 1-\frac{\delta}{k}.
\end{align*}
\end{lemma}
\begin{proof}

We first prove that the lemma holds for those $f^*$ where there exists a $j\in [B]$ such that
\begin{align*}
f^*\in [-b+\frac{j}{\sigma B}-\frac{1}{2\sigma B}, -b+\frac{j}{\sigma B}+\frac{1}{2\sigma B}]+\frac{1}{\sigma} \Z.
\end{align*}
For such $f^*$, we have 

\begin{align*}
     \widehat{G}^{(j)}_{\sigma,b}(f^*) =&~ \sum_{i\in \Z}\widehat{G}( \sigma f^* +\sigma b-i-\frac{j}{B})\\
     \geq &~ \widehat{G}(( \sigma f^* +\sigma b-\frac{j}{B}) \mod{1}) \\
     \geq&~ 1-\frac{\delta}{k},
\end{align*}
where the first step follows from the definition of $\widehat{G}^{(j)}_{\sigma,b}(f^*) $, the second step is straight forward, the third step follows from 
\begin{align*}
     \sigma f^* +\sigma b-\frac{j}{B}\mod{\frac{1}{\sigma}} \in [-\frac{1}{2 B}, \frac{1}{2 B}]
\end{align*}
and Lemma \ref{lem:property_of_filter_G} Property \RN{1} and Definition \ref{def:G_j_sigma_b}.

It remains to show that for an arbitrary $f\in \mathbb{R}$, the condition still holds. Let 
\begin{align*}
    j:=\mathrm{round}((\sigma f + \sigma b \mod{1} )  \cdot B).
\end{align*}
We have
\begin{align*}
    j\in [(\sigma f + \sigma b \mod{1} )  \cdot B-\frac{1}{2}, (\sigma f + \sigma b \mod{1} )  \cdot B + \frac{1}{2}],
\end{align*}
which implies that
\begin{align*}
    j\in [(\sigma f + \sigma b )  \cdot B-\frac{1}{2}, (\sigma f + \sigma b  )  \cdot B + \frac{1}{2}] + B \Z.
\end{align*}
Thus,
\begin{align*}
    f \in [-b+\frac{j}{\sigma B}-\frac{1}{2\sigma B}, -b+\frac{j}{\sigma B}+\frac{1}{2\sigma B}]+\frac{1}{\sigma} \Z.
\end{align*}

The lemma is then proved.
\end{proof}

\begin{lemma}\label{lem:filterGlarge}
For any $f\in \R$,
\begin{align*}
     \sum_{j=1}^B |\widehat{G}^{(j)}_{\sigma,b}(f) |^2\eqsim  1.
\end{align*}
\end{lemma}
\begin{proof}
By Lemma \ref{lem:GlargerEverywhere}, we have that for any $f\in\R$, there exist at least a $ j_0\in [B]$ such that
\begin{align}
    \widehat{G}^{(j_0)}_{\sigma,b}(f) \geq \frac{1}{2}. \label{eq:GlargerEverywhere}
\end{align}

Moreover, we have that
\begin{align}
    B \frac{\delta}{k} = O(\delta) \leq 0.01. 
    \label{eq:GlargerEverywhere_0}
\end{align}
where the first step follows from $B=O(k)$, the second step follows from $\delta=o(1)\leq 0.01$.

In the followings, we give lower and upper bounds for $\sum_{j=1}^B |\widehat{G}^{(j)}_{\sigma,b}(f) |^2$.

\paragraph{Lower bound:}
\begin{align*}
\sum_{j=1}^B |\widehat{G}^{(j)}_{\sigma,b}(f)|^2\geq |\widehat{G}^{(j_0)}_{\sigma,b}(f)|^2 
\gtrsim 1.
\end{align*}
where the first step follows from Lemma \ref{lem:property_of_filter_G} Property \RN{1}, \RN{2}, and \RN{3}, the second step follows from Eq.~\eqref{eq:GlargerEverywhere}, the third step follows from Eq.~\eqref{eq:GlargerEverywhere_0} and ${\delta}/{k}\leq 1$.

\paragraph{Upper bound:}
\begin{align*}
\sum_{j=1}^B (\widehat{G}^{(j)}_{\sigma,b}(f))^2 \leq  2  + B (\frac{\delta}{k})^2
\lesssim 1
\end{align*}
where the first step follows from the definition of $\widehat{G}^{(j)}_{\sigma,b}(f)$, the second step follows from Eq.~\eqref{eq:GlargerEverywhere_0} and ${\delta}/{k}\leq 1$.

Combining them together, the lemma follows.
\end{proof}

\section{Hashing the Frequencies}
\label{sec:hash_stratgy}

In this section, we review the \textsc{HashToBins} strategy, which an important tool for Sparse Fourier Transform \cite{hikp12a, ikp14, ps15, ckps16,k16, k17, jls23}. %
Ideally, the \textsc{HashToBins} procedure randomly splits the frequency domain into $B$ bins so that each bin contains at most one frequency. Then, the $k$-sparse Fourier reconstruction problem is reduced to a much easier one-sparse Fourier reconstruction problem. 

We first describe the hashing strategy(see Section~\ref{sec:hash_stratgy:describe_proc}). However, there are two kinds of bad events such that the \textsc{HashToBins} procedure cannot work as good as we want: two frequencies are hashed to the same bin, or some frequency lies close to the boundary of a bin. 
We show that these bad events only happen with small probabilities (see Sections~\ref{sec:hash_stratgy:freq_iso} and \ref{sec:hash_stratgy:large_off}) .

\subsection{\textsc{HashToBins} procedure}\label{sec:hash_stratgy:describe_proc}

Here, we introduce the hash function and how to compute the resulting signal of the \textsc{HashToBins} procedure.

We first give the definition of the hashing function: 
\begin{definition}[Hash function, \cite{ckps16}]
Let $\pi_{\sigma,b}(f) = \sigma(f+b) \pmod {1}$ and $h_{\sigma,b}(f) = \mathrm{round} (\pi_{\sigma,b}(f) \cdot {B})$ be the hash function that maps frequency $f \in [-F,F]$ into bins $\{0,\cdots,B-1\}$.
\end{definition}
Intuitively, the $j$-th bin corresponding to $f$ such that $\wh{G}^{(j)}_{\sigma, b}(f) \geq 1- \delta/ k$. In general, we set $B= \Theta(k)$ and $\sigma \in [\frac{1}{B \Delta},\frac{2}{B \Delta}]$ chosen uniformly at random, where $\Delta=k\cdot |\supp(\wh{H}(f))|$. 

Then, we show how to compute the \textsc{HashToBins}:
\begin{lemma}[Lemma 6.9 in \cite{ckps16}]\label{lem:hashtobins}

Let $z_j(t) = x(t) * G^{ (j)}_{\sigma,b}(t)$. Let $a:=t/\sigma$
Let $ u\in \C^B$ and for $j \in [B]$,
\begin{align*}
u_j := \sum_{i\in \Z}x(\sigma(a-j-iB )) e^{-2\pi\i \sigma b(j+iB)} G(j+iB).
\end{align*}
Then, we have that for all $j\in [B]$,
\begin{align*}
      \wh{u}_j = z_j({\sigma a}).
\end{align*}
\end{lemma}
Note that when we apply Lemma \ref{lem:hashtobins}, we take $x(t)=x(t)\cdot H(t)$, where the latter $x(t)$ is the observable signal, the $H(t)$ is the filter of time domain (see Section~\ref{sec:t_filter}).

\subsection{Frequency isolation}
\label{sec:hash_stratgy:freq_iso}

\begin{figure}[!ht]
    \centering
    {\includegraphics[width=\textwidth]{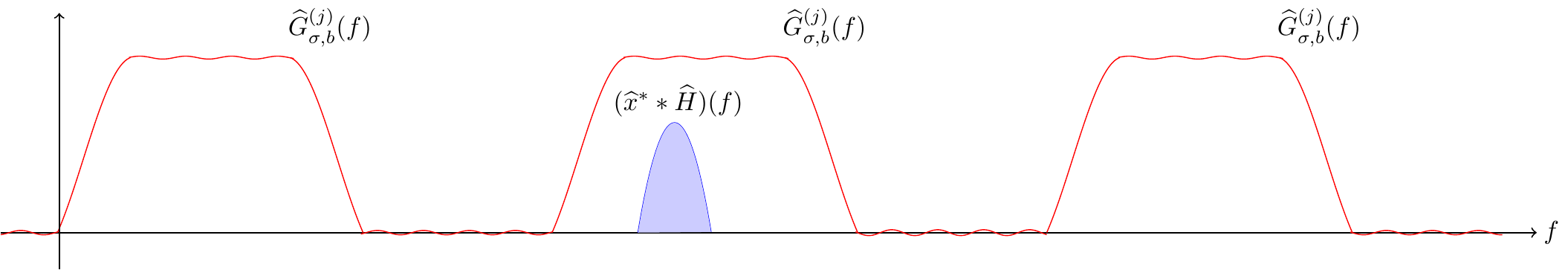}}
\caption{An example of the well-isolation event in the frequency domain. $\wh{G}^{(j)}_{\sigma, b}$ (the red curve) is the filter of frequency domain for the $j$-th bin 
and $H(t)$ is the filter of time domain. $\wh{x}^*(f)*\wh{H}(f)$ (the blue curve) is the filtered ground-truth signal. 
Under the well-isolation event, there is one small interval that contains most of the energy. In other words, each bin only contains one-cluster of frequencies.
}\label{fig:isolation}
\end{figure}

The goal of this section is to define and analyze the Frequency Isolation event. Frequency Isolation requires that the energy of the hashed signal in each bin is concentrated in a small band in the frequency domain. This condition is roughly equivalent to say that each bin only contains one cluster of frequencies. %
This condition is very useful in proving the concentration of the filtered signal in the frequency domain, which serves as one of the basic assumptions of our significant sample generation procedure.

We first introduce Claim \ref{cla:PS15_hash_claims}. This claim states that if two frequencies are not close to each other, with large probability, they also not hashed into the same bin.

\begin{claim}[Collision probability,  \cite{ckps16}]\label{cla:PS15_hash_claims}
For any $\Delta_0>0$, let $\sigma$ be a sample uniformly at random from $[\frac{1}{4B\Delta_0}, \frac{1}{2B\Delta_0}]$. Then, we have:

\begin{enumerate}
\item If $4\Delta_0 \leq|f^+ - f^-| < {2(B-1)\Delta_0} $, then $\mathsf{Pr}[h_{\sigma,b}(f^+) = h_{\sigma,b}(f^-)]=0$.

\item If ${2(B-1)\Delta_0} \leq |f^+ - f^-|$, then
$\mathsf{Pr}[h_{\sigma,b}(f^+) = h_{\sigma,b}(f^-)] \lesssim \frac{1}{B}$.
\end{enumerate}
\end{claim}

We then provide the formal definition of the well-isolation event:
\begin{definition}[Well-isolation condition]\label{def:k_signal_recovery_z}
  We say that a frequency $f^*$ is \emph{well-isolated} under the
  hashing parameters $(\sigma, b)$ if, for $j = h_{\sigma, b}(f^*)$, the hashed signal (in frequency domain) $\wh{z}^{(j)}(f) := \widehat{x\cdot H}(f) \cdot \widehat{G}^{ (j)}_{\sigma,b}(f)$
  satisfies
\begin{align*}
      \int_{\overline{I_{f^*}}} |{\wh{z}^{(j)}(f)}|^2df \lesssim \epsilon \cdot T\N^2/k,
\end{align*}
over the interval $\overline{I_{f^*}} = (-\infty, \infty) \setminus (f^* -\Delta, f^* + \Delta)$.
\end{definition}

The following lemma shows the probability of the Frequency Isolation event under the randomized hashing functions:
\begin{lemma}[Lemma 7.19 in \cite{ckps16}]\label{lem:often-well-isolated}
  Let $f^*$ be any frequency.  Then $f^*$ is well-isolated by
  hashing parameters $(\sigma, b)$ with probability $\ge 0.9$. 
\end{lemma}

\subsection{Large offset event}
\label{sec:hash_stratgy:large_off}

Large offset event is another kind of bad event for the \textsc{HashToBins} procedure, which happens when a ground-truth frequency is hashed into the changing edge of the filter $G^{(j)}_{\sigma, b}$. The large offset event breaks the guarantee of our signal equivalent method, and thus affects the performance of our significant sample generation and frequency estimation. Fortunately, this bad event only happens with a small probability.

We first state a tool for analyzing the hashing procedure, which intuitively says that the modular of a random sampling from a long interval is almost uniformly distributed:
\begin{lemma}[\cite{ps15, ckps16}] \label{lem:wrapping}
For any $\widetilde{T}$, and $ 0 \leq \widetilde{\epsilon}, \widetilde{\delta} \leq \widetilde{T}  $, if we sample $\widetilde{\sigma}$ uniformly at random from $[A,2A]$, then
\begin{equation}%
\frac{2\widetilde{\epsilon} }{ \widetilde{T} } - \frac{2\widetilde{\epsilon} }{ A } \leq \mathsf{Pr} \left[ \widetilde{\sigma}  {\pmod {\widetilde{T}} } \in [ \widetilde{\delta} - \widetilde{\epsilon}, \widetilde{\delta} + \widetilde{\epsilon} ~] \right] \leq \frac{2\widetilde{\epsilon} }{ \widetilde{T} } + \frac{4\widetilde{\epsilon} }{ A }.
\end{equation} 
\end{lemma}

Then, we define the large offset event:
\begin{definition}[Large offset event]
Given $ \sigma\in\R_+, b \in \R$. Let ${G}^{(j)}_{\sigma,b}$ and $\delta$ be defined as in Definition \ref{def:define_G_filter}. For any $k$-Fourier-sparse signal $x$, we say the \emph{Large Offset event} happens, if for any $f\in \supp(\wh{x\cdot H})$ and any $j\in [B]$, 
\begin{align*}
    \widehat{G}^{ (j)}_{\sigma,b}(f)\in \Big[\frac{\delta}{k}, 1-\frac{\delta}{k}\Big].
\end{align*}
\end{definition}

We analyze the probability of large offset event in the following lemma:

\begin{lemma}\label{lem:large_off_not_happen}
Let $\Delta_0=O(\Delta)$, $ \wh{\sigma}=1/\Delta_0$. Given $b=O(\max\{F, 1/\wh{\sigma}\})$, suppose $ \sigma \sim [0.5 \wh{\sigma}, \wh{\sigma}]$ uniformly at random. Then, with probability at least $0.99$, the Large Offset event does not happen. 

Furthermore, with probability at least $0.99$, for any $j\in [k]$, for any $f\in f_j+\supp(\wh{H})$, it holds that $\hat{G}_{\sigma,b}^{(j)}(f)\notin [\delta/k, 1-\delta/k]$.
\end{lemma}
\begin{proof}
Let $\alpha$ be defined as in Definition \ref{def:define_G_filter}. Let $I_G :=\{f\in\R~|~ \widehat{G}^{ (j)}_{\sigma,b}(f)\in [\delta /k, 1-\delta / k]\}$. Following from Lemma \ref{lem:property_of_filter_G} Property \RN{2}, we have that $ s_G:=|I_G{\pmod {1/\sigma }}|\leq 10\alpha \Delta_0/B$. 

Let $\delta_{f^*}(f)$ be the Dirichlet function at $f^*$. For any $f_j$ with $j\in[k]$, let $I_{f_j}:=\supp(\wh{H}*\delta_{f_j})$.
We also define 
\begin{align*}
    I'_{f_j}:= \{f\in\R~|~ [f-s_G, f+s_G]\cap  I_{f_j} \neq \emptyset \}.
\end{align*}
Since $\supp(\wh{x\cdot H})=\supp(\wh{H}*\wh{x})\subseteq \bigcup_{j=1}^k \supp(\wh{H}*\delta_{f_j})$, we know that the Large Offset event happens if 
\begin{align*}
    \Big(\bigcup_{j=1}^k I_{f_j}\Big)\cap I_{G} \neq \emptyset.
\end{align*}
Thus, it suffices to bound $\mathsf{Pr} [(\cup_{j=1}^k I_{f_j})\cap I_{G} \neq \emptyset]$.

First, for any $j\in [k]$, we have
\begin{align}
    |I'_{f_j} | \leq &~ |I_{f_j}|+ 2s_G
    \leq \Delta /B + 2s_G
     \leq O(\Delta /B) \label{eq:0_large_off_not_happen}
\end{align}
where the first step follows from the definition of $\Delta$, the second step follows from  $ s_G\leq 10\alpha \Delta_0/B$ and the setting of $\alpha$. 

We have that
\begin{align}
&~\mathsf{Pr} \Big[ \frac{1}{2B\sigma}+\frac{j}{B\sigma}-b   {\pmod {1/\sigma }} \in I'_{f_j} {\pmod {1/\sigma }}\Big] \notag \\
=&~ \mathsf{Pr} \Big[ \frac{1}{2B}+\frac{j}{B}   {\pmod {1 }} \in \sigma b+ \sigma I'_{f_j} {\pmod {1}} \Big] \notag \\
=&~ \mathsf{Pr} \Big[ \sigma b+ \sigma f_j  {\pmod {1 }} \in \frac{1}{2B}+\frac{j}{B} + \sigma [-|I'_{f_j}|/2, |I'_{f_j}|/2] {\pmod {1}} \Big] \notag \\
\leq &~ \mathsf{Pr} \Big[ \sigma b+ \sigma f_j  {\pmod {1 }} \in \frac{1}{2B}+\frac{j}{B} + \wh{\sigma} [-|I'_{f_j}|/2, |I'_{f_j}|/2] {\pmod {1}} \Big] \notag \\
\leq&~ {\wh{\sigma}|I'_{f_j}| }+ \frac{2\wh{\sigma}|I'_{f_j}| }{ 0.5\wh{\sigma} b+0.5 \wh{\sigma} f_j } \notag \\
\leq&~ 2 {\wh{\sigma}|I'_{f_j}| } \notag \\
\leq&~ 2 {\wh{\sigma}\cdot O(\Delta/B)} \notag \\
\leq&~ O(1/B) \label{eq:1_large_off_not_happen}
\end{align}
where the first steps are straightforward, the second step follows from the center of $I'_{f_j}$ is $f_j$, the length of the interval $I'_{f_j}$ is $|I'_{f_j}|$, and $a\in[c-b,c+b]\Rightarrow c\in[a-b,a+b]$,
the third step follows from $ \sigma \leq \wh{\sigma}$,
the forth step follows by applying Lemma \ref{lem:wrapping} with the following parameters setting: 
\begin{align*}
    \wt{T}=&~ 1 ,\notag \\
    \wt{\delta}=&~ \frac{1}{2B}+\frac{j}{B},\notag \\
    \wt{\eps}=&~ \wh{\sigma} |I'_{f_j}|/2,\notag \\
    A=&~  0.5\wh{\sigma} b+0.5 \wh{\sigma} f_j,\notag \\
    \wt{\sigma}=&~  \sigma b+ \sigma f_j,
\end{align*}
the fifth step follows from $ 0.5 b\geq F \geq f_j$ and $ 0.5 b \wh{\sigma} \geq 1$, the sixth step follows from Eq.~\eqref{eq:0_large_off_not_happen}, the last step follows from the definition of $\wh{\sigma}$. %

Similarly, we have that 
\begin{align}
\mathsf{Pr} [ -\frac{1}{2B\sigma}+\frac{j}{B\sigma}-b   {\pmod {1/\sigma }} \in I'_{f_j} {\pmod {1/\sigma }}] \leq&~ O(1/B)\label{eq:2_large_off_not_happen}
\end{align}

Note that $I_G$ is the edge of filter $G^{(j)}_{\sigma, b}$, under the meaning of module $1/\sigma$, the center of $G^{(j)}_{\sigma, b}$ is $\frac{j}{B\sigma}-b$, the length of $G^{(j)}_{\sigma, b}$ is $\frac{1}{B\sigma}$, the length of the edge is $s_G$. 
Moreover, $I_{f_j}$ is an interval center at $f_j$ and length $|\supp(\wh{H})|$. We can judge whether two interval have intersect $I_{f_j}\cap I_{G} \neq \emptyset$ by moving the length of one interval to another and judging whether $-\frac{1}{2B\sigma}+\frac{j}{B\sigma}-b   {\pmod {1/\sigma }}$, $\frac{1}{2B\sigma}+\frac{j}{B\sigma}-b   {\pmod {1/\sigma }}$ (the end point of $I_G$) contains in $I'_{f_j} $. 
By combining Eq.~\eqref{eq:1_large_off_not_happen} and Eq.~\eqref{eq:2_large_off_not_happen}, we have that
\begin{align}
    \mathsf{Pr} [I_{f_j}\cap I_{G} \neq \emptyset] \leq O(1/B)+O(1/B) = O(1/B).
    \label{eq:3_large_off_not_happen}
\end{align}

Therefore, by a union bound over all $j\in [k]$, we get that
\begin{align*}
   \mathsf{Pr} [(\cup_{j=1}^k I_{f_j})\cap I_{G} \neq \emptyset]
   \leq  \sum_{j=1}^k \mathsf{Pr} [ I_{f_j}\cap I_{G} \neq \emptyset] 
      \leq  \sum_{j=1}^k O(1/B) 
      \leq  0.01,
\end{align*}
where the first step is by union bound, the second step follows from Eq.~\eqref{eq:3_large_off_not_happen}, and the last step follows from $B=O(k)$. 
By the definitions of $I_{f_j}$ and $I_G$, it implies that with probability at least $0.99$, for any $j\in [k]$, and any $f\in f_j+\supp(\wh{H})$, $\wh{G}_{\sigma,b}^{(j)}(f)\notin [\delta/k, 1-\delta/k]$. 

The proof of the lemma is then completed.
\end{proof}

\begin{lemma}\label{lem:lo_and_fstar2Glarge}
For $x^*(t)$ be a $k$-Fourier-sparse signal. For frequency $f^*\in\supp(\wh{x}^*)$, let $j=h_{\sigma, b}(f^*)$ be the bin that $f^*$ hashed into. If Large Offset event not happens, then for $f\in \supp(\wh{x}^**\wh{H})$, 
\begin{align*}
    \wh{G}^{(j)}_{\sigma, b}(f) \in [1-\delta/k, 1]
\end{align*}
\end{lemma}
\begin{proof}

Since Large Offset event not happens, for $f\in\R$, 
\begin{align*}
    \wh{G}^{(j)}_{\sigma, b}(f) \geq 1-\delta/k \text{ or }\wh{G}^{(j)}_{\sigma, b}(f) \leq \delta/k.
\end{align*}

Since Large Offset event not happens and $j=h_{\sigma, b}(f^*)$, we have that for $f\in \supp(\wh{x}^**\wh{H})$,
\begin{align*}
    \wh{G}^{(j)}_{\sigma, b}(f) \geq 1-\delta /k.
\end{align*}

By Lemma \ref{lem:property_of_filter_G} Property \RN{1}, \RN{2}, and \RN{3}, we have that
\begin{align*}
    \wh{G}^{(j)}_{\sigma, b}(f) \in [ 1-\delta/k , 1].
\end{align*}
\end{proof}

\section{Filter in Time Domain}
\label{sec:t_filter}

\begin{figure}[!ht]
    \centering
    {\includegraphics[width=0.6\textwidth]{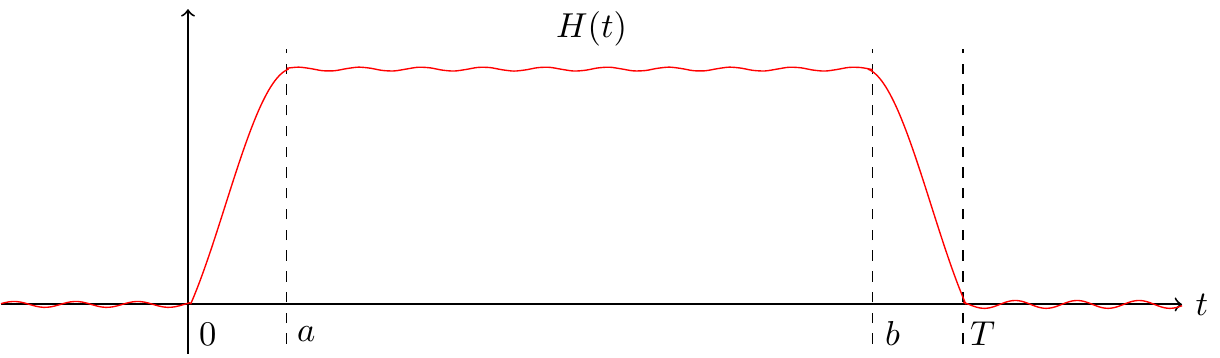}}
\caption{The filter $H(t)$ of time domain. We use decay region to refer $[0, a]$ and $[b, T]$. We use fluctuation region to refer $[a, b]$. 
}\label{fig:time_filter}
\end{figure}

In this section, we discuss the filter $H(t)$ of time domain, %
which is an analogous of the ideal filter $\rect_T(t)$. In the Fourier interpolation problem, we only care about the time duration $[0,T]$. Thus, applying the filter $\rect_T(t)$ to the observation signal $x(t)$ can cut-off the unobservable part and much simplify the analysis. Since $\rect_T(t)$ have an infinite band width, for efficient computation, we need to truncate $\rect_T(t)$'s frequency domain to a $\poly(k)/T$-length interval. However, the frequency truncation loses the high frequency components of the ideal filter $\rect_T(t)$, and the resulting filter $H(t)$ is no longer sharp around the boundary of $[0, T]$. More specifically, $H(t)$ is exponentially close to $1$ within $[T/\poly(k),T(1-1/\poly(k))]$, and exponentially close to $0$ outside $[0, T]$. 

We first provide the construction of $H(t)$ in \cite{cp19_icalp} and review some known properties (see Section \ref{sec:t_filter_pre}). Then, we discuss the normalization factor of the filter and provide a polynomial upper bound of it (see Section \ref{sec:t_filter_norm}). This bound is crucial for our fluctuation bound of the $H(t)$. Next,   we bound the fluctuation of $H(t)$ during a shrinking interval $[T/\poly(k),T(1-1/\poly(k))]$ and prove that $H(t)$ is exponentially close to $1$ in that range (see Section~\ref{sec:t_filter_range}). Furthermore,  we prove that $H(t)$ preserve the energy of Fourier sparse signal in the duration $[ 0, T]$ (see Section \ref{sec:_t_time_ene}).

\subsection{Time domain filter construction}
\label{sec:t_filter_pre}

We first introduce an growth rate bound from \cite{cp19_icalp}. This theorem bound the growth of Fourier sparse signal outside of the duration $[0, T]$ by an exponent function of base $t$. This bound in this theorem is high related to the size of the support set of $\wh{H}(f)$. 
\begin{theorem}[\cite{cp19_icalp}]\label{thm:bound_growth_around1}
There exists $S=O(k^2 \log k)$ such that for any $|t|>T$ and $g(t)=\sum_{j=1}^k v_j \cdot e^{2 \pi \i f_j t}$, $|g(t)|^2 \le \poly(k) \cdot \underset{x \in [-T,T]}{\E}[|g(x)|^2] \cdot |\frac{t}{T}|^{S}$.
\end{theorem}

The definition of the time domain filter $H(t)$ in \cite{cp19_icalp} is given in below. Intuitively, it uses some powers of $\sinc(t)$ to approximate $ \delta_0(t)*\rect_1(t) = \rect_1(t)$, thus one can get finite band-limit and good approximation at the same time. %
\begin{definition}[Definition 4.1 in \cite{cp19_icalp}]\label{def:filter_func_general}
Given an energy bound $R$ satisfying
\begin{align*}
    |x(t)|^2\lesssim R \|x(t)\|_T^2, ~~~\forall t\in [0, T] \text{ and $k$-Fourier sparse signal $x(t)$},
\end{align*}
the growth rate $S$ a power of two, $C\in 2\Z$, and $C_0\in \pi\Z$,%
we define the filter function:
\begin{align*}
    H_1(t)= s_0 \cdot \left( \sinc(C_0 R \cdot t)^{C \log R} \cdot \sinc\big(C_0 \cdot S \cdot t\big)^{C} \cdot \sinc\big(\frac{C_0 \cdot S}{2} \cdot t\big)^{2C} \cdots \sinc\big(C_0 \cdot t \big)^{C \cdot S} \right)*\rect_1(t),
\end{align*}
 where $s_0 \in \mathbb{R}^+$ is a parameter to normalize $H_1(0)=1$. 
Its Fourier transform is as follows:
\begin{align*}
    \wh{H}_1(f) = s_0 \cdot \left( \rect_{C_0 R}(f)^{*C \log R} * \rect_{C_0 \cdot S}(f)^{*C} * \rect_{\frac{C_0 \cdot S}{2}}(f)^{*2C} * \cdots * \rect_{C_0}(f)^{*C S} \right) \cdot \sinc(f/2).
\end{align*}
\end{definition}

We then state some basic properties of the time domain filter in \cite{cp19_icalp}. The following lemma bounds the support size of $\wh{H}_1(f)$:
\begin{lemma}[\cite{cp19_icalp}]\label{lem:supp_H}
Let $C_0=\Theta(C)$. we have that
\begin{align*}
    |\supp(\wh{H}_1(f))| = O(C^2 R \log R + C^2 S \log S). 
\end{align*}
\end{lemma}

The following theorem shows some time domain properties of the filter: 
\begin{theorem}[Theorem 4.2 in \cite{cp19_icalp}]\label{thm:HwithScaling}
  Let $R, S > 0$, let $C\in 2\Z$, $C_0\in\pi\Z$, $C_0=\Theta(C)$, and define
  $\alpha=(\frac{1}{2}+\frac{1.2}{\pi C_0 R})$. Consider any function
  $x$ satisfying the following two conditions:
 \begin{enumerate}
\item $\underset{t \in [-1,1]}{\sup} \big[ |x(t)|^2 \big] \le R \cdot \underset{t \in [-1,1]}{\E} \big[ |x(t)|^2 \big] $,
\item  $\poly(R) \cdot \underset{t \in [-1,1]}{\E} \big[ |x(t)|^2 \big] \cdot |{t}|^{S}$ for $t \notin [-1,1]$, %
\end{enumerate}
Then, we have that the filter function $H(t)=H_1\big(\alpha t \big)$  satisfies %
\begin{itemize}
\item Part 1. $\int_{-1}^1 |x(t) \cdot H\big(  t \big)|^2 \mathrm{d} t \ge 0.9 \int_{-1}^1 |x(t)|^2 \mathrm{d} t$,
\item Part 2. $\int_{-1}^1 |x(t) \cdot H\big(  t \big)|^2 \mathrm{d} t \ge 0.95 \int_{-\infty}^{\infty} |x(t) \cdot H\big(  t \big)|^2 \mathrm{d} t$,
\item Part 3. $|H(t)| \le 1.01$ for any $t$.
\end{itemize}
\end{theorem}

Throughout this paper, we denote $H(t)$ as the following re-scaling of $H_1(t)$: 
\begin{definition}
\label{def:effect_H_k_sparse}
Let $\alpha=(\frac{1}{2}+\frac{1.2}{\pi C_0 R})$. Let $H_1(t) $ be defined as in Definition \ref{def:filter_func_general}. The filter $H(t)$ is defined as: 
\begin{align*}
    H(t):=H_1(\alpha t).
\end{align*}
and setting $R=S=O(k^2)$,  $R=2^s$, where $s\in\Z_+$, %
$C=O(\log(1/\delta_1))$, $C\in2\Z$, $C_0=\Theta(C)$ and  $C_0\in\pi\Z$. 
\end{definition}

\subsection{Normalization factor of the filter}
\label{sec:t_filter_norm}

The goal of this section is to prove Lemma \ref{lem:bound_s_0}, an upper-bound for the normalization factor $s_0$. This lemma will be used later to ensure that the scaling factor will not break the exponential small fluctuation of Section \ref{sec:t_filter_range}. We note that the same result has been proved in \cite{cp19_icalp}, and we reprove it below for completeness.

\begin{lemma}[Lemma 7.2 in \cite{cp19_icalp}]
\label{lem:bound_s_0}
It holds that
\begin{align*}
    s_0 \leq O({CR \sqrt{{C \log R}}}).
\end{align*}
\end{lemma}
\begin{proof}

We first have that, %
\begin{align*}
   H_1(t) &~ = s_0 \cdot ( \sinc(C_0 R \cdot t)^{C \log R} \cdot \prod_{i=0}^{\log(S)} \sinc\big(\frac{C_0 \cdot S}{2^i} \cdot t\big)^{2^i \cdot C}) * \rect_1(t) \notag\\ 
  &~ = s_0 \cdot \int_{-\infty}^\infty \sinc(C_0 R \cdot \tau)^{C \log R} \cdot \prod_{i=0}^{\log(S)} \sinc\big(\frac{C_0 \cdot S}{2^i} \cdot \tau \big)^{2^i \cdot C} \cdot  \rect_1(t- \tau) \d\tau  \notag\\
  &~ = s_0 \cdot \int_{t-0.5}^{t+0.5} \sinc(C_0 R \cdot \tau)^{C \log R} \cdot \prod_{i=0}^{\log(S)} \sinc\big(\frac{C_0 \cdot S}{2^i} \cdot \tau \big)^{2^i \cdot C}  \d\tau , 
\end{align*}
where the first step follows from the definition of $H_1(t)$, the second step follows from the definition of the convolution, the third step follows from the definition of $\rect_2(t) $ function. 
Thus,
\begin{align*}
    H_1(0) = s_0 \cdot \int_{-0.5}^{+0.5} \sinc(C_0 R \cdot \tau)^{C \log R} \cdot \prod_{i=0}^{\log(S)} \sinc\big(\frac{C_0 \cdot S}{2^i} \cdot \tau \big)^{2^i \cdot C}  \d\tau
\end{align*}

Let $U := ({2 C_0 \log R})^{-1/2}$. We have that
\begin{align}
    &~\int_{-0.5}^{+0.5} \sinc(C_0 R \cdot \tau)^{C \log R} \cdot \prod_{i=0}^{\log(S)} \sinc\big(\frac{C_0 \cdot S}{2^i} \cdot \tau \big)^{2^i \cdot C}  \d\tau \notag \\
    \geq &~ \int_{-\frac{1}{\pi C_0R}}^{+\frac{1}{\pi C_0R}} \sinc(C_0 R \cdot \tau)^{C \log R} \cdot \prod_{i=0}^{\log(S)} \sinc\big(\frac{C_0 \cdot S}{2^i} \cdot \tau \big)^{2^i \cdot C}  \d\tau \notag \\
    = &~ \frac{1}{\pi C_0R} \int_{-{1}}^{+{1}} \sinc(\frac{\upsilon}{\pi})^{C \log R} \cdot \prod_{i=0}^{\log(S)} \sinc\big(\frac{\upsilon}{2^i \pi} \big)^{2^i \cdot C}  \d\upsilon \notag \\
    \geq &~  \frac{1}{\pi C_0R} \int_{-U}^{+U} \sinc(\frac{\upsilon}{\pi})^{C \log R} \cdot \prod_{i=0}^{\log(S)} \sinc\big(\frac{\upsilon}{2^i \pi} \big)^{2^i \cdot C}  \d\upsilon \notag \\
    \geq &~  \frac{1}{\pi C_0R} \int_{-U}^{+U} (1-\frac{\upsilon^2}{6})^{C \log R} \cdot \prod_{i=0}^{\log(S)} (1-\frac{\upsilon^2}{4^i 6})^{2^i \cdot C}  \d\upsilon \notag \\
    \geq &~  \frac{1}{\pi C_0R} \int_{-U}^{+U} (1-{C \log R} \cdot \frac{\upsilon^2}{6}-\sum_{i=0}^{\log(S)} C \cdot \frac{\upsilon^2}{2^i 6})  \d\upsilon \notag \\
    \geq &~ \frac{1}{\pi C_0R} \int_{-U}^{+U} (1-{C \log R} \cdot \frac{\upsilon^2}{3})  \d\upsilon \notag \\
    = &~ \frac{1}{\pi C_0R} (2U-{2 C \log R} \cdot \frac{U^3}{9}) \notag \\
    \geq &~ \frac{1}{\pi C_0R \sqrt{{2 C_0 \log R}}}  , \label{eq:H_1_0_right_part}
\end{align}
where the first step follows from $ R=O(k^2)$, $C_0=O(\log(1/\delta_1))$, $0.5 \geq \frac{1}{\pi C_0R}$,
second step follows from changing the variable $\nu=\pi C R\cdot \tau$, the third step follows from $ U < 1 $, the forth step follows from Fact~\ref{fac:bound_sinc}, the fifth step is follows from $ (1-a)(1-b) \geq 1-a-b$, the sixth step follows from $ \log(R) > 2$, the seventh step is straight forward, the eighth step follows from setting $ U = ({2 C_0 \log R})^{-1/2} $ and $C_0=\Theta(C)$.

As a result, we have that
\begin{align*}
    s_0 \leq  H_1(0) \cdot {\pi CR \sqrt{{2 C \log R}}} = {\pi CR \sqrt{{2 C \log R}}}
\end{align*}
where the first step follows from Eq.~\eqref{eq:H_1_0_right_part}, the second step follows from $ H_1(0) = 1 $. 

\end{proof}

\begin{fact}\label{fac:bound_sinc}
For any $t\in \mathbb{R}$,
\begin{align*}
    1 - \frac{(\pi t)^2}{3!} \le \sinc(t) \le 1.
\end{align*}
\end{fact}

\subsection{Fluctuation bound}
\label{sec:t_filter_range}
The idea filter $\rect_T(t)$ has a constant value $1$ in the interval $[0,T]$. Due to the frequency domain truncation in $H(t)$, it deviates from $\rect_T(t)$ with different magnitudes in different regions. In this section, we prove the Lemma \ref{lem:bound_H_1}, which shows that $H(t)$ is fluctuating near $1$ in the ``interior'' of $[0,T]$ (i.e., $[0+\frac{T}{\poly(k)},T-\frac{T}{\poly(k)}]$). 
It serves as an important tool for analyzing the error in our signal equivalent method. 

\begin{lemma}\label{lem:bound_H_1}

For filter $ H_1(t)$ defined in Definition \ref{def:filter_func_general} with the parameters $ C = \log(1/\delta_1)$, $C_0=\Theta(C)$, $R=S$,  and $ S= 2^s$ (where $s\in \Z_+$), $C_0\in \pi \Z$, we have that
\begin{align*}
    H_1(t) \in [1- \delta_1, 1], \forall |t| < 0.5-\frac{\pi}{C_0R}. 
\end{align*}
Moreover, $H(t)\in [1- \delta_1, 1]$ for any $t\in [\frac{T}{2}-\alpha^{-1}(\frac{1}{2}-\frac{\pi}{C_0R})\frac{T}{2}, \frac{T}{2}+\alpha^{-1}(\frac{1}{2}-\frac{\pi}{C_0R})\frac{T}{2}]$. 
\end{lemma}
\begin{proof}
The proof consists of two parts: upper bound and lower bound. For the upper bound, the idea is to compare the value of $H(t)$ with $H(0)$ by analyzing the gradient of $H(t)$. And the lower bound follows from directly estimating the integral of the product of $\sinc$ functions.

\paragraph{Upper bound:}
We have that $H_1(0) = 1$ by definition. We will show $H_1(t) \leq 1 $ by proving $H_1(t)$ is monotonically decreasing in $t$. %

We have that, %
\begin{align}
   H_1(t) &~ = s_0 \cdot ( \sinc(C_0 R \cdot t)^{C \log R} \cdot \prod_{i=0}^{\log(S)} \sinc\big(\frac{C_0 \cdot S}{2^i} \cdot t\big)^{2^i \cdot C}) * \rect_1(t) \notag\\ 
  &~ = s_0 \cdot \int_{-\infty}^\infty \sinc(C_0 R \cdot \tau)^{C \log R} \cdot \prod_{i=0}^{\log(S)} \sinc\big(\frac{C_0 \cdot S}{2^i} \cdot \tau \big)^{2^i \cdot C} \cdot  \rect_1(t- \tau) \d\tau  \notag\notag \\
  &~ = s_0 \cdot \int_{t-0.5}^{t+0.5} \sinc(C_0 R \cdot \tau)^{C \log R} \cdot \prod_{i=0}^{\log(S)} \sinc\big(\frac{C_0 \cdot S}{2^i} \cdot \tau \big)^{2^i \cdot C}  \d\tau , \label{eq:expand_H_1}
\end{align}
where the first step follows from the definition of $H_1(t)$, the second step follows from the definition of the convolution, the third step follows from the definition of $\rect_1(t) $ function.

Since $C \log R \in 2\Z$, $2^i \cdot C \in 2\Z$, we have that
\begin{align*}
    \sinc(C_0 R \cdot \tau)^{C \log R} \geq 0,
\end{align*}
and 
\begin{align*}
    \sinc\big(\frac{C_0 \cdot S}{2^i} \cdot \tau \big)^{2^i \cdot C} \geq 0.
\end{align*}

Moreover,
by setting $C_0=\pi \Z$,
we have that 
\begin{align}
   2 \mod{\frac{2\pi}{C_0}} = 0. \label{eq:sinc_periodic}
\end{align}

By Eq.~\eqref{eq:sinc_periodic}, we have that
\begin{align*}
    \sin\big(\frac{C_0 \cdot S}{2^i} \cdot (t+0.5) \big) = &~ \sin\big(\frac{C_0 \cdot S}{2^i} \cdot (t-0.5) \big), ~~~\forall i \in \{0,\cdots, \log(S)\},~~\text{and} \notag \\
    \sin(C_0 R \cdot (t+0.5)) = &~ \sin(C_0 R \cdot (t-0.5)) .
\end{align*}
Furthermore, for any $t > 0 $,
\begin{align*}
    \big(\frac{C_0 \cdot S}{2^i} \cdot (t+0.5) \big)^{-1} \leq &~ \big(\frac{C_0 \cdot S}{2^i} \cdot (t-0.5) \big)^{-1}, \forall i \in \{0,\cdots, \log(S)\}, ~~\text{and}\notag \\
    (C_0 R \cdot (t+0.5))^{-1} \leq &~ (C_0 R \cdot (t-0.5))^{-1}.
\end{align*}
Thus,
\begin{align}
    |\sinc\big(\frac{C_0 \cdot S}{2^i} \cdot (t+0.5) \big)| \leq &~ |\sinc\big(\frac{C_0 \cdot S}{2^i} \cdot (t-0.5) \big)|, ~~~\forall i \in \{0,\cdots, \log(S)\} ,~~\text{and}\notag \\
    |\sinc(C_0 R \cdot (t+0.5))| \leq &~ |\sinc(C_0 R \cdot (t-0.5))| \label{eq:sinc_mono}.
\end{align}

Then, we have that
\begin{align*}
   \frac{H_1'(t)}{s_0}  
   = &~ \sinc(C_0 R \cdot (t+0.5))^{C \log R} \cdot \prod_{i=0}^{\log(S)} \sinc\big(\frac{C_0 \cdot S}{2^i} \cdot (t+0.5) \big)^{2^i \cdot C}   \\
   &~\quad -  \sinc(C_0 R \cdot (t-0.5))^{C \log R} \cdot \prod_{i=0}^{\log(S)} \sinc\big(\frac{C_0 \cdot S}{2^i} \cdot (t-0.5) \big)^{2^i \cdot C} \\
   <&~ 0,
\end{align*}
where the first step is straight forward, the second step follows from Eq.~\eqref{eq:sinc_mono}.

Thus, $H_1(t)<H_1(0)=1$ for any $t>0$. 

Similarly, we also have that $H_1(t)<H_1(0)=1$ for any $t\leq 0$ since $H_1(t)$ is symmetric with respect to $t$. 

\paragraph{Lower bound:}

We have that, for any $|t| < 0.5 - \frac{  \pi}{C_0 R}$, 
\begin{align}
    &~ \int_{-\infty}^{t-0.5} \sinc(C_0 R \cdot \tau)^{C \log R} \cdot \prod_{i=0}^{\log(S)} \sinc\big(\frac{C_0 \cdot S}{2^i} \cdot \tau \big)^{2^i \cdot C} \d \tau \notag \\
    = &~  \int_{0.5-t}^{\infty} \sinc(C_0 R \cdot \tau)^{C \log (R)} \cdot \prod_{i=0}^{\log(S)} \sinc\big(\frac{C_0 \cdot S}{2^i} \cdot \tau \big)^{2^i \cdot C}  \d\tau  \notag \\
    \leq &~ \int_{0.5-t}^{\infty} \sinc(C_0 R \cdot \tau)^{C \log (R)}  \d\tau  \notag \\
    \leq &~ \int_{0.5-t}^{\infty} (C_0 R \cdot \tau)^{-C \log (R)}  \d\tau  \notag \\
    \leq &~ \int_{\frac{  \pi}{C_0R}}^{\infty} (C_0 R \cdot \tau)^{-C \log (R)}  \d\tau  \notag \\
    = &~ \frac{1}{C_0R}\int_{  \pi}^{\infty} \upsilon^{-C \log (R)}  \d\upsilon  \notag \\
    = &~ \frac{1}{C_0R} \frac{1}{C \log (R)-1} \pi^{-C \log (R)+1}   \notag \\
    \lesssim &~ \frac{1}{C_0^2R\log (R)}  \delta_1 ,\label{eq:bound_1side_H_1}
\end{align}
where the first step is straight forward, the second step follows from $ \sinc(x) \leq 1$, the third step is follows from $ \sinc(x) \leq 1 / x$, the forth step follows follows from $ 0.5-t \geq \frac{  \pi}{C_0R}$, the fifth step follows from $ \upsilon := CR \tau$, the sixth step is straight forward, the seventh step follows from $ \log(R ) > 1$, $ C \geq \log(1/\delta_1)$.

Hence, for any $t > 0$, 
\begin{align*}
    H_1(t) = &~ s_0 \cdot \int_{t-0.5}^{t+0.5} \sinc(C_0 R \cdot \tau)^{C \log R} \cdot \prod_{i=0}^{\log(S)} \sinc\big(\frac{C_0 \cdot S}{2^i} \cdot \tau \big)^{2^i \cdot C}  \d\tau  \\
    = &~  H_1(0) + s_0 \cdot \int_{1}^{t+0.5} \sinc(C_0 R \cdot \tau)^{C \log R} \cdot \prod_{i=0}^{\log(S)} \sinc\big(\frac{C_0 \cdot S}{2^i} \cdot \tau \big)^{2^i \cdot C}  \d\tau  \\
    &~ \quad - s_0 \cdot \int_{-1}^{t-0.5} \sinc(C_0 R \cdot \tau)^{C \log R} \cdot \prod_{i=0}^{\log(S)} \sinc\big(\frac{C_0 \cdot S}{2^i} \cdot \tau \big)^{2^i \cdot C}  \d\tau \\
    \geq &~  H_1(0) - s_0 \cdot \int_{-\infty}^{t-0.5} \sinc(C_0 R \cdot \tau)^{C \log R} \cdot \prod_{i=0}^{\log(S)} \sinc\big(\frac{C_0 \cdot S}{2^i} \cdot \tau \big)^{2^i \cdot C}  \d\tau \\
    \geq &~  H_1(0) - s_0 O(\frac{1}{C_0^2R\log (R)})  \delta_1   \\
    = &~  1 - s_0 O(\frac{1}{C_0^2R\log (R)} ) \delta_1  \\
    \geq &~ 1 - O(  \delta_1 ),
\end{align*}
where the first step follows from Eq.~\eqref{eq:expand_H_1}, the second step is straight forward, the third step follows from 
\begin{align*}
     \sinc(C_0 R \cdot \tau)^{C \log R} \cdot \prod_{i=0}^{\log(S)} \sinc\big(\frac{C_0 \cdot S}{2^i} \cdot \tau \big)^{2^i \cdot C} \geq 0, ~~~\forall \tau\in \R
\end{align*}
the forth step follows from Eq.~\eqref{eq:bound_1side_H_1}, the fifth step follows from $ H_1(0) = 1 $, the sixth step follows from Lemma \ref{lem:bound_s_0}. 

By re-scaling $\delta_1$, we get that $H(t)\geq 1-\delta_1$ for any $|t| < 0.5 - \frac{  \pi}{C_0 R}$.

The lemma then follows from the upper and lower bounds.

\end{proof}

\subsection{Energy preserving of the time domain filter}
\label{sec:_t_time_ene}

In this section, we show the properties of $H(t)$ that we use in the rest of the paper. 

We first prove Lemma \ref{lem:property_of_filter_H}, which summarizes the results in above sections and prove the energy preserving property of $H(t)$. %

\begin{lemma}\label{lem:property_of_filter_H}
The filter function $(H(t),\widehat{H}(f))$ has the following properties:
\begin{eqnarray*}
&\mathrm{Property~\RN{1}} : & |H(t)|\leq 1.01 , ~ \forall t\in\R \\ %
&\mathrm{Property~\RN{2}} : & 1- \delta_1 \leq H(t) \leq 1, ~ \forall |t| < \alpha^{-1}(\frac{1}{2}-\frac{\pi}{CR}) \\ %
&\mathrm{Property~\RN{3}} : & |\supp(\widehat{H}(f) ) |\leq O(k^2 \log^2 (k) \log^2(1/\delta_1)) \\ %
&\mathrm{Property~\RN{4}} : &\int_{-\infty}^{+\infty} \bigl|x^*(t) \cdot H(t) \cdot (1- \rect_2(t) ) \bigr|^2 \mathrm{d} t < 0.1 \int_{-\infty}^{+\infty} | x^*(t) \cdot \rect_2(t) |^2 \mathrm{d} t \\
&\mathrm{Property~\RN{5}} : &\int_{-\infty}^{+\infty} |x^*(t) \cdot H(t) \cdot \rect_2(t) |^2 \mathrm{d} t \in  [0.9 , 1.1]\cdot \int_{-\infty}^{+\infty} |x^*(t) \cdot \rect_2(t) |^2 \mathrm{d} t
\end{eqnarray*}
\end{lemma}

\begin{proof}
We prove each of the five properties in below.

\paragraph{Property I:}
By Theorem \ref{thm:HwithScaling} Part 3, we have that 
\begin{align*}
    |H(t)|\leq 1.01 .
\end{align*}

\paragraph{Property II:}
By Lemma \ref{lem:bound_H_1}, we have that
\begin{align*}
1- \delta_1 \leq H(t) \leq 1, ~ \forall |t| < \alpha^{-1}(\frac{1}{2}-\frac{\pi}{CR}). 
\end{align*}

\paragraph{Property III:}
By the $k$-Fourier-sparse signals' energy bound (Theorem \ref{thm:energy_bound}), we have that
\begin{align*}
    R = O(k^2).
\end{align*}
By Theorem \ref{thm:bound_growth_around1}, we have that
\begin{align*}
      S=O(k^2 \log k).
\end{align*}
Then, by Lemma \ref{lem:supp_H}, we have that
\begin{align*}
    |\supp(\wh{H}(f))| =&~ C^2R \log R + C^2S \log S \\
    =&~ O(\log(1/\delta_1))^2 \cdot O(k^2 \log k \log (k^2 \log k)) \\
    =&~  O(k^2 \log^2 k \log^2(1/\delta_1)). 
\end{align*}

\paragraph{Property IV:}
We have that
\begin{align*}
&~ \int_{-\infty}^{+\infty} \bigl|x^*(t) \cdot H(t) \cdot (1- \rect_1(t) ) \bigr|^2 \mathrm{d} t \\
= &~ \int_{-\infty}^{+\infty} \bigl|x^*(t) \cdot H(t)  \bigr|^2 \mathrm{d} t-\int_{-1}^{1} \bigl|x^*(t) \cdot H(t)  \bigr|^2 \mathrm{d} t \\
\leq &~ 0.06 \int_{-1}^{+1} \bigl|x^*(t) \cdot H(t)  \bigr|^2 \mathrm{d} t \\
\leq &~ 0.1 \int_{-1}^{+1} \bigl|x^*(t) \bigr|^2 \mathrm{d} t \\
=&~ 0.1 \int_{-\infty}^{+\infty} | x^*(t) \cdot \rect_1(t) |^2 \mathrm{d} t, 
\end{align*}
where the first step is straight forward, the second step follows from Theorem \ref{thm:HwithScaling} Part 2, the third step follows from Theorem  \ref{thm:HwithScaling} Part 3, the forth step is straight forward.

\paragraph{Property V:}
We first prove the upper bound:
\begin{align*}
    &~ \int_{-\infty}^{+\infty} |x^*(t) \cdot H(t) \cdot \rect_1(t) |^2 \mathrm{d} t \\
    = &~ \int_{-1}^{+1} |x^*(t) \cdot H(t)  |^2 \mathrm{d} t \\
    \leq &~ 1.1 \int_{-1}^{+1} |x^*(t)   |^2 \mathrm{d} t \\
    = &~ 1.1 \int_{-\infty}^{+\infty} |x^*(t) \cdot \rect_1(t) |^2 \mathrm{d} t, 
\end{align*}
where the first step is straight forward, the second step follows from Theorem  \ref{thm:HwithScaling} Part 3, the third step is straight forward. 

Then, we prove the lower bound:
\begin{align*}
    &~ \int_{-\infty}^{+\infty} |x^*(t) \cdot H(t) \cdot \rect_1(t) |^2 \mathrm{d} t \\
    = &~ \int_{-1}^{+1} |x^*(t) \cdot H(t)  |^2 \mathrm{d} t \\
    \geq &~ 0.9 \cdot \int_{-1}^{+1} |x^*(t)  |^2 \mathrm{d} t \\
    = &~ 0.9 \cdot \int_{-\infty}^{+\infty} |x^*(t) \cdot \rect_1(t) |^2 \mathrm{d} t
\end{align*}
where the first step is straight forward, the second step follows from Theorem  \ref{thm:HwithScaling} Part 1, the third step is straight forward. 

\end{proof}

The following lemma 
bounds the length of the fluctuation region (where $H(t)$ is close to $1$) in the time domain. %

\begin{lemma}
\label{lem:relation_R_L_T}
Let $\Delta = k |\supp(\wh{H}(f))| $, $\beta = O(1/ \Delta)$, 
$L= \frac{T}{2}-\alpha^{-1}(\frac{1}{2}-\frac{\pi}{C_0R})\frac{T}{2}, R=\frac{T}{2}+\alpha^{-1}(\frac{1}{2}-\frac{\pi}{C_0R})\frac{T}{2}-\beta $,
we have that
\begin{align*}
    T - k^2 (T+L-R) \eqsim T, 
\end{align*}
and 
\begin{align*}
    R-L \eqsim T.
\end{align*}
\end{lemma}
\begin{proof}
Let $U:=[L,R]$. By Lemma \ref{lem:bound_H_1}, we have that for any $t_0\in U$,
\begin{align*}
    H(t) > 1-\delta_1, \forall t\in [t_0,t_0+\beta].
\end{align*}

We have that
\begin{align*}
    R- L =&~ |U| \\
    = &~ ((\frac{1}{2}+\frac{1.2}{\pi C R})^{-1}\cdot (\frac{1}{2}-\frac{\pi}{CR}) - \beta )\cdot T \\
    \eqsim&~ T,
\end{align*}
where the first step follows from the definition of $ L, R$, the second step follows from Lemma \ref{lem:property_of_filter_H} Property \RN{2}, the third step follows from $ \Delta, CR \gg 1$.

We have that
\begin{align}
    T+L-R =&~ T- |U| \notag\\
    = &~ (1-(\frac{1}{2}+\frac{1.2}{\pi C R})^{-1}\cdot (\frac{1}{2}-\frac{\pi}{CR}) )\cdot T+\beta T \notag\\
    = &~ \frac{2\pi + {2.4}/{\pi}}{CR+{2.4}/{\pi}}\cdot T + \beta T ,\label{eq:TLR_000}
\end{align}
where the first step follows from the definition of $ L, R$, the second step follows from Lemma \ref{lem:property_of_filter_H} Property \RN{2}, the third step is straight forward.

Then, we have that
\begin{align*}
    T-k^2(T+L-R) = T - k^2 \cdot( \frac{2\pi + {2.4}/{\pi}}{CR+{2.4}/{\pi}}+\beta )\cdot T \eqsim T,
\end{align*}
where the first step follows from Eq.~\eqref{eq:TLR_000}, the second step follows from $ C=O(\log(1/\delta_1))$, $R=k^2$, $ k^2 \beta < 1/ k $. 

\end{proof}

\section{Ideal Filter Approximation}
\label{sec:approx_f_SFS}

As we discussed in previous sections, the filtered signal $z^{(j)}(t)=(x\cdot H)*G_{\sigma, b}^{(j)}(t)$ is the signal in the $j$-th bin by the \text{HashToBins} procedure. In this section, we consider an approximation of the frequency domain filter $G_{\sigma, b}^{(j)}$ by the \emph{ideal filter} $I_{\sigma,b}^{(j)}(t)$ defined by its Fourier transform:
\begin{align}\label{eq:def_ideal_filter}
    \wh{I}_{\sigma,b}^{(j)}(f):=\begin{cases}
1, &~ \wh{G}^{(j)}_{\sigma,b}(f) > 1-\delta_1 \\
0, &~ \text{otherwise}
\end{cases}    
\end{align}
Intuitively, $I_{\sigma,b}^{(j)}$ is ``denoising''  the frequency domain filter $\wh{G}^{(j)}_{\sigma,b}$ in the sense that it rounds the heavy Fourier coefficients of  $\wh{G}^{(j)}_{\sigma,b}$ to 1 and rounds the remaining Fourier coefficients to 0. The main purpose of this section is to show that $(x\cdot H)*I_{\sigma,b}^{(j)}$ is a good approximation of $z^{(j)}$. For simplicity, we will use $I$ to denote $I_{\sigma,b}^{(j)}$ when $\sigma, b, j$ are clear from context. 

We first show a commuting property of the ideal filter (see Section~\ref{sec:approx_f_SFS:swap}). Then, we derive the approximation error bound of the ideally filtered signals (see Section~\ref{sec:approx_f_SFS:approximation}).

\subsection{Swap the order of filtering}\label{sec:approx_f_SFS:swap}
We first prove a good property of the ideal filter that $I_{\sigma,b}^{(j)}$ ``commutes'' with the time domain filter $H$ with high probability over the random hashing function. 

\begin{lemma}\label{lem:xIH2xHI}
Let $\delta_1$ be the $\delta$ defined in Lemma \ref{lem:property_of_filter_G}.%
Let $H$ be defined as in Definition~\ref{def:effect_H_k_sparse}, ${G}^{(j)}_{\sigma,b}$ be defined as in Definition \ref{def:G_j_sigma_b}. Let the ideal filter $I=I_{\sigma,b}^{(j)}$ be defined as in Eq.~\eqref{eq:def_ideal_filter}. %

Then, for any $x\in {\cal F}_{k,F}$, with probability $0.9$ over the choice of $(\sigma,b)$, for any $j\in [B]$, %
\begin{align*}
 (x\cdot H)*I(t) = (x*I)(t)\cdot H(t)~~~\forall t\in \mathbb{R}. 
\end{align*}
\end{lemma}
\begin{proof}

By Fourier transformation, we have
\begin{align*}
 (x\cdot H)*I(t) = &~ \int_{-\infty}^\infty (\wh{x}*\wh{H})(f)\cdot \wh{I}(f)\cdot \exp(2\pi\i ft) \d f.
\end{align*}

We will show that $(\hat{x}*\wh{H})(f)\cdot \hat{I}(f)=((\hat{x}\cdot \hat{I})*\hat{H})(f)$ with high probability.

On the one hand,
\begin{align*}
    (\wh{x}*\wh{H})(f)\cdot \wh{I}(f) = &~ \sum_{j=1}^k v_j \cdot (\delta_{f_j}*\hat{H})(f)\cdot \wh{I}(f)\\
    = &~ \sum_{j=1}^k v_j \cdot \hat{H}(f-f_j)\cdot \wh{I}(f)\\
    = &~ \sum_{j=1}^k v_j \cdot \hat{H}(f-f_j)\cdot {\bf 1}_{f\in \supp(\wh{I})},
\end{align*}
where the first step follows from $\wh{x}(f)=\sum_{j=1}^k v_j \cdot \delta_{f_j}(f)$, the second step follows from the convolution property of Delta function. By Lemma~\ref{lem:large_off_not_happen}, we get that with probability at least $0.9$, for any $j\in [k]$ and any $f\in \supp(\hat{H})+f_j$, either $\hat{G}_{\sigma,b}^{(j)}<\delta_1$ or $\hat{G}_{\sigma,b}^{(j)}>1-\delta_1$. In other words, either $f_j+\supp(\hat{H})\subseteq \supp(\hat{I})$ or $f_j+\supp(\hat{H})\cap \supp(\hat{I})=\emptyset$. Since $0\in \supp(\hat{H})$, we get that for any $f\in f_j+\supp(\hat{H})$,
\begin{align*}
    f\in \supp(\wh{I})~~\Longleftrightarrow~~f_j\in \supp(\hat{I}).
\end{align*}
Hence, we have
\begin{align*}
    (\wh{x}*\wh{H})(f)\cdot \wh{I}(f) = \sum_{j\in [k]:f_j\in \supp(\hat{I})}v_j\cdot \hat{H}(f-f_j).
\end{align*}

On the other hand,
\begin{align*}
    (\hat{x}\cdot \hat{I})*\hat{H}(f)
    = &~ \sum_{j\in [k]:f_j\in \supp(\hat{I})} v_j\cdot \delta_{f_j} * \hat{H}(f)\\
    = &~ \sum_{j\in [k]:f_j\in \supp(\hat{I})} v_j\cdot \hat{H}(f-f_j)\\
    = &~ (\wh{x}*\wh{H})(f)\cdot \wh{I}(f).
\end{align*}

Therefore,
\begin{align*}
(x\cdot H)*I(t) = &~ \int_{-\infty}^\infty (\wh{x}*\wh{H})(f)\cdot \wh{I}(f)\cdot \exp(2\pi\i ft) \d f\\
 =&~ \int_{-\infty}^\infty (\wh{x}\cdot \wh{I})*\wh{H}(f)\cdot \exp(2\pi\i ft) \d f \notag \\
=&~  (x*I)(t)\cdot H(t),  
\end{align*}
where the last step follows from the definition of Fourier transform.

The lemma is then proved.

\end{proof}

\subsection{Approximation error bounds}\label{sec:approx_f_SFS:approximation}
We analyze the approximation error due to replacing $\hat{G}_{\sigma,b}^{(j)}$ with the ideal filter $I_{\sigma,b}^{(j)}$ defined by Eq.~\eqref{eq:def_ideal_filter}.  The following lemma gives a point-wise error bound. 

\begin{lemma}\label{lem:xHG_sub_XHI_is_small}
Let $\delta_1$ be defined as in Lemma \ref{lem:property_of_filter_G}. Let $H$ be defined as in Definition~\ref{def:effect_H_k_sparse}, ${G}^{(j)}_{\sigma,b}$ be defined as in Definition \ref{def:G_j_sigma_b}. 

For any $x\in {\cal F}_{k,F}$, we have that with probability $0.9$, for any $j\in [B]$,
\begin{align*}
|(x\cdot H)*{G}^{(j)}_{\sigma,b}(t) - (x\cdot H)*I(t)| \lesssim \delta_1  \sqrt{T|S|} \cdot  \|x(t) \|_T~~~\forall t\in \R.
\end{align*}
\end{lemma}
\begin{proof}
Let $S:=\supp(\wh{x}*\wh{H})$ be defined as the support set of $\wh{x}*\wh{H}$. 
Then $|S|\leq \Delta$. 

We have that
\begin{align}
&~|(x\cdot H)*{G}^{(j)}_{\sigma,b}(t) - (x\cdot H)*I(t)|\\
=&~ | (x\cdot H)*({G}^{(j)}_{\sigma,b} - I)(t)| \notag \\
=&~ \Big| \int_{-\infty}^\infty (\wh{x}* \wh{H})(f) \cdot (\wh{G}^{(j)}_{\sigma,b} - \hat{I})(f)\cdot e^{2\pi\i f t} \d f\Big| \notag \\
\leq&~  \int_{-\infty}^\infty |(\wh{x}* \wh{H})(f) \cdot (\wh{G}^{(j)}_{\sigma,b} - I)(f)|\d f\notag \\
=&~  \int_S |(\wh{x}* \wh{H})(f) \cdot (\wh{G}^{(j)}_{\sigma,b} - I)(f) |\d f\notag \\
\leq&~  \int_{S} |(\wh{x}* \wh{H})(f) \cdot \delta_1 |\d f\notag \\
\leq&~   \delta_1 \sqrt{|S|}\cdot \sqrt{\int_{-\infty}^\infty |(\wh{x}* \wh{H})(f)  |^2\d f}\notag \\
=&~   \delta_1\sqrt{|S|} \cdot \sqrt{\int_{-\infty}^\infty |(x\cdot H)(t)  |^2\d t }\notag \\
\lesssim &~   \delta_1  \sqrt{T|S|} \cdot  \|(x\cdot H)(t)  \|_T\notag \\
\lesssim &~   \delta_1  \sqrt{T|S|} \cdot  \|x(t) \|_T
\end{align}
where the first step is straight forward, the second step follows from the definition of Fourier transform, the third step follows from triangle equality, the forth step follows from the definition of $S$. For the fifth step, by Lemma~\ref{lem:large_off_not_happen} that with probability at least 0.9, the Large Offset event does not happen (i.e., for any $f\in S$, either $\wh{G}^{(j)}_{\sigma,b}(f)<\delta_1$ or $\wh{G}^{(j)}_{\sigma,b}(f)>1-\delta_1$). Then, by Lemma~\ref{lem:property_of_filter_G}, we know that $-\delta_1\leq \wh{G}^{(j)}_{\sigma,b}(f)\leq 1$. Thus, we get that $|(\wh{G}^{(j)}_{\sigma,b}-\hat{I})(f)|\leq \delta_1$. The sixth step follows from Cauchy–Schwarz inequality, the seventh step follows from Parseval's theorem, the eighth step follows from Lemma \ref{lem:property_of_filter_H} Property \RN{4} and \RN{5}, 
the last step follows from Lemma \ref{lem:property_of_filter_H} Property \RN{5}. 
\end{proof}

The following lemma gives a $T$-norm bound for the approximation error.

\begin{lemma}\label{lem:xHI_sub_XHG_T_norm_is_small}
Let $\delta_1$ be defined as in Lemma \ref{lem:property_of_filter_G}. Let $H$ be defined as in Definition~\ref{def:effect_H_k_sparse}, ${G}^{(j)}_{\sigma,b}$ be defined as in Definition \ref{def:G_j_sigma_b}. 

Then, for any $x\in {\cal F}_{k,F}$, with probability $0.9$, for any $j\in [B]$,
\begin{align*}
\int_{-\infty}^\infty | (x\cdot H)*I(t) - (x\cdot H)*G^{(j)}_{\sigma, b}(t) |^2 \d t \lesssim \delta_1^2 T \|x(t) \|_T^2.
\end{align*}
In particular,
\begin{align*}
\| (x\cdot H)*I(t) - (x\cdot H)*G^{(j)}_{\sigma, b}(t) \|_T \lesssim \delta_1 \|x(t) \|_T.
\end{align*}
\end{lemma}
\begin{proof}
Let $S:=\supp(\wh{x}*\wh{H})$ be defined as the support set of $\wh{x}*\wh{H}$.

We have that
\begin{align*}
&~T \| (x\cdot H)*I(t) - (x\cdot H)*G^{(j)}_{\sigma, b}(t) \|_T^2 \\
=&~ \int_0^T |(x\cdot H)*I(t) - (x\cdot H)*G^{(j)}_{\sigma, b}(t)|^2\d t \notag \\
\leq &~ \int_{-\infty}^\infty |(x\cdot H)*I(t) - (x\cdot H)*G^{(j)}_{\sigma, b}(t)|^2\d t \notag \\
\leq &~ \int_{-\infty}^\infty |(\wh{x}*\wh{H})(f)\cdot(\wh{I}(f) -\wh{G}^{(j)}_{\sigma, b}(f))|^2\d f \notag \\
= &~ \int_{S} |(\wh{x}*\wh{H})(f)\cdot(I(f) -\wh{G}^{(j)}_{\sigma, b}(f))|^2\d f \notag \\
\leq &~ \int_{S} |(\wh{x}*\wh{H})(f)\cdot\delta_1 |^2\d f \notag \\
\leq &~ \int_{-\infty}^\infty |(\wh{x}*\wh{H})(f)\cdot\delta_1 |^2\d f \notag \\
= &~ \delta_1^2\int_{-\infty}^\infty |(x\cdot H)(t) |^2\d t \notag \\
\lesssim&~ \delta_1^2 T \|x(t) \|_T^2
\end{align*}
where the first step follows from the definition of the norm, the second step is straight forward, the third step follows from Parseval's theorem, the forth step follows from the definition of $S$, the fifth step follows from Lemma~\ref{lem:large_off_not_happen} and  Lemma~\ref{lem:property_of_filter_G}, the sixth step is straight forward, the seventh step follows from Parseval's theorem, the eighth step follows from Lemma \ref{lem:property_of_filter_H} Property \RN{4} and Property \RN{5}.

\end{proof}

\section{Concentration Property of the Filtered Signal}
\label{sec:concentr}

Recall that a frequency $f^*$ is \emph{heavy} if it satisfies the following condition:
\begin{align*}
    \int_{f^*-\Delta_h}^{f^*+\Delta_h} | \widehat{x^*\cdot H}(f) |^2 \mathrm{d} f \geq  T\N^2/k.
\end{align*}
In this section, we consider the filtered signal in a hashing bin that contains a heavy frequency; that is, $z(t)=(x^*\cdot H)*\wh{G}_{\sigma,b}^{(j)}$ where $j=h_{\sigma,b}(f^*)$ is the index of the bin containing $f^*$. We will prove that $z(t)$ form a one-cluster signal around $f^*$, which means that in the frequency domain most energy are concentrated around $f^*$, and in the time domain, most energy are contained in the observation window $[0,T]$. The formal definition are given as follows:
\begin{definition}[$(\epsilon,\Delta)$-one-cluster signal]\label{def:one_cluster}
We say that a signal $z(t)$ is an $(\epsilon,\Delta)$-one-cluster signal around $f_0$ if and only if $z(t)$ and $\wh{z}(f)$ satisfy the following two properties:
\begin{eqnarray*}
\mathrm{Property ~\RN{1}} &:& \int_{f_0-\Delta}^{f_0+\Delta} | \widehat{z}(f) |^2 \mathrm{d} f \geq (1-\epsilon) \int_{-\infty}^{+\infty} | \widehat{z}(f) |^2 \mathrm{d} f \\
\mathrm{Property ~\RN{2}} &:& \int_0^T | z(t) |^2 \mathrm{d} t \geq (1-\epsilon) \int_{-\infty}^{+\infty} |z(t) |^2 \mathrm{d} t.
\end{eqnarray*}
\end{definition}

We first prove the energy preservation in the time domain: 
\begin{lemma}[Time domain energy preservation]\label{lem:full_proof_of_3_properties_true_for_z} 
Let $\Delta_h=|\supp(\wh{H})|$. Let $f^*$ satisfy
\begin{align*}
    \int_{f^*-\Delta_h}^{f^*+\Delta_h} | \widehat{x^*\cdot H}(f) |^2 \mathrm{d} f \geq  T\N^2/k
\end{align*}
 and $\wh{z} = \widehat{x^* \cdot H} \cdot \widehat{G}^{ (j)}_{\sigma,b}$ where $j=h_{\sigma,b}(f^*)$. 
Suppose the Large Offset event does not happen.
Then, we have that, 
\begin{align*}
      \int_{-\infty}^{+\infty} |z(t) |^2 \mathrm{d} t \leq 1.35 \int_{0}^{T} | z(t) |^2 \mathrm{d} t. 
\end{align*}
\end{lemma}
\begin{proof}
Let $I(f)$ be the ideal filter defined by Eq.~\eqref{eq:def_ideal_filter}.

We first have
\begin{align*}
    \|z(t)\|_{L_2} = &~ \|(x^*\cdot H)(t) * G^{(j)}_{\sigma, b} (t)\|_{L_2}\\
    \leq &~ \|(x^*\cdot H)(t) * I (t)\|_{L_2} + \|(x^*\cdot H)(t) * (I-G^{(j)}_{\sigma, b}) (t)\|_{L_2},
\end{align*}
where the second step follows from triangle inequality.

Then, we bound the two terms separately.

For the first term, by Lemma \ref{lem:xIH2xHI}, if the Large Offset event does not happen, we have that
\begin{align}
    (x^*\cdot H)*I(t) = (x^**I)(t)\cdot H(t). \label{eq:lem:xIH2xHI}
\end{align}
It implies that
\begin{align*}
    \|(x^*\cdot H)(t) * I (t)\|_{L_2} = &~ \|(x^** I)(t) \cdot H (t)\|_{L_2}
\end{align*}
Let $y(t):=(x^**I)(t)$. It's easy to see that $y(y)$ is $k$-Fourier-sparse. Then, we have
\begin{align}
&~ \int_{-\infty}^\infty |y(t) \cdot H(t)|^2 \d t  \notag \\
= &~ \int_0^T |y(t) \cdot H(t)|^2 \d t + \int_{[-\infty,\infty]\backslash [0, T]} |y(t) \cdot H(t)|^2 \d t  \notag \\
\leq &~ \int_0^T |y(t) \cdot H(t)|^2 \d t + 0.1 \int_0^T |y(t)|^2 \d t \notag \\
\leq &~ 1.1 \int_0^T |y(t) |^2 \d t  \notag \\
\leq &~ 1.3 \int_0^T |y(t) \cdot H(t)|^2 \d t,  \label{eq:xH_infty_T}
\end{align}
where the first step is straight forward, the second step follows from Lemma \ref{lem:property_of_filter_H} Property \RN{4}, the third step follows from Lemma \ref{lem:property_of_filter_H} Property \RN{5}, the forth step follows from Lemma \ref{lem:property_of_filter_H} Property \RN{5}. 
Hence,
\begin{align*}
    \|(x^*\cdot H)(t) * I (t)\|_{L_2} = \|y(t)\cdot H(t)\|_{L_2}\leq \sqrt{1.3T}\cdot \|(x^**I)(t)\cdot H(t)\|_T.
\end{align*}
By Eq.~\eqref{eq:lem:xIH2xHI} again, we can swap the order of $I$ and $H$ and obtain:
\begin{align*}
    \|(x^*\cdot H)(t) * I (t)\|_{L_2} \leq \sqrt{1.3T}\cdot \|(x^*\cdot H)(t)* I(t)\|_T.
\end{align*}

For the second term, by Lemma \ref{lem:xHI_sub_XHG_T_norm_is_small}, we have that
\begin{align}
    \| (x^*\cdot H)*I(t) - (x^*\cdot H)*G^{(j)}_{\sigma, b}(t) \|_{L_2} \lesssim \delta_1 \sqrt{T}\|x^*(t) \|_T.
    \label{eq:lem:xHI_sub_XHG_T_norm_is_small_1}
\end{align}

Therefore, we get that
\begin{align*}
    \|z(t)\|_{L_2} \leq &~ \sqrt{1.3T}\cdot \|(x^*\cdot H)(t)* I(t)\|_T + O(\delta_1 \sqrt{T}  \|x^*(t) \|_T)\\
    \leq&~ \sqrt{1.3T}\cdot \| (x^*\cdot H) * G^{(j)}_{\sigma, b}(t) \|_T + \sqrt{1.3T}\cdot\| (x^*\cdot H) * (I-G^{(j)}_{\sigma, b})(t) \|_T  + O(\delta_1 \sqrt{T} \|x^*(t) \|_T)\\
\leq&~ \sqrt{1.3T}\cdot \|z(t)\|_T + O(\delta_1 \sqrt{T} \|x^*(t) \|_T),
\end{align*}
where the second step follows from triangle inequality, and the last step follows from Lemma \ref{lem:xHI_sub_XHG_T_norm_is_small} again.

We claim that the second term can be bounded by $o(1)\cdot \|z(t)\|_T$. 
Indeed, we have
\begin{align}
  &~ \int_{-\infty}^\infty |(x^*\cdot H)*G^{(j)}_{\sigma, b}(t) |^2 \d t \notag\\
 = &~ \int_{-\infty}^\infty |(\wh{x}^**\wh{H})\cdot \wh{G}^{(j)}_{\sigma, b}(f) |^2 \d f \notag\\
 \geq &~ \int_{f^*-\Delta_h}^{f^*+\Delta_h} |(\wh{x}^**\wh{H})\cdot \wh{G}^{(j)}_{\sigma, b}(f) |^2 \d f \notag\\
 \gtrsim &~ \int_{f^*-\Delta_h}^{f^*+\Delta_h} |(\wh{x}^**\wh{H})(f) |^2 \d f \notag\\
 \gtrsim &~ \int_{f^*-\Delta_h}^{f^*+\Delta_h} |(\wh{x}^**\wh{H})(f) |^2 \d f \notag\\
 \geq &~ \frac{T \delta \|x^*\|_T^2}{k} \label{eq:xstarTnorm_xstarHGinfinity}
\end{align}
where the first step follows from Parseval's theorem, the second step is straightforward, the third step follows from Lemma \ref{lem:large_off_not_happen} and Lemma \ref{lem:property_of_filter_G} Property \RN{1}, 
the forth step follows from our assumption that there exists a heavy frequency $f^*$ hashing into the $j$-th bin, the fifth step follows from $f^*$ satisfying 
\begin{align*}
    \int_{f^*-\Delta}^{f^*+\Delta} | \widehat{x^*\cdot H}(f) |^2 \mathrm{d} f \geq  T\N^2/k.
\end{align*}
Thus, $\|x^*\|_T\leq O(\sqrt{k/(T\delta)})\|z(t)\|_{L_2}$ and we have
\begin{align*}
    O(\delta_1\sqrt{T} \|x^*\|_T) = O\Big(\delta_1\sqrt{\frac{k}{\delta}} \|z(t)\|_{L_2}\Big)\leq O\Big(\sqrt{\frac{\delta}{k}}\Big)\|z(t)\|_{L_2}=o(1)\cdot \|z(t)\|_{L_2},
\end{align*}
where the second step follows from $\delta_1\leq \delta / k$.

Finally, we have
\begin{align*}
    \|z(t)\|_{L_2} \leq \sqrt{1.3T}\cdot \|z(t)\|_T + o(1)\cdot \|z(t)\|_{L_2},
\end{align*}
which implies that
\begin{align*}
    \int_{-\infty}^{+\infty} |z(t) |^2 \mathrm{d} t \leq 1.35 \int_{0}^{T} | z(t) |^2 \mathrm{d} t.
\end{align*}
The lemma is then proved.

\end{proof}

We next show the frequency domain energy concentration in the following lemma. Together with Lemma~\ref{lem:full_proof_of_3_properties_true_for_z}, we conclude that $z(t)$ is a one-cluster signal.
\begin{lemma}[Frequency domain energy concentration]\label{lem:z_satisfies_two_properties}
Let $x^*$ be a $k$-Fourier-sparse signal. Let $f^*\in [-F,F]$ satisfy  the following property:
\begin{equation}\label{eq:heavyfrequency_000}
\int_{f^*-\Delta_h}^{f^*+\Delta_h} | \widehat{x^*\cdot H}(f) |^2 \mathrm{d} f \geq  T\N^2/k.
\end{equation}
Let $\sigma,b$ be the parameter of the hashing function. Suppose that Large Offset event not happened and $f^*$ is well-isolated.  Let $j=h_{\sigma,b}(f^*)$ be the bucket that $f^*$ maps to under the hash such that $z=(x^* \cdot H)*G^{(j)}_{\sigma,b}$ and $\wh{z}=\wh{x^* \cdot H} \cdot \wh{G}^{(j)}_{\sigma,b}$. Then, we have
\begin{align*}
    \int_{f^*-\Delta}^{f^*+\Delta} | \widehat{z}(f) |^2 \mathrm{d} f \geq 0.7 \int_{-\infty}^{+\infty} | \widehat{z}(f) |^2 \mathrm{d} f.
\end{align*}
Furthermore, $z(t)$ is a $(0.3,\Delta)$-one-cluster signal around $f^*$.
\end{lemma}
\begin{proof}
Define region $I_{f^*} = (f^* - \Delta, f^* + \Delta)$ with the complement   $\overline{I_{f^*}} = (-\infty, \infty)\setminus I_{f^*}$.  We have that
  
 \begin{align*}
 \int_{I_{f^*}} | \widehat{z}(f) |^2 \mathrm{d} f \geq (1-\delta/k)\int_{I_{f^*}} | \widehat{x^*\cdot H}(f) |^2 \mathrm{d} f \geq (1-o(1)) T\N^2/k
  \end{align*}
 where the first step follows from Lemma \ref{lem:lo_and_fstar2Glarge}, 
 the second step follows from Eq.~\eqref{eq:heavyfrequency_000}.
  
 On the other hand, $f^*$ is well-isolated. Thus, by the definition of well-isolation (Definition~\ref{def:k_signal_recovery_z}), we have that
  \begin{equation*}
    \int_{\overline{I_{f^*}}} | \widehat{z}(f) |^2 \mathrm{d} f \lesssim  \epsilon\cdot  T\N^2/k\leq 0.1 T\N^2/k.
  \end{equation*}
  
Combining them together, we get that
  \begin{align*}
    \int_{f_0-\Delta}^{f_0+\Delta} | \widehat{z}(f) |^2 \mathrm{d} f \geq 0.7 \int_{-\infty}^{+\infty} | \widehat{z}(f) |^2 \mathrm{d} f
\end{align*}
  
For the furthermore part, Lemma \ref{lem:full_proof_of_3_properties_true_for_z} implies that
 \begin{align*}
    \int_0^T | z(t) |^2 \mathrm{d} t \geq (1-0.3) \int_{-\infty}^{+\infty} |z(t) |^2 \mathrm{d} t.
\end{align*}
Hence, by Definition~\ref{def:one_cluster}, $z(t)$ is a $(0.3,\Delta)$-one-cluster.
\end{proof}

\section{Energy Bound for Filtered Fourier Sparse Signals}
\label{sec:struct_FS}

In this section, we prove an energy bound for the filtered signals $(x\cdot H)*G_{\sigma,b}^{(j)}$, which upper bounds the magnitude of any such signal at a point $t$ by its energy in the time duration $[0,T]$. We first prove an energy bound for untruncated ideally filtered signals (see Section~\ref{sec:struct_FS:untruncated}). Then, we prove an energy bound for filtered signals (see Section~\ref{sec:struct_FS:filtered}). In addition, we prove a technical claim (see Section~\ref{sec:struct_FS:technical_claim}).

\subsection{Energy bound for untruncated ideally filtered signals}\label{sec:struct_FS:untruncated}

In Section~\ref{sec:approx_f_SFS}, we show that the ideally filtered signal $(x\cdot H)*I_{\sigma,b}^{(j)}$, where $I_{\sigma,b}^{(j)}$ defined as Eq.~\eqref{eq:def_ideal_filter} is the ideal filter, is close to the true filtered signal. Here, we further simplify the signal by ignoring the truncation filter $H(t)$, and prove an energy bound for the signals of the form $(x*I_{\sigma,b}^{(j)})(t)$:

\begin{lemma}\label{lem:condition_number_pre}

Let $H$ be defined as in Definition~\ref{def:effect_H_k_sparse}, ${G}^{(j)}_{\sigma,b}$ be defined as in Definition \ref{def:G_j_sigma_b} and the corresponding ideal filter $I=I_{\sigma,b}^{(b)}$ be defined as in Eq.~\eqref{eq:def_ideal_filter}. 
Let $D(t):=\mathrm{Uniform}([-1, 1])$. 

For any $x\in {\cal F}_{k,F}$, we have that with probability $0.6$, for any $t\in(-1, 1)$ 
\begin{align*}
 | (x*I)(t) |^2  \lesssim  {\min} \Big\{\frac{k }{1-|t|}, k^2\Big\} \cdot  \| (x*I)(t) \|_D^2 
\end{align*}
\end{lemma}

\begin{proof}
Since $\wh{(x*I)}(f) =\wh{x}\cdot \wh{I}(f)$ and $x$ is $k$-Fourier-sparse, we know that $(x*\wh{I})(t)$ is also a $k$-Fourier-sparse signal.

On the one hand, by the $k$-Fourier-sparse signal's location-dependent energy bound (Theorem \ref{thm:bound_k_sparse_FT_x_middle}), we have
\begin{align}
   | (x*I)(t) |^2 
   \lesssim &~ \frac{k }{1-|t|} \| (x*I)(t) \|_D^2 
   \label{eq:xIH_k_1t_energy_bound}
\end{align}

On the other hand, by the location-independent energy bound (Theorem \ref{thm:energy_bound}), we have that
\begin{align}
   | (x*I)(t) |^2
   \lesssim  k^2 \| (x*I)(t) \|_D^2  \label{eq:xIH_k2_energy_bound}
\end{align}

Combine Eq.~\eqref{eq:xIH_k_1t_energy_bound} and Eq.~\eqref{eq:xIH_k2_energy_bound} together, we prove the lemma:
\begin{align*}
 | (x*I)(t) |^2  \lesssim  {\min} \Big\{\frac{k }{1-|t|}, k^2\Big\} \cdot  \| (x*I)(t) \|_D^2 .
\end{align*}

\end{proof}

\subsection{Energy bound for filtered signals}\label{sec:struct_FS:filtered}

Based on Lemma~\ref{lem:condition_number_pre}, we can relate the magnitude of the filtered signal with its own energy plus the original Fourier-sparse signal's energy.

\begin{lemma}\label{lem:condition_number_xI}
Let $H$ be defined as in Definition~\ref{def:effect_H_k_sparse}, ${G}^{(j)}_{\sigma,b}$ be defined as in Definition \ref{def:G_j_sigma_b} and the corresponding ideal filter $I=I_{\sigma,b}^{(b)}$ be defined as in Eq.~\eqref{eq:def_ideal_filter}.

For any $x\in {\cal F}_{k,F}$, $j\in [B]$, and $(\sigma,b)$ such that Large Offset event does not happen, let $z(t)=(x\cdot H)*{G}^{(j)}_{\sigma,b}(t)$. It holds that:
\begin{align*}
     {|z(t)|^2} \lesssim {\min}\Big\{ \frac{k \cdot H(t)}{1-|2t/T-1|} , k^2\Big\} \cdot {\|z(t)\|_T^2} +\delta_1 \|x(t)\|_T^2~~~\forall t\in (-1,1).
\end{align*}
\end{lemma}
\begin{proof}
Let $S:=\supp(\wh{x}*\wh{H})$ be defined as the support set of $\wh{x}*\wh{H}$. 
Then $|S|\leq \Delta$.

First, by the ideally untruncated filtered signal's energy bound (Lemma \ref{lem:condition_number_pre}), we have
\begin{align}
 | (x*I)(t)\cdot H(t) |^2  \lesssim &~ H^2 (t)\cdot {\min} \Big\{\frac{k  }{1-|2t/T-1|}, k^2\Big\} \cdot  \| (x*I)(t) \|_T^2\notag \\
 \lesssim&~  {\min} \Big\{\frac{k \cdot H(t)}{1-|2t/T-1|}, k^2\Big\} \cdot  \| (x*I)(t) \|_T^2 ,
 \label{eq:xIH_k2_k_1t_energy_bound}
\end{align}
where the second step follows from $H(t)\lesssim 1$ (Lemma \ref{lem:property_of_filter_H} Property \RN{1}, \RN{2}).

Then, we bound the magnitude of the ideal filtered signal as follows:
\begin{align}
  |(x\cdot H)*I(t)|^2 =&~   |(x*I)(t)\cdot H(t)|^2 \notag \\
  \lesssim&~ {\min} \Big\{\frac{k \cdot H(t)}{1-|2t/T-1|}, k^2\Big\} \cdot  \| (x*I)(t) \|_T^2 \notag \\
\lesssim&~ {\min} \Big\{\frac{k \cdot H(t)}{1-|2t/T-1|}, k^2\Big\} \cdot (\|(x\cdot H)*G^{(j)}_{\sigma, b}(t)\|_T^2 + \delta_1^2 \|x(t)\|_T^2)\notag \\
\lesssim&~ {\min} \Big\{\frac{k \cdot H(t)}{1-|2t/T-1|}, k^2\Big\} \cdot \|(x\cdot H)*G^{(j)}_{\sigma, b}(t)\|_T^2 + \delta_1  \|x(t)\|_T^2
\label{eq:bound_xHI2xHG}
\end{align}
where the first step follows from Lemma \ref{lem:xIH2xHI}, the second step follows from Eq.~\eqref{eq:xIH_k2_k_1t_energy_bound}, the third step follows from Claim \ref{clm:xI2xHG},
the forth step follows from $k^2\delta_1 \leq 1 $. 

Next, we consider the difference between the signals filtered by ${G}^{(j)}_{\sigma,b}(t)$ and  $I(t)$:
\begin{align}
    |(x\cdot H)*{G}^{(j)}_{\sigma,b}(t)-(x\cdot H)*I(t)|^2 \leq &~ \delta_1^2{T|S|} \cdot \|x(t)\|_T^2 \notag \\
    \leq &~ \delta_1  \cdot \|x(t)\|_T^2 \label{eq:bound_xHG_sub_xHI_to_delta_x_Tnorm}
\end{align}
where the first step follows from Lemma \ref{lem:xHG_sub_XHI_is_small}, the second step follows from $ \delta_1 T|S| \leq 1 $.

Finally, we have that
\begin{align*}
|(x\cdot H)*{G}^{(j)}_{\sigma,b}(t)|^2 \leq&~ 2|(x\cdot H)*I(t)|^2+2|(x\cdot H)*{G}^{(j)}_{\sigma,b}(t)-(x\cdot H)*I(t)|^2 \notag \\
\lesssim&~ |(x\cdot H)*I(t)|^2+\delta_1 \|x(t)\|_T^2 \notag \\
\lesssim  &~ {\min} \Big\{\frac{k \cdot H(t)}{1-|2t/T-1|}, k^2\Big\} \cdot \|(x\cdot H)*\wh{G}^{(j)}_{\sigma, b}(t)\|_T^2  +\delta_1 \|x(t)\|_T^2, 
\end{align*}
where the first step follows from $(a+b)^2\leq 2a^2+2b^2$, the second step follows from Eq.~\eqref{eq:bound_xHG_sub_xHI_to_delta_x_Tnorm}, the third step follows from Eq.~\eqref{eq:bound_xHI2xHG}.

The lemma is then proved.

\end{proof}

The energy bound in Lemma~\ref{lem:condition_number_xI} not only depends on $\|z(t)\|_T$, but also on $\|x(t)\|_T$. The following lemma show that assuming the filtered signal contains a heavy frequency, $\|x(t)\|_T$ can be upper bounded by $\|z(t)\|_T$.

\begin{lemma}\label{lem:xHG_large_than_exp_small}

Given $k\in \Z_+, F\in\R_+$. Let $H$ be defined as in Definition~\ref{def:effect_H_k_sparse}, ${G}^{(j)}_{\sigma,b}$ be defined as in Definition \ref{def:G_j_sigma_b}. Let $x\in {\cal F}_{k,F}$ be any $k$-Fourier sparse signal.
For $j\in [B]$ such that there exists a $f^*$ satisfying: $j=h_{\sigma, b}(f^*)$ and
\begin{align}
    \int_{f^*-\Delta_h}^{f^*+\Delta_h} | \widehat{x\cdot H}(f) |^2 \mathrm{d} f \geq  T\N^2/k, 
    \label{eq:assumption_heavy_freq_1}
\end{align}
where $\N^2\geq \delta \|x\|_T^2$ and $\Delta_h=|\supp(\wh{H})|$.

For any $(\sigma, b)$ that Large Offset event does not happen, we have that 
\begin{align*}
      \|(x\cdot H)*G^{(j)}_{\sigma, b}(t)  \|^2_T \gtrsim \frac{ \delta \|x\|_T^2}{k}.
\end{align*}
\end{lemma}
\begin{proof}

We have that
\begin{align*}
 T\|(x\cdot H)*G^{(j)}_{\sigma, b}(t)  \|^2_T =&~     \int_0^T |(x\cdot H)*G^{(j)}_{\sigma, b}(t) |^2 \d t \notag\\
 \gtrsim &~ \int_{-\infty}^\infty |(x\cdot H)*G^{(j)}_{\sigma, b}(t) |^2 \d t \notag\\
 = &~ \int_{-\infty}^\infty |(\wh{x}*\wh{H})\cdot \wh{G}^{(j)}_{\sigma, b}(f) |^2 \d f \notag\\
 \geq  &~ \int_{f^*-\Delta_h}^{f^*+\Delta_h} |(\wh{x}*\wh{H})\cdot \wh{G}^{(j)}_{\sigma, b}(f) |^2 \d f \notag\\
 \gtrsim &~ \int_{f^*-\Delta_h}^{f^*+\Delta_h} |(\wh{x}*\wh{H})(f) |^2 \d f \notag\\
 \geq &~ \frac{T \delta \|x\|_T^2}{k}
\end{align*}
where the first step follows from the definition of norm, the second step follows from Lemma \ref{lem:z_satisfies_two_properties}, the third step follows from Parseval's theorem, the forth step is straight forward,
the fifth step follows from Lemma \ref{lem:lo_and_fstar2Glarge}, 
the sixth step follows from Eq.~\eqref{eq:assumption_heavy_freq_1}.

\end{proof}
 
Lemma~\ref{lem:condition_number_xI} and Lemma~\ref{lem:xHG_large_than_exp_small} implies the following energy bound:
\begin{corollary}[Energy bound for filtered signals]\label{cor:condition_number_z}

Given $k\in \mathbb{N}$ and $F\in \R_+$. Let $x\in {\cal F}_{k,F}$.
Let $H$ be defined as in Definition~\ref{def:effect_H_k_sparse}, ${G}^{(j)}_{\sigma,b}$ be defined as in Definition \ref{def:G_j_sigma_b} with $(\sigma,b)$ such that Large Offset event does not happen.

For any $j\in [B]$, suppose there exists an $f^*$ with $j=h_{\sigma, b}(f^*)$ satisfying: 
\begin{align*}
    \int_{f^*-\Delta}^{f^*+\Delta} | \widehat{x\cdot H}(f) |^2 \mathrm{d} f \geq  T\N^2/k, 
\end{align*}
where $\N^2\geq \delta \|x\|_T^2$.
Then, for $z(t)=(x\cdot H)*{G}^{(j)}_{\sigma,b}(t)$, it holds that: 
\begin{align*}
     {|z(t)|^2} \lesssim {\min}\Big\{ \frac{k \cdot H(t)+\delta}{1-|2t/T-1|}, k^2 \Big\}\cdot {\|z(t)\|_D^2}~~~\forall t\in (0,T).
\end{align*}
\end{corollary}
\begin{proof}
We have that
\begin{align*}
|z(t)|^2 \lesssim  &~ {\min} \Big\{\frac{k \cdot H(t)}{1-|2t/T-1|}, k^2\Big\} \cdot \|z(t)\|_T^2  +\delta_1 \|x(t)\|_T^2 \notag \\
\lesssim  &~ {\min} \Big\{\frac{k \cdot H(t)}{1-|2t/T-1|}, k^2\Big\} \cdot \|z(t)\|_T^2  +\delta^2 k^{-1} \|x(t)\|_T^2 \notag \\
\lesssim  &~ {\min} \Big\{\frac{k \cdot H(t)}{1-|2t/T-1|}, k^2\Big\} \cdot \|z(t)\|_T^2  + \delta \|(x\cdot H)*G^{(j)}_{\sigma, b}(t)\|_T^2 \notag \\
\lesssim  &~ {\min} \Big\{\frac{k \cdot H(t)+\delta}{1-|2t/T-1|}, k^2\Big\} \cdot \|z(t)\|_T^2
\end{align*}
where the first step follows from Lemma \ref{lem:condition_number_xI}, the second step follows from $\delta_1 \leq \delta^2 k^{-1}$, the third step follows from Lemma \ref{lem:xHG_large_than_exp_small}, the forth step is straight forward.

\end{proof}

\subsection{Technical claim}\label{sec:struct_FS:technical_claim}

\begin{claim}\label{clm:xI2xHG}
Given $k\in \Z_+, F\in\R_+$. Let $\delta_1$ be defined as the $\delta$ of Lemma \ref{lem:property_of_filter_G}.%
Let $H$ be defined as in Definition~\ref{def:effect_H_k_sparse}, ${G}^{(j)}_{\sigma,b}$ be defined as in Definition \ref{def:G_j_sigma_b}, and $I=I_{\sigma,b}^{(j)}$ be the ideal filter defined by Eq.~\eqref{eq:def_ideal_filter}. 

Then, for any $x\in {\cal F}_{k,F}$ and $j\in [B]$,  with probability $0.6$ over $(\sigma,b)$, we have that %
\begin{align*}
\|(x*I)(t) \|_T^2 \lesssim \|(x\cdot H)*\wh{G}^{(j)}_{\sigma, b}(t)\|_T^2 + \delta_1^2 \|x(t)\|_T^2. %
\end{align*}
\end{claim}
\begin{proof}

We have that 
\begin{align*}
    \|(x*I)(t) \|_T^2 \lesssim&~ \|(x*I)(t)\cdot H\|_T^2\notag\\
    = &~ \|(x\cdot H)*I(t)\|_T^2\notag\\
    \leq &~ 2 \|(x\cdot H)*G^{(j)}_{\sigma, b}(t)\|_T^2 + 2 \| (x\cdot H)*G^{(j)}_{\sigma, b}(t)-(x\cdot H)*I(t)\|_T^2\notag\\
    \lesssim &~ \|(x\cdot H)*G^{(j)}_{\sigma, b}(t)\|_T^2 + \delta_1^2 \|x(t)\|_T^2 
\end{align*}
where the first step follows from $(x*I)(t)$ is a $k$-Fourier-sparse signal and Lemma \ref{lem:property_of_filter_H} Property \RN{5}, the second step follows from Lemma~\ref{lem:xIH2xHI} conditioning on Large Offset event not happening, the third step follows from $(a+b)^2\leq 2a^2+2b^2$, the forth step follows from Lemma \ref{lem:xHI_sub_XHG_T_norm_is_small}. 
\end{proof}

\section{Local-Test Signal}

\label{sec:twist_struct_FS}

Recall that the filtered signal in the $j$-th bin of the \textsc{HashToBins} procedure can be written as $z(t)=(x\cdot H)*G_{\sigma,b}^{(j)}(t)$. The next step of the frequency estimation algorithm is to extract a significant frequency from $z(t)$ by considering a so-called \emph{local-test signal}:
\begin{align}\label{eq:def_d_z}
    d_z(t):=z(t)e^{2\pi f_0\beta}-z(t+\beta),
\end{align}
where $f_0\in \supp(\wh{x}^*)$, and $j=h_{\sigma, b}(f_0)$, where $\beta\in\R_+$ is a parameter such that $\beta \leq O(1/\Delta)$ with $\Delta=O(k\cdot |\supp(\wh{H})|)$.

In this section, we will study some properties of $d_z(t)$ and its ideal versions (see Section~\ref{sec:twist_struct_FS:local} and Section~\ref{sec:twist_struct_FS:post}) and derive an energy bound for it (See Section~\ref{sec:twist_struct_FS:energy}).

\subsection{Ideal local-test signal}\label{sec:twist_struct_FS:local}

In previous section, we've shown that ideal filter $I_{\sigma,b}^{(j)}$ can be used to approximate $G_{\sigma,b}^{(j)}$ such that the ideally filtered signal is close to the true filtered signal. We will show that under the ideal filter approximation, the \emph{ideal local-test signal} is also close to the true local-test signal. More formally, we define the ideal filtered signal and the ideal local-test signal as follows: 
\begin{align}\label{eq:def_ideal_local_test_signal}
    z_I(t):= &~ (x\cdot H)*I(t),\notag\\
    d_{I,z}(t):= &~ z_I(t)e^{2\pi\i f_0 \beta}-z_I(t+\beta),
\end{align}
The following lemma bounds the point-wise distance between $d_z(t)$ and $d_{I,z}(t)$. 

\begin{lemma}\label{lem:z_diff_xHG_sub_XHI_is_small}
Let $\delta_1$ be defined as in Lemma \ref{lem:property_of_filter_G}. Let $H$ be defined as in Definition~\ref{def:effect_H_k_sparse}, ${G}^{(j)}_{\sigma,b}$ be defined as in Definition \ref{def:G_j_sigma_b} and $I=I_{\sigma,b}^{(j)}$ be the corresponding ideal filter as in Eq.~\eqref{eq:def_ideal_filter}. 
 
For any $x\in {\cal F}_{k,F}$ and $(\sigma,b)$ such that Large Offset event does not happen, for any $j\in [B]$, let $z(t)=(x\cdot H)*G_{\sigma,b}^{(j)}(t)$, $d_z(t)$ be defined as Eq.~\eqref{eq:def_d_z}, $z_I(t)$ and $d_{z,I}(t)$ be defined as Eq.~\eqref{eq:def_ideal_local_test_signal}. %

Then, we have
\begin{align*}
|d_z(t) - d_{I,z}(t)| \lesssim \delta_1  \sqrt{T|S|} \cdot  \|x(t) \|_T~~~\forall t\in \R.
\end{align*}
\end{lemma}
\begin{proof}
\begin{align*}
   |d_z(t) - d_{I,z}(t)| = &~ |z(t)e^{2\pi\i f_0 \beta}-z_I(t)e^{2\pi\i f_0 \beta}-(z(t+\beta) -z_I(t+\beta))| \\ 
   \leq &~ |z(t)e^{2\pi\i f_0 \beta}-z_I(t)e^{2\pi\i f_0 \beta}|+|z(t+\beta)  -z_I(t+\beta)| \\ 
   = &~ |z(t)-z_I(t)|+|z(t+\beta) -z_I(t+\beta)| \\ 
   \lesssim &~ \delta_1  \sqrt{T|S|} \cdot  \|x(t) \|_T,
\end{align*}
where the first step follows from the definition of $d_z(t)$ and $d_{I, z}(t)$, the second step follows from triangle inequality, the third step follows from $|e^{2\pi\i f_0 \beta}|=1$, the forth step follows from Lemma \ref{lem:xHG_sub_XHI_is_small}.  
\end{proof}

The following lemma bounds the $L_2$-distance between $d_z(t)$ and $d_{z,I}(t)$.

\begin{lemma}\label{lem:z_diff_xHI_sub_XHG_T_norm_is_small}
Let $\delta_1$ be defined as in Lemma \ref{lem:property_of_filter_G}. Let $H$ be defined as in Definition~\ref{def:effect_H_k_sparse}, ${G}^{(j)}_{\sigma,b}$ be defined as in Definition \ref{def:G_j_sigma_b} and $I=I_{\sigma,b}^{(j)}$ be the corresponding ideal filter as in Eq.~\eqref{eq:def_ideal_filter}. 

For any $x\in {\cal F}_{k,F}$ and $(\sigma,b)$ such that Large Offset event does not happen, for any $j\in [B]$, let $z(t)=(x\cdot H)*G_{\sigma,b}^{(j)}(t)$, $d_z(t)$ be defined as Eq.~\eqref{eq:def_d_z}, $z_I(t)$ and $d_{z,I}(t)$ be defined as Eq.~\eqref{eq:def_ideal_local_test_signal}.%
Then, 
\begin{align*}
\int_{-\infty}^\infty  |d_{I, z}(t) - d_z(t) |^2 \d t \lesssim \delta_1^2 T \|x(t) \|_T^2
\end{align*}
\end{lemma}
\begin{proof}

We first have that,
\begin{align}
    &~ \int_{-\infty}^\infty  |z_I(t)e^{2\pi\i f_0 \beta}-z(t)e^{2\pi\i f_0 \beta} |^2 \d t \notag \\
    = &~ \int_{-\infty}^\infty  |z_I(t)-z(t) |^2 \d t \notag \\
    \leq &~ \delta_1^2 T \|x(t) \|_T^2,\label{eq:bound_z_diff_e}
\end{align}
where the first step follows from $|e^{2\pi\i f_0 \beta}|=1$, the second step follows from Lemma \ref{lem:xHI_sub_XHG_T_norm_is_small}.

Then, we complete the proof as follows:
\begin{align*}
&~ \int_{-\infty}^\infty  |d_{I, z}(t) - d_z(t) |^2 \d t \\
= &~ \int_{-\infty}^\infty  |z_I(t)e^{2\pi\i f_0 \beta}-z(t)e^{2\pi\i f_0 \beta}-(z_I(t+\beta)-z(t+\beta)) |^2 \d t \\
\leq &~ 2\int_{-\infty}^\infty  |z_I(t)e^{2\pi\i f_0 \beta}-z(t)e^{2\pi\i f_0 \beta} |^2 \d t + 2\int_{-\infty}^\infty  |z_I(t+\beta)-z(t+\beta) |^2 \d t \\
\lesssim &~ \delta_1^2 T \|x(t) \|_T^2 + \int_{-\infty}^\infty  |z_I(t+\beta)-z(t+\beta) |^2 \d t \\
\lesssim &~ \delta_1^2 T \|x(t) \|_T^2 
\end{align*}
where the first step follows from the definition of $d_{I, z}(t)$ and $d_z(t)$, the second step follows from $(a+b)^2 \leq 2a^2+2b^2$, the third step follows from Eq.~\eqref{eq:bound_z_diff_e}, the forth step follows from Lemma \ref{lem:xHI_sub_XHG_T_norm_is_small}. 

\end{proof}

\subsection{Ideal post-truncated local-test signal}\label{sec:twist_struct_FS:post}
It is still difficult to directly study the energy bound for $d_{z,I}(t)$. In this section, we further simplify the ideally filtered signal by removing the $H$ filter and consider the untruncted ideally filtered signal $(x*I)(t)$. Then, in the local-test signal, we perform a post-truncation. More specifically, the untruncated ideally filtered signal and the \emph{ideal post-truncated local-test signal} are defined as follows:
\begin{align}\label{eq:def_x_I_d_I_x}
    x_I(t):=&~ (x*I)(t),\notag\\
    d_{I, x}(t):= &~ x_I(t)\cdot H(t)\cdot e^{2\pi\i f_0 \beta}-x_I(t+\beta)\cdot H(t+\beta).
\end{align}
Intuitively, $d_{I,x}(t)$ can be viewed as swapping the order of the $I$ and $H$ filters in $d_{I,z}(t)$. 

The following lemma shows that $d_{I,z}(t)$ and $d_{I,x}(t)$ are actually the same!
\begin{lemma}\label{lem:z_diff_xIH2xHI}
Let $\delta_1$ be defined as in Lemma \ref{lem:property_of_filter_G}. Let $H$ be defined as in Definition~\ref{def:effect_H_k_sparse}, ${G}^{(j)}_{\sigma,b}$ be defined as in Definition \ref{def:G_j_sigma_b} and $I=I_{\sigma,b}^{(j)}$ be the corresponding ideal filter as in Eq.~\eqref{eq:def_ideal_filter}.

For any $x\in {\cal F}_{k,F}$, and $(\sigma,b)$ such that Large Offset event does not happen, let $z_I(t)$ and $d_{z,I}(t)$ be defined as Eq.~\eqref{eq:def_ideal_local_test_signal}, $x_I(t)$ and $d_{x,I}(t)$ be defined as Eq.~\eqref{eq:def_x_I_d_I_x}.

Then, we have
\begin{align*}
 d_{I, z}(t) = d_{I, x}(t)~~~\forall t\in \R.
\end{align*}
\end{lemma}
\begin{proof}
We have that
\begin{align*}
 d_{I, z}(t) = &~ z_I(t)\cdot e^{2\pi\i f_0 \beta}-z_I(t+\beta) \\
 = &~ x_I(t)\cdot H(t)\cdot e^{2\pi\i f_0 \beta}-z_I(t+\beta) \\
 = &~ x_I(t)\cdot H(t)\cdot e^{2\pi\i f_0 \beta}-x_I(t+\beta)\cdot H(t+\beta) \\
 = &~ d_{I, x}(t),
\end{align*}
where the first step follows from the definition of $  d_{I, z}(t)$, the second step follows from Lemma \ref{lem:xIH2xHI}, the third step follows from Lemma \ref{lem:xIH2xHI}, the last step follows from the definition of $ d_{I, x}(t)$.
\end{proof}

The structure of $d_{I,x}(t)$ makes it easy to study its magnitude at any ``good point'': 

\begin{lemma}\label{lem:z_diff_condition_number_xI_pre}

Let $H$ be defined as in Definition~\ref{def:effect_H_k_sparse}, ${G}^{(j)}_{\sigma,b}$ be defined as in Definition \ref{def:G_j_sigma_b} and $I=I_{\sigma,b}^{(j)}$ be the corresponding ideal filter as in Eq.~\eqref{eq:def_ideal_filter}. 
Let $U:=\{t_0\in \R~|~ H(t) > 1-\delta_1~ \forall t\in [t_0,t_0+\beta]\}$. 

For any $x\in {\cal F}_{k,F}$, and $(\sigma,b)$ such that Large Offset event does not happen, let $x_I(t), d_{x,I}(t)$ be defined as Eq.~\eqref{eq:def_x_I_d_I_x}. Then, we have
\begin{align*}
|d_{I, x}(t)| \lesssim \left|x_I(t)\cdot  e^{2\pi\i f_0 \beta} - x_I(t+\beta)\right|+ \delta_1 k \|x_I(t)\|_{T}~~~\forall t\in U. 
\end{align*}
\end{lemma}
\begin{proof}

First, for any $t\in U$,
\begin{align}
|x_I(t)\cdot H(t)\cdot e^{2\pi\i f_0 \beta}- x_I(t)\cdot  e^{2\pi\i f_0 \beta}| = &~ |x_I(t)\cdot H(t)- x_I(t)| \notag \\
 = &~ |x_I(t)| \cdot |1-H(t)| \notag \\
 \leq &~  \delta_1 |x_I(t)| \notag \\
 \lesssim &~  \delta_1 k \|x_I(t)\|_{T}
 \label{eq:1_z_diff_condition_number_xI}
\end{align}
where the first step follows from $|e^{2\pi\i f_0 \beta}|=1$, the second step is straight forward, the third step follows from $H(t)\leq 1$ (Lemma \ref{lem:property_of_filter_H} Property \RN{1}, \RN{2}) and $ \forall t\in U, H(t)>1-\delta_1$, and the last step follows from Lemma \ref{lem:condition_number_pre}.

Second, for any $t\in U$,
\begin{align}
|x_I(t+\beta) - x_I(t+\beta)\cdot H(t+\beta)| = &~ |x_I(t+\beta)| \cdot |1-H(t+\beta)| \notag \\
 \leq &~  \delta_1 |x_I(t+\beta)| \notag \\
 \lesssim &~  \delta_1 k \|x_I(t)\|_{T} 
 \label{eq:2_z_diff_condition_number_xI}
\end{align}
where the first step is straight forward, the second step follows from $H(t)\leq 1$ (Lemma \ref{lem:property_of_filter_H} Property \RN{1}, \RN{2}) and $ \forall t\in U, H(t+\beta)>1-\delta_1$, the last step follows from Lemma \ref{lem:condition_number_pre}.

Combining them together, we have that for any $t\in U$,
\begin{align*}
|d_{I, x}(t)| 
=&~ |x_I(t)\cdot H(t)\cdot e^{2\pi\i f_0 \beta}-x_I(t+\beta)\cdot H(t+\beta)| \notag \\
\leq&~ |x_I(t)\cdot H(t)\cdot e^{2\pi\i f_0 \beta}- x_I(t)\cdot  e^{2\pi\i f_0 \beta}|+|x_I(t)\cdot  e^{2\pi\i f_0 \beta} - x_I(t+\beta)|\\
&+|x_I(t+\beta) - x_I(t+\beta)\cdot H(t+\beta)| \notag \\ 
\lesssim&~ |x_I(t)\cdot  e^{2\pi\i f_0 \beta} - x_I(t+\beta)|+|x_I(t+\beta) - x_I(t+\beta)\cdot H(t+\beta)| + \delta_1 k \|x_I(t)\|_{T}  \notag \\ 
\lesssim&~ |x_I(t)\cdot  e^{2\pi\i f_0 \beta} - x_I(t+\beta)|+ \delta_1 k \|x_I(t)\|_{T}  
\end{align*}
where the first step follows from the definition of $d_{I, x}(t)$, the second step follows from triangle inequality, the third step follows Eq.~\eqref{eq:1_z_diff_condition_number_xI}, the forth step follows from Eq.~\eqref{eq:2_z_diff_condition_number_xI}.

\end{proof}

Furthermore, we can show that the ideal post-truncated local-test signal is close to the ideal local-test signal without truncation on most of ``good points''.

\begin{lemma}\label{lem:z_diff_norm_condition_number_xI_pre}
Let $H$ be defined as in Definition~\ref{def:effect_H_k_sparse}, ${G}^{(j)}_{\sigma,b}$ be defined as in Definition \ref{def:G_j_sigma_b} and $I=I_{\sigma,b}^{(j)}$ be the corresponding ideal filter as in Eq.~\eqref{eq:def_ideal_filter}. 
Let $U:=\{t_0\in \R~|~ H(t) > 1-\delta_1, \forall t\in [t_0,t_0+\beta]\}$.  Let $D_U(t):=\mathrm{Uniform}(U)$ and $D_{U+\beta}(t):=\mathrm{Uniform}(U+\beta)$.  

For any $x\in {\cal F}_{k,F}$, and $(\sigma,b)$ such that Large Offset event does not happen, let $x_I(t), d_{x,I}(t)$ be defined as Eq.~\eqref{eq:def_x_I_d_I_x}. Then, we have
\begin{align*}
\|d_{I, x}(t) - (x_I(t)\cdot e^{2\pi\i f_0 \beta}-x_I(t+\beta))\|_{D_U} \lesssim \delta_1 \|x_I(t)\|_{T}. %
\end{align*}
\end{lemma}
\begin{proof}

First, 
\begin{align}
\|x_I(t)\cdot H(t)\cdot e^{2\pi\i f_0 \beta}- x_I(t)\cdot  e^{2\pi\i f_0 \beta}\|_{D_U} = &~ \|x_I(t)\cdot H(t)- x_I(t)\|_{D_U} \notag \\
 \leq &~ {\max}_{t\in U}\{|1-H(t)|\} \cdot\|x_I(t)\|_{D_U}   \notag \\
 \leq &~  \delta_1 \cdot\|x_I(t)\|_{D_U} \notag \\
 \lesssim &~  \delta_1 \cdot \sqrt{\frac{T}{|U|}} \|x_I(t)\|_{T} \notag \\
 \lesssim &~  \delta_1 \cdot  \|x_I(t)\|_{T} 
 \label{eq:1_d_I_x_d_x}
\end{align}
where the first step follows from $|e^{2\pi\i f_0 \beta}|=1$, the second step is straight forward, the third step follows from $H(t)\leq 1$ (Lemma \ref{lem:property_of_filter_H} Property \RN{1}, \RN{2}) and $ \forall t\in U, H(t)>1-\delta_1$, the forth step follows from the definition of the norm 
\begin{align*}
    \|x(t)\|^2_{D_U}=\frac{1}{|U|} \int_{U} |x(t)|^2 \d t \leq \frac{1}{|U|} \int_{[0, T]} |x(t)|^2 \d t = \frac{T}{|U|}\|x(t)\|^2_{T},
\end{align*}
and the last step follows from Lemma \ref{lem:condition_number_pre}.

Second, 
\begin{align}
\|x_I(t+\beta) - x_I(t+\beta)\cdot H(t+\beta)\|_{D_U} \leq &~ {\max}_{t\in U}\{|1-H(t+\beta)|\} \cdot\|x_I(t+\beta)\|_{D_U} \notag \\
 \leq &~  \delta_1\cdot  \|x_I(t)\|_{D_{U+\beta}} \notag \\
 \lesssim &~  \delta_1\cdot \frac{1}{|U+\beta|} \|x_I(t)\|_{D_{1}} \notag \\
 \lesssim &~  \delta_1\cdot  \|x_I(t)\|_{D_{1}}
 \label{eq:2_d_I_x_d_x}
\end{align}
where the first step is straight forward, the second step follows from $H(t)\leq 1$ (Lemma \ref{lem:property_of_filter_H} Property \RN{1}, \RN{2}) and $ \forall t\in U, H(t+\beta)>1-\delta_1$, the third step follows from the definition of the norm, the forth step follows from $|U+\beta|=|U|\gtrsim 1$.

Then, we have that,
\begin{align*}
&~ \|d_{I, x}(t) - (x_I(t)\cdot e^{2\pi\i f_0 \beta}-x_I(t+\beta))\|_{D_U} \\
=  &~ \|x_I(t)\cdot H(t)\cdot e^{2\pi\i f_0 \beta}-x_I(t+\beta)\cdot H(t+\beta) - (x_I(t)\cdot e^{2\pi\i f_0 \beta}-x_I(t+\beta))\|_{D_U}  \\
\leq  &~ \|x_I(t)\cdot H(t)\cdot e^{2\pi\i f_0 \beta} - x_I(t)\cdot e^{2\pi\i f_0 \beta}\|_{D_U} + \|x_I(t+\beta)\cdot H(t+\beta) - x_I(t+\beta)\|_{D_U}  \\
\lesssim  &~ \delta_1\cdot  \|x_I(t)\|_{D_{1}} + \|x_I(t+\beta)\cdot H(t+\beta) - x_I(t+\beta)\|_{D_U}  \\
\lesssim  &~ \delta_1\cdot  \|x_I(t)\|_{D_{1}} 
\end{align*}
where the first step follows from the definition of $d_{I, x}(t)$, the second step follows from triangle inequality, the third step follows from Eq.~\eqref{eq:1_d_I_x_d_x}, the forth step follows from Eq.~\eqref{eq:2_d_I_x_d_x}.

\end{proof}

\subsection{Energy bound for local-test signals}\label{sec:twist_struct_FS:energy}
In this section, we prove the following lemma, which gives an energy bound for local-test signals.
\begin{lemma}[Energy bound for local-test signals]\label{lem:z_diff_condition_number_xI}
Let $H$ be defined as in Definition~\ref{def:effect_H_k_sparse}, ${G}^{(j)}_{\sigma,b}$ be defined as in Definition \ref{def:G_j_sigma_b}.
Let $U$, $D_U$ be defined as in Lemma~\ref{lem:z_diff_norm_condition_number_xI_pre}. 

For any $x\in {\cal F}_{k,F}$, and $(\sigma,b)$ such that Large Offset event does not happen, let $z(t)=(x\cdot H)*G_{\sigma,b}^{(j)}(t)$ and $d_z(t)$ be defined as Eq.~\eqref{eq:def_d_z}. Then, we have
\begin{align*}
     {|d_z(t)|^2} \lesssim {\min}\Big\{ \frac{k}{1-|2t/T-1|} , k^2 \Big\}\cdot {\|d_z(t)\|_{D_U}^2} +\delta_1 \|x(t)\|_{T}^2~~~\forall t\in U.
\end{align*}
\end{lemma}
\begin{proof}
Let $I=I_{\sigma,b}^{(j)}$ be the corresponding ideal filter as in Eq.~\eqref{eq:def_ideal_filter}. 
Let $S:=\supp(\wh{x}*\wh{H})$ be the support set of $\wh{x}*\wh{H}$. 
We have $|S|\leq \Delta$. 

Let $z_I(t)$, $d_{z,I}(t)$ be defined as in Lemma~\ref{lem:z_diff_xHG_sub_XHI_is_small} and $x_I(t)$, $d_{I,x}(t)$ be defined as in Lemma~\ref{lem:z_diff_xIH2xHI}.

Before proving the energy bound for $|d_z(t)|$, we first consider the signal $x_I(t)\cdot  e^{2\pi\i f_0 \beta} - x_I(t+\beta)$. 
By Fourier transformation, we know that its Fourier coefficient of a frequency $f$ is:
\begin{align*}
    \wh{x}_I(f)  e^{2\pi\i f_0 \beta} - \wh{x}_I(f) e^{2\pi\i f \beta} = \wh{x}(f)\cdot \wh{I}(f)  e^{2\pi\i f_0 \beta} - \wh{x}(f)\cdot \wh{I}(f) e^{2\pi\i f \beta}
\end{align*}
Thus, $x_I(t)\cdot  e^{2\pi\i f_0 \beta} - x_I(t+\beta)$ is at most $k$-Fourier-sparse. 

Let $[L,R]:=U$. By Fourier-sparse signals' energy bound (Theorem \ref{thm:bound_k_sparse_FT_x_middle} and Theorem \ref{thm:energy_bound}), we have
\begin{align}
  |  x_I(t)\cdot  e^{2\pi\i f_0 \beta} - x_I(t+\beta) |^2 \lesssim &~ {\min} \Big\{\frac{k }{\min\{R-t, t-L\}}, k^2\Big\} \cdot  \|  x_I(t)\cdot  e^{2\pi\i f_0 \beta} - x_I(t+\beta) \|_{D_U}^2\notag \\
  \lesssim &~ {\min} \{\frac{k }{1-|2t/T-1|}, k^2\} \cdot  \|  x_I(t)\cdot  e^{2\pi\i f_0 \beta} - x_I(t+\beta) \|_{D_U}^2 
  \label{eq:1_d_x_I_z_diff_condition_number_xI}
\end{align}
where the first step follows from applying Theorem \ref{thm:bound_k_sparse_FT_x_middle} with $x(t)=x(Tt/2+T/2) $ and applying Theorem \ref{thm:energy_bound} with $T=|U|, x(t)=x(t+L) $, the second step follows from $ [-1+0.5/k, 1-0.5/k] \subseteq [L,R]$, which implies that $k(\min\{R-t, t-L\})^{-1} \lesssim k(1-|2t/T-1|)^{-1} $ for any $|t|\in [L+1/k, R-1/k]$. Moreover, for any $|t|\in [L, L+1/k]\cup [R-1/k, R]$, $k^2 \lesssim k(1-|2t/T-1|)^{-1} $. 

The RHS can be upper bounded by:
\begin{align}
    \|  x_I(t)\cdot  e^{2\pi\i f_0 \beta} - x_I(t+\beta) \|_{D_U}^2 \leq &~ 2 \| d_{I, x}(t)\|_{D_U}^2 + 2 \| d_{I, x}(t) -( x_I(t)\cdot  e^{2\pi\i f_0 \beta} - x_I(t+\beta)) \|_{D_U}^2  \notag \\
    \lesssim &~  \| d_{I, x}(t)\|_{D_U}^2 +   \delta_1^2 \|x_I(t)\|^2_{T}\notag \\
    = &~  \| d_{I, z}(t)\|_{D_U}^2 +   \delta_1^2 \|x_I(t)\|^2_{T}\notag \\
    \lesssim &~   \|d_z(t)\|_{D_U}^2 + \| d_{I, z}(t)-d_z(t)\|_{D_U}^2 +   \delta_1^2 \|x_I(t)\|^2_{T}\notag \\
    \lesssim &~   \|d_z(t)\|_{D_U}^2 + \| d_{I, z}(t)-d_z(t)\|_{D_U}^2 +   \delta_1^2 \|x(t)\|^2_{T}
    \label{eq:2_d_x_I_z_diff_condition_number_xI}
\end{align}
where the first step follows from $(a+b)^2\leq 2a^2+2b^2$, the second step follows from Lemma \ref{lem:z_diff_norm_condition_number_xI_pre}, the third step follows from Lemma \ref{lem:z_diff_xIH2xHI}, the forth step follows from $(a+b)^2\leq 2a^2+2b^2$, the last step follows from Claim \ref{clm:x_I2x}.
For the second term, we have that
\begin{align*}
    &~ \| d_{I, z}(t)-d_z(t)\|_{D_U}^2 \notag \\
   \lesssim &~    \frac{1}{|U|} \int_{U} | d_{I, z}(t)-d_z(t)|^2 \d t \notag \\
\lesssim &~    \frac{1}{|U|} \int_{-\infty}^\infty | d_{I, z}(t)-d_z(t)|^2 \d t \notag \\
 \lesssim &~    \frac{1}{|U|} \delta_1^2 \|x(t)\|^2_{T} \notag \\
  \lesssim &~     \delta_1^2 \|x(t)\|^2_{T} 
\end{align*}
where the first step follows from the definition of the norm, second step  is straight forward, the third step follows from Lemma \ref{lem:z_diff_xHI_sub_XHG_T_norm_is_small} with appropriate scaling, the forth step follows from $|U|\gtrsim 1$. Hence,
\begin{align}\label{eq:1p5_d_I_z_diff_condition_number_xI}
    \|  x_I(t)\cdot  e^{2\pi\i f_0 \beta} - x_I(t+\beta) \|_{D_U}^2 \lesssim \|d_z(t)\|_{D_U}^2 + \delta_1^2\|x(t)\|_{T}^2.
\end{align}

Therefore, we have that
\begin{align}
   |d_{I, z}(t)|^2 =&~ |d_{I, x}(t)|^2 \notag \\
    \lesssim &~ (|x_I(t)\cdot  e^{2\pi\i f_0 \beta} - x_I(t+\beta)|+ \delta_1 k \|x_I(t)\|_{T})^2 \notag \\
    \lesssim &~ |x_I(t)\cdot  e^{2\pi\i f_0 \beta} - x_I(t+\beta)|^2+ \delta_1^2 k^2 \|x_I(t)\|_{T}^2 \notag \\
    \lesssim &~ |x_I(t)\cdot  e^{2\pi\i f_0 \beta} - x_I(t+\beta)|^2+ \delta_1 \|x_I(t)\|_{T}^2 \notag \\
    \lesssim &~ {\min} \{\frac{k }{1-|2t/T-1|}, k^2\} \cdot  \|  x_I(t)\cdot  e^{2\pi\i f_0 \beta} - x_I(t+\beta) \|_{D_U}^2 + \delta_1 \|x_I(t)\|_{T}^2\notag \\
   \lesssim &~ {\min} \{\frac{k }{1-|2t/T-1|}, k^2\} \cdot  \|  x_I(t)\cdot  e^{2\pi\i f_0 \beta} - x_I(t+\beta) \|_{D_U}^2 + \delta_1 \|x(t)\|_{T}^2 \notag \\
\lesssim &~ {\min} \{\frac{k }{1-|2t/T-1|}, k^2\} \cdot  (\|d_z(t)\|_{D_U}^2  +   \delta_1^2 \|x(t)\|^2_{T}) + \delta_1 \|x(t)\|_{T}^2 \notag \\
\lesssim &~ {\min} \{\frac{k }{1-|2t/T-1|}, k^2\} \cdot  \|d_z(t)\|_{D_U}^2  + \delta_1 \|x(t)\|_{T}^2, 
\label{eq:2_d_I_z_diff_condition_number_xI}
\end{align}
where the first step follows from Lemma \ref{lem:z_diff_xIH2xHI}, the second step follows from Lemma \ref{lem:z_diff_condition_number_xI_pre}, the third step follows from $(a+b)^2\leq 2a^2+2b^2$, the forth step follows from $\delta_1 k^2 \leq 1$, the fifth step follows from Eq.~\eqref{eq:1_d_x_I_z_diff_condition_number_xI}, the six step follows from Claim \ref{clm:x_I2x}, the seventh step follows from Eq.~\eqref{eq:1p5_d_I_z_diff_condition_number_xI}, the last step follows from $\delta_1 k^2 \lesssim 1 $.

Finally, we have
\begin{align*}
    |d_z(t)|^2 \leq &~ 2|d_z(t)-d_{I,z}(t)|^2 + 2|d_{I, z}(t)|^2 \notag \\
    \leq &~ 2\delta_1^2 T|S| \cdot  \|x(t) \|_{T}^2 + 2|d_{I, z}(t)|^2 \notag \\
    \leq &~ 2\delta_1 \cdot  \|x(t) \|_{T}^2 + 2|d_{I, z}(t)|^2 \notag \\
    \lesssim &~ \delta_1 \cdot  \|x(t) \|_{T}^2 +  {\min} \{\frac{k }{1-|2t/T-1|}, k^2\} \cdot  \|d_z(t)\|_{D_U}^2   
\end{align*}
where the first step follows from $ (a+b)^2\leq 2a^2+2b^2$, the second step follows from Lemma \ref{lem:z_diff_xHG_sub_XHI_is_small}, the third step follows from $\delta_1 T|S| \leq 1$, the forth step follows from Eq.~\eqref{eq:2_d_I_z_diff_condition_number_xI}. 

The lemma is then proved.

\end{proof}

\begin{claim}[Energy Reduction by Ideal Filter]\label{clm:x_I2x}

Let $H$ be defined as in Definition~\ref{def:effect_H_k_sparse}, ${G}^{(j)}_{\sigma,b}$ be defined as in Definition \ref{def:G_j_sigma_b} and $I=I_{\sigma,b}^{(j)}$ be the corresponding ideal filter as in Eq.~\eqref{eq:def_ideal_filter}. 

For any $x\in {\cal F}_{k,F}$, for any $(\sigma,b)$ such that Large Offset event does not happen, then we have %
\begin{align*}
\|(x*I)(t) \|_T\lesssim \|x(t)\|_T
\end{align*}
\end{claim}
\begin{proof}
Let $S=\supp(\hat{x}* \hat{H})$.
We have that
\begin{align*}
T \|(x*I)(t) \|_T \lesssim &~ T \|(x*I)(t) \cdot H(t)\|_T \\
=&~ \int_0^T |(x*I)(t) \cdot H(t)|^2 \d t \\
\leq&~ \int_{-\infty}^\infty |(x*I)(t)\cdot H(t)|^2 \d t \\
=&~ \int_{-\infty}^\infty |(\wh{x}\cdot \wh{I})(f)*\wh{H}(f)|^2 \d f \\
=&~ \int_S |(\wh{x}\cdot \wh{I})(f)*\wh{H}(f)|^2 \d f \\
=&~ \int_S |\wh{x}(f)*\wh{H}(f)|^2 \d f \\
\leq &~ \int_{-\infty}^\infty |\wh{x}(f)*\wh{H}(f)|^2 \d f \\
= &~ \int_{-\infty}^\infty |x\cdot H(t)|^2 \d t \\
\lesssim &~ \int_0^T |x(t)|^2 \d t \\
= &~ T \|x(t)\|_T^2
\end{align*}
where the first step follows from Lemma \ref{lem:property_of_filter_H} Property \RN{5}, the second step follows from the definition of the norm, the third step is straight forward, the forth step follows from Parseval's theorem, the fifth and sixth steps follow from Large Offset event not happening, the seventh step is straight forward, the eighth step follows from Parseval's theorem, the ninth step follows from Lemma \ref{lem:property_of_filter_H} Property \RN{4} and \RN{6}, the last step follows from the definition of the norm.

\end{proof}

\begin{figure}%
    \centering
    {\includegraphics[width=\textwidth]{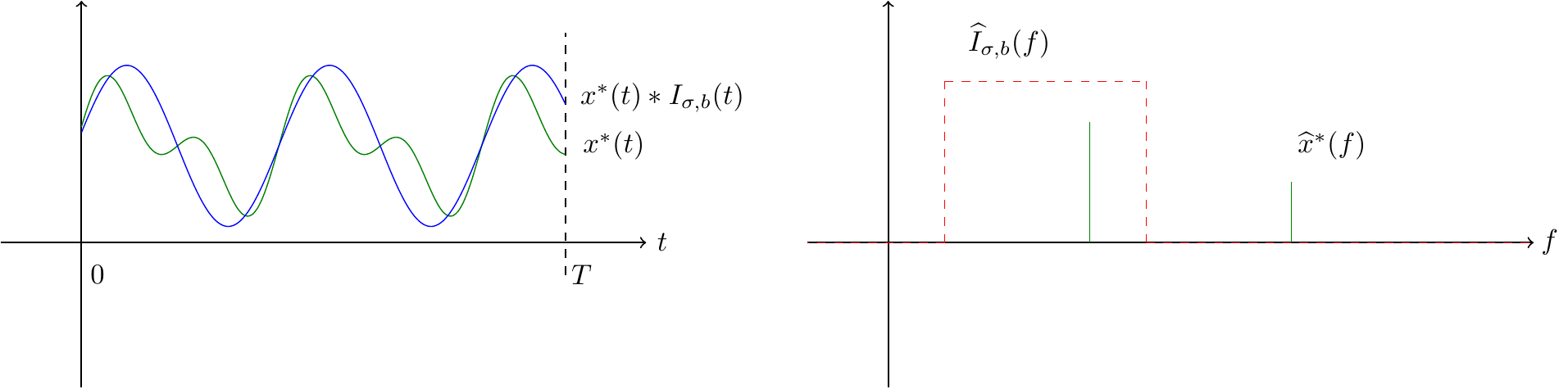}}
\caption{An illustration of the energy reduction by ideal filter. $I_{\sigma, b}$ is the ideal filter and $x^*(t)$ is a Fourier sparse signal. The energy of $x^*(t)$ in duration $[0, T]$ is reduced by applying the ideal filter, i.e., $\|x^**I_{\sigma,b}(t)\|_T\lesssim \|x^*(t)\|_T$. 
}\label{fig:mono}
\end{figure}

\section{Empirical Energy Estimation}

\label{sec:emp_energy_est}

The goal of this section is to show how to estimate a signal's energy using a few samples. We start with a general sampling and reweighing method (see Section~\ref{sec:emp_energy_est:sampling_reweighing}). Then, combining with the energy bounds derived in previous section, we obtain sample-efficient energy estimation methods for Fourier-sparse signals and filtered signals (see Section~\ref{sec:emp_energy_est:fourier_filtered}).
We further extend our methods to estimate the energy of filtered signals and local-test signals within a \emph{sub-interval} in the time duration (see Section~\ref{sec:emp_energy_est:filtered_local}). Finally, we prove several technical lemmas (see Section~\ref{sec:emp_energy_est:technical_lemmas}).

Throughout this section, for the convenience, we use a slightly different notation for the $T$-norm:
\begin{align*}
\|z\|_T^2:=\frac{1}{2T}\int_{-T}^T |z(t)|^2 \d t.  
\end{align*}
This results of using this $T$-norm is equivalent with the result of the norm taking on $[0, T]$, since we can always re-scaling the signal and transform the result into the new $T$-norm result. %

\subsection{Sampling and reweighing}\label{sec:emp_energy_est:sampling_reweighing}

In this section, we provide a generic sample-efficient method for estimating the energy of any function using discrete samples with proper weights.

\begin{lemma}\label{lem:link_weight_sampling_size__condition_num}

Let $k\in\mathbb{N}_+$ and $D$ be a probability distribution such that $\int_{-T}^T D(t) \d t= 1$. 
For any $\epsilon,\rho \in (0,1)$ and function $z:\R\rightarrow \C$, let $S_{D}=\{t_1, \cdots , t_{s}\}$ be a set of i.i.d. samples from $D$ of size
\begin{align*}
  s\ge \Big({\max}_{t\in [-T,T]}~\frac{|z(t)|^2}{D(t)}\Big)\cdot O\Big(\frac{\log(1/\rho)}{\eps^{2}T\|z(t)\|_T^2}\Big)  .
\end{align*}
Let the weight vector $w\in \R^s$ be defined by $w_i:=1/(2Ts D(t_i))$ for $i\in[s]$. 

Then with probability at least $1-\rho$, we have %
\begin{equation*}
(1-\epsilon)\| z(t)\|^2_T \leq \| z(t)\|^2_{S_D,w} \leq (1 + \epsilon)\| z(t)\|^2_T,
\end{equation*} 
where $\|z\|_T^2:=\frac{1}{2T}\int_{-T}^T |z(t)|^2 \d t$.
\end{lemma}
\begin{proof}
Let $M:={\max}_{t\in [-T,T]}\frac{|z(t)|^2}{D(t)}$. Let $z_D(t):=  \frac{1}{M} \frac{|z(t)|^2}{D(t)}$. By applying Chernoff bound (Lemma \ref{lem:chernoff_bound}) for the random variables $z_D(t_1),\dots,z_D(t_s)$, %
we get that, %
\begin{align}
    \Pr_{t_i\sim D}\Big[\Big| \sum_{i=1}^s  z_D(t_i) - \mu \Big| \leq  \epsilon\mu\Big]\geq 1-2\exp(-\epsilon^2 \mu / 3), \label{eq:norm_preseving_z_under_simplyfy}
\end{align}
where $\mu:=\sum_{i=1}^s \E_{t_i\sim D} [z_D(t_i)]=s\cdot \E_{t\sim D} [z_D(t)]$.

We first consider the expectation: 
\begin{align*}
    \E_{t\sim D} [z_D(t)] =&~ \int_{-T}^T D(t) \cdot \frac{1}{M}\frac{|z(t)|^2}{D(t) } \d t \\
    =&~ \frac{1}{ M} \int_{-T}^T | z(t)|^2 \d t\\
    =&~ \frac{2T}{M} \|z(t)\|_T^2
\end{align*}
where the first step follows from the definition of expectation, the second step is straightforward, the third step follows from the definition of the norm. Thus,
\begin{align}
    \mu =s\cdot \E_{t\sim D} [z_D(t)]= \frac{2Ts}{M} \|z(t)\|_T^2. \label{eq:sum_E_z_D}
\end{align}

Then, we consider the sum of samples:
\begin{align}
    \sum_{i=1}^s  z_D(t_i) =&~ \sum_{i=1}^s \frac{1}{M}\frac{|z(t_i)|^2}{D(t_i) }\notag \\
    = &~ \sum_{i=1}^s \frac{2w_iTs}{M} {|z(t_i)|^2} \notag\\
    = &~ \frac{2Ts}{M} \|z(t)\|^2_{S_D, w}\label{eq:sum_z_D}
\end{align}
where the first step follows from the definition of $z_D$, the second step follows from the definition of $w_i$, the last step follows from the definition of the norm.  

Putting Eqs.~\eqref{eq:norm_preseving_z_under_simplyfy} - \eqref{eq:sum_z_D} together, we get that with probability at least $1-2\exp(-\epsilon^2 \mu / 3)$,
\begin{align*}
    \Big| \frac{2Ts}{M} \|z(t)\|_{S_D, w}^2 - \frac{2Ts}{ M} \|z(t)\|_T^2 \Big| \leq  \epsilon\cdot  \frac{2Ts}{ M} \|z(t)\|_T^2,
\end{align*}
which can be simplified at follows:  
\begin{align*}
    |  \|z(t)\|_{S_D, w}^2 -  \|z(t)\|_T^2 | \leq  \epsilon\cdot   \|z(t)\|_T^2. 
\end{align*}

Finally, we need the success probability to be at least $1-\rho$, which requires that:
\begin{align*}
    1- 2 \exp\Big(-\frac{\epsilon^2}{3} \frac{2Ts}{ M} \|z(t)\|_T^2\Big)
    =&~  1- 2 \exp\Big(-\frac{\epsilon^2}{3} \frac{2Ts}{\cdot  {\max}_{t\in [-T,T]}\{{|z(t)|^2}/{D(t)}\}} \|z(t)\|_T^2\Big) \\
\ge &~  1- \rho.
\end{align*}
Hence, we need the sample complexity $s$ to be at least 
\begin{align*}
    s\geq \Big({\max}_{t\in [-T,T]}~\frac{|z(t)|^2}{D(t)}\Big)\cdot O\Big(\frac{\log(1/\rho)}{\eps^{2}T\|z(t)\|_T^2}\Big).
\end{align*}
\end{proof}

\subsection{Energy estimation for Fourier-sparse signals and filtered signals}\label{sec:emp_energy_est:fourier_filtered}
The goal of this section is to apply Lemma~\ref{lem:link_weight_sampling_size__condition_num} for Fourier-sparse signals and filtered signals. 

The following lemma defines the sampling distribution:
\begin{lemma}\label{lem:dist_well_define}
For $k\in\mathbb{N}_+$, define a probability distribution $D$ as follows:
\begin{align}\label{eq:def_dist_D}
D(t):=
\begin{cases}
{c}\cdot (1-|t/T| )^{-1}T^{-1}, & \text{ for } |t| \le T(1-{1}/k)\\
c \cdot  k T^{-1}, & \text{ for } |t|\in [T(1-{1}/k), T]
\end{cases} 
\end{align} 
where $c=\Theta(\log(k)^{-1})$ is a normalization factor such that $\int_{-T}^T D(t) \d t= 1$. Then, $D$ is well-defined.
\end{lemma}
\begin{proof}

We justify that $D$ can be normalized with $c=\Theta(\log(k)^{-1}))$.
 By the condition $\int_{-T}^T D(t) \d t= 1$, we have 
\begin{align*}
2\int_0^{T(1-{1}/k)}\frac{c}{(1-|t/T|)T}\d t + 2\int_{T(1-{1}/k)}^T  c \cdot \frac{k}{T} \d t= 1,
\end{align*}
which implies that
\begin{align*}
c^{-1} = &~2\int_0^{T(1-{1}/k)}\frac{1}{(1-|t/T|)T}\d t + 2\int_{T(1-{1}/k)}^T   \frac{k}{T} \d t\\
\eqsim &~ \log(k) +  1\\
= &~ \Theta(\log(k)).
\end{align*}
Thus, we get that $c=\Theta(\log(k)^{-1})$.

\end{proof}

The following lemma gives the sampling complexity for estimating the energy of a Fourier-sparse signal. The main idea is to apply the energy bounds in Section~\ref{sec:energy_bound}.
\begin{lemma}[Energy estimation for Fourier-sparse signals]\label{lem:x_energy_preserving_000}
Let $D$ be the probability distribution defined as Eq.~\eqref{eq:def_dist_D}.
Let $x\in {\cal F}_{k,F}$.
For any $\epsilon,\rho \in (0,1)$, let $S_{D}=\{t_1, \cdots , t_{s}\}$ be a set of i.i.d. samples from $D(t)$ of size $s \ge O(\eps^{-2}k\log(k)\log(1/\rho))$. Let the weight vector $w\in \R^s$ be defined by $w_i:=1/(2Ts D(t_i))$ for $i\in[s]$. 

Then with probability at least $1-\rho$, we have 
\begin{equation*}
(1-\epsilon)\| x(t)\|^2_T \leq \| x(t)\|^2_{S_D,w} \leq (1 + \epsilon)\| x(t)\|^2_T.
\end{equation*} 
\end{lemma}
\begin{proof}

By applying Lemma \ref{lem:link_weight_sampling_size__condition_num}, we have that the desired result satisfy when
\begin{align*}
  s\ge \Big({\max}_{t\in [-T,T]}~\frac{|x(t)|^2}{D(t)}\Big)\cdot O\Big(\frac{\log(1/\rho)}{\eps^{2}T\|x(t)\|_T^2}\Big).  
\end{align*}

By Fourier-sparse signals' energy bound (Theorem \ref{thm:energy_bound} and Theorem \ref{thm:bound_k_sparse_FT_x_middle} with  $x(t)=x(T\cdot t)$), we have that
\begin{align}
     {|x(t)|^2} 
    \lesssim&~ {\min}\Big\{ \frac{k}{1-|t/T|}, k^2\Big\} \cdot {\|x(t)\|_T^2}~~~\forall t\in [-T,T].\label{eq:_111_energy_bound_of_t}
\end{align}
Thus,
\begin{align}
    &~ {\max}_{t\in [-T,T]}\frac{|x(t)|^2}{D(t)} \notag\\
    \lesssim&~ {\max}_{t\in [-T,T]}~{\min}\Big\{ \frac{k}{1-|t/T|} , k^2  \Big\}\cdot \frac{\|x(t)\|_T^2}{D(t)} \notag\\
    \lesssim&~ {\max}_{t\in [-T,T]}~{\min}\Big\{ \frac{k}{1-|t/T|} \frac{T(1-|t/T|)}{c} , k^2 \frac{T}{ck} \Big\} \cdot \|x(t)\|_T^2 \notag\\
    =&~ kT \|x(t)\|_T^2 /c  \notag\\
    \simeq &~ k\log(k) T \|x(t)\|_T^2  ,\label{eq:_111_bound_z_div_D}
\end{align}
where the first step follows from Eq.~\eqref{eq:_111_energy_bound_of_t}, the second step follows from the definition of $D(t)$, the third step is straight forward, the forth step follows from $c=\Theta(\log(k)^{-1})$.

Hence, we get that 
\begin{align*}
    s\geq O(k\log(k) T \|x(t)\|_T^2) \cdot O\Big(\frac{\log(1/\rho)}{\eps^{2}T\|x(t)\|_T^2}\Big) = O(\eps^{-2} k\log(k)  \log(1/\rho)).
\end{align*}

The lemma is then proved.

\end{proof}

Using the energy bound for filtered signals, we immediately get the following lemma.

\begin{lemma}[Energy estimation for filtered signals]\label{lem:z_energy_preserving}
Let $D$ be the probability distribution defined as Eq.~\eqref{eq:def_dist_D}.
Let $x\in {\cal F}_{k,F}$.
Let $H$ be defined as in Definition~\ref{def:effect_H_k_sparse}. Let ${G}^{(j)}_{\sigma,b}$ be defined as in Definition \ref{def:G_j_sigma_b}.
Let $j\in [B]$ satisfying that there exists an $f^*$ with $h_{\sigma,b}(f^*)=j$ such that:
\begin{align*}
    \int_{f^*-{\Delta_h}}^{f^*+{\Delta_h}} | \widehat{x\cdot H}(f) |^2 \mathrm{d} f \geq  T\N^2/k,
\end{align*}
where $\N^2\geq \delta \|x\|_T^2$.
Let $z(t):=(x\cdot H)*{G}^{(j)}_{\sigma,b}(t)$ be the filtered signal. 

For any $\epsilon,\rho \in (0,1)$, let $S_{D}=\{t_1, \cdots , t_{s}\}$ be a set of i.i.d. samples from $D(t)$ of size $s \ge O(\eps^{-2}k\log(k)\log(1/\rho))$. Let the weight vector $w\in \R^s$ be defined by $w_i:=1/(2Ts D(t_i))$ for $i\in[s]$. 

Then when Large Offset event not happens, with probability at least $1-\rho$, we have %
\begin{equation*}
(1-\epsilon)\| z(t)\|^2_T \leq \| z(t)\|^2_{S_D,w} \leq (1 + \epsilon)\| z(t)\|^2_T.
\end{equation*} 
\end{lemma}
\begin{proof}

By applying Lemma \ref{lem:link_weight_sampling_size__condition_num}, we have that the desired result requires that
\begin{align*}
  s\ge \Big({\max}_{t\in [-T,T]}~\frac{|z(t)|^2}{D(t)}\Big)\cdot O\Big(\frac{\log(1/\rho)}{\eps^{2}T\|z(t)\|_T^2}\Big) .
\end{align*}

By the filtered signals' energy bound (Corollary \ref{cor:condition_number_z}), we have that
\begin{align}
     {|z(t)|^2} \lesssim&~ {\min}\Big\{ \frac{k\cdot H(t)+\delta }{1-|t/T|} , k^2\Big\}\cdot {\|z(t)\|_T^2}\notag\\
    \lesssim&~ {\min}\Big\{ \frac{k}{1-|t/T|} , k^2 \Big\} \cdot {\|z(t)\|_T^2}.\label{eq:energy_bound_of_t}
\end{align}
where the second step follows from $H(t)\lesssim 1$ (Lemma \ref{lem:property_of_filter_H} Property \RN{1}, \RN{2}).
Then, we get that
\begin{align}
    &~ {\max}_{t\in [-T,T]}~\frac{|z(t)|^2}{D(t)} \notag\\
    \lesssim&~ {\max}_{t\in [-T,T]}~{\min}\Big\{ \frac{k}{1-|t/T|} , k^2 \Big\} \cdot {\|z(t)\|_T^2} \notag\\
    \lesssim&~ {\max}_{t\in [-T,T]}~{\min}\Big\{ \frac{k}{1-|t/T|} \frac{T(1-|t/T|)}{c} , k^2 \frac{T}{ck} \Big\}\cdot \|z(t)\|_T^2  \notag\\
    =&~ kT \|z(t)\|_T^2 /c  \notag\\
    \simeq&~ k\log(k) T \|z(t)\|_T^2  ,\label{eq:bound_z_div_D}
\end{align}
where the first step follows from Eq.~\eqref{eq:energy_bound_of_t}, the second step follows from the definition of $D(t)$, the third step is straight forward, the forth step follows from $c=\Theta(\log(k)^{-1})$.

As a result,
\begin{align*}
    s\geq O(k\log(k) T \|z(t)\|_T^2) \cdot O\Big(\frac{\log(1/\rho)}{\eps^{2}T\|z(t)\|_T^2}\Big) = O(\eps^{-2} k\log(k)  \log(1/\rho)).
\end{align*}

The lemma is then proved.

\end{proof}

\subsection{Partial energy estimation for filtered signals and local-test signals}\label{sec:emp_energy_est:filtered_local}
In this section, we consider a variant version of energy estimation problem, which we are given a sub-interval $U\subseteq [-T,T]$ and we only want to estimate the energy within this interval. 

The following lemma gives the sampling distribution with respect to $U$.
\begin{lemma}\label{lem:dist_well_define_plus}
Let $U=[L,R]$ such that $ [-T(1-{1}/k) , T(1-{1}/k) ] \subseteq U\subseteq [-T,T]$. For $k\in\mathbb{N}_+$, define a probability distribution $D_U$ as follows:
\begin{align}\label{eq:def_dist_D_U}
D_U(t):=
\begin{cases}
{c}\cdot (1-|t/T| )^{-1}T^{-1}, & \text{ for } |t| \le T(1-{1}/k) \wedge t\in U \\
c \cdot  k T^{-1}, & \text{ for } |t|\in [T(1-{1}/k), T] \wedge t\in U
\end{cases} 
\end{align} 
where $c=\Theta(\log(k)^{-1})$ is a normalization factor such that $\int_{-T}^T D_U(t) \d t= 1$. Then, $D_U$ is well-defined.
\end{lemma}
\begin{proof}

We compute the normalization factor of $D_U$ in below. The condition that $\int_{-T}^T D_U(t) \d t= 1$ requires that 
\begin{align*}
2\int_0^{T(1-{1}/k)}\frac{c}{(1-|t/T|)T}\d t + \int_{T(1-{1}/k)}^R  c \cdot \frac{k}{T} \d t + \int_{L}^{-T(1-{1}/k)}  c \cdot \frac{k}{T} \d t= 1,
\end{align*}
which implies that
\begin{align*}
c^{-1} = &~2\int_0^{T(1-{1}/k)}\frac{1}{(1-|t/T|)T}\d t + \int_{T(1-{1}/k)}^R   \frac{k}{T} \d t + \int_{L}^{-T(1-{1}/k)}   \frac{k}{T} \d t\\
\eqsim &~ \log(k) +  1\\
= &~ \Theta(\log(k)).
\end{align*}
where the second step follows from $R\leq T$ and $L\geq -T$. 

Thus, we get that $c=\Theta(\log(k)^{-1})$.
\end{proof}

Similar to Lemma~\ref{lem:z_energy_preserving}, we have a sample-efficient approach for estimating the partial energy of a filtered signal.

\begin{lemma}[Partial energy estimation for filtered signals]\label{lem:z_energy_preserving_plus}

Let $U=[L,R]$ be such that $ [-T(1-{1}/k) , T(1-{1}/k) ] \subseteq U$. For $k\in\mathbb{N}_+$, let $D_U$ be the probability distribution defined as Eq.~\eqref{eq:def_dist_D_U}.

Let $x\in {\cal F}_{k,F}$.
Let $H$ be defined as in Definition~\ref{def:effect_H_k_sparse}, ${G}^{(j)}_{\sigma,b}$ be defined as in Definition \ref{def:G_j_sigma_b} with $(\sigma,b)$ such that Large Offset event does not happen.
For any $j\in [B]$, suppose there exists an $f^*$ with $j=h_{\sigma, b}(f^*)$ satisfying: 
\begin{align*}
    \int_{f^*-{\Delta_h}}^{f^*+{\Delta_h}} | \widehat{x\cdot H}(f) |^2 \mathrm{d} f \geq  T\N^2/k, 
\end{align*}
where $\N^2\geq \delta \|x\|_T^2$.
Let $z(t)=(x\cdot H)*{G}^{(j)}_{\sigma,b}(t)$ be the filtered signal.

For any $\epsilon,\rho \in (0,1)$, let $S_{D_U}=\{t_1, \cdots , t_{s}\}$ be a set of i.i.d. samples from $D_U$ of size $s \ge O(\eps^{-2}k\log(k)\log(1/\rho))$. Let the weight vector $w\in \R^s$ be defined by $w_i:=2/(Ts D_U(t_i))$ for $i\in[s]$. 

Then when Large Offset event not happens, with probability at least $1-\rho$, we have 
\begin{equation*}
(1-\epsilon)\| z\|^2_{U} \leq \| z\|^2_{S_{D_U},w} \leq (1 + \epsilon)\| z\|^2_{U},
\end{equation*} 
where $\|z\|_U^2:=\frac{1}{R-L}\cdot\int_{L}^R |z(t)|^2 \d t$.
\end{lemma}
\begin{proof}

By applying Lemma \ref{lem:link_weight_sampling_size__condition_num}, we have that the desired result requires that
\begin{align*}
  s\ge \Big({\max}_{t\in U}~\frac{|z(t)|^2}{D_U(t)}\Big)\cdot O\Big(\frac{\log(1/\rho)}{\eps^{2}T\|z(t)\|_T^2}\Big).
\end{align*}

The first term can be upper bounded as follows:
\begin{align}
    {\max}_{t\in U}~\frac{|z(t)|^2}{D_U(t)}
    \lesssim&~ {\max}_{t\in U}~{\min}\Big\{ \frac{k}{1-|t/T|} , k^2  \Big\}\cdot \frac{\|z(t)\|_{T}^2}{D_U(t)} \notag\\
    \lesssim&~ {\max}_{t\in U}~{\min}\Big\{ \frac{k}{1-|t/T|} \frac{T(1-|t/T|)}{c} , k^2 \frac{T}{ck} \Big\}\cdot \|z(t)\|_{T}^2 \notag\\
    =&~ kT \|z(t)\|_{T}^2 /c  \notag\\
    \lesssim&~ k\log(k) T \|z(t)\|_{T}^2  \notag\\
    \lesssim &~  k\log(k) \frac{R-L}{2T-k^2(2T+L-R)} \cdot T \|z(t)\|_{U}^2  \notag \\
    \leq &~  k\log(k)  \cdot T \|z(t)\|_{U}^2 
    ,\label{eq:bound_z_div_D_plus}
\end{align}
where the first step follows from Eq.~\eqref{eq:energy_bound_of_t}, the second step follows from the definition of $D_U(t)$, the third step is straight forward, the forth step follows from $c=\Theta(\log(k)^{-1})$, the fifth step follows from Lemma \ref{lem:norm_z_cut_preserve}, the sixth step follows from ${R-L}\leq {2T-k^2(2T+L-R)}$.  

Therefore, the sample complexity $s$ should be at least:
\begin{align*}
    s\geq O(k\log(k)  \cdot T \|z(t)\|_{U}^2)\cdot O\Big(\frac{\log(1/\rho)}{\eps^{2}T\|z(t)\|_T^2}\Big) = O(\eps^{-2} k\log(k)  \log(1/\rho)).
\end{align*}

The proof of the lemma is then completed.

\end{proof}

Recall that in Section~\ref{sec:twist_struct_FS}, we study the local-test signal $d_z(t)=z(t)e^{2 \pi \i f_0 \beta} - z(t+\beta)$. The following lemma gives a way to estimate the partial energy of a local-test signal.  It can be proved by the same  strategy with the energy bound in Lemma \ref{lem:z_diff_condition_number_xI}.

\begin{lemma}[Partial energy estimation for local-test signals]\label{lem:z_diff_energy_preserving}

Let $x\in {\cal F}_{k,F}$.
Let $H$ be defined as in Definition~\ref{def:effect_H_k_sparse}, ${G}^{(j)}_{\sigma,b}$ be defined as in Definition \ref{def:G_j_sigma_b} with $(\sigma,b)$ such that Large Offset event does not happen.
Let $z(t)=(x\cdot H)*{G}^{(j)}_{\sigma,b}(t)$ be the filtered signal.
Let $U:=\{t_0\in \R~|~ H(t) > 1-\delta_1, \forall t\in [t_0,t_0+\beta]\}$. Let $D_U$ be the probability distribution defined as Eq.~\eqref{eq:def_dist_D_U}. 
Let $D_H(t):=\mathrm{Uniform}(\{t\in \R~|~ H(t) > 1-\delta_1\})$. %

For any $\epsilon,\rho \in (0,1)$, let $S_{D_U}=\{t_1, \cdots , t_{s}\}$ be a set of i.i.d. samples from $D_U$ of size $s \ge O(k\log(k)\log(1/\rho))$. Let the weight vector $w\in \R^s$ be defined by $w_i:=2/(Ts D_U(t_i))$ for $i\in[s]$. 

Let $d_{z}(t)=z(t)e^{2 \pi \i f_0 \beta} - z(t+\beta)$ be the local-test signal. Then, with probability at least $1-\rho$, we have
\begin{align*}
\| d_z(t)\|^2_{S_{D_U},w} \leq  2\| d_z(t)\|^2_{U}
+  \sqrt{\delta_1} \|x(t)\|_{T}\cdot  \| d_z(t)\|_{U}.
\end{align*} 
\end{lemma}
\begin{proof}

By Lemma \ref{lem:link_weight_sampling_size__condition_num}, we have that when
\begin{align*}
  s\ge \Big({\max}_{t\in U}~\frac{|d_z(t)|^2}{D_U(t)}\Big)\cdot O\Big(\frac{\log(1/\rho)}{\xi^{2}|U|\cdot \|d_z(t)\|_U^2}\Big),
\end{align*}
 the following result holds with probability at least $1-\rho$,
 \begin{align}
 \|d_z(t)\|^2_{S_{D_U},w}\in (1\pm \xi)\|d_z(t)\|_{U}^2, \label{eq:z_diff_norm_preserving_condition}
 \end{align}
where $\xi$ is a parameter to be chosen later.

By the energy bound for local-test signals (Lemma \ref{lem:z_diff_condition_number_xI}), we have that for any $t\in U$,
\begin{align}
  {|d_z(t)|^2} \lesssim {\min}\Big\{ \frac{k}{1-|t/T|} , k^2  \Big\} \cdot {\|d_z(t)\|_{U}^2} +\delta_1 \|x(t)\|_{T}^2.
    \label{eq:z_diff_energy_bound_of_t}
\end{align}
Then, we get that
\begin{align}
&~ {\max}_{t\in U}~\frac{|d_z(t)|^2}{D_U(t)} \notag\\
\lesssim&~ {\max}_{t\in U}~{\min}\Big\{ \frac{k}{1-|t/T|} , k^2  \Big\}\cdot \frac{\|d_z(t)\|_{U}^2+  \delta_1 \|x(t)\|_{T}^2 }{D_U(t)} \notag\\
\lesssim&~ {\max}_{t\in U}~{\min}\Big\{ \frac{k}{1-|t/T|} \frac{T(1-|t/T|)}{c} , k^2 \frac{T}{ck} \Big\} \cdot (\|d_z(t)\|_{U}^2+\delta_1 \|x(t)\|_{T}^2) \notag\\
=&~ kT c^{-1} \cdot (\|d_z(t)\|_{U}^2+\delta_1 \|x(t)\|_{T}^2)  \notag\\
\simeq&~ k\log(k) T \cdot (\|d_z(t)\|_{D_U}^2+\delta_1 \|x(t)\|_{D_1}^2)  ,\label{eq:z_diff_bound_z_div_D}
\end{align}
where  the first step follows from Eq.~\eqref{eq:z_diff_energy_bound_of_t}, the second step follows from the definition of $D_U(t)$, the third step is straight forward, the forth step follows from $c=\Theta(\log(k)^{-1})$.

As a result, the sample complexity is
\begin{align*}
s\geq &~k\log(k) T \cdot (\|d_z(t)\|_{D_U}^2+\delta_1 \|x(t)\|_{D_1}^2)\cdot O\Big(\frac{\log(1/\rho)}{\xi^{2}|U|\cdot \|d_z(t)\|_U^2}\Big)\\
\simeq &~ \xi^{-2}\cdot k\log(k)\cdot ( 1+\frac{\delta_1 \|x(t)\|_{T}^2}{\|d_z(t)\|_{U}^2} ) \cdot \log(1/\rho) \\ 
= &~  k\log(k)\cdot \log(1/\rho) ,
\end{align*}
where the first step follows from Eq.~\eqref{eq:z_diff_bound_z_div_D}, the second step follows from $|U| \gtrsim T$, the third step follows by taking $\xi$ to be such that 
\begin{align*}
\xi^{-2}( 1+\frac{\delta_1 \|x(t)\|_{T}^2}{\|d_z(t)\|_{U}^2} ) \simeq  1 .
\end{align*}

It remains to bound the estimation error. We have that
 \begin{align*}
\|d_z(t)\|^2_{S_D,w}
 \leq &~ (1+\xi)\|d_z(t)\|_{U}^2 \\
 \simeq &~ \Big(1+\sqrt{1+\frac{\delta_1 \|x(t)\|_{D_1}^2}{\|d_z(t)\|_{U}^2}}\Big)\|d_z(t) \|_{U}^2 \\
  \leq &~ \Big(2+{\frac{\sqrt{\delta_1}  \|x(t)\|_{T}}{\|d_z(t)\|_{U}}}\Big)\|d_z(t)\|_{U}^2 \\
  \leq &~ 2\|d_z(t)\|_{U}^2 +{{\sqrt{\delta_1}  \|x(t)\|_{T}}\cdot {\|d_z(t)\|_{U}}}
 \end{align*}
where the first step follows from Eq.~\eqref{eq:z_diff_norm_preserving_condition}, the second step follows from the setting of $\eps$, the third step follows from $\sqrt{a+b}\leq\sqrt{a}+\sqrt{b}$, the forth step is straight forward.

The lemma is then proved.
\end{proof}

\subsection{Technical lemmas}\label{sec:emp_energy_est:technical_lemmas}

We prove two technical lemmas in this section. 

The following lemma bounds the energy of a Fourier-sparse signal within time duration $[L,R]\subseteq [-T,T]$ by its total energy.

\begin{lemma}[Partial energy of Fourier-sparse signal]\label{lem:norm_x_cut_preserve}

Given $k\in \Z_+, F\in\R_+$.
For any $x\in {\cal F}_{k,F}$, $[L, R]\subseteq [-T(1-O(\frac{1}{k^2})), T(1-O(\frac{1}{k^2}))]$, we have that, 
\begin{align*}
    \frac{2T - k^2 (2T+L-R)}{R-L}\|x(t)\|_{T}^2 \lesssim \frac{1}{R-L}\int_{L}^R |x(t)|^2 \d t  \leq \frac{2T}{R-L}\|x(t)\|_{T}^2.
\end{align*}
\end{lemma}
\begin{proof}
For the upper bound, we have that
\begin{align*}
    \frac{1}{R-L}\int_{L}^R |x(t)|^2 \d t \leq \frac{1}{R-L}\int_{-T}^T |x(t)|^2 \d t \leq \frac{2T}{R-L}\|x(t)\|_{T}^2,
\end{align*}
where the first step is straight forward, the second step follows from the definition of the norm. 

For the lower bound, we have that
\begin{align*}
    \int_{L}^R |x(t)|^2 \d t = &~ \int_{-T}^T |x(t)|^2 \d t - \int_{-T}^L |x(t)|^2 \d t - \int_{R}^T |x(t)|^2 \d t \\
    \geq &~ 2T\|x(t)\|_T^2 - (L+T)\cdot  {\max}_{t\in[-T, L]} |x(t)|^2  - (T-R)\cdot {\max}_{t\in[R, T]} |x(t)|^2  \\
    \gtrsim &~ 2T\|x(t)\|_T^2  - (L+T)\cdot k^2 \|x(t)\|_T^2  - (T-R)\cdot k^2 \|x(t)\|_T^2  \\
    = &~(2T - k^2 (2T+L-R))\|x(t)\|_T^2,
\end{align*}
where the first step is straight forward, the second step follows from the definition of the norm, the third step follows from Theorem \ref{thm:energy_bound}, the forth step is straight forward.

\end{proof}

By replacing the energy bound for Fourier-sparse signals with the energy bound for filtered signals (Corollary \ref{cor:condition_number_z}), we obtain the following lemma:

\begin{lemma}[Partial energy of filtered signal]\label{lem:norm_z_cut_preserve}
Given $k\in \mathbb{N}$ and $F\in \R_+$. Let $x\in {\cal F}_{k,F}$.
Let $H$ be defined as in Definition~\ref{def:effect_H_k_sparse}, ${G}^{(j)}_{\sigma,b}$ be defined as in Definition \ref{def:G_j_sigma_b} with $(\sigma,b)$ such that Large Offset event does not happen.

For any $j\in [B]$, suppose there exists an $f^*$ with $j=h_{\sigma, b}(f^*)$ satisfying: 
\begin{align*}
    \int_{f^*-{\Delta_h}}^{f^*+{\Delta_h}} | \widehat{x\cdot H}(f) |^2 \mathrm{d} f \geq  T\N^2/k, 
\end{align*}
where $\N^2\geq \delta \|x\|_T^2$.
Then, for $z(t)=(x\cdot H)*{G}^{(j)}_{\sigma,b}(t)$, we have that %
\begin{align*}
\frac{2T - k^2 (2T+L-R)}{R-L}\|z(t)\|_{T}^2 \lesssim \frac{1}{R-L}\int_{L}^R |z(t)|^2 \d t  \leq \frac{2T}{R-L}\|z(t)\|_{T}^2.
\end{align*}
\end{lemma}
\begin{proof}
For the upper bound, we have that
\begin{align*}
    \frac{1}{R-L}\int_{L}^R |z(t)|^2 \d t \leq \frac{1}{R-L}\int_{-T}^T |z(t)|^2 \d t \leq \frac{2T}{R-L}\|z(t)\|_{T}^2,
\end{align*}
where the first step is straight forward, the second step follows from the definition of the norm. 

For the lower bound, we have that
\begin{align*}
    \int_{L}^R |z(t)|^2 \d t = &~ \int_{-T}^T |z(t)|^2 \d t - \int_{-T}^L |z(t)|^2 \d t - \int_{R}^T |z(t)|^2 \d t \\
    \geq &~ T\|z(t)\|_T^2 - (L+T)\cdot  {\max}_{t\in[0, L]} |z(t)|^2  - (T-R)\cdot {\max}_{t\in[R, T]} |z(t)|^2  \\
    \gtrsim &~ T\|z(t)\|_T^2  - (L+T)\cdot k^2 \|z(t)\|_T^2  - (T-R)k^2 \|z(t)\|_T^2  \\
    = &~(2T - k^2 (2T+L-R))\|z(t)\|_T^2
\end{align*}
where the first step is straight forward, the second step follows from the definition of the norm, the third step follows from Corollary \ref{cor:condition_number_z}, the forth step is straight forward.

\end{proof}

\newpage

\section{Generate Significant Samples}
\label{sec:gen_sigi_sam}
In this section, we show our significant sample generation procedure for noisy signals.
Recall that we use $x^*(t)$ to denote the ground-truth $k$-Fourier-sparse signal and $x(t)=x^*(t)+g(t)$ to denote the observation signal. We first generalize the energy estimation method in previous section to the noisy signals (see Section~\ref{sec:gen_sigi_sam:noisy}). Then, we give our significant sample generation algorithm for a single bin (see Section~\ref{sec:gen_sigi_sam:single_bin}). Next, we show how to adapt our significant sample generation algorithm for multiple bins (see Section~\ref{sec:gen_sigi_sam:multiple_bins}). In addition, we provide some technical claims (see Section~\ref{sec:gen_sigi_sam:technical_claims}).

\subsection{Energy estimation for noisy signals}\label{sec:gen_sigi_sam:noisy}
In this section, we generalize our methods in Section~\ref{sec:emp_energy_est} to estimate the (partial) energy of the true observing signals, which contains some noise. 

In the following lemma, we show that the energy of the filtered signal $z(t)$ can be estimated with a few samples, assuming it contains a small fraction of noise.

\begin{lemma}\label{lem:significant_samples_z_below}
Let $x^*\in {\cal F}_{k,F}$ be the ground-truth signal and $x(t)=x^*(t)+g(t)$ be the noisy observation signal.
Let $H$ be defined as in Definition~\ref{def:effect_H_k_sparse}, ${G}^{(j)}_{\sigma,b}$ be defined as in Definition \ref{def:G_j_sigma_b} with $(\sigma,b)$ such that Large Offset event does not happen.
For any $j\in [B]$, suppose there exists an $f_0$
with $j=h_{\sigma, b}(f_0)$ satisfying: 
\begin{align*}
    \int_{f_0-{\Delta_h}}^{f_0+{\Delta_h}} | \widehat{x^*\cdot H}(f) |^2 \mathrm{d} f \geq  T\N^2/k, 
\end{align*}
where $\N^2\geq \delta \|x^*\|_T^2$.
Let $z^*(t):=(x^*\cdot H)*G^{(j)}_{\sigma, b}(t)$ and $z(t)=(x\cdot H)*{G}^{(j)}_{\sigma,b}(t)$.
Let $g_z(t):=z(t)-z^*(t)$. Let $U=\{t_0\in \R~|~ H(t) > 1-\delta_1~ \forall t\in [t_0,t_0+\beta]\}$.  Suppose that $ \|g_z(t)\|_T^2 \leq c \|z^*(t)\|_{U}^2$, where $c\in (0, 0.001)$ is a small universal constant. 

For $s\geq O(k\log(k)\log(1/\rho))$, let $S_{D_U}=\{t_1,\dots,t_s\}$ be a set of i.i.d. samples from the distribution $D_U$ defined as Eq.~\eqref{eq:def_dist_D_U}. Let the weights $w_i=1/(TsD_{U}(t_i))$ for $i\in  [s]$.

Then, with probability at least $0.85$,
\begin{equation*}
\|z(t)\|_{S_{D_U},w}^2 \geq ( 0.2 - 20 c )\cdot \|z^*(t)\|_{U}^2
\end{equation*}
\end{lemma}
\begin{proof}

We consider the expectation of $\|g_z(t)\|_{S_{U}, w}^2$ first.
\begin{align}
 \E\Big[\sum_{j=1}^s w_i|g_z(t_j)|^2\Big] =&~ \sum_{j=1}^s\E_{t_j\sim D_U} [w_i|g_z(t_j)|^2] \notag \\
  =&~ \sum_{j=1}^s\E_{t_j\sim D_U} \Big[\frac{1}{T s D_U(t_i)}|g_z(t_j)|^2\Big] \notag\\
  \leq &~ \E_{t\sim D_U} \Big[\frac{1}{T D_U(t)}|g_z(t)|^2\Big] \notag \\
  \leq &~ \int_U \frac{1}{T }|g_z(t)|^2 \d t\notag \\
  \leq &~  \frac{1}{T } \int_0^T|g_z(t)|^2 \d t \notag \\
  \leq &~ \|g_z(t)\|_T^2  \label{eq:3_good_samples_estimate_z_diff}
\end{align}
where the first step is straight forward, the second step follows from the definition of $w_i$, the third step is straight forward, the forth step follows from the definition of expectation, the fifth step follows from $U\subseteq [0, T]$, the sixth step follows from the definition of the norm.

By Eq.~\eqref{eq:3_good_samples_estimate_z_diff} and Markov inequality, we have that with probability at least $0.9$, 
\begin{align}
    \sum_{j=1}^s w_i|g_z(t_j)|^2 \leq 20 \|g_z(t)\|_T^2. \label{eq:22_good_samples_estimate_z_diff}
\end{align}

Then, we have that
\begin{align*}
\sum_{j=1}^s w_i|z(t_j)|^2 \geq &~  0.5 \sum_{j=1}^s w_i|z^*(t_j)|^2 -  \sum_{j=1}^s w_i|g_z(t_j)|^2\notag \\   
\geq &~  0.5 \sum_{j=1}^s w_i|z^*(t_j)|^2 - 20 \|g_z(t)\|_T^2 \notag \\ 
\geq &~  0.2 \|z^*(t)\|_{U}^2  - 20 \|g_z(t)\|_T^2 \notag \\
\geq &~ ( 0.2 - 20 c )\cdot \|z^*(t)\|_{U}^2
\end{align*}
where the first step follows from $ (a+b)^2\geq 0.5 a^2-b^2$, the second step follows from Eq.~\eqref{eq:22_good_samples_estimate_z_diff}, the third step follows from Lemma \ref{lem:z_energy_preserving_plus}, the forth step follows from  $\|g_z(t)\|_T^2 \leq c \|z^*(t)\|_{U}^2 $. 

The total success probability follows from a union bound: $0.9-\rho>0.85$.

The lemma is then proved.
\end{proof}

The following lemma shows how to estimate the energy of a noisy local-test signal.

\begin{lemma}\label{lem:significant_samples_z_above}

Let $x^*\in {\cal F}_{k,F}$ be the ground-truth signal and $x(t)=x^*(t)+g(t)$ be the noisy observation signal.
Let $H$ be defined as in Definition~\ref{def:effect_H_k_sparse}, ${G}^{(j)}_{\sigma,b}$ be defined as in Definition \ref{def:G_j_sigma_b} with $(\sigma,b)$ such that Large Offset event does not happen.
For any $j\in [B]$, suppose there exists an $f_0$
with $j=h_{\sigma, b}(f_0)$ satisfying: 
\begin{align*}
    \int_{f_0-{\Delta_h}}^{f_0+{\Delta_h}} | \widehat{x^*\cdot H}(f) |^2 \mathrm{d} f \geq  T\N^2/k, 
\end{align*}
where $\N^2\geq \delta \|x^*\|_T^2$.
Let $z^*(t):=(x^*\cdot H)*G^{(j)}_{\sigma, b}(t)$ and $z(t)=(x\cdot H)*{G}^{(j)}_{\sigma,b}(t)$.
Let $g_z(t):=z(t)-z^*(t)$. Let $U=\{t_0\in \R~|~ H(t) > 1-\delta_1~ \forall t\in [t_0,t_0+\beta]\}$.  Suppose that $ \|g_z(t)\|_T^2 \leq c \|z^*(t)\|_{U}^2$, where $c\in (0, 0.001)$ is a small universal constant. 

For $s\geq O(k\log(k)\log(1/\rho))$, let $S_{D_U}=\{t_1,\dots,t_s\}$ be a set of i.i.d. samples from the distribution $D_U$ defined as Eq.~\eqref{eq:def_dist_D_U}. Let $w_i=1/(TsD_{U}(t_i))$ for $i\in  [s]$. 

Then, %
with probability at least $0.85$,
\begin{equation*}
\|d_z(t)\|_{S_{D_U}, w}^2 \lesssim (c + \sqrt{\gamma^2 +\delta_1 }  )\cdot \|z^*(t)\|_{U}^2,
\end{equation*}
\end{lemma}

\begin{proof}

We first consider the expectation of $\|g_z(t) e^{2 \pi i f_0 \beta} - g_z(t+\beta)\|_{S_{D_U}, w}^2$.
We have that
\begin{align}
&~ \E\Big[\sum_{i=1}^s w_i|g_z(t_i) e^{2 \pi i f_0 \beta} - g_z(t_i+\beta)|^2\Big]\notag \\
=&~ \sum_{i=1}^s \E_{t_i\sim D_U}[w_i|g_z(t_i) e^{2 \pi i f_0 \beta} - g_z(t_i+\beta)|^2]\notag \\
= &~ \sum_{i=1}^s \E_{t_i\sim D_U}\Big[\frac{1}{Ts D_U(t_i)}|g_z(t_i) e^{2 \pi i f_0 \beta} - g_z(t_i+\beta)|^2\Big]\notag \\
= &~ \E_{t\sim D_U}\Big[\frac{1}{T D_U(t)}|g_z(t) e^{2 \pi i f_0 \beta} - g_z(t+\beta)|^2\Big]\notag \\
= &~ \int_{U} \frac{1}{T }|g_z(t) e^{2 \pi i f_0 \beta} - g_z(t+\beta)|^2 \d t \notag \\
\leq &~  \frac{4}{T} \int_{U} (|g_z(t)|^2+ |g_z(t+\beta)|^2) \d t \notag \\
\leq &~  \frac{8}{T} \int_{0}^T |g_z(t)|^2 \d t \notag \\
\leq &~ 10 \|g_z(t)\|_T^2, \label{eq:1_good_samples_estimate_z_diff}
\end{align}
where the first step is straight forward, the second step follows from the definition of $w_i$, the third step is straight forward, the forth step follows from the definition of expectation, the fifth step follows from $(a+b)^2\leq 2a^2+2b^2$, the sixth step follows from $U\subseteq [0,  T]$ and $U+\beta \subseteq [0,  T]$, the seventh step follows from the definition of the norm. 

By Eq.~\eqref{eq:1_good_samples_estimate_z_diff} and Markov inequality, we have that with probability at least $0.9$, 
\begin{align}
    \sum_{i=1}^s w_i|g_z(t_i) e^{2 \pi i f_0 \beta} - g_z(t_i+\beta)|^2 \leq 100 \|g_z(t)\|_T^2. \label{eq:2_good_samples_estimate_z_diff}
\end{align}

We have that
\begin{align*}
&~\sum_{i=1}^s w_i|z(t_i) e^{2 \pi i f_0 \beta} - z(t_i+\beta)|^2 \notag \\
\leq &~ \sum_{i=1}^s (2w_i|z^*(t_i) e^{2 \pi i f_0 \beta} - z^*(t_i+\beta)|^2 + 2w_i|g_z(t_i) e^{2 \pi i f_0 \beta} - g_z(t_i+\beta)|^2) \notag \\
\leq &~ 200 \|g_z(t)\|_T^2+ \sum_{i=1}^s 2w_i|z^*(t_i) e^{2 \pi i f_0 \beta} - z^*(t_i+\beta)|^2  \notag \\
\leq &~  200 \|g_z(t)\|_T^2+ 4\| z^*(t+\beta) - e^{2 \pi \i f_0 \beta} \cdot z^*(t)\|^2_{U} + 2 \sqrt{\delta_1} \|x(t)\|_{T} \| z^*(t+\beta) - e^{2 \pi \i f_0 \beta} \cdot z^*(t)\|_{U}  \notag \\
\lesssim &~  (200 c + 4\gamma^2 +4\delta_1) \|z^*(t)\|_{U}^2 + 2 \sqrt{\delta_1(\gamma^2 +\delta_1)} \|x^*(t)\|_{T} \|z^*(t)\|_{U} \notag \\
\lesssim &~  (c + \gamma^2 +\delta_1) \|z^*(t)\|_{U}^2 + \sqrt{\delta_1(\gamma^2 +\delta_1)\frac{k}{\delta}}  \|z^*(t)\|_{T} \|z^*(t)\|_{U} \notag \\
\lesssim &~  ((c + \gamma^2 +\delta_1)  + \sqrt{\delta_1(\gamma^2 +\delta_1)\frac{k}{\delta}  \frac{R-L}{T-k^2(T+L-R)}}  )\|z^*(t)\|_{U}^2\notag \\
\lesssim &~  (c + \sqrt{\gamma^2 +\delta_1 }  )\cdot \|z^*(t)\|_{U}^2 
\end{align*}
where the first step follows from $(a+b)^2\leq 2a^2+2b^2$, the second step follows from Eq.~\eqref{eq:2_good_samples_estimate_z_diff}, the third step follows from  the partial energy estimation for local-test signal (Lemma \ref{lem:z_diff_energy_preserving}) which holds with probability $1-\rho$, the forth step follows from the $\|g_z(t)\|_T^2 \leq c \|z^*(t)\|_{U}^2 $ and Claim \ref{clm:est_freq_z_in_expectation_L_R}, the fifth step follows from Lemma \ref{lem:xHG_large_than_exp_small}, the sixth step follows from $[L,R]:=U$ and Lemma \ref{lem:norm_z_cut_preserve}, the seventh step follows form ${R-L}\lesssim {T-k^2(T+L-R)}$, $\delta_1 \delta^{-1} k\lesssim 1$.  

The total success probability follows from a union bound $0.9-\rho>0.85$.

The lemma is then proved.
\end{proof}

\subsection{Significant sample generation for a single bin}\label{sec:gen_sigi_sam:single_bin}

Recall that we define a sample $t\in [0,T]$ is \emph{significant} if the magnitude of the local-test signal at $t$ is small, i.e., $|d_z(t)|\leq O(|z(t)|)$. The following lemma shows that a significant sample can be efficiently generated, provided that the filtered noisy signal does not contain too much noise.

\begin{lemma}[Generate Significant samples for filtered noisy  signals]\label{lem:significant_samples_z}

Let $x^*\in {\cal F}_{k,F}$ be the ground-truth signal and $x(t)=x^*(t)+g(t)$ be the noisy observation signal.
Let $H$ be defined as in Definition~\ref{def:effect_H_k_sparse}, ${G}^{(j)}_{\sigma,b}$ be defined as in Definition \ref{def:G_j_sigma_b} with $(\sigma,b)$ such that Large Offset event does not happen.
For any $j\in [B]$, suppose there exists an $f_0$
with $j=h_{\sigma, b}(f_0)$ satisfying: 
\begin{align*}
    \int_{f_0-{\Delta_h}}^{f_0+{\Delta_h}} | \widehat{x^*\cdot H}(f) |^2 \mathrm{d} f \geq  T\N^2/k, 
\end{align*}
where $\N^2\geq \delta \|x^*\|_T^2$.
Let $z^*(t):=(x^*\cdot H)*G^{(j)}_{\sigma, b}(t)$ and $z(t)=(x\cdot H)*{G}^{(j)}_{\sigma,b}(t)$.
Let $g_z(t):=z(t)-z^*(t)$. Let $U=\{t_0\in \R~|~ H(t) > 1-\delta_1~ \forall t\in [t_0,t_0+\beta]\}$.  Suppose that $ \|g_z(t)\|_T^2 \leq c \|z^*(t)\|_{U}^2$, where $c\in (0, 0.001)$ is a small universal constant.

Then, there is an algorithm  
that takes $O(k\log(k))$ samples in $z$, runs in $O(k\log(k))$ time, and output an $\alpha\in U$
such that with probability at least $0.6$,
\begin{equation*}
|z(\alpha + \beta) - z(\alpha)e^{2 \pi \i f_0 \beta}|^2 \le O({c + \sqrt{\gamma^2 +\delta_1 }  }) |z(\alpha)|^2\leq 0.01|z(\alpha)|^2.
\end{equation*}
\end{lemma}
\begin{proof}

The output $\alpha$ is sample in two steps:
\begin{enumerate}
    \item For $s\geq O(k\log(k))$, generate $s$ i.i.d. samples  $S_{D_U}=\{t_1,\dots,t_s\}$ bfrom the distribution $D_U$ defined as Eq.~\eqref{eq:def_dist_D_U}. Let $w_i=1/(TsD_{U}(t_i))$ for $i\in  [s]$ be the weights.
    \item Define a probability distribution $D_S$ such that
    \begin{align}\label{eq:def_dist_D_S}
        D_S(t_i):=\frac{w_i|z(t_i)|^2}{\sum_{i\in [s]}w_i |z(t_i)|^2}~~~\forall i\in [s].
    \end{align}
    And sample $\alpha$ according to $D_S$.
\end{enumerate}
The sample and time complexities of this procedure are straightforward. It remains to prove that $\alpha$ satisfies the significance requirement stated in the lemma.

By Lemma \ref{lem:significant_samples_z_below}, we have that with probability at least 0.85,
\begin{align}
 \sum_{j=1}^s w_i|z(t_j)|^2 \geq ( 0.2 - 20 c )\cdot \|z^*(t)\|_{U}^2\label{eq:down_good_samples_estimate_z_diff}
\end{align}

By Lemma \ref{lem:significant_samples_z_above}, we have that with probability at least 0.85,
\begin{align}
    \sum_{i=1}^s w_i|z(t_i) e^{2 \pi i f_0 \beta} - z(t_i+\beta)|^2 \lesssim (c + \sqrt{\gamma^2 +\delta_1 }  )\cdot \|z^*(t)\|_{U}^2 \label{eq:up_good_samples_estimate_z_diff}
\end{align}

Thus, with probability at least 0.7, 
\begin{align}
&~\frac{\sum_{i=1}^s w_i|z(t_i) e^{2 \pi i f_0 \beta} - z(t_i+\beta)|^2}{\sum_{j=1}^s w_i|z(t_j)|^2} \notag\\
\leq &~\frac{O(c + \sqrt{\gamma^2 +\delta_1 }  )\cdot \|z^*(t)\|_{U}^2}{\sum_{j=1}^s w_i|z(t_j)|^2} \notag\\
\leq &~  \frac{O(c + \sqrt{\gamma^2 +\delta_1 }  )\cdot \|z^*(t)\|_{U}^2}{( 0.2 - 20 c )\cdot \|z^*(t)\|_{U}^2 } \notag \\
= &~  O({c + \sqrt{\gamma^2 +\delta_1 }  })
\label{eq:bound_sum_frac_z_diff_z_magnitude}
\end{align}
where the first step follows from Eq.~\eqref{eq:up_good_samples_estimate_z_diff}, the second step follows from Eq.~\eqref{eq:down_good_samples_estimate_z_diff}, the third step is straight forward.

For a random sample $\alpha \sim D_S$, we bound the following expectation:
\begin{align*}
&~\E_{\alpha \sim D_S}\left[ \frac{|z(\alpha) e^{2 \pi i f_0 \beta} - z(\alpha+\beta)|^2}{|z(\alpha)|^2} \right]\\
=&~ \sum_{i=1}^s  \frac{w_i|z(t_i)|^2}{\sum_{j=1}^s w_j|z(t_j)|^2}\cdot \frac{|z(t_i) e^{2 \pi i f_0 \beta} - z(t_i+\beta)|^2}{|z(t_i)|^2}\\
=&~   \frac{\sum_{i=1}^s w_i |z(t_i) e^{2 \pi i f_0 \beta} - z(t_i+\beta)|^2}{\sum_{j=1}^s w_j|z(t_j)|^2} \\
\leq &~ O({c + \sqrt{\gamma^2 +\delta_1 }  }),
\end{align*}
where the first step follows from the definition of $D_m$, the second step is straightforward, the third step follows from Eq.~\eqref{eq:bound_sum_frac_z_diff_z_magnitude}.

Thus by Markov inequality, with probability $0.9$,
\begin{align*}
\frac{|z(\alpha) e^{2 \pi i f_0 \beta} - z(\alpha+\beta)|^2}{|z(\alpha)|^2}\le \frac{O({c + \sqrt{\gamma^2 +\delta_1 }  })}{0.1} = O({c + \sqrt{\gamma^2 +\delta_1 }  }).
\end{align*}

The success probability follows from a union bound.  And the second inequality follows from the range of the parameters $c,\gamma, \delta_1$.
\end{proof}

\subsection{Significant sample generation for multiple bins}\label{sec:gen_sigi_sam:multiple_bins}

In this section, we present our significant sample generation procedure that simultaneously works for all ``good bins''.

\begin{algorithm}[ht]\caption{Generate Significant Samples}\label{algo:significant_samples_z}
\begin{algorithmic}[1]
\Procedure{GenerateSignificantSamples}{$z$} %
\State $B\leftarrow O(k)$
\State $U\leftarrow \{t_0\in \R| H(t) > 1-\delta_1, \forall t\in [t_0,t_0+\beta]\}$
\State $D_z \leftarrow
\begin{cases}
{c}\cdot (1-|t/T| )^{-1}T^{-1}, & \text{ for } |t| \le T(1-{1}/k) \wedge t\in U \\
c \cdot  k T^{-1}, & \text{ for } |t|\in [T(1-{1}/k), T] \wedge t\in U
\end{cases}  $ 
\State \label{ln:genSigiSam_Sklogk} $S\leftarrow O(k\log(k))$ i.i.d. samples from $D_z$
\For{$ t_i \in S$} \label{ln:genSigiSam_t_i_S}
\For{$ j \in [B]$} \label{ln:genSigiSam_j_B_compute_u}
\State $a\leftarrow t_i/\sigma$
\State $u_j \leftarrow \sum_{i\in \Z}x\cdot H(\sigma(a-j-iB )) e^{-2\pi\i \sigma b(j+iB)} G(j+iB)$ \label{ln:genSigiSam_u} \Comment{$u\in \R^B $}
\State $u_j^\beta \leftarrow \sum_{i\in \Z}x\cdot H(\sigma(a+\beta-j-iB )) e^{-2\pi\i \sigma b(j+iB)} G(j+iB)$ \label{ln:genSigiSam_u_beta} \Comment{$u^\beta \in \R^B $}
\EndFor
\State \label{ln:genSigiSam_fft} $\wh{u}=\mathrm{FFT}(u)$ 
\State \label{ln:genSigiSam_fft_beta} $\wh{u}^\beta=\mathrm{FFT}(u^\beta)$
\For{$ j \in [B]$} \label{ln:genSigiSam_assign_z_u}
\State $z_j(t_i)\leftarrow \wh{u}_j$
\State $z_j(t_i+\beta)\leftarrow \wh{u}_j^\beta$
\EndFor
\EndFor
\State \label{ln:genSigiSam_assign_w_i} $w_i \leftarrow D_z(t_i), \forall t_i\in S$
\State \label{ln:genSigiSam_weight_z_sqrt} $W \leftarrow\sum_{t_i\in S}w_i |z_j(t_i)|^2$
\State $Z_{j,1}\leftarrow 0$, $Z_{j,2}\leftarrow 0$ \Comment{$Z\in \C^{B\times 2}$}
\For{$ j \in [B]$} \label{ln:genSigiSam_jB_final}
\State $D_S(t_i)\leftarrow w_i|z_j(t_i)|^2/W, \forall t_i\in S$
\State Sample $t_i \sim D_S$ \Comment{$\alpha\in \R^B $}
\State $Z_{j,1}\leftarrow z_j(t_i)$, $Z_{j,2}\leftarrow z_j(t_i+\beta)$
\EndFor
\State \Return $Z$
\EndProcedure
\end{algorithmic}
\end{algorithm}

We first prove the correctness of Algorithm~\ref{algo:significant_samples_z}.
\begin{lemma}[Generate significant samples for different bins simultaneously]\label{lem:significant_samples_for_each_bins}
Let $x^*\in {\cal F}_{k,F}$ be the ground-truth signal and $x(t)=x^*(t)+g(t)$ be the noisy observation signal.
Let $H$ be defined as in Definition~\ref{def:effect_H_k_sparse}, ${G}^{(j)}_{\sigma,b}$ be defined as in Definition \ref{def:G_j_sigma_b} with $(\sigma,b)$ such that Large Offset event does not happen.
Let $U=\{t_0\in \R~|~ H(t) > 1-\delta_1~ \forall t\in [t_0,t_0+\beta]\}$.  %

For $j\in [B]$, let $z_j^*(t):=(x^*\cdot H)*G^{(j)}_{\sigma, b}(t)$ and $z_j(t)=(x\cdot H)*{G}^{(j)}_{\sigma,b}(t)$.
Let $g_j(t):=z_j(t)-z_j^*(t)$.
Let 
\begin{align*}
  S_{g1}:=\big\{ j\in[B]~|~ \|g_j(t)\|_T^2 \leq c \|z^*_j(t)\|_{U}^2\big\},  
\end{align*}
where $c\in (0, 0.001)$ is a small universal constant.
Let 
\begin{align*}
  S_{g2}:=\left\{ j\in[B]~\Bigg|~ \exists f_0, h_{\sigma, b}(f_0)=j ~\text{and}~ \int_{f_0-{\Delta_h}}^{f_0+{\Delta_h}} | \widehat{x^*\cdot H}(f) |^2 \mathrm{d} f \geq  T\N^2/k  \right\},  
\end{align*}
where $\N^2\geq \delta \|x^*\|_T^2$. Let $S_g=  S_{g1} \cap S_{g2}$.

There is a Procedure \textsc{GenerateSignificantSamples} (Algorithm \ref{algo:significant_samples_z}) that takes $O(k^2\log^2(k/\delta_1))$ samples in $x$, runs in $O(k^2\log^3(k/\delta_1))$ time, and for each $j\in S_g$, output $\alpha_j$ such that with probability at least $0.6$, 
\begin{equation*}
|z_j(\alpha_j + \beta) - z_j(\alpha_j)e^{2 \pi \i f_0 \beta}|^2 \le 0.01  |z_j(\alpha_j)|^2~~~\forall j\in S_g.
\end{equation*}
\end{lemma}
\begin{proof}
For $k\in\mathbb{N}_+$, define a probability distribution $D(t)$ as follows:
\begin{align*}
D(t):=
\begin{cases}
{c}\cdot (1-|t/T| )^{-1}T^{-1}, & \text{ for } |t| \le T(1-{1}/k) \wedge t\in U \\
c \cdot  k T^{-1}, & \text{ for } |t|\in [T(1-{1}/k), T] \wedge t\in U
\end{cases} 
\end{align*} 
where $c=\Theta(\log(k)^{-1})$ is a normalization factor such that $\int_{-T}^T D(t) \d t= 1$. For any $\epsilon,\rho \in (0,1)$, let $S_{D}=\{t_1, \cdots , t_{s}\}$ be a set of i.i.d. samples from $D(t)$ of size $s \ge O(k\log(k)\log(1/\rho))$. Let the weight vector $w\in \R^s$ be defined by $w_i:=1/(Ts D(t_i))$ for $i\in[s]$.

Suppose all the bins can access the same set of  time points $S_D$. Then, by Lemma \ref{lem:significant_samples_z}, for any $j\in S_g$ with probability $0.6$, we have that, 
\begin{align*}
    |z_j(\alpha + \beta) - z_j(\alpha)e^{2 \pi \i f_0 \beta}|^2 \le 0.01  |z_j(\alpha)|^2.
\end{align*}

Then, we show that the value of $z_j(t),j=1,\cdots, B$ of same set of time points $S_D$ can be compute by accessing a same set of time points in $x(t)$. By Lemma \ref{lem:hashtobins} with setting $a = \alpha /\sigma$, we have that
\begin{align*}
    z_j(\alpha) = \wh{u}_j,
\end{align*}
which is computed by the algorithm.

As a result, for each $j \in S_g$, 
\begin{align*}
|z_j(\alpha_j + \beta) - z_j(\alpha_j)e^{2 \pi \i f_0 \beta}|^2 \le 0.01  |z_j(\alpha_j)|^2~~~\forall j\in S_g,
\end{align*}
holds with probably $0.6$.

\end{proof}

We compute the time and sample complexities of Algorithm~\ref{algo:significant_samples_z} in the following two lemmas.

\begin{lemma}[Running time of Procedure \textsc{GenerateSignificantSamples} in Algorithm \ref{algo:significant_samples_z}]\label{lem:time:algo:significant_samples_z}

 Procedure \textsc{GenerateSignificantSamples} in Algorithm \ref{algo:significant_samples_z} runs in $O(k^2\log(k)\log(k/\delta_1) ) $ times.
\end{lemma}
\begin{proof}
In each call of Procedure \textsc{GenerateSignificantSamples} in Algorithm \ref{algo:significant_samples_z}, 
\begin{itemize}
    \item In line \ref{ln:genSigiSam_Sklogk}, taking $|S|$ samples runs $O(|S|)$ times.
    \item In line \ref{ln:genSigiSam_t_i_S}, the for loop repeats $|S|$ times, 
    \begin{itemize}
        \item In line \ref{ln:genSigiSam_j_B_compute_u}, the for loops repeats $B$ times and $j$ iterate from $1$ to $B$, in each loop line \ref{ln:genSigiSam_u} and line \ref{ln:genSigiSam_u_beta}, computing the summation runs in $|\{ j + i B | i\in \Z \wedge j + i B \in  \supp(G) \}| $ times. 
        \item In line \ref{ln:genSigiSam_fft} and \ref{ln:genSigiSam_fft_beta}, running Fast Fourier Transform algorithm takes $ O(B \log(B))$ time, where $B$ is the length of the vector $u$ and $u^\beta$.
        \item In line \ref{ln:genSigiSam_assign_z_u}, the for loop repeats $B$ times, each loop runs in $O(1)$ times. 
    \end{itemize}
    \item In line \ref{ln:genSigiSam_assign_w_i}, assigning $w_i$ runs in $|S|$ times. 
    \item In line \ref{ln:genSigiSam_weight_z_sqrt}, computing $ \sum_{t_i\in S}w_i |z_j(t_i)|^2$ runs in $|S|$ times. 
    \item In line \ref{ln:genSigiSam_jB_final}, the for loop repeats $B$ times, each loop runs in $O(1)$ times. 
\end{itemize}

Notice that
\begin{align*}
    \sum_{j \in [B]} |\{ j + i B | i\in \Z \wedge j + i B \in  \supp(G) \}| \leq | \supp(G)|.
\end{align*}
In the algorithm, we set the parameters:
\begin{align}
    B = ~ O(k),~~\text{and}~~~
    |S| =~ k \log (k). \label{eq:time_GDS_BS}
\end{align}
Thus,
\begin{align}
    |\supp(G)| =  O(l B/\alpha) =  k \log(k/\delta_1),\label{eq:time_GDS_G}
\end{align}
where the first step follows from Lemma \ref{lem:property_of_filter_G} Property \RN{4}, the second step follows from $\alpha \eqsim 1$ and $ l = \Theta(\log(k/\delta_1))$.  

Therefore, the time complexity in total is
\begin{align*}
    &~ O(O(|S|)+|S|\cdot (| \supp(G)| + O(B\log(B)) + B\cdot O(1) ) + |S| + |S| + B \cdot O(1))\\
    \leq &~ O(|S|\cdot (| \supp(G)| + B\log(B)))  \\
    \leq &~ O(k\log(k) \cdot (| \supp(G)| + k\log(k)))  \\
    \leq &~ O(k\log(k) \cdot (k\log(k/\delta_1) + k\log(k)))  \\
    \leq &~ O(k^2\log(k)\log(k/\delta_1) ),
\end{align*}
where the first step is straightforward, the second step follows from Eq.~\eqref{eq:time_GDS_BS}, the third step follows from Eq.~\eqref{eq:time_GDS_G}, the forth step is straight forward. 
\end{proof}

\begin{lemma}[Sample complexity of Procedure \textsc{GenerateSignificantSamples} in Algorithm \ref{algo:significant_samples_z}] \label{lem:sample:algo:significant_samples_z}
 Procedure \textsc{GenerateSignificantSamples} in Algorithm \ref{algo:significant_samples_z} takes $O(k^2\log(k)\log(k/\delta_1) )$ samples.
\end{lemma}
\begin{proof}
In each call of Procedure \textsc{GenerateSignificantSamples} in Algorithm \ref{algo:significant_samples_z}, 
\begin{itemize}
    \item In line \ref{ln:genSigiSam_t_i_S}, the for loop repeats $|S|$ times, 
    \item In line \ref{ln:genSigiSam_j_B_compute_u}, the for loops repeats $B$ times and $j$ iterate from $1$ to $B$, in each loop line \ref{ln:genSigiSam_u} and line \ref{ln:genSigiSam_u_beta}, computing the summation takes $O(|\{ \sigma(a-j-iB) | i\in \Z \wedge j + i B \in  \supp(G) \}|) $ samples. 
\end{itemize}

Following from the setting in the algorithm, we have that 
\begin{align}
    |S| =&~ k \log (k). \label{eq:samples_GDS_BS}
\end{align}
Thus,
\begin{align}
    |\supp(G)| =  O(l B/\alpha) =  k \log(k/\delta_1).\label{eq:sample_GDS_G}
\end{align}
where the first step follows from Lemma \ref{lem:property_of_filter_G} Property \RN{4}, the second step follows from $\alpha \eqsim 1$ and $ l = \Theta(\log(k/\delta_1))$.  

So, the samples complexity of Procedure \textsc{GenerateSignificantSamples} in Algorithm \ref{algo:significant_samples_z} is
\begin{align*}
    &~ |S| \cdot \sum_{j \in [B]} O(|\{ \sigma(a-j-iB) ~|~ i\in \Z \wedge j + i B \in  \supp(G) \}|) \\
    \leq &~ O(|S| \cdot | \supp(G)|) \\
    \leq &~ O(|S| \cdot k \log(k/\delta_1)) \\
      \leq &~ O(k \log (k) \cdot k \log(k/\delta_1)) \\
    = &~ O(k^2\log(k)\log(k/\delta_1) )
\end{align*}
where the first step is straight forward, the second step follows from Eq.~\eqref{eq:sample_GDS_G}, the third step follows from Eq.~\eqref{eq:samples_GDS_BS}, the forth step is straight forward. 
\end{proof}

\subsection{Technical claims}\label{sec:gen_sigi_sam:technical_claims}

We prove two technical claims in below about the local-test signals' energy reduction.

\begin{claim}[Energy decay of local-test signals]\label{clm:est_freq_z_in_expectation}

Let $x^*\in {\cal F}_{k,F}$. For any $(\sigma,b)$ such that Large Offset event does not happen and any $j\in [B]$,  suppose there exists an $f_0$ with $j=h_{\sigma, b}(f_0)$ satisfying: well-isolation conditions and
\begin{align*}
    \int_{f_0-{\Delta_h}}^{f_0+{\Delta_h}} | \widehat{x^*\cdot H}(f) |^2 \mathrm{d} f \geq  T\N^2/k, 
\end{align*}
where $\N^2\geq \delta \|x^*\|_T^2$. Let $z(t)=(x^*\cdot H)*G^{(j)}_{\sigma, b}(t)$ be the filtered signal. 

For $\beta \leq \gamma / \Delta_0$ with $\Delta_0 = O(\Delta)$, let $d_z(t)=z(t)e^{2 \pi \i f_0 \beta}-z(t+\beta)$ be the local-test signal. We have that
\begin{align*}
\|d_z(t)\|_T^2\lesssim  (\gamma^2 + \delta_1) \cdot \|z(t)\|_T^2.
\end{align*}
\end{claim}
\begin{proof}
Let $S:=\supp(\wh{x}^**\wh{H})$ be the support set of $\wh{x}^**\wh{H}$. Let 
\begin{align*}
    V:=\{f\in\R~|~\wh{G}^{(j)}_{\sigma, b}(f)\geq 1-\delta_1\}.
\end{align*}

Note that $\|d_z(t)\|_T^2$ can be expressed as follows (ignoring the $\frac{1}{T}$ factor):
\begin{align*}
    &~ \int_{0}^T |z(t+\beta) - e^{2 \pi \i f_0 \beta} \cdot z(t)|^2\d t\notag \\
\leq &~ \int_{-\infty}^\infty |z(t+\beta) - e^{2 \pi \i f_0 \beta} \cdot z(t)|^2\d t \notag \\
= &~ \int_{-\infty}^\infty |\wh{z}(f) e^{2 \pi \i f \beta} - e^{2 \pi \i f_0 \beta} \cdot \wh{z}(f)|^2\d f \notag \\
= &~ \int_{-\infty}^\infty |\wh{z}(f)|^2\cdot |e^{2 \pi \i f \beta} - e^{2 \pi \i f_0 \beta} |^2\d f \notag \\
= &~ \int_{S} |\wh{z}(f)|^2\cdot |e^{2 \pi \i f \beta} - e^{2 \pi \i f_0 \beta} |^2\d f \notag \\
= &~ \int_{S\cap V} |\wh{z}(f)|^2\cdot |e^{2 \pi \i f \beta} - e^{2 \pi \i f_0 \beta} |^2\d f + \int_{S\backslash V} |\wh{z}(f)|^2\cdot |e^{2 \pi \i f \beta} - e^{2 \pi \i f_0 \beta} |^2\d f,
\end{align*}
where the first step is straight forward, the second step follows from Parseval’s theorem, the third step is straight forward, the forth step follows from the assumption that Large Offset event does not happen, the fifth step is straight forward.

Then, for the first term, we have that
\begin{align}
   \int_{S\cap V} |\wh{z}(f)|^2\cdot |e^{2 \pi \i f \beta} - e^{2 \pi \i f_0 \beta} |^2\d f \leq&~ \int_{f_0-\Delta_0}^{f_0+\Delta_0} |\wh{z}(f)|^2\cdot |e^{2 \pi \i f \beta} - e^{2 \pi \i f_0 \beta} |^2\d f \notag \\
   \lesssim &~ \int_{f_0-\Delta_0}^{f_0+\Delta_0} |\wh{z}(f)|^2\cdot \gamma^2\d f \notag \\
   \leq&~ \gamma^2\cdot \int_{-\infty}^\infty |\wh{z}(f)|^2\d f \notag \\
   \leq&~ \gamma^2\cdot \int_{-\infty}^\infty |z(t)|^2\d t \notag \\
\leq&~ \gamma^2\cdot \int_{0}^T |z(t)|^2\d t \label{eq:bound_S_and_V}
\end{align}
where the first step follows from $S\cap V\subset [f_0-\Delta_0, f_0+\Delta_0]$ by Claim \ref{cla:PS15_hash_claims}, the second step follows from $$|e^{2 \pi \i f \beta} - e^{2 \pi \i f_0 \beta}| \leq 4\pi \beta |f-f_0|\leq 4\pi \beta{\Delta_h}_0 \lesssim \gamma,$$ 
the third step is straight forward, the forth step follows from Parseval’s theorem, the fifth step follows from Lemma \ref{lem:z_satisfies_two_properties}.%

For the second term, we have that
\begin{align}
    \int_{S\backslash V} |\wh{z}(f)|^2\cdot |e^{2 \pi \i f \beta} - e^{2 \pi \i f_0 \beta} |^2\d f = &~ \int_{S\backslash V} |(\wh{x}^**\wh{H})(f)|^2\cdot |\wh{G}^{(j)}_{\sigma, b}(f)|^2\cdot |e^{2 \pi \i f \beta} - e^{2 \pi \i f_0 \beta} |^2\d f \notag \\
    \leq &~ \int_{S\backslash V} |(\wh{x}^**\wh{H})(f)|^2\cdot \delta_1^2 \cdot |e^{2 \pi \i f \beta} - e^{2 \pi \i f_0 \beta} |^2\d f \notag \\
    \lesssim &~ \int_{S\backslash V} |(\wh{x}^**\wh{H})(f)|^2\cdot \delta_1^2 \d f \notag \\
    \leq &~ \int_{-\infty}^\infty |(\wh{x}^**\wh{H})(f)|^2\cdot \delta_1^2 \d f \notag \\
    =  &~  \delta_1^2\cdot \int_{-\infty}^\infty |(x^*\cdot H)(t)|^2 \d t \notag \\
    \lesssim  &~  \delta_1^2\cdot \int_0^T |x^*(t)|^2 \d t \label{eq:bound_S_sub_V}
\end{align}
where the first step follows from the definition of $z$, the second step follows from the definition of $V$, the third step follows from $|e^{2 \pi \i f \beta} - e^{2 \pi \i f_0 \beta} |^2\lesssim 1$, the forth step is straight forward, the fifth step follows from Parseval’s theorem, the sixth step follows from Lemma \ref{lem:property_of_filter_H} Property \RN{4} and \RN{5}.

Putting them together, we get that
\begin{align*}
\int_{0}^T |z(t+\beta) - e^{2 \pi \i f_0 \beta} \cdot z(t)|^2\d t %
\lesssim &~ \gamma^2\cdot \int_{0}^T |z(t)|^2\d t +  \delta_1^2\cdot \int_0^T |x(t)|^2 \d t \\
\lesssim &~ (\gamma^2 + \delta_1^2 \delta^{-1} k )\cdot \int_{0}^T |z(t)|^2\d t  \\
\lesssim &~ (\gamma^2 + \delta_1)\cdot \int_{0}^T |z(t)|^2\d t,  
\end{align*}
where the first step follows from Eq.~\eqref{eq:bound_S_and_V} and Eq.~\eqref{eq:bound_S_sub_V}, the second step follows from Lemma \ref{lem:xHG_large_than_exp_small}, the last step follows from $ \delta_1 \delta^{-1} k \lesssim 1$.

The proof of the lemma is then completed.
\end{proof}

Similar result also holds for the partial energy:

\begin{claim}[Partial energy decay of local-test signals]\label{clm:est_freq_z_in_expectation_L_R}

Let $x^*(t)$ be a $k$-Fourier-sparse signal.   Let $z(t):=(x^*\cdot H)*G^{(j)}_{\sigma, b}(t)$ and $d_z(t)=z(t+\beta) - e^{2 \pi \i f_0 \beta} \cdot z(t)$. Let $U=\{t_0\in \R~|~ H(t) > 1-\delta_1~ \forall t\in [t_0,t_0+\beta]\}=:[L,R]$. For $\beta \leq \gamma / \Delta_0$ with $\Delta_0 = O(\Delta)$, we have that
\begin{align*}
\int_{L}^R |d_z(t)|^2
\mathrm{d} t \lesssim (\gamma^2 +\delta_1)  \cdot \int_{L}^R |z
(t)|^2 \mathrm{d} t.
\end{align*}
\end{claim}
\begin{proof}
We have that
\begin{align*}
\int_{L}^R |z(t+\beta) - e^{2 \pi \i f_0 \beta} \cdot z(t)|^2 \mathrm{d} t \leq &~ \int_{0}^T |z(t+\beta) - e^{2 \pi \i f_0 \beta} \cdot z(t)|^2
\mathrm{d} t  \\
\leq &~ (\gamma^2 +\delta_1) \int_{0}^T | z(t)|^2
\mathrm{d} t  \\
\leq &~ (\gamma^2 +\delta_1) \frac{T}{T-k^2(T+L-R)} \int_{L}^R | z(t)|^2
\mathrm{d} t \\
\lesssim &~ (\gamma^2 +\delta_1)  \int_{L}^R | z(t)|^2
\mathrm{d} t,
\end{align*}
where the first step is straight forward, the second step follows from Claim \ref{clm:est_freq_z_in_expectation}, the third follows from Lemma \ref{lem:norm_z_cut_preserve} with time duration changed from $[-T,T]$ to $[0,T]$, the forth step follows from ${T-k^2(T+L-R)}\gtrsim T$. 
\end{proof}

\section{Frequency Estimation}
\label{sec:freq_est}

We introduce our improved frequency estimation algorithm in this section. We first show that given significant samples, we are able to estimate a specific target frequency (see Section~\ref{sec:freq_est:frequency_one_bin}). Then, we show to generalize it to simultaneously estimate frequencies for multiple bins and give our main frequency estimation algorithm (see Section~\ref{sec:freq_est:frequency_different_bins}). Next, we prove several technical claims on the votes distribution in the \textsc{ArySearch} procedure (see Section~\ref{sec:freq_est:vote}).

\subsection{Frequency estimation via significant samples}\label{sec:freq_est:frequency_one_bin}
In this section, we show an algorithm such that for a target frequency $f_0$, it can use several significant samples to estimate it with high accuracy.
The main idea is as follows: %
for a significant sample $\alpha$, since $|z(\alpha+\beta)-z(\alpha)e^{2\pi \i f_0\beta}|$ is very small, the angle of $\frac{z(\alpha+\beta)}{z(\alpha)}$ will be close to $2\pi f_0 \beta$. That is, 
\begin{align*}
    \arg\Big(\frac{z(\alpha+\beta)}{z(\alpha)}\Big)\approxeq 2\pi f_0 \beta \pmod{2\pi}.
\end{align*}
Solving the congruence equation gives that
\begin{align*}
    f_0\approx \frac{1}{2\pi \beta}\Big(\arg\Big(\frac{z(\alpha+\beta)}{z(\alpha)}\Big) + 2\pi s\Big)
\end{align*}
for some unknown $s\in \mathbb{Z}$.

To find the unknown $s$, we use the same strategy as in \cite{ps15}: perform a $D$-round searching procedure to narrow the possible range of $f_0$. More specifically, at the beginning, the possible range of $f_0$ (frequency interval) is $[-F, F]$.
And after $D$ rounds, $f_0$ is located in a frequency interval of length $O(\Delta)$, resulting in an estimate with $\Delta$ accuracy.

For $d\in [D]$, consider the $d$-th round of searching, where the frequency interval is: 
\begin{align*}
   [\L_d, \L_d + \len_d]. 
\end{align*}
We equally partition the frequency interval into $\num$ parts and do a $\num$-ary search. We generate $R$ significant samples, and for each sample $\alpha$, we enumerate all possible $s$ and compute 
\begin{align*}
    \frac{1}{2\pi \beta}\Big(\arg\Big(\frac{z(\alpha+\beta)}{z(\alpha)}\Big) + 2\pi s\Big).
\end{align*}
Then, we find which part this quantity falls in and add a vote to that part. For robustness, we also add votes to that part's left and right neighbors. In the end, the frequency interval for the next round is the part with more than $R/2$ votes. %
It is easy to see that at most $5$ parts can be selected in the new frequency interval. 
Hence, we have
\begin{align}\label{eq:len_dacay}
    \len_{d+1}\leq \frac{\len_d}{\num/5},
\end{align}
i.e., the length of the possible range  of $f_0$ decays at a constant rate.

More formally, we have the following lemma:

\begin{lemma}[significant sample to frequency estimation]\label{lem:sample2freq_est}

Suppose that there is an algorithm \textsc{GetSignificantSample} that
\begin{itemize}
    \item takes $z(t), \beta$ as input where $\beta \leq O(1/\Delta)$,
    \item takes $ {\cal S}$ samples in $z(t)$,
    \item runs in ${\cal T}$ time,
    \item outputs an $\alpha$ such that with probability $0.9$,
\begin{equation*} 
|z(\alpha + \beta) - z(\alpha)e^{2 \pi \i f_0 \beta}|^2 \le 0.0001 |z(\alpha)|^2.
\end{equation*}
\end{itemize}

Then, there is an Procedure \textsc{FrequencyEstimationZ} in Algorithm \ref{alg:locateksignal_locatekinner} that 
\begin{itemize}
    \item takes $ O(\log(FT)\log(\log(FT))\log(1/\rho_1){\cal S})$ samples,
    \item  runs in $ O(\log(FT)\log(\log(FT))\log(1/\rho_1){\cal T})$ times,
    \item and outputs $\wt{f}_0$ such that with probability at least $1-\rho_1$, 
\begin{align*}
    | \wt{f}_0 - f_0 | \lesssim \Delta. 
\end{align*}
\end{itemize}
\end{lemma}

\begin{proof}
We prove the correctness, time/sample complexity of Algorithm \ref{alg:locateksignal_locatekinner} in below. 
\paragraph{Correctness:} We first compute the value of $D$, the number of rounds needed for the searching procedure. Note that the \textsc{GetSignificantSample} procedure requires that $\beta\leq O(1/\Delta)$. In our algorithm, we take $\beta_d = O(\num/\len_d)$ for the $d$-th round. Hence, for the last round $d=D$, we have that, %
\begin{align*}
    O(\num/\len_D)=O(1/\Delta)~~\Longrightarrow~~\len_D \geq  \num \Delta. 
\end{align*}
Then, by Eq.~\eqref{eq:len_dacay} and $\len_1=F$, we get that
\begin{align}
D = \log_\num (\frac{FT}{\num (T\Delta)} ) \lesssim \log(FT)/\log(\num).\label{eq:b_D}
\end{align}

Then, we calculate the success probability. For $d\in [D]$, by Claim~\ref{clm:angle_between_xgamma_and_xgammabeta_is_betaf_sum_2}, with probability at least $1-O(c+\rho)^{R/6}$, the true part containing $f_0$ and its left and right neighbor will get $R$ votes in total, and the other far away parts will get at most $R/2$ votes. In this case, the new frequency interval will contain the true part, and we consider this round being success.  
Since the search procedure takes $D$ rounds, by a union bound, all rounds will succeed with probability at least 
\begin{align*}
1- D \cdot O(c+\rho)^{R/6}  \geq &~  1- \frac{\log(FT)}{\log(\num)}\cdot O(c+\rho)^{R/6} 
\geq 1- \rho_1
\end{align*}
where the first step follows from Eq.~\eqref{eq:b_D}, the second step follows from 
\begin{align*}
 R \geq O\Big(\frac{\log(\log(FT)/\rho_1 )}{\log(1/(c+\rho))}\Big) \geq O\Big(\frac{\log(\log(FT)/(\rho_1 \log(\num)))}{\log(1/(c+\rho))}\Big) .    
\end{align*}

Therefore, with probability at least $1-\rho$, the final frequency interval of length $O(\Delta)$ will contain the target frequency $f_0$, which means that the output $\wt{f}_0$ satisfies $|\wt{f}_0-f_0|\lesssim \Delta$. And the correctness of the algorithm is proved.

\paragraph{Time complexity:} We show that Procedure \textsc{FrequencyEstimationZ} in Algorithm \ref{alg:locateksignal_locatekinner} runs in $O(\log (FT) \cdot \log({\log(FT)}/{\rho_1}) ) $ times.

In each call of the Procedure \textsc{FrequencyEstimationZ} in Algorithm \ref{alg:locateksignal_locatekinner}, 
\begin{itemize}
    \item The for-loop repeats $ D$ times.
    \item In each loop, line \ref{ln:callArySearch} call Procedure  \textsc{ArySearch}. 
\end{itemize}
Then, in the $d$-th call of the Procedure \textsc{ArySearch}, 
\begin{itemize}
    \item In line \ref{ln:for_R_boost_prob}, the for-loop repeats $R$ times.
    \item In line~\ref{ln:call_get_sig_sample}, the Procedure \textsc{GetSignificantSample} is called.
    \item In  line \ref{ln:for_s_round_of_a_circle}, the for-loop repeats $\beta_d \len_d + O(1) $ times.
    \item In line \ref{ln:for_num_of_ary}, the for-loop repeats $ \num$ times.
\end{itemize}
Thus, the total time complexity is dominated by:
\begin{align*}
    D\cdot R \cdot (\beta_d \len_d + O(1))\cdot \num + D \cdot R \cdot {\cal T}.
\end{align*}

By the parameter settings in Algorithm \ref{alg:locateksignal_locatekinner}, we have that
\begin{align*}
\num = &~ O(1) ,\notag \\
    D=&~ O(\log (\frac{FT}{ \Delta} )/\log(\num)), \notag \\
    R =&~ O(\log(\frac{\log(FT)}{\rho_1\log(\num)})) ,\notag \\
    \beta_d =&~  O(\frac{\num}{\len_d}),
\end{align*}
In particular, we have
\begin{align*}
 D=&~ O(\log (\frac{FT}{ \Delta} )/\log(\num)) \leq O(\log (\frac{FT}{ \Delta} ))   \leq O(\log (FT))   ,
\end{align*}
where the first step follows from the setting of $D$, the second step follows from $ \num = O (1)$, the third step follows from $ \Delta = \poly(k)$. 

Hence, the total time complexity of Algorithm \ref{alg:locateksignal_locatekinner} is
\begin{align*}
    &~O(D\cdot R \cdot (\beta_d \len_d + O(1))\cdot \num) + D\cdot R \cdot {\cal T}\\
    = &~ O(D\cdot R \cdot (O(\num) + O(1))\cdot \num)+ D\cdot R \cdot {\cal T}\\
    = &~ O(D\cdot R\cdot {\cal T})\\
    = &~ O(\log (FT)\cdot \log(\log(FT)/\rho_1)\cdot {\cal T}),
\end{align*}
where the first step follows from $\beta_d =  O(\frac{\num}{\len_d})$, the second step follows from $\num=O(1)$, the third step follows from the choices of $D$ and $R$.

\paragraph{Sample complexity:} Each call of the Procedure \textsc{GetSignificantSample} takes ${\cal S}$ samples, and it is called $DR$ times. Thus, the total sample complexity of Algorithm \ref{alg:locateksignal_locatekinner} is
\begin{align*}
    DR\cdot {\cal S} =  O(\log (FT)\cdot \log(\log(FT)/\rho_1)\cdot {\cal S}).
\end{align*}

The proof of the lemma is completed.

\end{proof}

\begin{algorithm}[!ht]
\caption{Frequency Estimation of the Filtered Signal}\label{alg:locateksignal_locatekinner}
\begin{algorithmic}[1]
\Procedure{$\textsc{FrequencyEstimationZ}$}{$x,H,G^{(j)}_{\sigma, b}$}%
	\State $\num \leftarrow O(1)$, $D \leftarrow O(\log (\frac{FT}{ \Delta} )/\log(\num))$, $R \leftarrow O(\log(\frac{\log(FT)}{\rho_1\log(\num)})) $ %
	\State $\L_1\leftarrow -F$, $\len_1 \leftarrow 2F$
	\For{$d\in [D]$}
	\State $\len_d\leftarrow 5 \frac{\len_{d-1}}{ \num} $
	\State \label{ln:callArySearch} $\L_{d+1} \leftarrow  \textsc{ArySearch}$($x,H,G^{(j)}_{\sigma, b},F, T,\Delta, \L_d, \len_d,  \num$)
	\EndFor
	\State \Return $\L_D$
\EndProcedure

\Procedure{$\textsc{ArySearch}$}{$x,H,G^{(j)}_{\sigma, b},F, T,\Delta, \L_i, \len_i,  \num$}
	\State Let $v\in \Z_+^\num$ and $v_{q}\leftarrow 0$ for $q\in [\num]$
	\For {$r=1 \to R$}\label{ln:for_R_boost_prob}
\State $z(\alpha+\beta), z(\alpha) \leftarrow \textsc{GetSignificantSample}(x,H,G^{(j)}_{\sigma, b}, r, d)$ \label{ln:call_get_sig_sample}
	
	\For {$ s \in [\beta \L_d -10, \beta(\L_d+\len_d )+10]\cap \Z$} \label{ln:for_s_round_of_a_circle}
	    \State $\wt{f} = \frac{1}{2\pi\beta} ( \arg(\frac{z(\alpha+\beta)}{ z(\alpha)}) +2\pi s)$
	    \For {$q \in [\num]$} \label{ln:for_num_of_ary}
		
		\If{$ \wt{f} \in [\L_d+(q-1)\len_d/\num, \L_d+q\len_d/\num] $}
        \State $ v_{q}\leftarrow v_q + 1$
		\EndIf
	\EndFor
	\EndFor
\EndFor
\For {$q \in [\num]$}
	\If{$ v_{q} +  v_{q+1} + v_{q+2} \geq R / 2$}
	    \State $\L_{d+1} \leftarrow  \L_d+(q-1)\len_d/\num$
	    \State \Return $\L_{d+1}$
	\EndIf
\EndFor
\State \Return $\emptyset $
\EndProcedure
\end{algorithmic}
\end{algorithm}

\begin{algorithm}
\caption{Pre-computation of the Significant Samples}\label{alg:pre_compute}
\begin{algorithmic}[1]
\Procedure{$\textsc{SamplingSignificantSample}$}{$x,H,G^{(j)}_{\sigma, b},F, T, \Delta, \num, D, R $} 
\State $\len_1 \leftarrow 2F$, ${\cal L}\in \C^{D\times R \times B \times 2} $
\For{$d\in [D]$} \label{ln:alg:pre_compute_d_D}
\State $\len_d=5 \frac{\len_{d-1}}{ \num}$
\State $\wh{\beta} \leftarrow O(\frac{\num}{\len_d}) $%
	\For {$r\in [R]$} \label{ln:alg:pre_compute_r_R}
		\State Sample $\beta \in \text{Uniform}([\frac{1}{2}\wh{\beta}, \wh{\beta} ])$
\State $Z \leftarrow \textsc{GenerateSignificantSamples}(x,H,G)$ \label{ln:alg:pre_compute_gensigsam} \Comment{$Z\in \C^{B\times 2}$, see Algorithm \ref{algo:significant_samples_z}}
\For{$j\in [B]$} \label{ln:alg:pre_compute_j_B}
\State ${\cal L}_{d,r,j,1}\leftarrow Z_{j,1}$ \Comment{$ Z_{j,1}=z^{(j)}(\alpha+\beta)$}
\State ${\cal L}_{d,r,j,2}\leftarrow Z_{j,2}$ \Comment{$ Z_{j,2}=z^{(j)}(\alpha)$}
\EndFor
\EndFor
\EndFor
\State \Return ${\cal L}$
\EndProcedure

\Procedure{$\textsc{GetSignificantSample}$}{${\cal L}, d, r, j $} 
\State \Return $({\cal L}_{d,r,j,1}, {\cal L}_{d,r,j,2})$
\EndProcedure
\end{algorithmic}
\end{algorithm}

\subsection{Simultaneously estimate frequencies for different bins}\label{sec:freq_est:frequency_different_bins}

Combining the significant sample generation procedure discussed in Section~\ref{sec:gen_sigi_sam} with Algorithm~\ref{alg:locateksignal_locatekinner}, we obtain the frequency estimation algorithm that improves the algorithms in \cite{ps15} and \cite{ckps16}.

\begin{algorithm}[ht]\caption{Frequency Estimation}\label{algo:freqEstX}
\begin{algorithmic}[1]
\Procedure{FrequencyEstimationX}{$x$}
\State \label{ln:samSigiSam} ${\cal L}\leftarrow \textsc{SamplingSignificantSample}(x)$
\For{$j\leftarrow 1,\cdots,B$}
\State \label{ln:freqEstZcallfreqEstZ} $\wt{f}_j \leftarrow\textsc{FrequencyEstimationZ}(x, H, G^{(j)}_{\sigma, b})$  \Comment{$z^{(j)} = (x\cdot H)* G^{(j)}_{\sigma, b}(t)$} \label{ln:call_freq_est_z}
\EndFor
\State $L\leftarrow\{\wt{f}_1,\cdots,\wt{f}_B\}$
\State \Return $L$
\EndProcedure
\end{algorithmic}
\end{algorithm}

\begin{theorem}[Better frequency estimation algorithm]\label{thm:frequency_recovery_k_better}
Let $x^*(t) = \overset{k}{ \underset{j=1}{\sum} } v_j e^{2\pi\i f_j t}$ and $x(t)= x^*(t) +g(t)$ be the observation signal where $g(t)$ is arbitrary noise. Let $\Delta_h := O(|\supp(\wh{H})|) $, $\Delta:=O(k\cdot \Delta_h )$ and $\N^2 := \| g(t) \|_T^2 + \delta \| x^*(t) \|_T^2$. Let $H$ be defined as in Definition~\ref{def:effect_H_k_sparse}, ${G}^{(j)}_{\sigma,b}$ be defined as in Definition \ref{def:G_j_sigma_b} with $(\sigma,b)$ such that Large Offset event does not happen.
Let $U=\{t_0\in \R~|~ H(t) > 1-\delta_1~ \forall t\in [t_0,t_0+\beta]\}$.

For $j\in [B]$, let $z_j^*(t):=(x^*\cdot H)*G^{(j)}_{\sigma, b}(t)$ and $z_j(t)=(x\cdot H)*{G}^{(j)}_{\sigma,b}(t)$.
Let $g_j(t):=z_j(t)-z_j^*(t)$. Let 
\begin{align*}
  S_{g1}=\{ j\in[B]~|~ \|g_j(t)\|_T^2 \leq c \|z^*_j(t)\|_{U}^2\},  
\end{align*}
where $c\in (0, 0.001)$ is a small universal constant.
Let 
\begin{align*}
  S_{g2}=\Big\{ j\in[B]~\Big|~ \exists f_0, h_{\sigma, b}(f_0)=j ~\text{and}~ \int_{f_0-\Delta_h}^{f_0+\Delta_h} | \widehat{x^*\cdot H}(f) |^2 \mathrm{d} f \geq  T\N^2/k  \Big\}.  
\end{align*}
Let $S_g=  S_{g1} \cap S_{g2}$. Let  $S_f = \{ f_i ~|~ \exists j\in S_g:h_{\sigma, b}(f_i)=j~\forall i\in [k]\}$.

There is a Procedure \textsc{FrequencyEstimationX} in Algorithm \ref{algo:freqEstX} such that:
\begin{itemize}
    \item takes $k^2 \log(1/\delta) \log(FT)$ samples,
    \item runs in $k^2 \log(1/\delta) \log^2 (FT)$ time,
    \item returns a set $L$ of $O(k)$ frequencies such that with probability $1- \rho_0$, for any $f\in S_f$, there exists an $\wt{f} \in L$ satisfying
\begin{equation*}
|f-\widetilde{f} | \lesssim \Delta.
\end{equation*}
\end{itemize}
\end{theorem}
\begin{proof}
We prove the correctness, time complexity, and sample complexity of Algorithm \ref{algo:freqEstX} in below.

\paragraph{Correctness:}
By Lemma \ref{lem:significant_samples_for_each_bins}, we know that the Procedure \textsc{GenerateSignificantSamples} in Algorithm \ref{algo:significant_samples_z}  takes ${\cal S}=O(k^2\log^2(k/\delta_1))$ samples in $x$, runs in ${\cal T}= O(k^2\log^3(k/\delta_1))$ time, and for each $j\in S_g$, and outputs $\alpha_j$ such that for each $j\in S_g$ with probability $0.6$, \begin{equation*} |z_j(\alpha_j + \beta) - z_j(\alpha_j)e^{2 \pi \i f_0 \beta}|^2 \le 0.01  |z_j(\alpha_j)|^2,
\end{equation*}
where $f_0$ satisfies  \begin{equation}%
    \int_{f_0-\Delta}^{f_0+\Delta} | \widehat{x^*\cdot H}(f) |^2 \mathrm{d} f \geq T\N^2/k,
  \end{equation}
  and $j=h_{\sigma, b}(f_0)$.

In the line~\ref{ln:call_freq_est_z}, we call the algorithm \textsc{FrequencyEstimationZ}$(x, H, G^{(j)}_{\sigma, b})$. By Lemma \ref{lem:sample2freq_est}, \textsc{FrequencyEstimationZ}$(x, H, G^{(j)}_{\sigma, b})$ output $\wt{f}$ for each $ f_j \in S_f$ such that with probability at least $1-\rho_1$
\begin{align*}
| \wt{f} - f_j | \lesssim \Delta . 
\end{align*}

As a result, for all the $f\in S_f$, there is a $\wt{f} \in L $ such that
\begin{align*}
| \wt{f} - f | \lesssim \Delta  
\end{align*} 
holds with probability at least 
\begin{align*}
    1 - B \rho_1 \geq 1- \rho_0. 
\end{align*}

\paragraph{Time complexity:} We show that the Procedure \textsc{FrequencyEstimationX} in Algorithm \ref{algo:freqEstX} runs in 
\begin{align*}
    O(k^2\log(k)\log(k/\delta_1)\log (FT)\log({\log(FT)}/{\rho_1} )) 
\end{align*}
time.

In each call of the Procedure \textsc{FrequencyEstimationX} in Algorithm \ref{algo:freqEstX}, 
\begin{itemize}
    \item Line \ref{ln:samSigiSam} call Procedure \textsc{SamplingSignificantSample}, which runs in
    \begin{align*}
        O(k^2\log(k)\log(k/\delta_1)\log (FT)\log({\log(FT)}/{\rho_1} ))
    \end{align*}
    time by Lemma \ref{lem:time:alg:pre_compute}.
    \item The for-loop repeats $ B $ times:
    \begin{itemize}
        \item In each loop, line \ref{ln:freqEstZcallfreqEstZ} call Procedure \textsc{FrequencyEstimationZ}, which runs in $$ O(\log (FT) \cdot \log({\log(FT)}/{\rho_1}) ) $$ time by Lemma \ref{lem:sample2freq_est}. 
    \end{itemize}
\end{itemize}

Thus, the total time complexity is 
\begin{align*}
 &~ O(k^2\log(k)\log(k/\delta_1)\log (FT)\log({\log(FT)}/{\rho_1} )) + B \cdot O(\log (FT) \cdot \log({\log(FT)}/{\rho_1}) ) \\
 \leq &~ O(k^2\log(k)\log(k/\delta_1)\log (FT)\log({\log(FT)}/{\rho_1} )) + O(k) \cdot O(\log (FT) \cdot \log({\log(FT)}/{\rho_1}) ) \\
  \leq &~ O(k^2\log(k)\log(k/\delta_1)\log (FT)\log({\log(FT)}/{\rho_1} )) ,
\end{align*}
where the first step follows from $ B = O(k)$, the second step is straight forward. 

\paragraph{Sample complexity:} We show that the Procedure \textsc{FrequencyEstimationX} in Algorithm \ref{algo:freqEstX} takes 
\begin{align*}
    O(k^2\log(k)\log(k/\delta_1)\log (FT)\log({\log(FT)}/{\rho_1} ))
\end{align*}
samples.

In each call of the Procedure \textsc{FrequencyEstimationX} in Algorithm \ref{algo:freqEstX}, Line \ref{ln:samSigiSam} call Procedure \textsc{SamplingSignificantSample}, which takes $ O(k^2\log(k)\log(k/\delta_1)\log (FT)\log({\log(FT)}/{\rho_1} ))$ samples by Lemma \ref{lem:sample:alg:pre_compute}.

So, the sample complexity of Procedure \textsc{FrequencyEstimationX} in Algorithm \ref{algo:freqEstX} is 
\begin{align*}
 O(k^2\log(k)\log(k/\delta_1)\log (FT)\log({\log(FT)}/{\rho_1} )).
\end{align*}
\end{proof}

The following two lemmas shows the time complexity and sample complexity of the significant sample generation procedure in Algorithm \ref{alg:pre_compute}.

\begin{lemma}[Running time of Procedure \textsc{SamplingSignificantSample} in Algorithm \ref{alg:pre_compute}]\label{lem:time:alg:pre_compute}
Procedure \textsc{SamplingSignificantSample} in Algorithm \ref{alg:pre_compute} runs in \begin{align*}
    O(k^2\log(k)\log(k/\delta_1)\log (FT)\log({\log(FT)}/{\rho_1} ))
\end{align*} times.
\end{lemma}
\begin{proof}
In each call of Procedure \textsc{SamplingSignificantSample} in Algorithm \ref{alg:pre_compute},  in line \ref{ln:alg:pre_compute_d_D}, the for loop repeats $D$ times, in line \ref{ln:alg:pre_compute_r_R}, the for loops repeats $R$ times,
\begin{itemize}
    \item In line \ref{ln:alg:pre_compute_gensigsam}, by Lemma \ref{lem:time:algo:significant_samples_z}, each call of Procedure \textsc{GenerateSignificantSamples} takes $ O(k^2\log(k)\log(k/\delta_1) )$ times.
    \item In line \ref{ln:alg:pre_compute_j_B}, the for loop repeats $B$ times, each iteration runs in $ O(1)$ times.
\end{itemize}

Following from the setting in the algorithm, we have that
\begin{align} 
\num = &~ O(1) ,\notag \\ 
D=&~ O(\log (\frac{FT}{ \Delta} )/\log(\num)), \notag \\ 
R =&~ O(\log(\frac{\log(FT)}{\rho_1\log(\num)})) ,\notag \\
B = &~ O(k). %
\label{eq:time_SSS_BS} \end{align} 

We have that
\begin{align}
 D=&~ O(\log (\frac{FT}{ \Delta} )/\log(\num)) \leq O(\log (\frac{FT}{ \Delta} ))   \leq O(\log (FT)), \label{eq:D_time_SSS}
\end{align}
where the first step follows from the setting of $D$, the second step follows from $ \num = O (1)$, the third step follows from $ \Delta = \poly(k)$. 

We also have that
\begin{align}
\label{eq:R_time_SSS}
R =&~ O(\log(\frac{\log(FT)}{\rho_1\log(\num)})) \leq O(\log({\log(FT)}/{\rho_1})) , %
\end{align}
where the first step follows from the setting of $R$, the second step follows from $\num=O(1)$.

So, the time complexity of Procedure \textsc{SamplingSignificantSample} in Algorithm \ref{alg:pre_compute} is 
\begin{align*}
 &~ D \cdot R \cdot (O(k^2\log(k)\log(k/\delta_1) ) + B \cdot O(1)) \\
 \leq &~ O(\log (FT)) \cdot R \cdot (O(k^2\log(k)\log(k/\delta_1) ) + B \cdot O(1)) \\
 \leq &~ O(\log (FT)) \cdot O(\log({\log(FT)}/{\rho_1})) \cdot (O(k^2\log(k)\log(k/\delta_1) ) + B \cdot O(1)) \\
 \leq &~ O(\log (FT)) \cdot O(\log({\log(FT)}/{\rho_1})) \cdot (O(k^2\log(k)\log(k/\delta_1) ) + O(k) \cdot O(1)) \\
 = &~ O(k^2\log(k)\log(k/\delta_1)\log (FT)\log({\log(FT)}/{\rho_1} )) ,
\end{align*}
where the first step follows from Eq.~\eqref{eq:D_time_SSS}, the second step follows from Eq.~\eqref{eq:R_time_SSS},  the third step follows from Eq.~\eqref{eq:time_SSS_BS}, the forth step is straightforward.

\end{proof}

\begin{lemma}[Sample complexity of Procedure \textsc{SamplingSignificantSample} in Algorithm \ref{alg:pre_compute}]\label{lem:sample:alg:pre_compute}
Procedure \textsc{SamplingSignificantSample} in Algorithm \ref{alg:pre_compute} takes 
\begin{align*}
    O(k^2\log(k)\log(k/\delta_1)\log (FT)\log({\log(FT)}/{\rho_1} ))
\end{align*}
samples.
\end{lemma}
\begin{proof}
In each call of Procedure \textsc{SamplingSignificantSample} in Algorithm \ref{alg:pre_compute},  In line \ref{ln:alg:pre_compute_d_D}, the for loop repeats $D$ times, in line \ref{ln:alg:pre_compute_r_R}, the for loops repeats $R$ times,
\begin{itemize}
    \item In line \ref{ln:alg:pre_compute_gensigsam}, by Lemma \ref{lem:time:algo:significant_samples_z}, each call of Procedure \textsc{GenerateSignificantSamples} takes $ O(k^2\log(k)\log(k/\delta_1) )$ samples.
\end{itemize}

Following from the setting in the algorithm, we have that
\begin{align*} 
\num = &~ O(1) ,\notag \\ 
D=&~ O(\log (\frac{FT}{ \Delta} )/\log(\num)), \notag \\ 
R =&~ O(\log(\frac{\log(FT)}{\rho_1\log(\num)})). %
\end{align*} 

We have that
\begin{align}
 D=&~ O(\log (\frac{FT}{ \Delta} )/\log(\num)) \leq O(\log (\frac{FT}{ \Delta} ))   \leq O(\log (FT)), \label{eq:D_sample_SSS}
\end{align}
where the first step follows from the setting of $D$, the second step follows from $ \num = O (1)$, the third step follows from $ \Delta = \poly(k)$. 

We also have that
\begin{align}
\label{eq:R_sample_SSS}
R =&~ O(\log(\frac{\log(FT)}{\rho_1\log(\num)})) \leq O(\log({\log(FT)}/{\rho_1})) , %
\end{align}
where the first step follows from the setting of $R$, the second step follows from $\num=O(1)$.

So, the sample complexity of Procedure \textsc{SamplingSignificantSample} in Algorithm \ref{alg:pre_compute} is 
\begin{align*}
 &~ D \cdot R \cdot O(k^2\log(k)\log(k/\delta_1) ) \\
 \leq &~ O(\log (FT)) \cdot R \cdot O(k^2\log(k)\log(k/\delta_1) ) \\
 \leq &~ O(\log (FT)) \cdot O(\log({\log(FT)}/{\rho_1})) \cdot O(k^2\log(k)\log(k/\delta_1) ) \\
 = &~ O(k^2\log(k)\log(k/\delta_1)\log (FT)\log({\log(FT)}/{\rho_1} )) ,
\end{align*}
where the first step follows from Eq.~\eqref{eq:D_sample_SSS}, the second step follows from Eq.~\eqref{eq:R_sample_SSS},  the third step is straightforward.

\end{proof}

\subsection{Vote distributions in \textsc{ArySearch}}\label{sec:freq_est:vote}

In this section, we prove several claims on the distributions of votes when we perform the $\num$-ary search on the frequency interval.

We first consider a single voter, i.e., one significant sample.
The following claim shows that, if $f_0$ is in the $q$-th part, then this part or its left neighbor or its right neighbor will get at least one vote.

\begin{claim}\label{clm:angle_between_xgamma_and_xgammabeta_is_betaf}
For $\len\in\R_+, \num\in\Z_+, q\in[1, \num]$,
  let $f_0 \in [\L + (q-1)\frac{\len}{\num} , \L + q\frac{\len}{\num} ]$. 
  For any $\beta\in [\frac{c}{2}\cdot \frac{\num}{\len}, c\cdot \frac{\num}{\len}]$ with constant $c\in(0, 0.01)$, for any constant $\epsilon\in(0, 0.01\cdot c^2)$,
let $\alpha\in \R$ such that 
\begin{align*}
    |z(\alpha + \beta) - z(\alpha)e^{2 \pi \i f_0 \beta}|^2 \le \eps  |z(\alpha)|^2.
\end{align*}
Let 
\begin{align*}
    \Theta = \Big\{ \frac{1}{2\pi  \beta }(\arg(\frac{z(\alpha + \beta)}{z(\alpha)}) + 2\pi s)  ~\Big|~ s \in [ \beta \L -10 , \beta( \L + \len  ) + 10] \cap \Z \Big\}.
\end{align*}
 
 Then, we have that 
 \begin{align*}
   \left|\Theta\cap  \Big[\L + (q-2)\frac{\len}{\num} , \L + (q+1)\frac{\len}{\num} \Big) \right|=1.   
 \end{align*}
\end{claim}
\begin{proof}
We have that
\begin{align*}
    \Big|\frac{z(\alpha + \beta)}{z(\alpha)}  - e^{2 \pi \i f_0 \beta}\Big| \le  \sqrt{\eps} .
\end{align*}
Since $|e^{2 \pi \i f_0 \beta}|=1$ and $\sin(x)\approx x$ for small $x$, it indicates that%
\begin{align}
    \Big\|\arg(\frac{z(\alpha + \beta)}{z(\alpha)})  - 2 \pi f_0 \beta \Big\|_{\bigcirc} \lesssim  \sqrt{\eps} ,  \label{eq:est_arg_z_diff_temp}
\end{align}
where $\| a\|_{\bigcirc} = \underset{x\in \mathbb{Z} }{\min}~| a + 2\pi x|$. Thus,  Eq.~\eqref{eq:est_arg_z_diff_temp} can be rewritten as:
\begin{align}
   \underset{x\in \mathbb{Z} }{\min}~ \Big|\arg(\frac{z(\alpha + \beta)}{z(\alpha)})  - 2 \pi f_0 \beta + 2\pi x \Big| \lesssim  \sqrt{\eps} . \label{eq:est_arg_z_diff}
\end{align}

Let $s_0$ be defined as
\begin{align*}
  s_0:= \arg\min_{x\in \mathbb{Z}}~ \arg(\frac{z(\alpha + \beta)}{z(\alpha)})  - 2 \pi f_0 \beta + 2\pi x.
\end{align*}

We first show that $s_0$ falls in interval in the definition of $\Theta$. We have that 
\begin{align}
    |2\pi s_0  - 2\pi f_0 \beta | \leq&~ \Big| \arg(\frac{z(\alpha + \beta)}{z(\alpha)})  - 2 \pi f_0 \beta + 2\pi s_0\Big| + \Big|\arg(\frac{z(\alpha + \beta)}{z(\alpha)})\Big| \notag \\
    \leq&~ \Big| \arg(\frac{z(\alpha + \beta)}{z(\alpha)})  - 2 \pi f_0 \beta + 2\pi s_0\Big| + 2\pi \notag \\
    \leq&~ 2\pi + O(\sqrt{\eps}) \label{eq:b_s_0_to_f_0}
\end{align}
where the first step follows from the triangle inequality, the second step follows from $|\arg(\frac{z(\alpha + \beta)}{z(\alpha)})|\leq 2\pi $, the third step follows from Eq.~\eqref{eq:est_arg_z_diff}. 

As a result, $s_0$ has the following upper bound:
\begin{align*}
    s_0 \leq&~ f_0 \beta + 1 + O(\sqrt{\eps}) \\
    \leq&~ \beta (\L + \len) + 1 + O(\sqrt{\eps}) \\
    \leq&~ \beta (\L + \len) + 2 
\end{align*}
where the first step follows from Eq.~\eqref{eq:b_s_0_to_f_0}, the second step follows from $f_0\in[\L + (q-1)\frac{\len}{\num} , \L + q\frac{\len}{\num} ]\subseteq [\L, \L + \len]$, the third step follows from the setting of $\eps$. 

Also, $s_0$ has the following lower bound:
\begin{align*}
    s_0 \geq&~ f_0 \beta - 1 - O(\sqrt{\eps}) \\
    \geq&~ \beta \L - 1 - O(\sqrt{\eps}) \\
    \geq&~ \beta \L - 2 
\end{align*}
where the first step follows from Eq.~\eqref{eq:b_s_0_to_f_0}, the second step follows from $f_0\in [\L, \L + \len]$, the third step follows from the setting of $\eps$. 

Combining the lower and upper bounds of $s_0$ together, and by the definition of the set $\Theta$, we know that
\begin{align*}
    \frac{1}{2\pi  \beta }\Big(\arg(\frac{z(\alpha + \beta)}{z(\alpha)}) + 2\pi s_0\Big) \in \Theta.
\end{align*}

Then, we show that
\begin{align*}
    \L + (q-2)\frac{\len}{\num}\leq \frac{1}{2\pi  \beta }\Big(\arg(\frac{z(\alpha + \beta)}{z(\alpha)}) + 2\pi s_0\Big) < \L + (q+1)\frac{\len}{\num}.
\end{align*}
Eq.~\eqref{eq:est_arg_z_diff} also implies that
\begin{align}
\Big|\frac{1}{2\pi  \beta }(\arg(\frac{z(\alpha + \beta)}{z(\alpha)}) + 2\pi s_0) - f_0 \Big|   %
\leq  &~  O(\frac{\sqrt{\eps}}{\beta }  ). \label{eq:b_z_est_to_f_0}
\end{align}

Then, we have the following upper bound:
\begin{align*}
\frac{1}{2\pi  \beta }(\arg(\frac{z(\alpha + \beta)}{z(\alpha)}) + 2\pi s_0) \leq &~ f_0 + O(\frac{\sqrt{\eps}}{\beta }  ) \\ 
\leq &~ \L + q \frac{\len}{\num}   + O(\frac{\sqrt{\eps}}{\beta }  ) \\  
\leq &~ \L + q \frac{\len}{\num}   + O(\frac{2\sqrt{\eps}}{c} \frac{\len}{\num }  ) \\
\leq &~ \L + (q+1) \frac{\len}{\num},  
\end{align*}
where the first step follows from Eq.~\eqref{eq:b_z_est_to_f_0}, the second step follows from $f_0 \in [\L + (q-1)\len /\num , \L + q \len /\num ]$, the third step follows from $\beta \in [c\frac{\num}{2 \len}, c \frac{\num}{\len}] $, the forth step follows from $ O(\sqrt{\eps}/c)\leq 1$.

We also have the following lower bound:
\begin{align*}
\frac{1}{2\pi  \beta }(\arg(\frac{z(\alpha + \beta)}{z(\alpha)}) + 2\pi s_0) \geq &~ f_0 - O(\frac{\sqrt{\eps}}{\beta }  ) \\ 
\geq &~ \L + (q-1) \frac{\len}{\num}   - O(\frac{\sqrt{\eps}}{\beta }  ) \\  
\geq &~ \L + (q-1) \frac{\len}{\num}   - O(\frac{\sqrt{\eps}}{c} \frac{\len}{\num }  ) \\
> &~ \L + (q-2) \frac{\len}{\num}  ,
\end{align*}
where the first step follows from Eq.~\eqref{eq:b_z_est_to_f_0}, the second step follows from $f_0 \in [\L + (q-1)\len /\num , \L + q \len /\num ]$, the third step follows from $\beta \in [c\frac{\num}{2 \len}, c \frac{\num}{\len}] $, the forth step follows from $ O(\sqrt{\eps}/c)< 1$.

Moreover, since 
\begin{align*}
     \frac{1}{ \beta} \geq \frac{\len}{\num},
\end{align*}
we have that there is at most $1$ element in the intersection
 \begin{align*}
   \left|\Theta\cap  \Big[\L + (q-2)\frac{\len}{\num} , \L + (q+1)\frac{\len}{\num} \Big) \right|\leq 1.   
 \end{align*}

The lemma then follows.
\end{proof}

The following claim shows that for those parts far away from the true part containing $f_0$, they will get no vote.
\begin{claim}\label{clm:angle_between_xgamma_and_xgammabeta_is_betaf_second_part}

For $\len\in\R_+, \num\in\Z_+, q\in[1, \num]$,
  let $f_0 \in [\L + (q-1)\frac{\len}{\num} , \L + q\frac{\len}{\num} ]$. 
  Let $\beta\sim \text{Uniform}( [\frac{c}{2}\cdot \frac{\num}{\len}, c\cdot \frac{\num}{\len}])$ with constant $c\in(0, 0.01)$. For any constant $\epsilon\in(0, 0.01\cdot c^2)$, 
let $\alpha\in \R$ such that 
\begin{align*}
    |z(\alpha + \beta) - z(\alpha)e^{2 \pi \i f_0 \beta}|^2 \le \eps  |z(\alpha)|^2.
\end{align*}
Let 
\begin{align*}
    \Theta = \Big\{ \frac{1}{2\pi  \beta }(\arg(\frac{z(\alpha + \beta)}{z(\alpha)}) + 2\pi s)  ~\Big|~ s \in [ \beta \L -10 , \beta( \L + \len  ) + 10] \cap \Z \Big\}.
\end{align*}

Then, we have that for any $q' \in [0, \num - 1]$ such that $|q-q'|> 1$, with probability at least $1-O(c)$, 
\begin{align*}
\Theta \cap  \Big[\L + (q'-1)\frac{\len}{\num} , \L + q'\frac{\len}{\num} \Big]  = \emptyset.    
\end{align*}

\end{claim}
\begin{proof}

Let $s_0$ be defined as
\begin{align*}
  s_0:= ~\arg\min_{x\in \mathbb{Z}}~ \arg(\frac{z(\alpha + \beta)}{z(\alpha)})  - 2 \pi f_0 \beta + 2\pi x.
\end{align*}

By Claim \ref{clm:angle_between_xgamma_and_xgammabeta_is_betaf}, we have that 
\begin{align}
\frac{1}{2\pi  \beta }(\arg(\frac{z(\alpha + \beta)}{z(\alpha)}) + 2\pi s_0) \in  \Big[\L + (q-2)\frac{\len}{\num} , \L + (q+1) \frac{\len}{\num} \Big) .    \label{eq:angle_between_xgamma_and_xgammabeta_is_betaf}
\end{align}

Then, we discuss two cases based on the range of $q'$.

\textbf{Case 1: $1<|q-q'|< 1/(4c) $.} 

For the ease of discussion, suppose $1<q'-q<  1/(4c) $. We have that
\begin{align*}
 \L + (q'-1)\frac{\len}{\num} \geq &~    \L + (q+1)\frac{\len}{\num}\\
 > &~ \frac{1}{2\pi  \beta }(\arg(\frac{z(\alpha + \beta)}{z(\alpha)}) + 2\pi s_0),
\end{align*}
where the first step follows from $q,q'\in \Z$, and the second step follows from Eq.~\eqref{eq:angle_between_xgamma_and_xgammabeta_is_betaf}.

Moreover, we also have that
\begin{align*}
 \L + q'\frac{\len}{\num} \leq &~ \L + (q+\frac{1}{4c}-1) \frac{\len}{\num}\\
 \leq &~ \frac{1}{2\pi  \beta }(\arg(\frac{z(\alpha + \beta)}{z(\alpha)}) + 2\pi s_0) + (\frac{1}{4c}+1) \frac{\len}{\num}\\
  < &~ \frac{1}{2\pi  \beta }(\arg(\frac{z(\alpha + \beta)}{z(\alpha)}) + 2\pi s_0) + \frac{1}{\beta}\\
 \leq &~ \frac{1}{2\pi  \beta }(\arg(\frac{z(\alpha + \beta)}{z(\alpha)}) + 2\pi (s_0 + 1)),
\end{align*}
where the second step follows from Eq.~\eqref{eq:angle_between_xgamma_and_xgammabeta_is_betaf}, the third step follows from $ (\frac{1}{4c}+1) \len /\num< 1/\beta $. 

Hence, we get that $\Big[\L + (q'-1)\frac{\len}{\num} , \L + q'\frac{\len}{\num} \Big]$ is contained in the following interval:
\begin{align*}
    \Big(\frac{1}{2\pi  \beta }(\arg(\frac{z(\alpha + \beta)}{z(\alpha)}) + 2\pi s_0), \frac{1}{2\pi  \beta }(\arg(\frac{z(\alpha + \beta)}{z(\alpha)}) + 2\pi (s_0+1))\Big).
\end{align*}
Since $s_0$ and $s_0+1$ are two consecutive integers, by the definition of $\Theta$, there is no element of $\Theta$ in this open interval. Hence, we know that in this case,
\begin{align*}
\Theta \cap  \Big[\L + (q'-1)\frac{\len}{\num} , \L + q'\frac{\len}{\num} \Big]  = \emptyset.    
\end{align*}

{\bf Case 2: $|q-q'|\geq 1/(4c)$.}

For the ease of discussion, suppose that $ q'-q \geq 1/(4c)$. We have that, 
\begin{align}
    c (q'-q) \geq \frac{1}{4}.\label{eq:angle_between_xgamma_and_xgammabeta_is_betaf_0_5} 
\end{align}

Moreover, we have that
\begin{align}
\arg(\frac{z(\alpha + \beta)}{z(\alpha)}) - 2\pi\beta\left(\L + (q-\frac{1}{2})\frac{\len}{\num} \right) {\pmod {2\pi } }  \in  \Big[-\frac{3\pi\beta\len}{\num} , \frac{3\pi\beta\len}{\num} \Big],    \label{eq:angle_between_xgamma_and_xgammabeta_is_betaf_1} 
\end{align}
which follows from Eq.~\eqref{eq:angle_between_xgamma_and_xgammabeta_is_betaf}. %

We also have that,
\begin{align}
    \beta \frac{\len}{\num} \leq c. \label{eq:angle_between_xgamma_and_xgammabeta_is_betaf_2}
\end{align}

Then, we have that, 
\begin{align}
&~\mathsf{Pr} \Big[2\pi\beta(\L + (q'-\frac{1}{2})\frac{\len}{\num}  ) -\arg(\frac{z(\alpha + \beta)}{z(\alpha)})  {\pmod {2\pi } } \in [   - \frac{\pi\beta\len}{ \num },   \frac{\pi\beta\len}{ \num } ] \Big] \notag \\
\leq&~ \mathsf{Pr} \Big[2\pi\beta (q'-q)\frac{\len}{\num} {\pmod {2\pi } } \in [   - \frac{4\pi\beta\len}{ \num },   \frac{4\pi\beta\len}{ \num } ] \Big] \notag \\
\leq&~ \mathsf{Pr} \Big[2\pi\beta (q'-q)\frac{\len}{\num} {\pmod {2\pi } } \in [   - 4\pi c,   + 4\pi c ] \Big] \notag \\
\leq&~ 4c + \frac{16}{q'-q} \notag \\
\leq&~ 100 c \label{eq:bound_prob_q_q_prime_not_close}
\end{align}
where the first step follows from Eq.~\eqref{eq:angle_between_xgamma_and_xgammabeta_is_betaf_1}, the second step follows from Eq.~\eqref{eq:angle_between_xgamma_and_xgammabeta_is_betaf_2}, the third step follows from Lemma \ref{lem:wrapping} with the following parameters: 
\begin{align*}
    \widetilde{T}=&~  {2\pi },\\
    \widetilde{\sigma}=&~ 2\pi\beta (q'-q)\frac{\len}{\num},\\
    \widetilde{\epsilon} = &~ 4\pi c,\\
    \widetilde{\delta}=&~0 , \\
    A=&~ \pi c (q'-q),
\end{align*}
the forth step follows from Eq.~\eqref{eq:angle_between_xgamma_and_xgammabeta_is_betaf_0_5}.

By Eq.~\eqref{eq:bound_prob_q_q_prime_not_close}, we have that
\begin{align*}
\mathsf{Pr} \Big[ \exists s_0\in\Z, \frac{1}{2\pi  \beta }(\arg(\frac{z(\alpha + \beta)}{z(\alpha)}) + 2\pi s_0) \in  \Big[\L + (q'-1)\frac{\len}{\num} , \L + q' \frac{\len}{\num} \Big]   \Big] \leq 100 c
\end{align*}

As a result, we know that in this case, with probability at least $1-O(c)$, 
\begin{align*}
\Theta \cap  \Big[\L + (q'-1)\frac{\len}{\num} , \L + q'\frac{\len}{\num} \Big]  = \emptyset.    
\end{align*}

\end{proof}

Then, we consider $R$ independent voters, i.e., $R$ significant samples $\alpha_1,\dots,\alpha_R$. The following claim shows that the true part and its left and right neighbors will get at least $R$ votes. Meanwhile, those parts far away from the true part will get at most $R/2$ votes with high probability.
\begin{claim}\label{clm:angle_between_xgamma_and_xgammabeta_is_betaf_sum}
For For $\len\in\R_+, \num\in\Z_+, q\in[1, \num]$, 
let $f_0 \in [\L + (q-1)\frac{\len}{\num} , \L + q\frac{\len}{\num} ]$. 
Let $\beta\sim \text{Uniform}( [\frac{c}{2}\cdot \frac{\num}{\len}, c\cdot \frac{\num}{\len}])$ with constant $c\in(0, 0.01)$. For any constant $\epsilon\in(0, 0.01\cdot c^2)$, 
Let $\alpha_1, \cdots, \alpha_R \in \R$ such that for any $i\in [R]$,
\begin{align*}
    |z(\alpha_i + \beta) - z(\alpha_i)e^{2 \pi \i f_0 \beta}|^2 \le \eps  |z(\alpha_i)|^2.
\end{align*}
For any $i\in [R]$, let 
\begin{align*}
    \Theta_i = \Big\{ \frac{1}{2\pi  \beta }(\arg(\frac{z(\alpha_i + \beta)}{z(\alpha_i)}) + 2\pi s)  ~\Big|~ s \in [ \beta \L -10 , \beta( \L + \len  ) + 10] \cap \Z \Big\}.
\end{align*}
 
 Then, it holds that:
 \begin{enumerate}
 \item 
 \begin{align*}
\sum_{i=1}^R \Big|\Theta_i\cap\Big[\L + (q-2)\frac{\len}{\num} , \L + (q+1)\frac{\len}{\num} \Big]\Big| \geq R   .
 \end{align*}
 
 \item For any $|q'-q|\geq 3$, with probability at least $ 1- O(c)^{R/6}$,
  \begin{align*}
\sum_{i=1}^R \Big|\Theta_i\cap\Big[\L + (q'-2)\frac{\len}{\num} , \L + (q'+1)\frac{\len}{\num} \Big]\Big| \leq \frac{R}{2} .
 \end{align*}
 \end{enumerate}
\end{claim}
\begin{proof}

{\bf Part 1.}

By applying Claim \ref{clm:angle_between_xgamma_and_xgammabeta_is_betaf}, we have that,
\begin{align*}
    \Big|\Theta_i\cap\Big[\L + (q-2)\frac{\len}{\num} , \L + (q+1)\frac{\len}{\num} \Big]\Big| \geq 1,
\end{align*}
which implies that
\begin{align*}
    \sum_{i=1}^R \Big|\Theta_i\cap\Big[\L + (q-2)\frac{\len}{\num} , \L + (q+1)\frac{\len}{\num} \Big]\Big| \geq R.
\end{align*}

{\bf Part 2.}
By applying Claim \ref{clm:angle_between_xgamma_and_xgammabeta_is_betaf_second_part}, we have that, for any $ |q-q'|> 1$, with probability at most $ O(c)$,
\begin{align*}
    \Big|\Theta_i\cap\Big[\L + (q'-1)\frac{\len}{\num} , \L + q'\frac{\len}{\num} \Big]\Big| \geq  1.
\end{align*}

By the setting of our parameter $ \frac{1}{ \beta} \geq \frac{\len}{\num}$, thus
\begin{align*}
    \Big|\Theta_i\cap\Big[\L + (q'-1)\frac{\len}{\num} , \L + q'\frac{\len}{\num} \Big]\Big| =  1.
\end{align*}
Then,  for any $ |q-q'|\geq 3$, by a union bound over $q'-1$, $q'$, and $q'+1$, with probability at most $O(c)$,
\begin{align*}
   3\geq  \Big|\Theta_i\cap\Big[\L + (q'-2)\frac{\len}{\num} , \L + (q'+1)\frac{\len}{\num} \Big]\Big| \geq  1.
\end{align*}

Then, we have that
\begin{align*}
&~ \Pr\Big[\sum_{i=1}^R  \Big|\Theta_i\cap\Big[\L + (q'-2)\frac{\len}{\num} , \L + (q'+1)\frac{\len}{\num} \Big]\Big| \geq \frac{R}{2} \Big] \\
\leq &~ \binom{ R}{ R/6 } O( c)^{R/6} \\
\leq &~ (\frac{e R}{R/6})^{R/6} O( c)^{R/6} \\
\leq &~ O(c)^{R/6} 
\end{align*}
where the first step follows from there should be at least $0.5R/3=R/6$ different $i\in [R]$ satisfying $ |\Theta_i \cap[\L + (q'-2)\len /\num , \L + (q'+1)\len /\num ]| \geq  1$,%
the second step follows from $ \binom{n}{k}\leq (\frac{en}{k})^k$, the third step is straight forward.

The lemma is then proved.
\end{proof}

Finally, we consider probabilistic voters, that is, for each sample $\alpha_i$, with probability $1-\rho$, it is significant. The following claim shows the votes distribution in this case. 

\begin{claim}\label{clm:angle_between_xgamma_and_xgammabeta_is_betaf_sum_2}
For $\len\in\R_+, \num\in\Z_+, q\in[1, \num]$, 
let $f_0 \in [\L + (q-1)\frac{\len}{\num} , \L + q\frac{\len}{\num} ]$. 
Let $\beta\sim \text{Uniform}( [\frac{c}{2}\cdot \frac{\num}{\len}, c\cdot \frac{\num}{\len}])$ with $c=\Theta(1)\in(0, 0.01)$, $\epsilon=\Theta(1)\in(0, 0.01\cdot c^2)$, 
let $\alpha_1, \cdots, \alpha_R \in \R$ such that for any $i\in [R]$ with probability at least $1-\rho$,
\begin{align*}
    |z(\alpha_i + \beta) - z(\alpha_i)e^{2 \pi \i f_0 \beta}|^2 \le \eps  |z(\alpha_i)|^2.
\end{align*}
For any $i\in [R]$, let 
\begin{align*}
    \Theta_i = \Big\{ \frac{1}{2\pi  \beta }(\arg(\frac{z(\alpha_i + \beta)}{z(\alpha_i)}) + 2\pi s)  ~\Big|~ s \in [ \beta \L -10 , \beta( \L + \len  ) + 10] \cap \Z \Big\}.
\end{align*}
 
Then, it holds that
\begin{enumerate}
\item With probability at least $ 1- O(\rho)^{R/3} $, 
 \begin{align*}
\sum_{i=1}^R \Big|\Theta_i\cap\Big[\L + (q-2)\frac{\len}{\num} , \L + (q+1)\frac{\len}{\num} \Big]\Big| \geq \frac{2R}{3}.
 \end{align*}
 
\item  For any $|q'-q|\geq 3$, with probability at least $ 1- O(c + \rho)^{R/6}$,
  \begin{align*}
\sum_{i=1}^R \Big|\Theta_i\cap\Big[\L + (q'-2)\frac{\len}{\num} , \L + (q'+1)\frac{\len}{\num} \Big]\Big| \leq \frac{R}{2} .
 \end{align*}
\end{enumerate}
\end{claim}

\begin{proof}

{\bf Part 1.}

By applying Claim \ref{clm:angle_between_xgamma_and_xgammabeta_is_betaf}, we have that with probability at most $\rho$,
\begin{align*}
    \Big|\Theta_i\cap\Big[\L + (q-2)\frac{\len}{\num} , \L + (q+1)\frac{\len}{\num} \Big]\Big| =0,
\end{align*}
then we have that,
\begin{align*}
    &~\Pr\Big[ \sum_{i=1}^R \Big|\Theta_i \cap [\L + (q-2)\frac{\len} {\num}, \L + (q+1)\frac{\len}{\num} ]\Big| \leq \frac{R}{3} \Big] \\
    \leq&~  \binom{ R}{ R/3 } O(\rho)^{R/3} \\
    \leq&~  O(\frac{e R}{R/3})^{R/3} O(\rho)^{R/3} \\
    \leq&~ O(\rho)^{R/3} 
\end{align*}
where the first step follows from $ |\Theta_i \cap[\L + (q-2)\len /\num , \L + (q+1)\len /\num ]| =0$ or $1$ by our parameter setting $1/\beta > 3\len /\num$ 
and there should be at least $R/3$ different $i\in [R]$ satisfying $ |\Theta_i \cap[\L + (q-2)\len /\num , \L + (q+1)\len /\num ]| =0$, 
the second step follows from $ \binom{n}{k}\leq (\frac{en}{k})^k$, the third step is straight forward.

{\bf Part 2.}

By applying Claim \ref{clm:angle_between_xgamma_and_xgammabeta_is_betaf_second_part}, we have that, for any $ |q-q'|> 1$, with probability at most $ O(c)+\rho$,
\begin{align*}
    \Big|\Theta_i \cap\Big[\L + (q'-1)\frac{\len}{\num} , \L + q'\frac{\len}{\num} \Big]\Big| = 1,
\end{align*}
where the probability follows from a union bound over the success of Claim \ref{clm:angle_between_xgamma_and_xgammabeta_is_betaf_second_part} and $\alpha_i$ being significant. %

Thus, for any $ |q-q'|\geq 3$, by a union bound, with probability at most $ 3 ((1-\rho)O(c)+\rho) = O(c+\rho)$,
\begin{align*}
   3\geq  \Big|\Theta_i \cap\Big[\L + (q'-2)\frac{\len}{\num} , \L + (q'+1)\frac{\len}{\num} \Big]\Big|\geq  1.
\end{align*}

Then, we have that
\begin{align*}
&~ \Pr\Big[ \sum_{i=1}^R \Big|\Theta_i\cap[\L + (q'-2)\len /\num , \L + (q'+1)\len /\num ]\Big| \geq \frac{R}{2} \Big] \\
\leq &~ \binom{ R}{ R/6 } O(c+\rho)^{R/6} \\
\leq &~ (\frac{e R}{R/6})^{R/6} O(c+\rho)^{R/6} \\
\leq &~ O(c+\rho)^{R/6} 
\end{align*}
where the first step follows from there should be at least $0.5R/3=R/6$ different $i\in [R]$ satisfying $ |\Theta_i \cap[\L + (q'-2)\len /\num , \L + (q'+1)\len /\num ]| \geq  1$, the second step follows from $ \binom{n}{k}\leq (\frac{en}{k})^k$, the third step is straight forward.

\end{proof}

\section{Signal Reconstruction}
\label{sec:sig_recontr}

In this section, we wrap up all technical tools developed in previous sections and present our main result: a Fourier interpolation algorithm with improved time complexity, sample complexity, and output sparsity.

This section consists of two parts. The first part is devoted to the signal estimation. We first provide some tools that are useful for signal estimation (see Section~\ref{sec:sig_recontr:preli}). Then, we formally define the heavy clusters and show their approximation property (see Section~\ref{sec:sig_recontr:heavy}). Next, we give a Fourier set query algorithm, which is a component in signal estimation (see Section~\ref{sec:sig_recontr:fourier}). We further show that it suffices to only reconstruct the signals in the bins satisfying the high SNR band condition (see Section~\ref{sec:sig_recontr:reduce}).  

The second part focuses on the Fourier interpolation algorithm. Combining the frequency estimation algorithm in Section~\ref{sec:freq_est} with the signal estimation method we just developed, we obtain a Fourier interpolation algorithm with a constant success probability (see Section~\ref{sec:sig_recontr:constant_prob}). Then, we introduce the min-of-median signal estimator used to boost the success probability (see Section~\ref{sec:sig_recontr:boost}). Finally, we prove our main theorem that gives a Fourier interpolation algorithm with high success probability (see Section~\ref{sec:sig_recontr:high_prob}).

\subsection{Preliminary}\label{sec:sig_recontr:preli}

We provide some technical tools in this section.

The following two lemma shows that Fourier-polynomial mixed signals and Fourier-sparse signals can approximate each other.  
\begin{lemma}[\cite{ckps16}]\label{lem:low_degree_approximates_concentrated_freq_ours}
For any $\Delta>0$, $\delta>0$, for any $n_1,\dots,n_k\in \Z_{\geq 0}$ with $\sum_{j\in [k]}n_j=k$, let
\begin{align*}
    x^*(t)=\sum_{j \in [k]} e^{2 \pi \i f_j t} \sum_{i=1}^{n_j} v_{j, i} e^{2\pi\i f'_{j, i} t}, 
\end{align*}
where $|f'_{j, i}| \le \Delta$ for each $j\in [k], i\in [n_j]$. There exist $k$ polynomials $P_j(t)$ for $j \in [k]$ of degree at most
\begin{align*}
    d=O(T \Delta + k^3 \log k + k \log (1/\delta) )
\end{align*} such that
\begin{align*}
    \Big\|\sum_{j\in [k]} e^{2 \pi \i f_j t} P_j(t) - x^*(t)\Big\|^2_T \le \delta \|x^*(t)\|^2_T.
\end{align*}
\end{lemma}

\begin{lemma}[{\cite[Lemma 8.8]{ckps16}}]\label{lem:polynomial_to_FT}
For any degree-$d$ polynomial $Q(t) = \overset{d}{\underset{j=0}{\sum}} c_j t^j$, any $T>0$ and any $\epsilon>0$, there always exist $\gamma>0$ and \begin{align*}
    x^*(t)=\sum_{j=1}^{d+1} \alpha_j e^{2\pi \i (\gamma j) t}
\end{align*} %
such that
\begin{equation*}
|x^*(t) - Q(t)| \le \epsilon~~~\forall t \in [0,T]. %
\end{equation*}
\end{lemma}

\begin{algorithm}[ht]
\caption{Multipoint evaluation of a polynomial}\label{alg:multipoint_evaluation}
\begin{algorithmic}[1]
\Procedure{\textsc{PolynomialEvaluation}}{$P, t$}  \Comment{Fact \ref{fac:multipoint_evaluation_of_polynomial}}
	\State \Return $(P(t_1), P(t_2), \cdots, P(t_d))$ \Comment{$t \in \C^d$}
\EndProcedure
\Procedure{\textsc{MixedPolynomialEvaluation}}{$\sum_{j=1}^k P_j(t)\exp(2\pi\i f_j t), t$} 
\For{$ j\in [k] $} \label{ln:for_mpe_j_k}
    \State $v_j\leftarrow (P_j(t_1), P_j(t_2), \cdots, P_j(t_d))$ \Comment{$t \in \C^d$} \label{ln:mpe_v_P}
\EndFor
\State \Return $(\sum_{j=1}^k v_{j, 1}\exp(2\pi\i f_j t_1), \sum_{j=1}^k v_{j, 2}\exp(2\pi\i f_j t_2), \cdots, \sum_{j=1}^k v_{j, 3}\exp(2\pi\i f_j t_3))$ 
\EndProcedure
\end{algorithmic}
\end{algorithm}

The following fact shows an efficient method  multi-point evaluation of a polynomial. 
\begin{fact}[{\cite[Chapter 10]{gg99}}]
\label{fac:multipoint_evaluation_of_polynomial}
Given a degree-$d$ polynomial $P(t)$, and a set of $d$ locations $\{t_1,t_2,\cdots, t_d\}$. There exists an algorithm that takes $O(d \log^2 d \log\log d)$ time to output the evaluations $\{ P(t_1), P(t_2), \cdots, P(t_d)\}$. %
\end{fact}

The following lemma shows the time complexity of evaluating a mixed polynomial.

\begin{lemma}[Time complexity of Algorithm \ref{alg:multipoint_evaluation}]\label{lem:time_mpe}
Procedure \textsc{MixedPolynomialEvaluation} in Algorithm \ref{alg:multipoint_evaluation} runs \begin{align*}
    O\Big(\sum_{j=1}^k {\max}\{d, \mathrm{deg}(P_j)\} \log^3({\max}\{d, \mathrm{deg}(P_j)\})\Big)
\end{align*} time.
\end{lemma}
\begin{proof}
Procedure \textsc{MixedPolynomialEvaluation} in Algorithm \ref{alg:multipoint_evaluation} consists of the following steps:
\begin{itemize}
\item In line \ref{ln:for_mpe_j_k}, the for loop repeats $ k $ times. 
\item In line \ref{ln:mpe_v_P},  multipoint evaluation of a polynomial takes $d_j \log^c(d_j)$ times by Fact \ref{fac:multipoint_evaluation_of_polynomial}, where $ d_j = {\max}\{d, \mathrm{deg}(P_j)\}$. 
\end{itemize}

Hence, the total time complexity is 
\begin{align*}
\sum_{j=1}^k O(d_j \log^3(d_j)) = O\Big(\sum_{j=1}^k {\max}\{d, \mathrm{deg}(P_j)\} \log^3({\max}\{d, \mathrm{deg}(P_j)\})\Big).
\end{align*}

\end{proof}

\subsection{Heavy cluster}\label{sec:sig_recontr:heavy}

In this section, we formally define the heavy clusters and show that using ``heavy frequencies'' only yields a good approximation of the ground-truth signal.
\begin{definition}[Heavy cluster]\label{def:heavy_clusters}
Let $x^*(t)= \overset{k}{\underset{j=1}{\sum} } v_j e^{2 \pi \i f_j t}$ and $\N>0$. Let the filter $H$ be defined as in Lemma \ref{lem:property_of_filter_H}. 
Let $\Delta_h = |\supp(\wh{H})|$. We say a frequency $f^*$ belongs to an ${\cal N}$-\emph{heavy cluster} if and only if
\begin{align*}
    \int_{f^*-\Delta_h}^{f^*+\Delta_h} |\wh{H \cdot x^*}(f)|^2 \mathrm{d} f \ge T \cdot \N^2/k.
\end{align*} 

\end{definition}

\begin{claim}\label{cla:guarantee_removing_x**_x*_}
Given $x^*(t)= \overset{k}{ \underset{j=1}{\sum} } v_j e^{2 \pi \i f_j t}$ and any $\N>0$. For the set of heavy frequencies:
\begin{equation*}
{S^*}=\left\{j \in [k]\bigg{|}\int_{f_j-\Delta_h}^{f_j+\Delta_h} |\wh{H \cdot x^*}(f)|^2 \mathrm{d} f \ge T \cdot \N^2/k  \right\},
\end{equation*}
and the signal $x_{S^*}(t)= \underset{j\in {S^*}}{\sum} v_j e^{2 \pi \i f_j t}$, it holds that
\begin{align*}
    \|x_{S^*} - x^* \|_T^2 \lesssim \N^2.
\end{align*}
\end{claim}

\begin{proof}
Let $x_{\overline{S^*}}(t) = \underset{j\in [k] \backslash S^*}{\sum} v_j e^{2 \pi \i f_j t}$. Then $\|x^*-x_{S^*}\|_T^2=\|x_{\overline{S^*}}\|^2_T$.

Then, we have that
\begin{align*}
 T\|x_{\overline{S^*}}(t)\|_T^2 =&~ \int_{0}^{T} |x_{\overline{S^*}} (t)|^2 \d t \notag\\  
 \lesssim &~ \int_{0}^{T} |x_{\overline{S^*}}(t)\cdot  H(t)|^2 \d t \notag\\  
 \leq &~ \int_{-\infty}^{\infty} |x_{\overline{S^*}}(t)\cdot H(t) |^2 \d t \notag\\  
 = &~ \int_{-\infty}^{\infty} |\wh{x}_{\overline{S^*}}(f) * \wh{H}(f) |^2 \d f \notag\\
 \leq &~ \sum_{j\in[k]\backslash S^*} \int_{f_j-\Delta_h}^{f_j+\Delta_h} |\wh{x}_{\overline{S^*}}(f) * \wh{H}(f) |^2 \d f \notag\\
\leq &~ \sum_{j\in[k]\backslash S^*}  T\N^2/k \notag\\
\leq &~  T\N^2 ,
\end{align*}
where the first step follows from the definition of the norm, the second step follows from Lemma \ref{lem:property_of_filter_H} Property \RN{5}, the third step is straight forward, the forth step follows from Parseval's theorem, the fifth step follows from the definition of $x_{\overline{S^*}}(t) $, the sixth step follows from the definition of heavy frequency, the seventh step is straightforward.

\end{proof}

\subsection{Fourier set query}\label{sec:sig_recontr:fourier}

In this section, we present a Fourier set query algorithm such that for a Fourier-polynomial mixed signal, given all of its frequencies, the algorithm can reconstruct the signal very efficiently. 

\begin{lemma}

\label{lem:magnitude_recovery_for_CKPS}
For $j\in [k]$, given a $d_j$-degree polynomial $P_j(t)$ and a frequency $f_j$. Let $x_S(t) = \sum_{j=1}^k P_j(t) \exp({2\pi\i  f_j  t  })$. Given observations of the form $x(t):=x_S(t) + g(t)$ for arbitrary noise $g(t)$ in time duration $t\in [0, T]$. Let $D:=\sum_{j=1}^k d_j$.

Then, there is an algorithm (Procedure \textsc{SignalEstimation} in Algorithm \ref{algo:sig_est_1d_accuracy}) such that
\begin{itemize}
    \item takes $O(D\log(D))$ samples from $x(t)$, %
    \item runs $O(D^\omega\log(D))$ time, 
    \item outputs $y(t) = \sum_{j=1}^k P'_j(t) \exp({2\pi\i  f_j  t  })$ with $d$-degree polynomial $P'_j(t)$, such that with probability at least $0.99$, we have
    \begin{align*}
        \|y - x_S\|_T^2  \lesssim \|g\|_T^2 .
    \end{align*}
\end{itemize}  
\end{lemma}
\begin{proof}

By Lemma \ref{lem:polynomial_to_FT}, we have that, for all $t\in [0, T]$, there exist $ D$-Fourier-sparse signals $y_1(t) $ and $x_{S,1}(t) $
\begin{align}
   | y(t) - y_1(t) | \leq \eps_1, \label{eq:y_y_1}
\end{align}
and
\begin{align}
   | x_S(t) - x_{S,1}(t) | \leq \eps_1. \label{eq:x_S_x_S_1}
\end{align}
Then, we have that
\begin{align}
\|y(t) - x_{S}(t)\|_T^2 \lesssim &~ \|y(t)-y_1(t)\|_T^2 + \|x_S(t)-x_{S,1}(t)\|_T^2 + \|y_1(t) - x_{S,1}(t)\|_T^2 \notag \\
\lesssim &~ 2\eps_1 + \|y_1(t) - x_{S,1}(t)\|_T^2 \notag \\
\lesssim &~ \|y_1(t) - x_{S,1}(t)\|_T^2  \label{eq:y_x_S_T}
\end{align}
where the first step follows from $(a+b)^2\leq 2a^2+2b^2$, the second step follows from Eq.~\eqref{eq:y_y_1} and Eq.~\eqref{eq:x_S_x_S_1}, the third step follows from $\eps_1 \lesssim  \|y_1(t) - x_{S,1}(t)\|_T^2 $.  

We also have that
\begin{align}
\|y_1(t) - x_{S,1}(t)\|_{S,w}^2 \lesssim &~ \|y_1(t)-y(t)\|_{S,w}^2 + \|x_{S,1}(t)-x_{S}(t)\|_{S,w}^2 + \|y(t) - x_{S}(t)\|_{S,w}^2 \notag \\
\lesssim &~ 2\eps_1 + \|y(t) - x_{S}(t)\|_{S,w}^2 \notag \\
\lesssim &~ \|y(t) - x_{S}(t)\|_{S,w}^2  \label{eq:y_x_S_Sw}
\end{align}
where the first step follows from $(a+b)^2\leq 2a^2+2b^2$, the second step follows from Eq.~\eqref{eq:y_y_1} and Eq.~\eqref{eq:x_S_x_S_1}, the third step follows from $\eps_1 \lesssim  \|y(t) - x_{S}(t)\|_{S,w}^2 $.  

By the definition of $y(t)$ in line \ref{ln:sig_est_regress_fast} in Procedure \textsc{SignalEstimation} of Algorithm \ref{algo:sig_est_1d_accuracy}, we have that
\begin{align}
    \|y(t) - x(t)\|_{S,w}^2 \leq  \|x_S(t) - x(t)\|_{S,w}^2  \label{eq:sig_est_regress_fast_def_y}
\end{align}

We have that
\begin{align}
   \E [\|x - x_{S}\|_{S,w}^2] =&~ \E \Big[\sum_{i\in [|S|]} w_i |x(t_i) - x_{S}(t_i)|^2\Big] \notag\\
   =&~ \E \Big[\sum_{i\in [|S|]} \frac{1}{2T |S| D(t)} |x(t_i) - x_{S}(t_i)|^2\Big]\notag\\
   =&~  \sum_{i\in [|S|]} \E_{t_i\sim D(t)} \Big[\frac{1}{2T |S| D(t)} |x(t_i) - x_{S}(t_i)|^2\Big]\notag\\
   =&~  |S|\cdot \int_{-T}^T D(t) \frac{1}{2T |S| D(t)} |x(t) - x_{S}(t)|^2 \d t\notag\\
   =&~  \int_{-T}^T \frac{1}{2T} |x(t) - x_{S}(t)|^2 \d t\notag\\
   =&~  \|x(t) - x_{S}(t)\|_T^2 
   \label{eq:E_x_x_Sw_T}
\end{align}
where the first step follows from the definition of the norm, the second step follows from the definition of $w_i$, the third step is straightforward, the forth follows from the definition of expectation, the fifth  step follows from the definition of the norm.

We have that
\begin{align*}
  \|y - x_{S}\|_T^2 
  \lesssim&~ \|y_1 - x_{S,1}\|_T^2 \\
  \lesssim&~ \|y_1 - x_{S,1}\|_{S,w}^2 \\
  \lesssim&~ \|y - x_{S}\|_{S,w}^2 \\
  \lesssim&~ \|y - x\|_{S,w}^2 + \|x - x_{S}\|_{S,w}^2 \\
  \lesssim&~ \|x - x_{S}\|_{S,w}^2 \\
  \lesssim&~ \|x - x_{S}\|_{T}^2 ,
\end{align*}
where the first step follows from Eq.~\eqref{eq:y_x_S_T}, the second step follows from Lemma \ref{lem:x_energy_preserving_000}, the third step follows from Eq.~\eqref{eq:y_x_S_Sw}, the forth step follows from $(a+b)^2\leq 2a^2+2b^2$, the fifth step follows from Eq.~\eqref{eq:sig_est_regress_fast_def_y}, the sixth step follows from Eq.~\eqref{eq:E_x_x_Sw_T} by Markov inequality with probability at least $0.99$.

\end{proof}

\subsection{High signal-to-noise ratio band approximation}\label{sec:sig_recontr:reduce}

The goal of this section is to prove the following lemma, which roughly states that for the heavy frequencies, it suffices to only reconstruct those in the bins with high SNRs.

\begin{lemma}\label{lem:xSfsubxSleqg}
Let $x^*(t)=\sum_{j=1}^k v_j e^{2\pi \i f_j t}$ be the ground-truth signal and $x(t)=x^*(t)+g(t)$ be the noisy observation signal.
Let $H$ be defined as in Definition~\ref{def:effect_H_k_sparse}, ${G}^{(j)}_{\sigma,b}$ be defined as in Definition \ref{def:G_j_sigma_b} with $(\sigma,b)$ such that Large Offset event does not happen.
Let $U:=\{t_0\in \R~|~ H(t) > 1-\delta_1~ \forall t\in [t_0,t_0+\beta]\}$. Let
\begin{align*}
 S:=\Big\{ j\in[k]~\Big|~  \int_{f_j-\Delta}^{f_j+\Delta} | \widehat{H\cdot x^*}(f) |^2 \mathrm{d} f \geq T\N^2/k \Big\},
\end{align*}
and $x_S(t)=\sum_{j\in S}v_j  e^{2\pi \i f_j t}$.

For $j\in [B]$, let $z_j^*(t):=(x^*\cdot H)*G^{(j)}_{\sigma, b}(t)$ and $z_j(t)=(x\cdot H)*{G}^{(j)}_{\sigma,b}(t)$.
Let $g_j(t):=z_j(t)-z_j^*(t)$.
Let 
\begin{align}\label{eq:define_S_g_1}
  S_{g1}:=\big\{ j\in[B]~|~ \|g_j(t)\|_T^2 \leq c \|z^*_j(t)\|_{U}^2\big\},  
\end{align}
where $c\in (0, 0.001)$ is a small universal constant.
Let 
\begin{align*}
  S_{g2}:=\left\{ j\in[B]~\Bigg|~ \exists f_0\in \{f_1,\dots,f_k\}, ~\text{and}~ h_{\sigma, b}(f_0)=j, ~\text{and}~ \int_{f_0-\Delta}^{f_0+\Delta} | \widehat{x^*\cdot H}(f) |^2 \mathrm{d} f \geq  T\N^2/k  \right\}.
\end{align*}
Let $S_g=  S_{g1} \cap S_{g2}$.
Let  $S_f := \{ j\in[k] ~|~ h_{\sigma, b}(f_j) \in S_g\} \cap S$ and  $x_{S_f}(t):= \underset{j\in S_f}{\sum} v_j e^{2 \pi \i f_j t}.$

Then, we have
\begin{align*}
\|x_{S_f}(t)-x_S(t)\|_T^2 
\lesssim  \|{g}(t)\|_T^2.    
\end{align*}
\end{lemma}
\begin{proof}
By the definition of $S$ and $S_f$, we have that
\begin{align*}
    S_f\subseteq S.
\end{align*}

Let $ [L,R]:=U$. We have that for any $f\in S\backslash S_f$, $ j = h_{\sigma, b }(f)$,
\begin{align}
\|(g(t)\cdot H(t))*G^{(j)}_{\sigma, b}(t)\|_T^2 \geq&~ c \|(x^*(t)\cdot H(t))*G^{(j)}_{\sigma, b}(t)\|_{U}^2\notag \\
\geq&~ c \frac{T - k^2 (T+L-R)}{R-L} \|(x^*(t)\cdot H(t))*G^{(j)}_{\sigma, b}(t)\|_{T}^2\notag \\
\geq&~ O(c) \cdot \|(x^*(t)\cdot H(t))*G^{(j)}_{\sigma, b}(t)\|_{T}^2, \label{eq:gHGgeqcxstarHG}
\end{align}
where the first step follows from Eq.~\eqref{eq:define_S_g_1}, the second step follows from Lemma \ref{lem:norm_z_cut_preserve}, the third step follows from the Lemma \ref{lem:relation_R_L_T}.

Let $ {\cal T} = S\backslash S_f $. And for $j\in[B]$, let 
\begin{align*}
    {\cal T}_j :=\begin{cases}
    \left\{i\in S ~|~  h_{\sigma, b } (f_i)= j  \right\}, &~ \forall j\in [B]\backslash S_g,\\
    \emptyset, &~ \text{otherwise}.
    \end{cases}
\end{align*}
It is easy to see that
\begin{align*}
    {\cal T} = \bigcup_{i=1}^B {\cal T}_i.
\end{align*}
Moreover, by Lemma \ref{lem:property_of_filter_G} Property \RN{1} and \RN{3}, the definition of ${\cal T}_j$ and $\wh{G}^{(j)}_{\sigma, b}(f) $, and the Large Offset event not happening, we have that for any $f\in \supp(\wh{x}_{{\cal T}_j}*\wh{H})$,
\begin{align}
\wh{G}^{(j)}_{\sigma, b}(f) \geq 1-\frac{\delta}{k},\label{eq:inside_T_big}
\end{align}
where $x_{{\cal T}_j}=\sum_{i\in {\cal T}_j}v_i e^{2\pi \i f_i t}$ and $\wh{x}_{{\cal T}_j}$ is its Fourier transform.

Then, we have that 
\begin{align}
&~ T\|(x^*(t)\cdot H(t))*G^{(j)}_{\sigma, b}(t)\|_{T}^2 \notag \\
=&~ \int_0^T |(x^*(t)\cdot H(t))*G^{(j)}_{\sigma, b}(t)|^2 \d t \notag \\
\gtrsim &~ \int_{-\infty}^\infty |(x^*(t)\cdot H(t))*G^{(j)}_{\sigma, b}(t)|^2 \d t \notag \\
= &~ \int_{-\infty}^\infty |(\wh{x}^*(f)* \wh{H}(f))\cdot \wh{G}^{(j)}_{\sigma, b}(f)|^2 \d f \notag \\
= &~ \int_{-\infty}^\infty |(\wh{x}_{{\cal T}_j}(f)* \wh{H}(f))\cdot\wh{G}^{(j)}_{\sigma, b}(f)|^2 \d f + \int_{-\infty}^\infty |(\wh{x}_{[k]\backslash {\cal T}_j}(f)* \wh{H}(f))\cdot \wh{G}^{(j)}_{\sigma, b}(f)|^2 \d f \notag \\
\geq &~ \int_{-\infty}^\infty |(\wh{x}_{{\cal T}_j}(f)* \wh{H}(f))\cdot\wh{G}^{(j)}_{\sigma, b}(f)|^2 \d f \notag \\
\gtrsim &~ \int_{-\infty}^\infty |\wh{x}_{{\cal T}_j}(f)* \wh{H}(f)|^2 \d f \label{eq:xHGgeqxTjH}
\end{align}
where the first step follows from the definition of the norm, the second step follows from Lemma \ref{lem:z_satisfies_two_properties}, third step follows from Parseval's theorem, the forth step follows from the Large Offset event not happening and the definition of ${\cal T}_j$, the fifth step is straight forward, the sixth step follows from Eq.~\eqref{eq:inside_T_big}.

Thus, we have that 
\begin{align}
   &~T\|x_{S_f}(t)-x_S(t)\|_T^2 \notag \\  
   = &~T\|x_{\cal T}(t)\|_T^2 %
   \notag \\
   \lesssim &~T\|x_{\cal T}(t)\cdot H(t)\|_T^2 \notag \\
   = &~\int_0^T |x_{\cal T}(t)\cdot H(t)|^2 \d t \notag \\
   \leq &~\int_{-\infty}^\infty |x_{\cal T}(t)\cdot H(t)|^2 \d t \notag \\
   = &~\int_{-\infty}^\infty |\wh{x}_{\cal T}(f)* \wh{H}(f)|^2 \d f \notag \\
   = &~\sum_{j=1}^B \int_{-\infty}^\infty |\wh{x}_{{\cal T}_j}(f)* \wh{H}(f)|^2 \d f \notag \\
  \lesssim &~\sum_{j\in[B] \backslash S_g} T\|(x^*(t)\cdot H(t))*G^{(j)}_{\sigma, b}(t)\|_{T}^2 \notag \\
  \lesssim &~\sum_{j\in[B] \backslash S_g} T\|(g(t)\cdot H(t))*G^{(j)}_{\sigma, b}(t)\|_{T}^2 \label{eq:xSfsubxSleqSumBgHG}
\end{align}
where the first step follows from the definition of ${\cal T}$, the second step follows from $ x_{\cal T}$ is a $k$-Fourier-sparse signal and Lemma \ref{lem:property_of_filter_H} Property \RN{5}, the third step follows from the definition of the norm, the forth step is straight forward, the fifth step follows from Parseval's theorem, the sixth step follows from the definition of ${\cal T}_j$ and the Large Offset event not happened, the seventh step follows from Eq.~\eqref{eq:xHGgeqxTjH}, the eighth step follows from Eq.~\eqref{eq:gHGgeqcxstarHG}.

Eq.~\eqref{eq:xSfsubxSleqSumBgHG} can be upper bounded by the summation over all bins, which can be further upper bounded as follows:
\begin{align}
&~\sum_{j\in[B]} T \cdot \|(g(t)\cdot H(t))*G^{(j)}_{\sigma, b}(t)\|_{T}^2 \notag \\
= &~\sum_{j\in[B]}  \int_0^T |(g(t)\cdot H(t))*G^{(j)}_{\sigma, b}(t)|^2 \d t \notag \\
\leq &~\sum_{j\in[B]}  \int_{-\infty}^\infty |(g(t)\cdot H(t))*G^{(j)}_{\sigma, b}(t)|^2 \d t \notag \\
\leq &~\sum_{j\in[B]}  \int_{-\infty}^\infty |(\wh{g}(f)* \wh{H}(f))\cdot \wh{G}^{(j)}_{\sigma, b}(f)|^2 \d f \notag \\
= &~ \int_{-\infty}^\infty |(\wh{g}(f)* \wh{H}(f))|^2\cdot \sum_{j\in[B]}| \wh{G}^{(j)}_{\sigma, b}(f)|^2 \d f \notag \\
\lesssim &~ \int_{-\infty}^\infty |(\wh{g}(f)* \wh{H}(f))|^2 \d f \notag \\
= &~ \int_{-\infty}^\infty |({g}(t) \cdot {H}(t))|^2 \d t \notag \\
= &~ \int_0^T |({g}(t) \cdot {H}(t))|^2 \d t \notag \\
\lesssim &~ \int_0^T |{g}(t)|^2 \d t \notag \\
= &~ T \|{g}(t)\|_T^2 \label{eq:SumgHGleqg}
\end{align}
where the first step follows from the definition of the norm, the second step is straightforward, the third step follows from Parseval's theorem, the forth step is straightforward, the fifth step follows from Lemma \ref{lem:filterGlarge}, the sixth step follows from Parseval's theorem, the seventh step follows from $g(t)=0,\forall t\in \R\backslash [0, T]$, the eighth step follows from Lemma \ref{lem:property_of_filter_H} Property \RN{1}, \RN{2}, the ninth step follows from the definition of the norm.

Therefore, we get that 
\begin{align*}
&~ T\|x_{S_f}(t)-x_S(t)\|_T^2 \notag \\
   \lesssim&~ \sum_{j\in[B] \backslash S_g} T\|(g(t)\cdot H(t))*G^{(j)}_{\sigma, b}(t)\|_{T}^2 \notag \\
\leq &~ \sum_{j\in[B] } T\|(g(t)\cdot H(t))*G^{(j)}_{\sigma, b}(t)\|_{T}^2 \notag \\
\lesssim &~ T \|{g}(t)\|_T^2, 
\end{align*}
 where the first step follows from Eq.~\eqref{eq:xSfsubxSleqSumBgHG}, the second step is straight forward, the third step follows from Eq.~\eqref{eq:SumgHGleqg}.

The lemma is then proved.

\end{proof}

\subsection{Fourier interpolation with constant success probability}\label{sec:sig_recontr:constant_prob}

In this section, we give an algorithm for Fourier interpolation by combining our frequency estimation algorithm with a signal estimation algorithm. However, it only succeeds with a constant probability.

\begin{theorem}\label{thm:main_constant_prob}
Let $x(t) = x^*(t) + g(t)$, where $x^*(t)\in {\cal F}_{k,F}$ and $g(t)$ is arbitrary noise.  Given samples of $x$ over $[0, T]$, there is an algorithm (Procedure \textsc{ConstantProbFourierInterpolation} in Algorithm~\ref{alg:main_k}%
) that uses \begin{align*}
    O(k^{4} \log^{3} (k) \log^{2}(1/\delta_1)\log(\log(1/\delta_1))\log (FT)\log({\log(FT)})) 
\end{align*} samples, runs in \begin{align*}
    O(k^{4\omega} \log^{2\omega+1} (k) \log^{2\omega}(1/\delta_1)\log(\log(1/\delta_1))\log (FT)\log({\log(FT)})) 
\end{align*} time, and outputs an $O(k^4  \log^4(k/\delta))$-Fourier-sparse signal $y(t)$ such that with probability at least $0.6$,
\begin{align*}
    \|{y - x^*}\|_T \lesssim \|{g}\|_T + \delta\|{x^*}\|_T.
\end{align*}
\end{theorem}

\begin{proof}
Let $\N^2 := \| g(t) \|_T^2 + \delta \| x^*(t) \|_T^2$ be the noisy level of the observation signal. %

\paragraph{Heavy-clusters approximation.}Let $S$ be the set of heavy frequencies:
\begin{align*}
    S=\Big\{j \in [k]\bigg{|}\int_{f_j-\Delta_h}^{f_j+\Delta_h} |\wh{H \cdot x^*}(f)|^2 \mathrm{d} f \ge T \cdot \N^2/k\Big\},
\end{align*}
where $\Delta_h=|\supp(\wh{H})|$, and let $x_S(t)=\sum_{j\in S} v_j e^{2\pi \i f_j t}$.
By Claim \ref{cla:guarantee_removing_x**_x*_}, we have %
\begin{align}
    \|x_S - x^* \|_T \lesssim \N, \label{eq:fourier_intor:x_S_x*}
\end{align}
which implies that it suffices to reconstruct $x_S$, instead of $x^*$.

\paragraph{Frequency estimation.}Conditioning on Large Offset event not happening, which holds with probability at least 0.6 by Lemma \ref{lem:large_off_not_happen}, let $S_f\subseteq S$ be defined as in Lemma~\ref{lem:xSfsubxSleqg} and $x_{S_f}(t)=\sum_{j\in S_f}v_j e^{2\pi \i f_j t}$. By Lemma~\ref{lem:xSfsubxSleqg}, we have
\begin{align}\label{eq:lem_M_8}
    \|x_{S_f}(t)-x_S(t)\|_T^2 \lesssim  \|{g}(t)\|_T^2.
\end{align}
Furthermore, by Theorem \ref{thm:frequency_recovery_k_better}, there is an algorithm that outputs a set of frequencies $L\subset \R$ of size $B$ such that with probability at least $1-2^{-\Omega(k)}$,  for any $j\in S_f$, there exists an $\wt{f}\in L$ such that,
\begin{align*}
|f_j-\widetilde{f} |\lesssim \Delta.
\end{align*}

\paragraph{Fourier-polynomial mixed signal approximation.}We define a map $p:\R\rightarrow L$ as follows:
\begin{align*}
    p(f):=\arg \min_{\wt{f}\in L} ~ |f-\wt{f}|~~~\forall f\in \R.
\end{align*}
Then, $x_{S_f}(t)$ can be expressed as
\begin{align*}
x_{S_f}(t)= &~\sum_{j\in {S_f}}v_je^{2\pi \i f_j t}\\
= &~ \sum_{j\in {S_f}}v_j e^{2\pi \i \cdot p(f_j)t} \cdot e^{2\pi \i \cdot (f_j - p(f_j))t}\\
= &~ \sum_{\wt{f} \in L} e^{2 \pi \i \widetilde{f} t} \cdot \sum_{j\in {S_f}:~ p(f_j)=\wt{f}} v_j e^{2\pi \i ( f_j - \widetilde{f})t},
\end{align*}
where the first step follows from the definition of $x_{S_f}$, the last step follows from interchanging the summations.

For each $\wt{f}_i\in L$, by Lemma \ref{lem:low_degree_approximates_concentrated_freq_ours} with $ x^*=x_{S_f}$, there exists a degree $ d=O(T \Delta + k^3 \log k + k \log 1/\delta)$ polynomial $P_i(t)$  such that, 
\begin{align}
\Big\|x_{S_f}(t)-\sum_{\wt{f}_i \in L} e^{2 \pi \i \widetilde{f}_i t} P_i(t)\Big\|_T \leq \sqrt{\delta} \|x_{S_f}(t)\|_T\label{eq:fourier:xS_sum}
\end{align}

\paragraph{Reconstructing the polynomials.}Define the following function family: 
\begin{align*}
  \mathcal{F} := \mathrm{span}\Big\{e^{2\pi \i \widetilde{f} t} \cdot t^j~{|}~  \wt{f}\in L, j \in \{0,1,\dots,d\} \Big\}.
\end{align*}
Note that $\sum_{\wt{f}_i \in L} e^{2 \pi \i \widetilde{f}_i t} P_i(t)\in {\cal F}$.

Let $ D:= d \cdot |L|$. By Lemma \ref{lem:magnitude_recovery_for_CKPS}, there is an algorithm that runs in $O(\eps^{-1}D^\omega\log^3(D)\log(1/\rho))$-time using $O(\eps^{-1}D\log^3(D)\log(1/\rho)) $ samples, and outputs $y'(t)\in {\cal F}$ such that, with probability $1-\rho$, 
\begin{align}
    \Big\|y'(t) - \sum_{\wt{f}_i \in L} e^{2 \pi \i \widetilde{f}_i t} P_i(t)\Big\|_T\leq (1+\eps)\Big\|x(t)-\sum_{\wt{f}_i \in L} e^{2 \pi \i \widetilde{f}_i t} P_i(t)\Big\|_T\label{eq:fourier:y_sum_0}
\end{align}
Thus, we have that
\begin{align}
\Big\|y'(t) - \sum_{\wt{f}_i \in L} e^{2 \pi \i \widetilde{f}_i t} P_i(t)\Big\|_T
\lesssim &~  \Big\|\sum_{\wt{f}_i \in L} e^{2 \pi \i \widetilde{f}_i t} P_i(t) - x^*(t)\Big\|_T  + \|x(t)-x^*(t)\|_T \notag \\
= &~  \Big\|\sum_{\wt{f}_i \in L} e^{2 \pi \i \widetilde{f}_i t} P_i(t) - x^*(t)\Big\|_T  + \|g(t)\|_T ,
\label{eq:fourier:y_sum}
\end{align}
where the first step  follows from triangle inequality, the second step follows from the definition of $g(t)$.

For the first term, we have that
\begin{align}
\Big\|\sum_{\wt{f}_i \in L} e^{2 \pi \i \widetilde{f}_i t} P_i(t) - x^*(t)\Big\|_T  
\lesssim &~  \Big\|\sum_{\wt{f}_i \in L} e^{2 \pi \i \widetilde{f}_i t} P_i(t) - x_{S_f}(t)\Big\|_T+\|x_{S_f}(t)-x^*(t)\|_T  \notag \\
\lesssim &~   \sqrt{\delta} \|x_{S_f}(t)\|_T +\|x_{S_f}(t)-x^*(t)\|_T  \notag \\
\leq &~   \sqrt{\delta} (\|x_{S_f}(t)-x^*(t)\|_T  + \|x^*(t)\|_T) +\|x_{S_f}(t)-x^*(t)\|_T  \notag \\
\lesssim &~   \|x_{S_f}(t)-x^*(t)\|_T + \sqrt{\delta}\|x^*(t)\|_T \notag\\
\leq&~ \|x_{S_f}(t)-x_S(t)\|_T + \|x_S(t)-x^*(t)\|_T+ \sqrt{\delta}\|x^*(t)\|_T \notag\\
\lesssim &~ \|x_{S_f}(t)-x_S(t)\|_T + {\cal N}+ \sqrt{\delta}\|x^*(t)\|_T \notag\\
\lesssim &~ \|g(t)\|_T + {\cal N} + \sqrt{\delta}\|x^*(t)\|_T,\label{eq:PsubxstarleqN}
\end{align}
where the first step follows from triangle  inequality,  the second step follows from Eq.~\eqref{eq:fourier:xS_sum}, the third step follows from triangle inequality, the forth step follows is straightforward, the fifth step follows from triangle inequality, the sixth step follows from Eq.~\eqref{eq:fourier_intor:x_S_x*}, and the last step follows from Eq.~\eqref{eq:lem_M_8}.  

Hence, we get that
\begin{align}
\Big\|y'(t)-\sum_{\wt{f}_i \in L} e^{2 \pi \i \widetilde{f}_i t} P_i(t)\Big\|_T 
\leq &~ \Big\|\sum_{\wt{f}_i \in L} e^{2 \pi \i \widetilde{f}_i t} P_i(t) - x^*(t)\Big\|_T  + \|g(t)\|_T  \notag \\
\lesssim &~ \|g(t)\|_T + \N + \sqrt{\delta}\|x^*(t)\|_T   \label{eq:yprimesubPleqN}
\end{align}
where the first step follows from Eq.~\eqref{eq:fourier:y_sum}, the second step follows from Eq.~\eqref{eq:PsubxstarleqN}. 

\paragraph{Transforming back to Fourier-sparse signal.}By Lemma \ref{lem:polynomial_to_FT}, we have that there is a $O(kd) $-Fourier-sparse signal $y(t)$, such that
\begin{align}\label{eq:approx_y}
    \|y(t)-y'(t)\|_T \leq \delta'
\end{align}
where $\delta'>0$ is any positive real number. Thus, $y(t)$ can be arbitrarily close to $y'(t)$. 
Moreover, the sparsity of $y(t)$ is 
\begin{align*}
    O(kd )= O(k\cdot (T \Delta + k^3 \log k + k \log 1/\delta)) = O(k^4  \log^4(k/\delta)),
\end{align*}
which follows from Lemma \ref{lem:property_of_filter_H} Property \RN{3}: 
\begin{align*}
    \Delta =  k \Delta_h = k |\supp(\wh{H})| = O(k^3 \log^2 (k) \log^2(1/\delta_1)/T) .
\end{align*}

Moreover,  we take
\begin{align}
    \N = \sqrt{\| g \|_T^2 + \delta \| x^* \|_T^2}\leq \|g\|_T+\sqrt{\delta}\|x^*\|_T. \label{eq:Nleqgdeltaxstar}
\end{align}

Therefore, the total approximation error can be bounded as follows:
\begin{align}
 &~ \|y(t)-x^*(t)\|_T\notag  \\
\leq &~ \|y(t)-y'(t)\|_T+ \Big\|y'(t)-\sum_{\wt{f}_i \in L} e^{2 \pi \i \widetilde{f}_i t} P_i(t)\Big\|_T + \Big\|\sum_{\wt{f}_i \in L} e^{2 \pi \i \widetilde{f}_i t} P_i(t) - x^*(t)\Big\|_T \notag \\
\lesssim &~ \Big\|y'(t)-\sum_{\wt{f}_i \in L} e^{2 \pi \i \widetilde{f}_i t} P_i(t)\Big\|_T + \Big\|\sum_{\wt{f}_i \in L} e^{2 \pi \i \widetilde{f}_i t} P_i(t) - x^*(t)\Big\|_T   \notag \\
\lesssim &~ \Big\|y'(t)-\sum_{\wt{f}_i \in L} e^{2 \pi \i \widetilde{f}_i t} P_i(t)\Big\|_T + \N + \|g(t)\|_T + \sqrt{\delta}\|x^*(t)\|_T\notag\\
\lesssim &~  \N + \|g(t)\|_T  + \sqrt{\delta}\|x^*(t)\|_T \notag \\
\lesssim &~  \|g(t)\|_T  + \sqrt{\delta}\|x^*(t)\|_T ,
\end{align}
where the first step follows from triangle inequality, the second step follows from Eq.~\eqref{eq:approx_y}, the third step follows from Eq.~\eqref{eq:PsubxstarleqN}, the forth step follows from Eq.~\eqref{eq:yprimesubPleqN}, the fifth step follows from $\N = \sqrt{\| g \|_T^2 + \delta \| x^* \|_T^2}\leq \|g\|_T+\sqrt{\delta}\|x^*\|_T$.

The correctness then follows by re-scaling $\delta$.

The running time of the algorithm follows from Lemma~\ref{lem:time_cpfi}, and the sample complexity follows from Lemma~\ref{lem:smaple_cpfi}.

The theorem is then proved.

\end{proof}

\begin{lemma}[Running time of Algorithm \ref{alg:main_k}]\label{lem:time_cpfi}
Procedure \textsc{ConstantProbFourierInterpolation} in Algorithm \ref{alg:main_k} runs in \begin{align*}
    O(k^{4\omega} \log^{2\omega+1} (k) \log^{2\omega}(1/\delta_1)\log(\log(1/\delta_1))\log (FT)\log({\log(FT)}))
\end{align*} times.
\end{lemma}
\begin{proof}
Procedure \textsc{ConstantProbFourierInterpolation} in Algorithm \ref{alg:main_k} consists of the following two steps:
\begin{itemize}
    \item Line \ref{ln:1_main_con} calls Procedure \textsc{FrequencyEstimationX}. By Theorem \ref{thm:frequency_recovery_k_better}, it runs in \begin{align*}
        O(k^2\log(k)\log(k/\delta_1)\log (FT)\log({\log(FT)} ))
    \end{align*} time.  %
    \item Line \ref{ln:2_main_con} calls Procedure \textsc{SignalEstimation}. By Lemma \ref{lem:magnitude_recovery_for_CKPS}, it runs in \begin{align*}
        O(\eps^{-1}D^\omega\log(D)\log(1/\rho))
    \end{align*} time, where $ \eps, \rho$ are set to be universal constants and $D= B\cdot d$.
\end{itemize}

Following from the setting in the algorithm, we have that 
\begin{align*}
    B =&~ O(k),\\
    d =&~ O( \Delta T + k^3 \log k + k \log 1/\delta).
\end{align*}

By Lemma \ref{lem:property_of_filter_H} Property \RN{3}, we have that
\begin{align*}
    \Delta =  k \Delta_h = k |\supp(\wh{H}(f))| = O(k^3 \log^2 (k) \log^2(1/\delta_1)/T) .
\end{align*}

As a result, we have that
\begin{align}
    D = B\cdot d = O(k^4 \log^2 (k) \log^2(1/\delta_1)) \label{eq:D_k_delta_000}
\end{align}

Thus, the time complexity of Procedure \textsc{ConstantProbFourierInterpolation} in Algorithm \ref{alg:main_k} is 
\begin{align*}
    &~ O(k^2\log(k)\log(k/\delta_1)\log (FT)\log({\log(FT)}/{\rho_1} ))  + O(\eps^{-1}D^\omega\log(D)\log(1/\rho))\\
   \leq  &~  O(k^2\log(k)\log(k/\delta_1)\log (FT)\log({\log(FT)}/{\rho_1} )) \\
   &~ \quad + O(\eps^{-1}(k^4 \log^2 (k) \log^2(1/\delta_1))^\omega\log(k^4 \log^2 (k) \log^2(1/\delta_1))\log(1/\rho))\\
   \leq  &~  O(k^{4\omega} \log^{2\omega+1} (k) \log^{2\omega}(1/\delta_1)\log(\log(1/\delta_1))\log (FT)\log({\log(FT)})) %
\end{align*}
where the first step follows from Eq.~\eqref{eq:D_k_delta_000}, the second step follows from $\eps=O(1), \rho=O(1), \rho_1=O(1)$. %

\end{proof}

\begin{lemma}[Sample complexity of  Algorithm \ref{alg:main_k}]\label{lem:smaple_cpfi}
 Procedure \textsc{ConstantProbFourierInterpolation} in Algorithm \ref{alg:main_k} takes \begin{align*}
     O(k^{4} \log^{3} (k) \log^{2}(1/\delta_1)\log(\log(1/\delta_1))\log (FT)\log({\log(FT)}))
 \end{align*} samples.
\end{lemma}
\begin{proof}

The sample complexity of each steps of Procedure \textsc{ConstantProbFourierInterpolation} in Algorithm \ref{alg:main_k} is as follows:
\begin{itemize}
    \item Line \ref{ln:1_main_con} calls Procedure \textsc{FrequencyEstimationX}. By Theorem \ref{thm:frequency_recovery_k_better}, it takes  \begin{align*}
        O(k^2\log(k)\log(k/\delta_1)\log (FT)\log({\log(FT)} )) 
    \end{align*} samples. %
    \item Line \ref{ln:2_main_con} calls Procedure \textsc{SignalEstimation}. By Lemma \ref{lem:magnitude_recovery_for_CKPS}, it takes in \begin{align*}
        O(\eps^{-1}D\log(D)\log(1/\rho))
    \end{align*} samples, where $ \eps, \rho$ are set to be a universal constant and $D= B\cdot d$. 
\end{itemize}

By Eq.~\eqref{eq:D_k_delta_000}, we have
\begin{align*}
D = O(k^4 \log^2 (k) \log^2(1/\delta_1)).
\end{align*}

Thus, the sample complexity of Procedure \textsc{ConstantProbFourierInterpolation} in Algorithm \ref{alg:main_k} is 
\begin{align*}
    &~ O(k^2\log(k)\log(k/\delta_1)\log (FT)\log({\log(FT)}/{\rho_1} ))  + O(\eps^{-1}D\log(D)\log(1/\rho))\\
   \leq  &~  O(k^2\log(k)\log(k/\delta_1)\log (FT)\log({\log(FT)}/{\rho_1} )) \\
   &~ \quad + O(\eps^{-1}(k^4 \log^2 (k) \log^2(1/\delta_1))\log(k^4 \log^2 (k) \log^2(1/\delta_1))\log(1/\rho))\\
   \leq  &~  O(k^{4} \log^{3} (k) \log^{2}(1/\delta_1)\log(\log(1/\delta_1))\log (FT)\log({\log(FT)})) %
\end{align*}
where the first step follows from Eq.~\eqref{eq:D_k_delta_000}, the second step follows from $\eps=O(1), \rho=O(1), \rho_1=O(1)$. %

\end{proof}

\subsection{Min-of-medians signal estimator}\label{sec:sig_recontr:boost}

In this section, we propose a ``min-of-medians'' estimator for signals that can exponentially boost the success probability. 

\begin{lemma}
Let $R_p\in \mathbb{N}$. For each $i\in [R_p]$, let $y_i(t)$ be a signal independently sampled from some distribution such that with probability at least $0.9$, %
\begin{align*}
    \| y_i(t) - x^*(t) \|_T^2 \lesssim \| g(t) \|_T^2 .
\end{align*}
Let $y(t) := y_{j^*}(t)$ where
\begin{equation*}
j^* := \underset{j \in [R_p] }{\arg\min}~\underset{i\in [R_p]}\median~ \| y_j(t) -y_i(t) \|_T^2.
\end{equation*}

Then, with probability at least $ 1-2^{-\Omega(R_p)}$, 
\begin{align*}
    \| y(t) - x^*(t) \|_T^2 \lesssim \| g(t) \|_T^2.
\end{align*}
\end{lemma}
\begin{proof}

Let $S=\{i ~|~ \| y_i(t) - x^*(t) \|_T^2 \lesssim \| g(t) \|_T^2\}$. %
By the Chernoff bound, we have that 
\begin{align*}
    \Pr[|S|\geq 3/4 R_p] \geq 1-2^{-\Omega(R_p)}. 
\end{align*}

For the ease of discussion, we suppose $ |S|\geq 3/4 R_p$ holds in the following proof.

Fix any $j\in S$. Then, for any $ q \in S$, we have that
\begin{align}
     \| y_j(t) -y_q(t) \|_T \leq \| y_j(t) -x^*(t) \|_T + \| x^*(t) -y_q(t) \|_T \lesssim \| g(t) \|_T^2, \label{eq:wty_sub_y_S_le}
\end{align}
where the first step follows from triangle inequality, the second step follows from the definition of $S$. 

In other words, there are at least $|S|\geq (3/4)R_p$ elements such that Eq.~\eqref{eq:wty_sub_y_S_le} holds. By the definition of median, we get that
\begin{align}
\underset{i \in [R_p] }{\median}~ \| y_j(t) -y_i(t) \|_T^2  \lesssim \| g(t) \|_T^2. \label{eq:bound_ygood_wty}
\end{align}

By definition of $y(t)$, we have that, 
\begin{align}
\underset{i \in [R_p] }{\median} ~ \| y(t) -y_i(t) \|_T^2  \leq \underset{i \in [R_p] }{\median} ~ \| y_j(t) -y_i(t) \|_T^2 \lesssim \| g(t) \|_T^2, \label{eq:med_le_g}
\end{align}
where the first step follows from the definition of $ y(t)$, the second step follows from Eq.~\eqref{eq:bound_ygood_wty}. 

By the definition of median, we know that there are $R_p/2$ elements $r\in [R_p]$ such that
\begin{align}
\| y(t) -y_r(t) \|_T^2 \leq \underset{i \in [R_p] }{\median} ~ \| y(t) -y_i(t) \|_T^2\lesssim \|g(t)\|_T^2 , \label{eq:exist_good_leq_med}
\end{align}
where the last step follows from Eq.~\eqref{eq:med_le_g}. Since $|S|\geq (3/4)R_p>(1/2)R_p$, there must exists an $r\in S$ such that Eq.~\eqref{eq:exist_good_leq_med} holds.

As a result, we have that
\begin{align*}
\| y(t) -x^*(t) \|_T^2 \leq&~ \| y(t) -y_r(t) \|_T^2  + \| y_r(t) -x^*(t) \|_T^2 \notag \\
\lesssim&~  \| g(t) \|_T^2  + \| y_r(t) -x^*(t) \|_T^2\notag \\
\lesssim&~  \| g(t) \|_T^2 ,
\end{align*}
where the first step follows from triangle inequality,  the second step follows from Eq.~\eqref{eq:med_le_g}, the third step follows from the definition of $S$. 

The lemma is then proved.
\end{proof}

One potential issue in applying the min-of-median signal estimator is that, we may not be able to compute the distances $\|y_i(t)-y_j(t)\|_T^2$ exactly, but we can only estimate then with high accuracy. Therefore, we show that our estimator is robust with respect to approximated distances.

We first show a fact about the approximation of min and median. 

\begin{fact}\label{fac:robust_min_median}
Let $ x_1,\cdots, x_n\in \R_+$, and $ y_1,\cdots, y_n \in \R_+$ such that
for any $ i\in [n]$, $ y_i \in [\alpha\cdot x_i, \beta\cdot x_i] $. Then, we have: 
\begin{itemize}
    \item $ \underset{i\in[n] }{\min} ~ y_i  \in \Big[\alpha\cdot \underset{i\in[n] }{\min}~ x_i, \beta \cdot \underset{i\in[n] }{\min}~x_i\Big] $.
    \item $\underset{i\in[n] }{\median} ~ y_i \in \Big[\alpha\cdot \underset{i\in[n] }{\median} ~ x_i, \beta\cdot \underset{i\in[n] }{\median}~ x_i \Big]$.  
\end{itemize}
\end{fact}
\begin{proof}
\textbf{Part 1:} 
Let $i^* =  \underset{i\in[n] }{\arg\min} ~ y_i$. 
We have that
\begin{align*}
    y_{i^*}
    \geq \alpha\cdot x_{i^*}
    \geq \alpha\cdot \underset{i\in[n] }{\min} ~  x_i.
\end{align*}

Let $j^* =  \underset{j\in[n] }{\arg\min} ~ x_j$.
We have that
\begin{align*}
\underset{j\in[n] }{\min} ~ y_j\leq y_{j^*}
 \leq  \beta \cdot x_{j^*}
 =  \beta \cdot \underset{j\in[n] }{\min} ~ x_j ,
\end{align*}

Hence, 
\begin{align*}
    \underset{i\in[n] }{\min} ~ y_i  \in \Big[\alpha\cdot \underset{i\in[n] }{\min}~ x_i, \beta \cdot \underset{i\in[n] }{\min}~x_i\Big].
\end{align*}

\textbf{Part 2:}
For any $ x_j \leq  \underset{i\in[n] }{\median}~ x_i$, we have that
\begin{align*}
    y_j \leq \beta\cdot x_j \leq \beta\cdot \underset{i\in[n] }{\median} ~ x_i.
\end{align*}
Thus, 
\begin{align*}
    |\{j\in [n] ~|~ y_j  \leq \beta\cdot \underset{i\in[n] }{\median} ~ x_i \}|\geq n/2, 
\end{align*} which implies that 
\begin{align*}
    \underset{i\in[n] }{\median} ~ y_i  \leq \beta\cdot \underset{i\in[n] }{\median} ~ x_i.
\end{align*}

For any $ x_j \geq  \underset{i\in[n] }{\median} ~ x_i$, we have that
\begin{align*}
    y_j \geq \alpha\cdot x_j \geq \alpha\cdot \underset{i\in[n] }{\median}~ x_i.
\end{align*}
Thus, \begin{align*}
    |\{j\in [n] ~|~ y_j  \geq \alpha\cdot \underset{i\in[n] }{\median}~ x_i\}|\geq n/2, 
\end{align*} which implies that 
\begin{align*}
      \underset{i\in[n] }{\median} ~ y_i \geq \alpha\cdot \underset{i\in[n] }{\median} ~ x_i .
\end{align*}

As a result, 
\begin{align*}
  \underset{i\in[n] }{\median}~ y_i \in \Big[\alpha\cdot \underset{i\in[n] }{\median} ~ x_i, \beta\cdot \underset{i\in[n] }{\median}~ x_i\Big].    
\end{align*}
\end{proof}

The following lemma shows that our min-and-median estimator can still exponentially boost the success probability given access to approximated distances.

\begin{lemma}[Robust min-of-median signal  estimator]\label{eq:med_boost_prob}
Let $R_p\in \mathbb{N}$. For each $i\in [R_p]$, let $y_i(t)$ be a signal independently sampled from some distribution such that with probability at least $0.9$,
\begin{align*}
    \| y_i(t) - x^*(t) \|_T^2 \lesssim \| g(t) \|_T^2 .
\end{align*}
Given $d\in \R_+^{R_p\times R_p}$ such that for any $i, j \in [R_p]$,
\begin{align*}
    d_{i,j} \in \Big[\alpha\cdot \| y_i(t) -y_j(t) \|_T^2 , \beta\cdot \| y_i(t) -y_j(t)\|_T^2\Big]. 
\end{align*}
Let $y(t) := y_{j^*}(t)$ where
\begin{equation*}
j^* := \underset{j \in [R_p] }{\arg\min}~\underset{i \in [R_p] }{\median} ~ d_{j,i}.
\end{equation*}

Then, we have that, with probability at least $ 1-2^{-\Omega(R_p)}$, 
\begin{align*}
    \| y(t) - x^*(t) \|_T^2 \lesssim \frac{\beta}{\alpha} \| g(t) \|_T^2 .
\end{align*}
\end{lemma}
\begin{proof}
Let $S=\{i ~|~ \| y_i(t) - x^*(t) \|_T^2 \lesssim \| g(t) \|_T^2\}$. %
By the Chernoff bound, we have that 
\begin{align*}
    \Pr[|S|\geq 3/4 R_p] \geq 1-2^{-\Omega(R_p)}. 
\end{align*}

For the ease of discussion, we suppose $ |S|\geq 3/4 R_p$ holds in the following proof.

Fix any $i^*\in S$. For any $ q \in S$, we have that
\begin{align}
     \| y_{i^*}(t) -y_q(t) \|_T \leq \| y_{i^*}(t) -x^*(t) \|_T + \| x^*(t) -y_q(t) \|_T \lesssim \| g(t) \|_T^2, \label{eq:000_wty_sub_y_S_le}
\end{align}
where the first step follows from triangle inequality, the second step follows from the definition of $S$. 

By the definition of median, since $|S|>R_p/2$, we know that
\begin{align}
\underset{i \in [R_p] }{\median} ~ \| y_{i^*}(t) -y_i(t) \|_T^2   \lesssim \| g(t) \|_T^2. \label{eq:000_bound_ygood_wty}
\end{align}

Then, we have that 
\begin{align}
\underset{i \in [R_p] }{\median} ~ d_{j^*,i} \leq &~\underset{i \in [R_p] }{\median} ~ d_{i^*,i} \notag \\
\leq &~
\beta \cdot \underset{i \in [R_p] }{\median} ~  \| y_{i^*}(t) -y_i(t) \|_T^2 \notag \\
\lesssim&~ \beta \cdot  \| g(t) \|_T^2, \label{eq:000_med_le_g}
\end{align}
where the first step follows from the definition of $ j^*$, the second step follows from Fact~\ref{fac:robust_min_median}, the third step follows from Eq.~\eqref{eq:000_bound_ygood_wty}. 

Since $|S|>R_p/2$, by the definition of median, there must exists an $r \in S$ such that
\begin{align}
d_{j^*, r} \leq \underset{i \in [R_p] }{\median}~  d_{j^*,i}\lesssim \beta \cdot \|g(t)\|_T^2.  \label{eq:000_exist_good_leq_med}
\end{align}

As a result, we have that
\begin{align*}
\| y(t) -x^*(t) \|_T^2 \leq&~ \| y(t) -y_r(t) \|_T^2  + \| y_r(t) -x^*(t) \|_T^2 \notag \\
\leq&~ \frac{1}{\alpha} v_{j^*, r}  + \| y_r(t) -x^*(t) \|_T^2 \notag \\
\lesssim&~  \frac{\beta}{\alpha}\| g(t) \|_T^2  + \| y_r(t) -x^*(t) \|_T^2\notag \\
\lesssim&~ \frac{\beta}{\alpha} \| g(t) \|_T^2 ,
\end{align*}
where the first step follows from triangle inequality, the second step follows from the definition of $d$, the third step follows from  Eq.~\eqref{eq:000_exist_good_leq_med}, the forth step follows from Eq.~\eqref{eq:000_med_le_g}, the fifth step follows from $r\in S$. 

The proof of the lemma is then completed.
\end{proof}

\subsection{Main algorithm for Fourier interpolation}\label{sec:sig_recontr:high_prob}

In this section, we present our main theorem---a time and sample efficient Fourier interpolation algorithm with high success probability. The pseudocode is given in Algorithm~\ref{alg:main_k}.

\begin{algorithm}[!ht]\caption{Signal estimation algorithm} \label{algo:sig_est_1d_accuracy}
\begin{algorithmic}[1]
\Procedure{\textsc{WeightedSketch}}{$m, k, T$}
\State $ c\leftarrow\Theta(\log(k)^{-1})$
\State $D(t)$ is defined as follows:
\begin{align*} D(t)\leftarrow \begin{cases} {c}\cdot (1-|t/T| )^{-1}T^{-1}, & \text{ for } |t| \le T(1-{1}/k) , \\ c \cdot  k T^{-1}, & \text{ for } |t|\in [T(1-{1}/k), T] . \end{cases} \end{align*} 
\State $S_0\gets m$ i.i.d. samples from $D$ \label{ln:s_0_smaple_D_ws}
\For{$t\in S_0$}\label{ln:t_s_0_ws}
    \State $w_t\gets \frac{1}{2T\cdot |S_0|\cdot D(t)}$
\EndFor
\State\label{step:new_dist} Set a new distribution $D'(t)\leftarrow w_t/\sum_{t'\in S_0} w_{t'}$ for all $t\in S_0$
\State \Return $D'$
\EndProcedure
\Procedure{\textsc{SignalEstimation}}{$x, F, T, L$} 

\State $\{f_1,f_2,\cdots,f_{B}\} \leftarrow L $ \Comment{$L\in \R^B$}
\State  $d \leftarrow O( \Delta T + k^3 \log k + k \log 1/\delta)$
\State $s, \{t_1,t_2,\cdots, t_s\},w  \leftarrow\textsc{WeightedSketch}(O(Bd\log(Bd)), Bd, T)$ \Comment{$w \in \R^{s }$} %
\State $A_{i, B\cdot j_2 + j_1} \leftarrow t_i^{j_2}\cdot \exp(2\pi\i f_{j_1} t_i)$, $A\in \C^{s \times B}$
\State $b  \leftarrow (x(t_1),x(t_2),\cdots, x(t_s))^\top$
\State Solving the following weighted linear regression\label{ln:sig_est_regress_fast}
\begin{align*}
v' \gets \underset{v' \in \C^{ B d } }{\arg\min}\| \sqrt{w} \circ (A v'- b)\|_2.  
\end{align*}
\State \Return $y(t) \leftarrow \sum_{j_1=1}^{B}\sum_{j_2=1}^d v_{B\cdot j_2 +j_1}' \cdot t^{j_2} \cdot \exp(2\pi \i f_{j_1}' t)$.
\EndProcedure

\end{algorithmic}
\end{algorithm}

\begin{algorithm}[!ht]
\caption{Fourier-sparse signal interpolation}\label{alg:main_k}
\begin{algorithmic}[1]
\Procedure{$\textsc{ConstantProbFourierInterpolation}$}{$x,H,G,T,F$}
	\State \label{ln:1_main_con} $ L \leftarrow \textsc{FrequencyEstimationX}(x,H,G,T,F)$ \Comment{$L\in \R^{B}$}
	\State \label{ln:2_main_con} $y(t) \leftarrow \textsc{SignalEstimation}(x,\eps, k, F, T , L)$ 
\State \Return $y(t)$
\EndProcedure
\Procedure{$\textsc{HighProbFourierInterpolation}$}{$x,H,G,T,F$}
\State $R_p \leftarrow \log(1/\rho)$
\For{$i\in [R_p]$}\label{ln:hpfi_for_i_R}
\State $y_i(t)\leftarrow \textsc{ConstantProbFourierInterpolation}(x,H,G,T,F)$ \label{ln:hpfi_cpfi}
\EndFor
\State \label{ln:hpfi_ms} $y(t) \leftarrow \textsc{MergeSignal}(y_1(t), y_2(t), \cdots, y_{R_p}(t))$ %
\State \Return $y(t)$
\EndProcedure
\Procedure{{MergeSignal}}{$y_1(t), y_2(t), \cdots, y_{R_p}(t)$} %
\State  $d \leftarrow O( \Delta T + k^3 \log k + k \log 1/\delta)$
\For{$i\in [R_p]$}\label{ln:mergeS_i_R_p}
\For{$j\in [R_p]$}\label{ln:mergeS_j_R_p}
\State $s, \{t_1,t_2,\cdots, t_s\},w  \leftarrow\textsc{WeightedSketch}(O(Bd\log(Bd)\log(R_p^2/\rho)), 2 \cdot B d, T)$ 
\State \Comment{$w \in \R^{s }$} \label{ln:merge_ws}
\State $S\leftarrow\{t_1,t_2,\cdots, t_s\}$
\State $Y \leftarrow \textsc{MixedPolynomialEvaluation}(y_i-y_j, S)$ \Comment{$Y\in \C^{s}$} \label{ln:merge_mpe}
\State $\|y_i(t)-y_j(t)\|^2_{S,w} \leftarrow \sum_{l=1}^s w_l \cdot |Y_l|^2$
\EndFor
\State $\textsf{med}_i \leftarrow {\median}_{j\in [R_p]}\{ \|y_i-y_j\|^2_{S,w}\}$
\EndFor
\State $ i^* \leftarrow {\arg\min}_{i\in [R_p]} \{\textsf{med}_i\}$ \label{ln:merge_med}
\State \Return $ y_{i^*}$
\EndProcedure
\end{algorithmic}
\end{algorithm}

\begin{theorem}[Main Fourier interpolation algorithm]\label{thm:main}
Let $x(t) = x^*(t) + g(t)$, where $x^*$ is $k$-Fourier-sparse signal with frequencies in $[-F, F]$.
Given samples of $x$ over $[0, T]$, there is an algorithm (Procedure \textsc{HighProbFourierInterpolation}) uses 
\begin{align*}
    O(k^{4}  \log^{6}(k/\delta)\log (FT)\log({\log(FT)}) \log(1/\rho))
\end{align*}
samples, runs in 
\begin{align*}
    O(k^{4\omega} \log^{4\omega+2}(k/\delta)\log (FT)\log({\log(FT)})\log^5(1/\rho))
\end{align*}
time, and
outputs an $O(k^4  \log^4(k/\delta))$-Fourier-sparse signal $y(t)$ such that with probability at least $1-\rho$,
\begin{align*}
    \|{y - x^*}\|_T \lesssim \|{g}\|_T + \delta\|{x^*}\|_T.
\end{align*}
\end{theorem}
\begin{proof}

We first prove the correctness of the algorithm.

Let
\begin{align*} D(t):= \begin{cases} {c}\cdot (1-|t/T| )^{-1}T^{-1}, & \text{ for } |t| \le T(1-{1}/k) , \\ c \cdot  k T^{-1}, & \text{ for } |t|\in [T(1-{1}/k), T] . \end{cases} \end{align*} 

Let $ y_1(t),\cdots, y_{R_p}(t)$ be the outputs of $R_p$ independent runs of Procedure \textsc{ConstantProbFourierInterpolation} in Algorithm~\ref{alg:main_k}. By Theorem \ref{thm:main_constant_prob}, we have that for any $j\in [R_p]$ with probability at least $ 0.9$, 
\begin{align*}
    \|y_{j}(t)-x^*(t)\|_T^2 \lesssim \|g(t)\|_T^2. 
\end{align*}

Let $S=\{t_1,\dots,t_s\}$ be $s=O(k\log(k)\log(R_p^2/\rho))$ i.i.d. samples from $D(t)$, and let $w_i=1/(TsD(t_i))$ for $i\in [s]$.
By Lemma \ref{lem:x_energy_preserving_000}, for any $i,j\in [R_p]$, %
with probability at least $ 1- \rho/ R_p^2$, 
\begin{align*}
    \|y_i(t)-y_j(t)\|_{S, w}^2 \in [1/2, 3/2] \cdot \|y_i(t)-y_j(t)\|_T^2. 
\end{align*}

Let
\begin{align*}
    j^* = \underset{j \in [R_p] }{\arg\min}~\underset{i\in [R_p]}{\median}~ \| y_j(t) -y_i(t) \|_T^2,
\end{align*}
and let $ y(t) := y_{j^*}(t)$. 

By Lemma \ref{eq:med_boost_prob}, we have that with probability at least $ 1- 2^{-\Omega(R_p)}$,
\begin{align*}
 \|y_{j}(t)-x^*(t)\|_T^2 \lesssim \frac{3/2}{1/2} \|g(t)\|_T^2 \eqsim \|g(t)\|_T^2.    
\end{align*}

By setting $ R_p = \log(1/\rho)$, we get the desired result. The correctness is then proved.

The time complexity follows from Lemma~\ref{lem:t_complexity_main_k}. And the sample complexity follows from Lemma~\ref{lem:s_complexity_main_k}.

The proof of the theorem is completed.
\end{proof}

In the remaining of this section, we prove the time and sample complexities of Procedure \textsc{HighProbFourierInterpolation} in Algorithm~\ref{alg:main_k}.

The following two lemmas show the time complexity of Procedure \textsc{MergeSignal} in Algorithm \ref{alg:main_k}, which is used to boost the success probability of Fourier interpolation algorithm.

\begin{lemma}[Time complexity of Procedure \textsc{WeightedSketch} in Algorithm \ref{algo:sig_est_1d_accuracy}]
 \label{lem:time_ws}
Procedure \textsc{WeightedSketch} in Algorithm \ref{algo:sig_est_1d_accuracy} runs in  \begin{align*}
    O(\eps^{-2}k\log(k)\log(1/\rho)) 
\end{align*} time.
\end{lemma}
\begin{proof}

Procedure \textsc{WeightedSketch} contains the following steps:
\begin{itemize}
    \item In line \ref{ln:s_0_smaple_D_ws}, sampling $ S_0$ takes $ O(\eps^{-2}k\log(k)\log(1/\rho)) $ times.
    \item In line \ref{ln:t_s_0_ws}, the for loop repeats $ |S_0|$ times, and each takes $ O(1)$ times. 
\end{itemize}

Following from the setting in the algorithm, we have that
\begin{align*}
   |S_0| = O(\eps^{-2}k\log(k)\log(1/\rho)). 
\end{align*}

So, the time complexity of Procedure \textsc{WeightedSketch} in Algorithm \ref{algo:sig_est_1d_accuracy} is 
\begin{align*}
O(\eps^{-2}k\log(k)\log(1/\rho)) + |S_0|\cdot O(1) 
=  O(\eps^{-2}k\log(k)\log(1/\rho)).
\end{align*}

\end{proof}

\begin{lemma}[Time complexity of Procedure \textsc{MergeSignal} in Algorithm \ref{alg:main_k}]
 \label{lem:time_ms}
 Procedure \textsc{MergeSignal} in Algorithm \ref{alg:main_k} runs in 
 \begin{align*}
      O(  k^5 \log^{6} (k) \log^4(1/\delta_1)\log^5(1/\rho) )
 \end{align*} time.
\end{lemma}
\begin{proof}

In each call of the Procedure \textsc{MergeSignal} in Algorithm \ref{alg:main_k},  the for loop (Line \ref{ln:mergeS_i_R_p}) repeats $ R_p$ times, each consisting of the following steps:
\begin{itemize}
    \item In Line \ref{ln:mergeS_j_R_p}, the for loop repeats $ R_p$ times and each iteration has the following steps:
    \begin{itemize}
        \item Line \ref{ln:merge_ws} calls Procedure \textsc{WeightedSketch}. By Lemma \ref{lem:time_ws}, it runs in $$O(Bd\log(Bd)\log(R_p^2/\rho)) $$ time.
        \item Line \ref{ln:merge_mpe} calls Procedure \textsc{MixedPolynomialEvaluation}. By Lemma \ref{lem:time_mpe}, it runs in \begin{align*}
            O\Big(\sum_{j=1}^k {\max}\{d', \mathrm{deg}(P_j)\} \log^3({\max}\{d', \mathrm{deg}(P_j)\})\Big)
        \end{align*} time, where $d' = O(Bd\log(Bd)\log(R_p^2/\rho))$ and $\mathrm{deg}(P_j) = d$.
    \end{itemize}
    \item Line \ref{ln:merge_med} computes the median in $R_p\log(R_p) $ time. 
\end{itemize}

Following from the parameter setting in the algorithm, we have that 
\begin{align*}
    B =&~ O(k),\\
    d =&~ O( \Delta T + k^3 \log k + k \log 1/\delta).
\end{align*}

By Lemma \ref{lem:property_of_filter_H} Property \RN{3}, we have that
\begin{align*}
    \Delta =  k \Delta_h = k |\supp(\wh{H}(f))| / T= O(k^3 \log^2 (k) \log^2(1/\delta_1)/T) .
\end{align*}

As a result, we have that 
\begin{align}
    B\cdot d = O(k^4 \log^2 (k) \log^2(1/\delta_1)) \label{eq:D_k_delta_002}
\end{align}

Moreover, we have that %
\begin{align}
 {\max}\{d', \mathrm{deg}(P_j)\} =&~   O(Bd\log(Bd)\log(R_p^2/\rho))\notag \\
 =&~ O(k^4 \log^3 (k) \log^2(1/\delta_1)\log( \log(1/\delta_1))\log(1/\rho)\log(\log(1/\rho)))\notag \\
 \leq &~ O(k^4 \log^3 (k) \log^3(1/\delta_1)\log^2(1/\rho)).\label{eq:max_d_deg_P}
\end{align}
where the first step follows from the definition of $d'$ and $ \mathrm{deg}(P_j)$, the second step follows from Eq.~\eqref{eq:D_k_delta_002}, the third step is straight forward.

So, the time complexity of Procedure \textsc{MergeSignal} in Algorithm \ref{alg:main_k} is 
\begin{align*}
  &~ R_p^2 \cdot \left(O(Bd\log(Bd)\log(R_p^2/\rho))  + O\Big(\sum_{j=1}^k {\max}\{d', \mathrm{deg}(P_j)\} \log^3({\max}\{d', \mathrm{deg}(P_j)\})\Big)\right)\\
  +&~ R_p\cdot O(R_p\log(R_p))\\
  \leq &~  O\Big( R^2_p  \sum_{j=1}^k {\max}\{d', \mathrm{deg}(P_j)\} \log^3({\max}\{d', \mathrm{deg}(P_j)\})\Big)\\
\leq &~  O( \log^2(1/\rho)\cdot k\cdot (k^4 \log^3 (k) \log^3(1/\delta_1)\log^2(1/\rho)) \log^3(k \log(1/\delta_1)\log(1/\rho)))\\
\leq &~  O(  k^5 \log^{6} (k) \log^4(1/\delta_1)\log^5(1/\rho) ),
\end{align*}
where the first step follows from Eq.~\eqref{eq:D_k_delta_002}, the second step follows from Eq.~\eqref{eq:max_d_deg_P}, the third step follows from Eq.~\eqref{eq:max_d_deg_P}, the forth step is straight forward.

\end{proof}

The following two lemmas show the time complexity and sample complexity of our main algorithm.

\begin{lemma}[Time complexity of the main algorithm]\label{lem:t_complexity_main_k}
 Procedure \textsc{HighProbFourierInterpolation} in Algorithm \ref{alg:main_k} runs in  \begin{align*}
     O(k^{4\omega} \log^{4\omega+2}(k/\delta)\log (FT)\log({\log(FT)})\log^5(1/\rho))
 \end{align*} times.
\end{lemma}
\begin{proof}

Procedure \textsc{HighProbFourierInterpolation} in Algorithm \ref{alg:main_k} consists of the following steps: 
\begin{itemize}
    \item In Line \ref{ln:hpfi_for_i_R}, the for loop repeats $ R_p$ times with the following step:
    \begin{itemize}
        \item Line \ref{ln:hpfi_cpfi} calls Procedure \textsc{ConstantProbFourierInterpolation}. By Lemma \ref{lem:time_cpfi}, it runs in \begin{align*}
            O(k^{4\omega} \log^{2\omega+1} (k) \log^{2\omega}(1/\delta_1)\log(\log(1/\delta_1))\log (FT)\log({\log(FT)}))
        \end{align*} time. 
    \end{itemize}
    \item Line \ref{ln:hpfi_ms} calls Procedure \textsc{MergeSignal}. By Lemma \ref{lem:time_ms}, it runs in \begin{align*}
        O(  k^5 \log^{6} (k) \log^4(1/\delta_1)\log^5(1/\rho) )
    \end{align*} time. 
\end{itemize}

Following from the setting in the algorithm, we have that 
\begin{align}
\delta_1 = \delta / \poly(k). \label{eq:delta_1}
\end{align}

So, the time complexity of Procedure \textsc{HighProbFourierInterpolation} in Algorithm \ref{alg:main_k} in Algorithm \ref{alg:main_k} is 
\begin{align*}
&~R_p \cdot O\left(k^{4\omega} \log^{2\omega+1} (k) \log^{2\omega}(1/\delta_1)\log(\log(1/\delta_1))\log (FT)\log({\log(FT)})\right) \\
& ~ + O(  k^5 \log^{c+3} (k) \log^4(1/\delta_1)\log^5(1/\rho) )\\
=&~ O(k^{4\omega} \log^{2\omega+1} (k) \log^{2\omega}(1/\delta_1)\log(\log(1/\delta_1))\log (FT)\log({\log(FT)})\log^5(1/\rho)) \\
\leq &~ O(k^{4\omega} \log^{4\omega+2}(k/\delta)\log (FT)\log({\log(FT)})\log^5(1/\rho)),
\end{align*}
where the first step follows from $ R_p = \log(1/\rho)$, the second step follows from Eq.~\eqref{eq:delta_1}.

\end{proof}

\begin{lemma}[Sample complexity of the main algorithm]\label{lem:s_complexity_main_k}
 Procedure \textsc{HighProbFourierInterpolation} in Algorithm \ref{alg:main_k} takes \begin{align*}
O(k^{4}  \log^{6}(k/\delta)\log (FT)\log({\log(FT)}) \log(1/\rho))
 \end{align*} samples.
\end{lemma}
\begin{proof}
Procedure Procedure \textsc{HighProbFourierInterpolation} in Algorithm \ref{alg:main_k} consists of the following steps: 
\begin{itemize}
    \item In Line \ref{ln:hpfi_for_i_R}, the for loop repeats $ R_p$ times:
    \begin{itemize}
        \item Line \ref{ln:hpfi_cpfi} calls Procedure \textsc{ConstantProbFourierInterpolation}. By Lemma \ref{lem:smaple_cpfi}, it takes \begin{align*}
            O(k^{4} \log^{3} (k) \log^{2}(1/\delta_1)\log(\log(1/\delta_1))\log (FT)\log({\log(FT)}))
        \end{align*} samples. 
    \end{itemize}
The remaining steps do not use any new sample.
\end{itemize}

Thus, the total sample complexity is 
\begin{align*}
&~ R_p \cdot O(k^{4} \log^{3} (k) \log^{2}(1/\delta_1)\log(\log(1/\delta_1))\log (FT)\log({\log(FT)}))\\
\leq &~ O(k^{4} \log^{3} (k) \log^{2}(1/\delta_1)\log(\log(1/\delta_1))\log (FT)\log({\log(FT)}) \log(1/\rho) )\\
\leq &~ O(k^{4}  \log^{6}(k/\delta)\log (FT)\log({\log(FT)}) \log(1/\rho) ),
\end{align*}
where the first step follows from $ R_p = \log(1/\rho)$, the second step follows from Eq.~\eqref{eq:delta_1}.

\end{proof}

\newpage
\section{Structure of Our Fourier Interpolation Algorithm}\label{sec:flowchart}
\tikzstyle{startstop} = [rectangle, rounded corners, minimum width=3cm, minimum height=1cm,text centered, draw=black, fill=red!30]
\tikzstyle{io} = [trapezium, trapezium left angle=70, trapezium right angle=110, minimum width=3cm, minimum height=1cm, text centered, draw=black, fill=blue!30]
\tikzstyle{process} = [rectangle, minimum width=3cm, minimum height=1cm, text centered, draw=black, fill=orange!30]
\tikzstyle{math} = [rectangle, minimum width=3cm, minimum height=1cm, text centered, draw=black, fill=blue!10]
\tikzstyle{math1} = [rectangle, rounded corners, minimum width=4cm, minimum height=2cm, text centered, draw=black, fill=yellow!30]
\tikzstyle{decision} = [diamond, minimum width=2cm, minimum height=1cm, text centered, draw=black, fill=green!30, aspect=2]
\tikzstyle{arrow} = [thick,->,>=stealth]

\begin{center}

\begin{tikzpicture}[node distance=2cm]
\node (main_thm) [startstop] {\begin{tabular}{c}Fourier interpolation algorithm \\  (Theorem~\ref{thm:main})\end{tabular}};

\node (FI_const_prob) [process, below of=main_thm, yshift=-0.5cm] {\begin{tabular}{c} Fourier interpolation with \\ constant success probability \\ (Theorem~\ref{thm:main_constant_prob})\end{tabular}}; 
\draw [arrow] (FI_const_prob) -- (main_thm);

\node (boost) [process, right of=main_thm, xshift=5cm]{\begin{tabular}{c}Boost success probability\\ (Lemma~\ref{eq:med_boost_prob})\end{tabular}};
\draw [arrow] (boost) -- (main_thm);

\node (sig_est) [process, right of=FI_const_prob, xshift=5cm]{\begin{tabular}{c}Signal estimation\\(Lemma~\ref{lem:magnitude_recovery_for_CKPS})\end{tabular}};
\draw [arrow] (sig_est) -- (FI_const_prob);

\node (freq_est) [process, below of=FI_const_prob, yshift=-1cm] {\begin{tabular}{c}Frequency estimation\\ algorithm   (Theorem~\ref{thm:frequency_recovery_k_better})\end{tabular}};
\draw [arrow] (freq_est) -- (FI_const_prob);

\node (hc_snr) [decision, right of=freq_est, xshift=6cm] {\begin{tabular}{c}Heavy-cluster \&\\High SNR band\end{tabular}};
\draw [arrow] (hc_snr) -- %
node[sloped, anchor=center, below] {\footnotesize{(Claim \ref{cla:guarantee_removing_x**_x*_}, Lemma~\ref{lem:xSfsubxSleqg})}}
node[sloped, anchor=center, above] {Sufficiency}
(FI_const_prob);
\draw [arrow] (hc_snr) -- node[anchor=south] {Assuming} (freq_est);

\node (sig_sample) [process, below of=freq_est, yshift=-0.25cm] {\begin{tabular}{c}Generate Significant Samples\\(Lemmas~\ref{lem:significant_samples_z}, \ref{lem:significant_samples_for_each_bins})\end{tabular}};
\draw [arrow] (sig_sample) -- node[anchor=west] {Lemma~\ref{lem:sample2freq_est}} (freq_est);

\node (noisy_filter_energy) [process, below of=sig_sample] {\begin{tabular}{c}Noisy filtered signal's\\ energy estimation (Lemma~\ref{lem:significant_samples_z_below})\end{tabular}};
\draw [arrow] (noisy_filter_energy) -- (sig_sample);

\node (noisy_test_energy) [process, right of=noisy_filter_energy, xshift=5cm] {\begin{tabular}{c}Noisy local-test signal's\\ energy estimation (Lemma~\ref{lem:significant_samples_z_above})\end{tabular}};
\draw [arrow] (noisy_test_energy) -- (sig_sample);

\node (filter_energy) [process, below of=noisy_filter_energy] {\begin{tabular}{c}Partial energy estimation\\ for filtered signals\\(Lemma~\ref{lem:z_energy_preserving_plus})\end{tabular}};
\draw [arrow] (filter_energy) -- (noisy_filter_energy);

\node (test_energy) [process, below of=noisy_test_energy] {\begin{tabular}{c}Partial energy estimation\\ for local-test signals\\(Lemma~\ref{lem:z_diff_energy_preserving})\end{tabular}};
\draw [arrow] (test_energy) -- (noisy_test_energy);

\node (partial_energy_filter) [math, below of=filter_energy, xshift=-2cm] {\begin{tabular}{c}Partial energy of filtered\\ signals (Lemma~\ref{lem:norm_z_cut_preserve})\end{tabular}};
\draw [arrow] (partial_energy_filter) -- (filter_energy);

\node (energy_est) [process, right of=partial_energy_filter, xshift=4cm] {\begin{tabular}{c}Sampling \& Reweighing\\ energy estimation  (Lemma~\ref{lem:link_weight_sampling_size__condition_num})\end{tabular}};
\draw [arrow] (energy_est) -- (filter_energy);
\draw [arrow] (energy_est) -- (test_energy);

\node (eb_test) [math, right of=energy_est, xshift=4.5cm]{\begin{tabular}{c}Energy bound for\\ local-test signals\\ (Lemma~\ref{lem:z_diff_condition_number_xI})\end{tabular}};
\draw [arrow] (eb_test) -- (test_energy);

\node (eb_filter) [math, below of=partial_energy_filter] {\begin{tabular}{c}Energy bound for\\ filtered signals (Corollary~\ref{cor:condition_number_z})\end{tabular}};
\draw [arrow] (eb_filter) -- (partial_energy_filter);

\node (sig_equiv) [math1, below of=energy_est, xshift=0cm, yshift=-2.5cm] {Signal Equivalent Method};
\draw [arrow] (sig_equiv) -- (eb_filter);
\draw [arrow] (sig_equiv) -- (eb_test);

\node (eb_fourier) [math, below of=eb_test, yshift=-2cm] {\begin{tabular}{c}Energy bound for\\ Fourier-sparse signals\\ (Theorems~\ref{thm:energy_bound}, \ref{thm:bound_k_sparse_FT_x_middle})\end{tabular}};
\draw [arrow] (eb_fourier) -- (eb_filter);
\draw [arrow] (eb_fourier) -- (eb_test);

\node (concentrate) [math, below of=eb_filter, yshift=-0.5cm] {\begin{tabular}{c}Concentration property\\ of filtered signals\\ (Lemmas~\ref{lem:full_proof_of_3_properties_true_for_z}, \ref{lem:z_satisfies_two_properties})\end{tabular}};
\draw [arrow] (concentrate) -- (eb_filter);
\draw [arrow] (sig_equiv) -- (concentrate);

\end{tikzpicture}
    
\end{center}

\ifdefined\isarxiv
\bibliographystyle{alpha}
\bibliography{ref}
\else
\bibliography{ref}
\bibliographystyle{alpha}

\fi

\end{document}